\documentclass[a4wide]{article}

\usepackage{graphicx}
\usepackage{amssymb} 
\usepackage{amsmath} 
\usepackage{amsthm}  
\usepackage{array} 
\usepackage[usenames,dvipsnames]{color} 
\usepackage{bbm} 
\usepackage{xfrac}
\usepackage{url} 
\usepackage{a4wide} 
\usepackage[citebordercolor={0.1 0.32 0.76}, linkbordercolor={0.1 0.7 0.3}]{hyperref} 
\usepackage{cite}
\usepackage[toc,page]{appendix}
\usepackage{multicol}

\theoremstyle{plain}
\newtheorem{theorem}{Theorem}[section]

\theoremstyle{definition}
\newtheorem{definition}[theorem]{Definition}
\newtheorem{example}[theorem]{Example}
\newtheorem{lemma}[theorem]{Lemma}
\newtheorem{thm}[theorem]{Theorem}

\newtheorem{prop}[theorem]{Proposition}

\newtheorem{postulate}{\textcolor{orange}{Postulate}}
    
\newtheorem{postulateapp}{\textcolor{orange}{Postulate}}
    
\newtheorem{postulateqsfirst}{\textcolor{orange}{Postulate}}
    
\newtheorem{postulateqsentropy}{\textcolor{orange}{Postulate}}

\numberwithin{equation}{section}
\numberwithin{figure}{section}

\usepackage{tikz}
\usetikzlibrary{arrows,calc}
\usepackage{relsize}
\usetikzlibrary{decorations.pathmorphing} 

\def\@xfootnote[#1]{%
  \protected@xdef\@thefnmark{#1}%
  \@footnotemark\@footnotetext}

\title{ 
Tangible phenomenological thermodynamics
}
\author{\normalsize Philipp Kammerlander and Renato Renner \\
\normalsize{Institute for Theoretical Physics, ETH Z\"urich, Wolfgang-Pauli-Str. 27, 8093 Z\"urich, Switzerland}}
\date{\normalsize \today}

\begin{document}

\maketitle

\begin{abstract}\noindent
In this paper, the foundations of classical pheno\-meno\-logical thermo\-dynamics are being thoroughly re\-visited.
A new rigorous basis for thermodynamics is laid out in the main text and presented in full detail in the appendix. 
All relevant concepts, such as work, heat, internal energy, heat reservoirs, reversibility, absolute temperature and entropy, are introduced on an abstract level 
and connected through traditional results, such as Carnot's Theorem, Clausius' Theorem and the Entropy Theorem.
The paper offers insights into the basic assumptions one has to make in order to formally introduce a phenomenological thermodynamic theory.
This contribution is of particular importance 
when applying phenomenological thermodynamics to systems, such as black holes, where the microscopic physics is not yet fully understood.
%
%
%
Altogether, this work can serve as a basis for a complete and rigorous introduction to thermodynamics in an undergraduate course which follows the traditional lines as closely as possible. 
\end{abstract}

\newpage

\begin{multicols}{2}
\tableofcontents
\end{multicols}

\newpage


\part*{Introduction}
\addcontentsline{toc}{part}{Introduction}
\label{sec:intro}

The first steps towards a theoretical formulation of thermo\-dynamics have been made more than 200 years ago \cite{Clausius50, Rankine50, Kelvin51, Clausius54, Maxwell71} and today we are still applying the theory to 
systems which are studied in the context of extending our currently best microscopic theories.
An example are black holes, with which we hope to find out how to come up with a theory of quantum gravity \cite{Bardeen73}.
In the past centuries the applications of thermodynamics have moved from steam engines and locomotives to systems of almost all orders of magnitude in size and levels of complexity, from single photons to solar systems and galaxy clusters.
And even with a focus on the traditional foundations, for the moment ignoring the most current investigations, it is worth having a closer look at the basis of a theory
which is so universally applicable.

\subsection*{Motivation}
\addcontentsline{toc}{subsection}{Motivation}

Among the many books that have been written and lectures held about the subject \cite{Carnot24, Planck97, Caratheodory09, Fowler39, Planck14, Fermi56, Feynman63, Giles64, Buchdahl66, Zemansky68, Pauli73, Adkins83, Callen85, Huang87, Neumaier07, Thess11, Hulse18, Graf05}, the two common ways to introduce thermodynamics are the \emph{phenomenological} and the \emph{statistical} approaches.
In the former, the laws of thermodynamics are postulated and the theory is built on them.
In the latter, the macroscopic thermodynamic properties of a system are derived by investigating its microscopic degrees of freedom using methods from statistical physics, e.g.\ \cite{Salem06}.
In particular, when taking the statistical approach, one of the main goals is to derive the fundamental laws from microscopic considerations. 
The work presented here is solely concerned with the phenomenological paradigm. Microscopic properties of a system may be intuitively used at times, but only to give intuition in certain examples, if at all.\\

Even when restricted to the phenomenological approach there exist many different views on the foundations of the theory and how it should be introduced. 
This may be unimportant for applications e.g.\ in engineering, but the different and seemingly contradicting alternatives are confusing and unsatisfactory from a conceptual point of view.
Modern literature on systematic approaches to thermodynamics is scarce. An exception is the approach presented in \cite{LY99,LY02,Thess11}.
However, in numerous conversations with fellow physicists the authors realized that their uneasy feeling about the foundations of phenomenological thermodynamics, and especially about how it is taught in undergraduate courses, is shared by many others.
Hence this paper.\\

Besides the educational value, the increasing interest in the new research field of quantum thermodynamics is another good reason for looking more deeply into the basics of classical phenomenological thermodynamics. 
Recently various ideas have been proposed on how to use thermodynamics for a description of small (quantum) systems. In particular, a considerable amount of effort has gone into the quest for finding thermodynamic applications where quantum systems can outperform their classical counterparts.
Nevertheless, the community is still far from reaching an agreement on the conclusions drawn from a quantum analysis of thermodynamic systems. 
An example is the ongoing controversy about the definition of work in the presence of quantum coherent states \cite{Talkner07, Aberg14, Perarnau17}.
In the authors' view this is also because a fair comparison between classical and quantum thermodynamics is bound by the limited agreement on the basics of the classical theory itself.\\

In this paper we will introduce a new way of laying out the foundations of thermodynamics allowing for a derivation of the theory along the general lines of the traditional approach.
Starting from the very basic concepts such as systems, processes and states, we develop a rigorous language to formulate the first and second laws of thermodynamics, which can be seen as the core postulates of the theory. 
These are then used to prove powerful standard results like Carnot's Theorem or Clausius' Theorem, which in turn allow us to define thermodynamic entropy with all its very useful properties.\\

Besides the achievement of making all basic thermodynamic concepts precise without losing contact to their intuitive meanings, this paper offers insights into the basic assumptions one has to make in order to formally introduce a phenomenological thermodynamic theory. 
For instance, we show that the zeroth law is redundant as a postulate, and we capture precisely what properties a heat reservoir must have. 
We also introduce a definition of quasistatic processes arising from initially discrete considerations which allows one to use the usual formulas for computing internal energy and entropy differences over piecewise continuously differentiable paths in the state space of a thermodynamic system.

\subsection*{Outline of the paper}
\addcontentsline{toc}{subsection}{Outline of the paper}

In Section~\ref{sec:systemsprocessesstates} the notions of thermodynamic \emph{systems}, \emph{processes} and \emph{states} are introduced. 
Section~\ref{sec:work} introduces the \emph{work cost} of thermodynamic processes and related notions, whereas Section~\ref{sec:firstlaw} uses all terms introduced up to this point to state the \emph{first law} and derive the \emph{internal energy} function from it. 
In Section~\ref{sec:equivalentsys} we intuitively describe the notion of what it means for two systems to be \emph{copies} of one another, which is made formally precise in the appendix. 
Section~\ref{sec:heat} treats the notions \emph{heat} and \emph{heat reservoirs} which are then used in Section~\ref{sec:secondlaw} to state the \emph{second law} and briefly discuss its immediate consequences.
The more complex but also more important consequences of the second law, namely \emph{Carnot's Theorem} and \emph{absolute temperature}, are discussed in Section~\ref{sec:carnot} and Section~\ref{sec:abstemp}, respectively.
Using the definition of absolute temperature for heat reservoirs it is then possible to define the \emph{absolute temperature of a heat flow} in Section~\ref{sec:theatflow}, which finally allows us to state and prove \emph{Clausius' Theorem} and introduce \emph{thermodynamic entropy} in Section~\ref{sec:clausius}. 

Having introduced the very basics of every thermodynamic theory, the paper then offers several insights into more common applications and advances the abstract framework in order to apply it to traditional settings. 
In Section~\ref{sec:quasistatic} the discrete considerations are extended to account for quasistatic processes with associated continuous and differentiable quantities. 
This is followed by the discussion of the notion of temperature for other systems than heat reservoirs in Section~\ref{sec:tsys}. 
The thoroughly worked out example of the ideal gas is presented in Section~\ref{sec:exideal}, 
while further explorations into topics such as scaling thermodynamic systems and the principle of maximum entropy is presented in Section~\ref{sec:scaling}.


\newpage
\part*{Thermodynamics}
\addcontentsline{toc}{part}{Thermodynamics}

\section{Systems, processes and states}
\label{sec:systemsprocessesstates}

\begin{center}
\fcolorbox{OliveGreen}{white}{\begin{minipage}{0.98\textwidth}
\centering
\begin{minipage}{0.95\textwidth}
\ \\
\textcolor{black}{\textbf{Postulates:}}
thermodynamic systems, thermodynamic processes and states, concatenation of processes\\
\ \\
\textcolor{black}{\textbf{New notions:}}
 atomic systems, involved, composition of systems, 
state space\\
\ \\
\textcolor{black}{\textbf{Technical results}} for this section can be found in Appendices~\ref{app:systems} and \ref{app:processesstates}.

\paragraph{\textcolor{black}{Summary:}}
Through postulates the notions {systems}, {composition of systems}, {thermodynamic processes} and {their input and output states} are introduced. A further postulate makes sure that two thermodynamic processes 
can be {concatenated}.
From these postulates further notions, among them {disjoint systems}, {subsystems} and {state space} are defined. 
The states of composite systems are defined by the states of their subsystems. 
Basic properties of the composition operation are discussed as well as the correspondence of systems and composition with finite set theory and set union. 
\vspace{.1cm}

\end{minipage}
\end{minipage}}
\end{center}

\ \\

We start by introducing all basic principles and postulates from which we will derive the consequences for the theory step by step. At times it is going to be abstract. Table \ref{tab:example} together with other examples discussed in the text compensates for the technical nature of this section by listing the main concepts together with different examples for an easier intuitive understanding. This table is recycled from our previous paper on the zeroth law \cite{Kammerlander18}.

Postulates are in orange. Definitions, lemmas, propositions, theorems and other derived concepts appear with black titles. \\

The section on the basic principles is about the notions of thermodynamic systems, thermodynamic processes and thermodynamic states. Admittedly, we could also call them systems, processes and states at this point, without the specification ``thermodynamic'', as there is nothing thermodynamic about them yet. The thermodynamic aspects will come in later.
At this point, we are concerned with the very basic notions needed to formulate thermodynamic assumptions, such as the first and second laws.


\subsection{Thermodynamic systems}

Sometimes physical systems are said to be a ``subset of the universe''. For the purposes presented here this view, taken literally, is valid. 
We think of systems as finite non-empty subsets of the ``world''. Composing two systems corresponds to uniting the sets.
Singletons are seen as indivisible systems.

\begin{postulate}[Thermodynamic systems]
\label{post:systems}
The \emph{thermodynamic world} is a 
set $\Omega$ 
and the \emph{set of thermodynamic systems} consists of finite non-empty subsets of the thermodynamic world, $\mathcal{S}:= \{ S\subset\Omega \,|\, 0<|S|<\infty \}$.
\end{postulate}
%

The structure of systems as sets is illustrated in Figure~\ref{fig:systems}. The natural way to \emph{compose} two systems is to unite the sets. The system composed from $S_1\in\mathcal{S}$ and $S_2\in\mathcal{S}$ is denoted by $S=S_1\vee S_2$. 
For the \emph{composition of systems} we introduce the notation $\vee$ instead of $\cup$, which is normally used to denote the union of sets. 
However, $\vee$ is nothing else but the union $\cup$ applied to finitely many representing sets of the systems.

\newpage

\vspace*{-2cm}
\begin{table}[h!]
\hskip-2cm
\begin{tabular}{p{2.3cm} | p{5cm} p{5.5cm} p{5cm}}
& $1$ mole of ideal gas $A$ 
& two ideal gases $S=A\vee A'$, $A\hat=A'$
& water tank, $N\gg1$ moles, $\tilde R$ \\
\hline
intuitive \newline description 
& A container filled with an ideal gas. It is possible to read off the pressure inside the container and one can vary the volume by moving a piston.
& Two such ideal gases composed. They can still be addressed individually but now they could also be thermally connected. The system is described by the composition $A\vee A'$.
& A very large water tank may be microscopically complex but is here described in the simplest non-trivial way, by its energy. It is an approximate reservoir. \\ 
 & & & \\
 state space  \newline $\Sigma_S$
& $\Sigma_{A} = (\mathbbm{R}_{>0})^2 \ni \sigma = (p,V)  $ \newline pressure and volume\protect\footnotemark 
& $\Sigma_{S} = (\mathbbm{R}_{>0})^4 \ni \sigma = (p_1,V_1,p_2,V_2)$ \newline both pressures and volumes 
& $\Sigma_{\tilde R} = [E_{\rm min},E_{\rm max}] \ni \sigma = E$ \newline (internal) energy\protect\footnotemark \\
 & & & \\
processes \newline (examples \newline from $\mathcal{P}$) 
& $\bullet$ connecting the gas thermally \newline
\phantom{$\bullet$} to another system (irreversible \newline
\phantom{$\bullet$} in general) \newline 
\phantom{$\bullet$} $(p_{\rm in}, V_{\rm in}) \mapsto (p_{\rm out}, V_{\rm out}=V_{\rm in})$ \newline 
$\bullet$ isothermal expansion or com- \newline
\phantom{$\bullet$} pression (generally reversible) \newline 
\phantom{$\bullet$} $(p_{\rm in}, V_{\rm in}) \mapsto (p_{\rm out}, V_{\rm out})$ s.t.\newline \phantom{$\bullet$} \,$p_{\rm in}V_{\rm in} = p_{\rm out}V_{\rm out}$  \newline
$\bullet$ Carnot cycle with the gas 
& $\bullet$ all processes from the left column \newline
\phantom{$\bullet$} applied to one of the gases individ- \newline
\phantom{$\bullet$} ually \newline 
$\bullet$ thermally connecting the first gas \newline
\phantom{$\bullet$} to the second and the second to a \newline
\phantom{$\bullet$} reservoir (typically irreversible) 
& 
$\bullet$ connecting the tank thermally \newline 
\phantom{$\bullet$} to another system, thereby let- \newline 
\phantom{$\bullet$} ting them exchange energy (ir- \newline 
\phantom{$\bullet$} reversible in general) \newline
$\bullet$ using the tank in a thermal en- \newline 
\phantom{$\bullet$} gine together with a cyclic ma- \newline
\phantom{$\bullet$} chine and another tank \\
 & & & \\
work processes \newline (examples \newline from $\mathcal{P}_S$) 
& $\bullet$ shaking the gas, or applying \newline
\phantom{$\bullet$} friction to it (irreversible) \newline 
\phantom{$\bullet$} $(p_{\rm in}, V_{\rm in}) \mapsto (p_{\rm out}, V_{\rm out}=V_{\rm in})$ \newline 
\phantom{$\bullet$} with $p_{\rm out} > p_{\rm in}$ \newline 
$\bullet$ expanding or compressing \newline
\phantom{$\bullet$} the otherwise perfectly iso- \newline
\phantom{$\bullet$} lated gas\protect\footnotemark (reversible, if done \newline
\phantom{$\bullet$} optimally) \newline 
\phantom{$\bullet$} $(p_{\rm in}, V_{\rm in}) \mapsto (p_{\rm out}, V_{\rm out})$ s.t. \newline \phantom{$\bullet$} $p_{\rm in}V_{\rm in}^\gamma = p_{\rm out}V_{\rm out}^\gamma$ 
& $\bullet$ all processes from the left column \newline
$\bullet$ thermally connecting the two \newline
\phantom{$\bullet$} gases (irreversible in general) \newline
$\bullet$ isolating the gases and compress- \newline
\phantom{$\bullet$} ing one while expanding the other \newline
\phantom{$\bullet$} (reversible) \newline
$\bullet$ applying friction to one of the \newline
\phantom{$\bullet$} gases while thermally connecting \newline
\phantom{$\bullet$} them, thereby heating up the \newline
\phantom{$\bullet$} other one, too (irreversible) 
& $\bullet$ raising the internal energy of \newline
\phantom{$\bullet$} $\tilde R$  by doing work on it: this \newline 
\phantom{$\bullet$} could be done for instance by \newline
\phantom{$\bullet$} shaking it, applying friction, \newline
\phantom{$\bullet$} doing electrical work by let- \newline 
\phantom{$\bullet$} ting a current flow through a \newline 
\phantom{$\bullet$} resistance\\
 & & & \\
work function & For classical mechanical systems the infinitesimal work done is always $\delta W = \vec{F}\cdot{\rm d}\vec{s}$. When shaking the gas or applying friction $\vec{F}$ is the applied force and d$\vec{s}$ the line segment of the trajectory along which the force acts. If the volume is changed without friction the infinitesimal work done can be simplified to $\delta W = -p{\rm d}V$.
& The total work done on the composite system is the sum of the amounts of work done on each individual system, $\delta W = \delta W_A + \delta W_{A'}$, where $\delta W_X$ is defined as in the left column for both gases individually.
& Depending on the  mechanism that is used to heat up the tank the work function will look different. For instance, in the case of increasing the internal energy by means of an electrical current $I(t)$ flowing through a resistance $Z$ the work done on $R$ can be written as $W = \int ZI^2 {\rm d}t$.

\end{tabular}
\caption{The basic concepts such as states, state spaces, processes, work processes, composition of systems, work cost functions and (an approximation of) reservoirs are illustrated with simple examples, recycled from a previous paper \cite{Kammerlander18}.}
\label{tab:example}
\end{table}

\vspace{-.6cm}

\addtocounter{footnote}{-2}
\footnotetext{We think of a truly ideal gas that remains an ideal gas even if $V$ becomes very small and $p$ very high. In the case of a real gas one could restrict the state space further, e.g.\ to $\Sigma_S = (0,p_{\rm max})\times(V_{\rm min},\infty)$.}


\addtocounter{footnote}{1}
\footnotetext{If one wants to use the tank as a reservoir, think of the energy range $[E_{\rm min},E_{\rm max}]$ as a relatively small interval in a much larger spectrum of the system. As long as the energy is in this interval, the tank can exchange energy with other systems while keeping its properties and behaviour relative to other systems invariant (to good approximation). Intuitively one can always achieve this by taking a large enough system. 
When taking the limit towards an infinitely large water tank, $N\rightarrow\infty$ such that \smash{$\tfrac{N}{V} = \rm const.$} ($V$ the volume) $R'$ becomes an exact reservoir. 
For instance, when choosing \smash{$E_{\rm min / max} \propto \pm\sqrt{N}$},
in this limit the state space $\Sigma_{R'}$ will be the real line, i.e.\ $E_{\rm min / max} \rightarrow \mp\infty$ but still occupy only a very small part of the actual spectrum of the tank.
Reservoirs should be thought of as infinitely large systems that do not change their behaviour when exchanging finite amounts of energy with other systems.}

\addtocounter{footnote}{1}
\footnotetext{A compression or expansion of a perfectly isolated ideal gas leads to the known state changes as given in the table, where $\gamma=\tfrac{c_p}{c_V}$ is the heat capacity ratio.}
\newpage

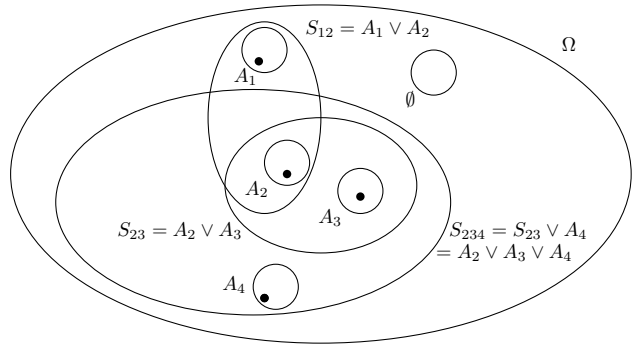
\begin{figure}
\hskip-1cm
\begin{center}
	\begin{tikzpicture}[scale=0.75, every node/.style={transform shape}]	
	
	\draw[] (0,0) ellipse (5.5cm and 3cm);
	\node at (4.4,2.3) {$\Omega$};
	
	\draw[] (-1,2.2) ellipse (.4cm and .4cm);
	\node at (-1.1,2) [circle,fill,inner sep=1.5pt]{};
	\node at (-1,2) [below left]{$A_1$};
	
	\draw[] (2,1.8) ellipse (.4cm and .4cm);
	\node at (1.8,1.6) [below left]{$\emptyset$};
	
	\draw[] (0,-.2) ellipse (1.7cm and 1.2cm);
	\node at (-2.5,-1) {$S_{23} = A_2\vee A_3$};

	\draw[] (-.6,.2) ellipse (.4cm and .4cm);
	\node at (-.6,0) [circle,fill,inner sep=1.5pt]{};
	\node at (-.8,0) [below left]{$A_2$};

	\draw[] (.7,-.3) ellipse (.4cm and .4cm);
	\node at (0.7,-.4) [circle,fill,inner sep=1.5pt]{};
	\node at (0.5,-.5) [below left]{$A_3$};

	\draw[] (-1.2,-.5) ellipse (3.5cm and 2cm);
	\node at (3.5,-1) {$S_{234}=S_{23}\vee A_4$};
	\node at (3.25,-1.4) {$=A_2\vee A_3\vee A_4$};
	
	\draw[] (-.8,-2) ellipse (.4cm and .4cm);
	\node at (-1,-2.2) [circle,fill,inner sep=1.5pt]{};
	\node at (-1.2,-1.7) [below left]{$A_4$};

	\draw[] (-1,1) ellipse (1cm and 1.7cm);
	\node at (-1,-2.2) [circle,fill,inner sep=1.5pt]{};
	\node at (-.4,2.3) [above right]{$S_{12}= A_1 \vee A_2$};

	\end{tikzpicture}
\end{center}
\caption{Abstract illustration of the structure of thermodynamic systems:
Systems are finite non-empty subset of the thermodynamic world $\Omega$. 
The empty set is not a thermodynamic system.
Singletons are called atomic systems and are usually called $A_i$ to indicate this in the notation. 
Just like the elements of a set determine the set, the atomic systems contained in a system (as subsets) determine the system. 
Multi-element subsets of $\Omega$ are composed systems.
The composition operation for systems is nothing else but the union of the sets which the systems represent. 
}
\label{fig:systems} 
\end{figure}

Since systems are finite sets, the composition operation applied to finitely many systems will output a system again.
To simplify the notation for the composition of more than two systems we sometimes write
\begin{align}
S = S_1\vee \cdots \vee S_n
\equiv \bigvee_{S\in\{S_1,\dots,S_n\}} S 
\equiv \bigvee \{ S_1,\dots,S_n \} \,.
\end{align}

The set of singletons of $\Omega$ denoted by $\mathcal{A} := \{ A\subset\Omega \,|\, |A|=1 \}$ is called the \emph{set of atomic systems}. 
This set is a subset of the set of thermodynamic systems. Non-atomic systems are those that are composed of more than one but finitely many different atomic systems, whereas atomic systems are considered indivisible, i.e.\ they are not composed of other systems. 

However, one should not take this as a statement about physical indivisibility. Rather it is an abstract structure the user of the theory needs to choose.
As an example, one user could define a container of gas as an atomic system. Another user might say that there are different compartments of gas within this container, which could have different pressures and could be separated into further subsystems. In the latter case, the system would not be atomic. 
The two users would then work with different thermodynamic theories.

Another aspect shown by this example is the fact that an atomic system need not be small. The gas container could contain an arbitrary amount of substance and still be seen as an atomic system.\\

The concept of an atomic system is helpful for the definition and postulates regarding basic concepts such as work or states. When defining and postulating these for atomic systems and in addition specifying how they behave under composition, these concepts automatically extend to arbitrary thermodynamic systems.\\


Just like we can compose two systems with $\vee$ (i.e.\ by set union) we can as well consider the \emph{intersection} of two systems. This corresponds to taking the intersection of the sets. referring to Figure~\ref{fig:systems}, we see for instance that 
$S_{12}\wedge S_{23} = A_2$.

Postulate~\ref{post:systems} requires that the empty set is not a thermodynamic system. To some extent, this is a matter of taste. We do not want ``nothing'' to be called a system, as it would be difficult to argue that this ``nothing'' has states on which processes can act. Without states and processes acting on them, one can not do anything physical with it, hence we can exclude it from the set of system from the beginning.

Nevertheless, we write $S_1\wedge S_2 = \emptyset$ to say that the two systems $S_1,S_2\in\mathcal{S}$ are \emph{disjoint}. This equality should be seen as notation only, as $\emptyset$ is not a system.
As an example, in Figure~\ref{fig:systems} we see that $A_1\wedge S_{23}=\emptyset$. \\

Given a system $S\in\mathcal{S}$, its \emph{subsystems} are defined as its non-empty subsets, i.e.\
the set of subsystems $\mathrm{Sub}(S) := \{ S'\in\mathcal{S} \ | \ S'\subset S \}$.

The set of subsystems $\mathrm{Sub}(S)$ contains $S$ itself since $S\vee S=S$. In particular, it does not only contain the proper subsystems and is therefore always non-empty. 
The subsystems of a system $S\in\mathcal{S}$, which are in addition atomic, are called \emph{atomic subsystems} and are summarized in the set $\mathrm{Atom}(S) := \{ A\in\mathcal{A} \,|\, A\in\mathrm{Sub}(S) \}$.
Examples from Figure~\ref{fig:systems} are $S_{23}\in\mathrm{Sub}(S_{234})$ and $A_1\in\mathrm{Atom}(S_{12})$.
Just like elements of a set determine the set, the atomic subsystems of a system determine the system.
In a more formal notation that is, for any $S\in\mathcal{S}$ it holds
\begin{align}
S = \bigvee \mathrm{Atom}(S)\,.
\end{align}

We conclude the abstract introduction of systems with some comments. 
First, a thermodynamic system is interpreted as a specific physical system rather than a type of system. For instance, $A_1$ may stand for the specific box containing one mole of gas standing in front of us, whereas $A_2$ may be the label for a box filled with half a mole of another gas standing in a friend's lab.\footnote{
In a previous paper \cite{Kammerlander18} the introduction and interpretation of systems was different. Instead thinking of systems as \emph{types}, we now think of them as specific instances.}
Nevertheless, one needs to be able to talk about systems of the same type, for instance when the two boxes in the different labs are identical in construction and contain the same amount of the same gas. The notion of two different systems being of the same type will be introduced later in Section~\ref{sec:equivalentsys} based on their thermodynamic properties. When two systems of the same type are discussed, their labels $A_1$ and $A_2$ can be seen as their names, while their type is their thermodynamic DNA.\footnote{
Identical twins need names, too. It is not possible to distinguish them just by their DNA.}

Second, we think of a single thermodynamic system $S$ as closed. Nevertheless, it is not necessary to formalize the term ``closed system'' as the coming postulates make sure that the properties of a thermodynamic system are such that the intuition one has about closed systems is captured also in this framework. 

Third, it is important to keep in mind that composing two systems does not mean that they are thermally or otherwise coupled. So far composition is independent of any thermodynamic properties. Two systems composed are simply described as one new system.
When disjoint systems are composed, we speak of \emph{disjoint composition}. 
In this case the composition is akin to the tensor product structure from quantum theory, where $\mathcal{H}_1\otimes\mathcal{H}_2$ allows for a composite description of the two systems with associated Hilbert spaces 
$\mathcal{H}_1$ and $\mathcal{H}_2$ but does not make any statements about possible interactions between the subsystems.
The composition of disjoint systems is the one that is more familiar with the usual ways of composing two systems in physics.
Many concepts concerning composition that will be introduced are about disjoint composition.

%

\subsection{Thermodynamic processes and states}

After having introduced the structure of systems, atomic systems and subsystems, as well as their interrelations, we move on to the postulate about thermodynamic processes and states.

\begin{postulate}[Thermodynamic processes and states]
\label{post:processesstates}
The non-empty \emph{set of thermodynamic processes} (also simply \emph{processes}) that the theory allows for is denoted by $\mathcal{P}$.
A thermodynamic process $p\in\mathcal{P}$ specifies the \emph{initial and final states} of a finite and non-zero number of
\emph{involved} atomic systems $A\in\mathcal{A}$ by means of 
the functions 
$\lfloor\cdot\rfloor_A: \mathcal{P}\rightarrow\Sigma_A$ 
and $\lceil\cdot\rceil_A: \mathcal{P}\rightarrow\Sigma_A$.  
If $A$ is \emph{not involved} the two functions are undefined.
\end{postulate}

The notion of \emph{being involved} is sometimes also phrased the other way around. That is, instead of saying that $A\in\mathcal{A}$ is involved in $p\in\mathcal{P}$, we may say that $p$ \emph{acts on} $A$.\\

Thermodynamic processes are thought of as procedures
but do not necessarily correspond to curves in the state space. 
However, for atomic systems that are involved in a process $p\in\mathcal{P}$, the process has a well-defined initial and final state captured by the functions $\lfloor \cdot \rfloor_A$ and
$\lceil \cdot \rceil_A$, respectively.
Since these functions are always either both defined or both undefined, we will only talk about the input state whenever we discuss the question whether a certain atomic system is involved in the process or not.
For an atomic system $A\in\mathcal{A}$ the co-domain of the functions $\lfloor \cdot \rfloor_A$ and $\lceil \cdot \rceil_A$ is the union of their well-defined outputs and denoted by $\Sigma_A$.
It is called the \emph{set of states of $A$}. \\

We want the set of states of different atomic systems to be disjoint.
W.l.o.g.\ if $A_1,A_2\in\mathcal{A}$ are different, $A_1\neq A_2$, then $\Sigma_{A_1}\cap\Sigma_{A_2}=\emptyset$. If this was not the case, simply extend the descriptions of the states with a unique label for each atomic system, e.g. its name. 
This assumption will allow us to easily extend the definition of state changes for states of composite systems. It also simplifies the analysis of state spaces of equivalent systems in Section~\ref{sec:equivalentsys}. \\

We do not explicitly require the states to be what is usually called an ``equilibrium state'' \cite{Callen85,Huang87,Thess11}.
The way we introduce the theory, thermodynamics does not rely on a basic notion of ``equilibrium''. 
Without the necessity to define this term, we are able to treat states as thermodynamic states which could intuitively be seen as non-equilibrium states. For instance, one could think of a system that has fluctuations on the time scale of $\sim 1\mathrm{ms}$, but thermodynamic processes on this system happen on a time scale of $\sim 1\mathrm{h}$. Intuitively this fluctuating state might not be counted as equilibrium, but in the context of processes that take orders of magnitudes more time, it is reasonable to work with average values, which are then constant over the longer time scale. 

Even if the fluctuations are visible ripples on the water surface of a swimming pool (i.e.\ they can be detected) the thermodynamic processes may ignore this and the state of the pool is a valid thermodynamic state.

Although the term ``equilibrium'' is not defined as a physical concept, thermodynamic states cannot be anything. For instance, a necessary property is that they must allow for identity processes (Definition~\ref{def:id}.
This will follow from the first law (Postulate~\ref{post:first}) and it implies that it is always possible to leave a thermodynamic state invariant.
In our view, it is an asset of the framework that it does not have to formally define the term ``equilibrium'' and is nevertheless able to capture all relevant properties of ``equilibrium states''.  \\

We are not assuming that a set $\Sigma_A$ contains a continuum of states. This may very well be the case for many examples, but it is not a necessary assumption for the theory.
The concept of a state is rather abstract up to here. It will become more comprehensible with further postulates and results derived from them. \\

Since processes are seen as procedures there may very well exist more than one process with the same input and output states. The initial and final states of a thermodynamic process do not contain all thermodynamically relevant information associated with it. It is helpful to keep in mind this view of processes in the remainder of the paper.\\



For an arbitrary thermodynamic system $S=A_1\vee\cdots\vee A_n\in\mathcal{S}$ with pairwise disjoint atomic systems $\{A_i\}_{i=1}^n$ the state change under a thermodynamic process $p\in\mathcal{P}$ is given by the state changes on its atomic subsystems in terms of the functions
\begin{align}
\label{eq:statecomp}
\begin{split}
\lfloor \cdot \rfloor_S :\,\, &\mathcal{P} \rightarrow \Sigma_S\\
&p \, \mapsto \lfloor p \rfloor_S := \{ \lfloor p \rfloor_{A_1},\dots,\lfloor p \rfloor_{A_n} \}
\end{split}
\end{align}
and likewise for $\lceil \cdot \rceil_S$. States of a non-trivially composite system are called \emph{joint states}.


As a consequence of this definition and the fact that different atomic systems have disjoint state spaces, the state space of a disjointly composite system can be seen as the Cartesian product of its subsystems, up to ordering.
Furthermore, the initial state on $\lfloor p \rfloor_S$ is defined if and only if $\lfloor p \rfloor_A$ is defined for all atomic subsystems $A\in\mathrm{Atom}(S)$.
This, on the other hand, is the case if and only if the final states $\lceil p \rceil_A$ are defined for all atomic subsystems $A\in\mathrm{Atom}(S)$ and thus $\lceil p \rceil_S$ is defined if and only if $\lfloor p \rfloor_S$ is.\\

To simplify the notation we use the same symbol for joint states as we used for composite systems.
That is, we write 
$\lfloor p \rfloor_S \equiv \bigvee_{A\in\mathrm{Atom}(S)} \lfloor p \rfloor_A \equiv 
\lfloor p \rfloor_{A_1}\vee\cdots\vee\lfloor p \rfloor_{A_n}$ 
for $S=A_1\vee\cdots\vee A_n$. 
For two disjoint systems $S_1\wedge S_2=\emptyset$ with $\sigma_1\in\Sigma_{S_1}$ and $\sigma_2\in\Sigma_{S_2}$ this reads 
$\sigma_1\vee\sigma_2 = \sigma_2\vee\sigma_1\in\Sigma_{S_1\vee S_2}$.
The fact that $\vee$ is commutative when composing states is no issue as state spaces of different systems are disjoint, i.e.\ it is always clear which symbol ($\sigma_1$ or $\sigma_2$) describes the state of which subsystem ($S_1$ or $S_2$).\\

Notice that the terms \emph{involved} and \emph{not involved} cannot be directly extended to arbitrary systems. These consist of atomic systems, but it could happen that an atomic subsystem is involved in a process, while another one is not. An arbitrary system containing the two would then intuitively be involved. However, its state change is undefined, as there is at least one atomic subsystem whose state change is undefined.\\


The view of composite systems as independent but collectively described subsystems is manifest in the representation of the state space of composite systems. It allows for all combinations of subsystem states, not just for pairs which are in ``equilibrium'' with each other -- a term that is anyway not defined at this point. 
Consequently, states of subsystems in composite systems can still be separately treated. 
The subsystems' states 
$\sigma_1\in\Sigma_{S_1}$ and $\sigma_2\in\Sigma_{S_2}$ can be extracted as the corresponding entries of the ``tuple'' $\sigma_1\vee \sigma_2 \in\Sigma_{S_1\vee S_2}$.
In particular, this implies that if $\lfloor p \rfloor_S$ 
is defined for a (composite) system $S$, then $\lfloor p \rfloor_{S'}$ 
is automatically well-defined for all subsystems $S'\in\mathrm{Sub}(S)$ as well. 
Likewise, if all subsystems of $S$ have a well-defined state change in some thermodynamic process, then so does $S$. 

On the other hand, the structure of state spaces also makes clear that if a proper subsystem $S'\in\mathrm{Sub}(S)$ has well-defined state changes in a thermodynamic process $p\in\mathcal{P}$ this does not imply that the same holds for $S$. 
There might exist other subsystems of $S$ with undefined state change.\\

We do not aim at describing correlations between thermodynamic systems in this framework. The Cartesian product of two state spaces does not allow for a representation of correlations.
Nevertheless, this is not to say that correlations between thermodynamic systems do not appear. Our description just ignores them by describing the total state of a composite system as the collection of the reduced states of its subsystems. Consequently, correlations between thermodynamic systems cannot be exploited by thermodynamic processes -- they have to function whether correlations are present or not. 
Even though it may sound paradoxical, this view does not automatically exclude applying the theory to problems regarding the thermodynamics of correlations \emph{inside} a thermodynamic system (e.g.~\cite{Perarnau15}). For instance, thinking of two qubits that may be correlated, these correlations may have a thermodynamic interpretation captured by the theory. However, thermodynamically the two ``quantum-subsystems'' would then have to be summarized in a single atomic system. Otherwise the correlations would necessarily have to be neglected in our thermodynamic description of the state. 

This example shows that the thermodynamic subsystem structure introduced by $\vee$ must not be confused with the subsystem structure from an underlying physical theory (such as quantum mechanics).

\subsection{Concatenating two processes}

We conclude the section with one more postulate. 
Individual thermodynamic processes are obviously necessary objects in order to formulate the theory. 
However, individual processes are not enough, as in many protocols it will be crucial to be able to describe what happens if one process is executed after another. For this, we postulate the existence of a concatenation operation, similar to concatenating functions. The result of a concatenated process will be another process, namely the one which consists of the consecutive execution of the two processes that are concatenated.

\begin{postulate}[Concatenation of processes]
\label{post:conc}
Let $p,p'\in\mathcal{P}$ such that for all $A\in\mathcal{A}_p\cap\mathcal{A}_{p'}$ it holds
$\lceil p \rceil_A = \lfloor p' \rfloor_A$.
Then $p$ and $p'$ can be \emph{concatenated} to form a new process denoted by $p'\circ p\in\mathcal{P}$,
which represents the consecutive execution of $p$ followed by $p'$. 
If $\mathcal{A}_p\cap\mathcal{A}_{p'} = \emptyset$,
then in addition $p'\circ p = p \circ p'$, i.e.\ concatenation commutes.
An atomic system $A\in\mathcal{A}$ is involved in the concatenated process $p'\circ p$ if and only if $A\in\mathcal{A}_p\cup\mathcal{A}_{p'}$.
For the involved atomic systems the initial and final states are
\begin{align}
\label{eq:inputconc}
\lfloor p'\circ p \rfloor_A = 
\begin{cases}
\lfloor p \rfloor_A\,, \quad &\text{if } A \in\mathcal{A}_p \\
\lfloor p' \rfloor_A\,, \quad &\text{otherwise} \\
\end{cases}
\end{align}
and the final state follows the same rules with swapped roles for $p$ and $p'$.
\end{postulate}


If there is an atomic system for which the initial state of $p'$ and the final state of $p$ are defined but do not match, it is not possible to concatenate the two processes. 
Hence, in this sense, $\circ: \mathcal{P}\times\mathcal{P} \rightarrow\mathcal{P}$ is a partially defined function and $\mathcal{P}$ is closed under $\circ$.\\

The assignment of input and output states for concatenated thermodynamic processes can intuitively be justified as follows. 
If input and output states of both $p$ and $p'$ are defined for an atomic system $A$ then the input state of the concatenation $p'\circ p$ of the two processes must be equal to the input of $p$ on 
$A$ and accordingly the output must be equal to the output of $p'$ on $A$.

Now suppose both $p$ and $p'$ have undefined input and output states on $A$. This is interpreted as ``system $A$ is involved in neither of the processes $p$ and $p'$''.
Thus, the concatenation of the two processes does not act on $A$ either and its input and output states are undefined as well.

Turning to the case where only one of the processes, say $p$, has undefined input state (and thus also undefined output state), it helps to think of an undefined state as a variable. The input and output state of $p$ on $A$ are in this case not determined as $A$ is not involved in $p$. However, when concatenating $p$ with $p'$ to $p' \circ p$ the total process does act on $A$, since $p'$ does. 
Therefore the input state on $A$ of the concatenation, which was variable according to $p$, is set to be the input state of $p'$. In the same way it is argued that for undefined $\lceil p' \rceil_A$ the output state $\lceil p'\circ p \rceil_A$ is set to be equal to $\lceil p \rceil_A$ if the latter is defined. \\

Every thermodynamic system is a composition of atomic systems and reduced states determine the global state of a composite system. Therefore, the rules for initial and final states of concatenated processes for atomic systems translate to rules for arbitrary systems.\\

This concludes the first section with basic postulates on the structure of the theory. 
A more formal discussion of processes and states can be found in Appendix~\ref{app:processesstates}.
We now move on to the ``real thermodynamic'' interplay of systems, states and processes.

\newpage
\section{Work and work processes}
\label{sec:work}

\begin{center}
\fcolorbox{OliveGreen}{white}{\begin{minipage}{0.98\textwidth}
\centering
\begin{minipage}{0.95\textwidth}
\ \\
\textcolor{black}{\textbf{Postulates:}}
work for atomic systems, additivity of work, freedom of description\\
\ \\
\textcolor{black}{\textbf{New notions:}}
work function, 
work process, 
identity process, cyclic, catalytic and re\-versible process\\
\ \\
\textcolor{black}{\textbf{Technical results}} 
for this section can be found in Appendix~\ref{app:work}.

\paragraph{\textcolor{black}{Summary:}}
Two postulates guarantee that for any atomic system there exists a {work function} assigning the work cost of a process on that system and that this function is {additive under concatenation}.
The notion of a work function is then extended to exist for all systems such that it is additive under composition. 
The terms {work process}, {identity process}, {cyclic process} and {catalytic process} are introduced. 
Regarding catalytic processes, a further postulate guarantees that {a catalytic system can also be left out of the description of a process}. The final postulate of this section essentially says that we have a freedom in where to draw the line between things that we explicitly describe in the theory and means that may be used to implement processes, but are not modelled in the theory.
The definition and discussion of {reversible work processes} concludes the section.
\vspace{.1cm}

\end{minipage}
\end{minipage}}
\end{center}

\ \\

So far we have discussed many aspects of systems, states and to some extent processes. However, the thermodynamic component in these considerations has been missing. 
This changes now, as we introduce work cost functions.
They open the discussion on energy flows and will lead to the laws of thermodynamics. 

\subsection{Thermodynamic work}

\begin{postulate}[Work]
\label{post:work}
For any atomic system $A\in\mathcal{A}$ exists a function 
$W_A:\mathcal{P}\rightarrow\mathbbm{R}$ that maps a thermodynamic process $p$ to $W_A(p)$,
the \emph{work done on system $A$} by performing $p$. 
The value $W_A(p)$ is positive 
whenever positive work is done on $A$ while executing $p$.
If the system $A$ is not involved in $p$, $A\notin\mathcal{A}_p$, then $W_A(p)=0$ necessarily.
\end{postulate}

Work is seen as the form of energy which can be controlled perfectly. It is in the operator's control to decide when what amount of work flows into or out of a specific system, and in what form this work is done on or drawn from it. She does so by deciding which state change is induced in what manner, i.e.\ by choosing the process to be implemented. The thermodynamic process contains the information about the work flows exchanged with the involved systems.

How to compute the work $W_A(p)$ of a certain process $p$ does not follow from the thermodynamic theory. On the contrary, it is something that the theory takes as an input. 
Typically the work cost is computed using a more fundamental theory such as classical mechanics, electrodynamics or maybe even quantum mechanics. 
For examples see Table \ref{tab:example}.\\

In some texts, e.g.\ in \cite{Callen85}, work is termed the energy coming from ``work reservoirs''. Here, the concept of a work reservoir is not fundamental and the physical system from which the energy termed work comes is usually not modelled explicitly as a thermodynamic system.
On the other hand, an explicit modelling is not excluded either. The choice of what to explicitly include in the thermodynamic description is up to the user of the theory, as long as all postulates are satisfied.

\begin{postulate}[Additivity of work under concatenation]
\label{post:addwork}
If for two processes $p,p'\in\mathcal{P}$ the concatenation $p'\circ p$ is well-defined, then the work cost of the concatenated process equals the sum of the work costs of the individual processes.
That is, for all atomic systems $A\in\mathcal{A}$ it holds that 
$W_A(p'\circ p) = W_A(p) + W_A(p')$ is \emph{additive}.
\end{postulate}

For atomic systems that are neither involved in $p'$ nor $p$ this statement is trivial, as $0=0+0$ always holds. 
However, for atomic systems that are involved in at least one of the concatenated processes the statement is important to relate the work costs of the individual processes to the one of the concatenated process. 
Additivity under concatenation supports the interpretation of work as a currency. If a quantity of work is ``invested'' now and some other quantity after that, in total the sum of the two has been invested.\\


Based on the work cost function for atomic systems it is possible to define the work cost function for arbitrary systems in $\mathcal{S}$ in an intuitive way.

\begin{definition}[Work function for arbitrary systems]
\label{def:work}
Let $S\in\mathcal{S}$ be an arbitrary thermodynamic system. We define its \emph{work cost function} (also simply \emph{work function}) $W_S:\mathcal{P}\rightarrow\mathbbm{R}$ by
\begin{align}
W_S = \sum_{A\in\mathrm{Atom}(S)} W_A\,.
\end{align}
\end{definition}

The work function $W_A(p)$ for atomic systems yields the work done on $A$ during a process $p$. Hence, the work done on an arbitrary system $S$ is the sum of the work done on all its atomic subsystems. 
By defining the work cost function of an arbitrary system as such we automatically obtain additivity for disjoint systems $S_1,S_2\in\mathcal{S}$ (Lemma~\ref{lemma:addworkcompapp}.
That is, for all $p\in\mathcal{P}$ it holds that $W_{S_1\vee S_2}(p) = W_{S_1}(p) + W_{S_2}(p)$.
Furthermore, additivity under concatenation naturally extends to arbitrary systems (Lemma~\ref{lemma:addworkconcapp})).
If $p'\circ p$ is defined, then for all $S\in\mathcal{S}$ it holds that 
$W_S(p'\circ p) = W_S(p)+W_S(p')$.\\

The existence of work cost functions $W_S$ for all thermodynamic systems $S\in\mathcal{S}$ gives insights into the interpretation of the structure of systems as the following example illustrates.

\begin{example}[Different views on the structure of systems]
\label{ex:A1A2}
Consider Figure~\ref{fig:exA1A2}, where two different views on a scenario are depicted. 
The panels (a) and (b) show the same physical situation with different thermodynamic descriptions.

In (a) there are two systems $A_1$ and $A_2$, depicted as cylinders filled with gases, each with a piston through which work can be done on or drawn from the gas. 
The individual work functions $W_{A_1}$ and $W_{A_2}$ sum up to the total work cost on the composite system $S=A_1\vee A_2$. 
Importantly, by means of the transmission it is possible to address the cylinders individually, in particular the individual work functions are both well-defined. 
A thermodynamic process involving $S$ must specify what part of the total work $W_S$ goes into which subsystem, e.g. in terms of determining the transmission ratio. 
States of $S$ specify the states of the subsystem and have the form $\sigma = \sigma_1\vee \sigma_2$ for $\sigma_i\in\Sigma_{A_i}$. 

In (b) the physically identical system to (a) is described by a thermodynamic system $A$ which cannot be decomposed into subsystems, i.e.\ $A$ is atomic. 
The user of the framework may not know about the more complex structure of $A$ or she may just ignore it. 
There is still a piston through which work can be done on or drawn from $A$. However, the transmission ratio from (a) is fixed and there are no individual work functions $W_{A_i}$ as there are no subsystems $A_i$ that could be addressed individually.
\end{example}

Importantly, both views (a) and (b) lead to valid thermodynamic descriptions within the framework, as long as the chosen view is held on to consistently. 
Which view to take is up to the user of the theory.

We conclude that in general, by Postulate~\ref{post:work}, if a system has subsystems then in any thermodynamic process the splitting of the total work is determined and one can compute the work done on any of the subsystems.

\begin{figure}
\hskip-1cm
\begin{center}
	\begin{tikzpicture}[scale=0.65, every node/.style={transform shape}]	

	\draw[] (-.5,5.7) node[above] {\Large{(a)}} ;
	\draw[xshift=14cm] (-.5,5.7) node[above] {\Large{(b)}} ;

	\draw[-] (0,0) -- (3,0) ;
	\draw[-] (0,0) -- (0,2) ;
	\draw[-] (0,2) -- (3,2) ;
	\draw[] (-.5,1) node[] {\Large{$A_2$}};
%
	\draw[-, very thick] (2.5,0) -- (2.5,2);
	\draw[-, very thick] (2.5,1) -- (5,1);
	\draw[<->] (3.5,1.2) node[above right] {\Large{$W_{A_2}$}} -- (4.5,1.2);

	\draw[-,yshift=3cm] (0,0) -- (3,0) ;
	\draw[-,yshift=3cm] (0,0) -- (0,2) ;
	\draw[-,yshift=3cm] (0,2) -- (3,2) ;
	\draw[yshift=3cm] (-.5,1) node[] {\Large{$A_1$}};

	\draw[-, very thick,yshift=3cm] (2,0) -- (2,2);
	\draw[-, very thick,yshift=3cm] (2,1) -- (5,1);
	\draw[<->,yshift=3cm] (3.5,1.2) node[above right] {\Large{$W_{A_1}$}} -- (4.5,1.2);

	\fill[draw=black,color=lightgray,opacity=0.5] (0.01,0.02) -- (2.47,0.02) -- (2.47,1.98) -- (0.01,1.98) ;
	\fill[yshift=3cm,draw=black,color=lightgray,opacity=0.9] (0.01,0.02) -- (1.97,0.02) -- (1.97,1.98) -- (0.01,1.98) ;

	\draw[-,xshift=-1cm] (6,2.4) node[rotate=90, below, yshift = -.1cm] {\LARGE{transmission}} -- (6,.5) -- (7,.5) -- (7,4.5) -- (6,4.5) -- (6,2.4);
	\draw[-, very thick] (6,2.5) -- (8.5,2.5);
	\draw[<->] (7,2.7) node[above right] {\Large{$W_S = W_{A_1} + W_{A_2}$}} -- (8,2.7);

	\draw[-] (6.5,2.7) -- (6.5,5.5) -- (-1,5.5) -- (-1,-.5) -- (6.5,-.5) -- (6.5,2.3);
	\draw[] (3,-1) node[] {\Large{$S = A_1\vee A_2$}};

	\begin{scope}[xshift = 14cm]
		
		\draw[-, color=lightgray] (0,0) -- (3,0) ;
		\draw[-, color=lightgray] (0,0) -- (0,2) ;
		\draw[-, color=lightgray] (0,2) -- (3,2) ;
		\draw[color=lightgray] (-.5,1) node[] {\Large{$A_2$}};
%
		\draw[-, very thick, color=lightgray] (2.5,0) -- (2.5,2);
		\draw[-, very thick, color=lightgray] (2.5,1) -- (5,1);
		\draw[<->, color=lightgray] (3.5,1.2) node[above right] {\Large{$W_{A_2}$}} -- (4.5,1.2);
	
		\draw[-,yshift=3cm, color=lightgray] (0,0) -- (3,0) ;
		\draw[-,yshift=3cm, color=lightgray] (0,0) -- (0,2) ;
		\draw[-,yshift=3cm, color=lightgray] (0,2) -- (3,2) ;
		\draw[yshift=3cm, color=lightgray] (-.5,1) node[] {\Large{$A_1$}};
	
		\draw[-, very thick,yshift=3cm, color=lightgray] (2,0) -- (2,2);
		\draw[-, very thick,yshift=3cm, color=lightgray] (2,1) -- (5,1);
		\draw[<->,yshift=3cm, color=lightgray] (3.5,1.2) node[above right] {\Large{$W_{A_1}$}} -- (4.5,1.2);

		\fill[draw=black,color=lightgray,opacity=0.05] (0.01,0.02) -- (2.47,0.02) -- (2.47,1.98) -- (0.01,1.98) ;
		\fill[yshift=3cm,draw=black,color=lightgray,opacity=0.1] (0.01,0.02) -- (1.97,0.02) -- (1.97,1.98) -- (0.01,1.98) ;

		\draw[-,xshift=-1cm, color=lightgray] (6,2.4) 
		-- (6,.5) -- (7,.5) -- (7,4.5) -- (6,4.5) -- (6,2.4);
		\draw[-, very thick] (6,2.5) -- (8.5,2.5);
		\draw[<->] (7,2.7) node[above right] {\Large{$W_A$}} -- (8,2.7);

		\draw[-] (6.5,2.7) -- (6.5,5.5) -- (-1,5.5) -- (-1,-.5) -- (6.5,-.5) -- (6.5,2.3);
		\draw[] (3,-1) node[] {\Large{$A$}};
	\end{scope}

	\end{tikzpicture}
\end{center}
\caption{Two different thermodynamic descriptions of the same physical setting. 
Both views are valid thermodynamic descriptions and the user has to decide which one is taken. The user must then consistently work with the chosen view. 
(a) The total system is seen as a composite system $S=A_1\vee A_2$. A process must thus specify what amount of the total work $W_S$ goes to which subsystem, e.g.\ by determining the transmission ratio. For instance, the total work done could be zero, but the work $W_{A_1}$ is drawn from the subsystem $A_1$ and $W_{A_2} = -W_{A_1}$ is done on $A_2$.
Importantly, each subsystem can individually be addressed.
(b) When describing the same setting as an atomic system $A$, i.e.\ a system without further subsystems, there is no transmission ratio to choose. Only the total work $W_A$ has a meaning and the internal structure of $A$, however complex it may be, is neglected in the description. 
}
\label{fig:exA1A2} 
\end{figure}
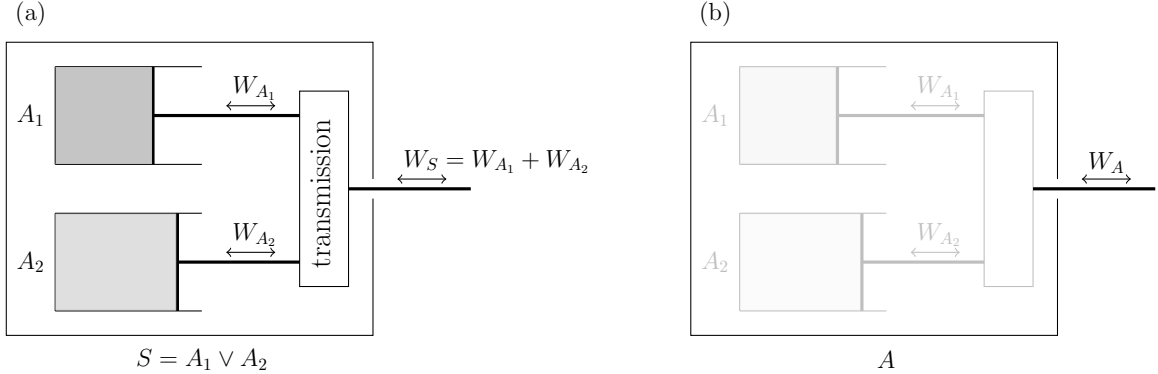

\subsection{Work processes}

Among all thermodynamic processes that act on a system $S$, there are those that act exactly on $S$ and on no other system. These processes deserve special attention as they are the ones addressed in the first law. 

\begin{definition}[Work process]
\label{def:wp}
For $S\in\mathcal{S}$ a process $p\in\mathcal{P}$ is a \emph{work process on $S$} 
if all its atomic subsystems are involved in $p$ and no other atomic systems are.
That is, $p$ is a work process on $S$ if $S=\bigvee \mathcal{A}_p$.
The \emph{set of work processes on $S$} is denoted by $\mathcal{P}_S$.
\end{definition}

In simple words the above definition says that a process is called a work process on a system $S$ if $S$ is the biggest involved system in $p$.
It is not difficult to see that the set $\mathcal{P}_S$ is closed under concatenation of work processes (Lemma~\ref{lemma:closedapp}).\\

Any thermodynamic process can be seen as a work process on a large enough system.
This is a direct consequence of the fact that for any thermodynamic process $p\in\mathcal{P}$ the set $\mathcal{A}_p$ 
is finite and non-empty, which is part of Postulate~\ref{post:processesstates} on thermodynamic processes and states. The system defined as $S:=\bigvee \mathcal{A}_p$ is then such that $p\in\mathcal{P}_S$.
Importantly, as the composition of finitely many atomic systems, $S\in\mathcal{A}$ is indeed a system. \\

This observation makes work processes the central tool to understand thermodynamics. They are the basic building blocks necessary to consistently describe thermodynamic processes. 
If we understand work processes, we understand all thermodynamic processes.
Also it allows us to make definitions and statements for work processes that, if necessary, can then be extended to notions and more general statements for thermodynamic processes.
For instance, work processes are relevant in the formulation of the first law, but of course the first law has consequences for any thermodynamic process.\\

The view of thermodynamic processes as work processes on a large enough system is reminiscent of completely positive trace preserving maps and unitary maps in quantum mechanics, where the Stinespring dilation makes sure that for any cptp map can be seen as a unitary map on a larger system. 
However, the analogy does not go as far as telling us which of the thermodynamic processes are reversible, as is the case with unitary and non-unitary evolution in quantum mechanics.\\


We next consider the special case of two disjoint systems $S_1 \wedge S_2=\emptyset$. Let $p_i\in\mathcal{P}_{S_i}$ for $i=1,2$ be two work processes on $S_1$ and $S_2$, respectively. Such work processes can always be concatenated both as $p_2\circ p_1$ and $p_1\circ p_2$ since the sets of involved atomic systems are disjoint. Namely they are $\mathrm{Atom}(S_1)$ for $p_1$ and on $\mathrm{Atom}(S_2)$ for $p_2$. 
Also, Postulate~\ref{post:conc} on concatenation requires $p_2\circ p_1 = p_1\circ p_2$ in this case.
Furthermore, by construction of the concatenation operation, the atomic systems involved in 
$p_2\circ p_1$ are exactly $\mathrm{Atom}(S_1)\cup\mathrm{Atom}(S_2) = \mathrm{Atom}(S_1\vee S_2)$.
Therefore, the concatenated process $p_1\circ p_2 = p_2\circ p_1\in\mathcal{P}_{S_1\vee S_2}$ is a work process on $S_1\vee S_2$. 
For such a \emph{joint work process} we use the notation $p_1\vee p_2 := p_1\circ p_2$ (Definition~\ref{def:jointwpapp}).
On the other hand, using this notation means that we are dealing with work processes and systems of the kind discussed in this paragraph.\\

%

This notation implies an embedding (an injective mapping) from $\mathcal{P}_{S_1}\times\mathcal{P}_{S_2}$ to $\mathcal{P}_{S_1\vee S_2}$ and images of this mapping $(p_1,p_2)\mapsto p_1\vee p_2$ stand for the parallel execution of the work process $p_1$ on subsystem $S_1$ and of $p_2$ on $S_2$.
Hence, what was achievable by means of work processes on the individual systems $S_1$ and $S_2$ can still be realized as work processes on the composite system $S_1\vee S_2$.
In this sense, composition respects work processes.\\

As a further consequence, the input and output states of a joint work process 
$p_1\vee p_2$ are given by 
$\lfloor p_1\vee p_2 \rfloor_{S_1\vee S_2}  
= \lfloor p_1\rfloor_{S_1} \vee \lfloor p_2 \rfloor_{S_2}
\in\Sigma_{S_1\vee S_2} $
and likewise for $\lceil\cdot\rceil$.
While the state space of a composite system consists of joint states only, this is not the case with work processes.
In general, the set $\mathcal{P}_{S_1\vee S_2}$ contains more work processes than just the joint work processes of its subsystems. An example of a more general work process on a composite system is thermally connecting two subsystems and letting them exchange energy.\\

For a joint work process $p_1\vee p_2$ the total work cost $W_{S_1\vee S_2}(p_1\vee p_2)$ is the sum of the local work costs, i.e.\ $W_{S}(p_1\vee p_2) = W_{S_1}(p_1) + W_{S_2}(p_2)$ (Lemma~\ref{lemma:wjointapp}).
This is a consequence of both the additivity of work under concatenation and under composition (for disjoint systems).

\subsection{Processes with special properties}

A special kind of work processes on a system are \emph{identity processes}. They are defined in the following.

\begin{definition}[Identity process]
\label{def:id}
An \emph{identity process on $S$}, where $S\in\mathcal{S}$ is an arbitrary thermodynamic system, is a work process $\mathrm{id}_S\in\mathcal{P}_S$ on $S$ with $\lfloor \mathrm{id}_S \rfloor_S = \lceil \mathrm{id}_S \rceil_S$ and zero work cost on all atomic subsystems of $S$, $W_A(\mathrm{id}_S)=0$ for all
$A\in\mathrm{Atom}(S)$. 
\end{definition}

We sometimes write $\mathrm{id}_S^\sigma$ for an identity process on $S$ with initial and final state $\sigma\in\Sigma_S$, to indicate on which state the identity process acts. 
This notation manifests that there is not a single identity map that can be applied to an arbitrary input state. An identity process, just like any other thermodynamic process, determines its input and output state.
Identity processes will be discussed further after the first law.\\

Besides identity processes, there are other ways in which a thermodynamic process can act trivially on a system. 

\begin{definition}[Cyclic and catalytic process]
\label{def:cyccat}
Given a system $C\in\mathcal{S}$ an arbitrary thermodynamic process $p\in\mathcal{P}$ is called \emph{cyclic on $C$} if 
$\lceil p \rceil_C = \lfloor p \rfloor_C$.\\
The process is called \emph{catalytic on $C$} if 
it is cyclic on $C$ and in addition $W_C(p)=0$.\footnote{Notice that the work costs of $p$ for subsystems of $C$ do not have to be zero, only the total work done on $C$ does.}
\end{definition}

The definition of a cyclic process on $S$ which is in addition a work process on $S$ differs from a identity process on $S$ by the missing requirement on the work costs on atomic subsystems.

As a consequence of Definition~\ref{def:cyccat}, if a process $p$ is cyclic on two systems $S_1$ and $S_2$ then it is also cyclic on their composition $S_1\vee S_2$. 
This holds independently of whether the two systems are disjoint or not.
If the systems are in addition disjoint, then the same holds for catalytic processes.\\

Consider now the case of a work process 
$p\in\mathcal{P}_{S\vee C}$ on a disjointly composite system $S\vee C$.
If the process is cyclic on $C$ it is still possible that non-zero work flows into $C$ occur. Hence, leaving $C$ out of the description would in this case leave open the question where these work flows go.
However, if the process is catalytic on $C$, i.e.\ $W_C(p)=0$ in addition, one may wonder why such a process is not considered a work process on $S$ alone. 
On the one hand technically it is not a work process on $S$ according to Definition~\ref{def:wp}. On the other hand, even though it acts on $C$, it does so in the most trivial sense, and could hence intuitively be considered a work process. 
We resolve this issue with an additional postulate on ``where to draw the line'', stating that whenever a process is catalytic on a part of a system, there exists a corresponding work process on rest with the same thermodynamic properties, but without acting on the catalytic part.

\begin{postulate}[Freedom of description]
\label{post:freedom}
For $S,C\in\mathcal{S}$ disjoint, let $p\in\mathcal{P}_{S\vee C}$ be such that $p$ is catalytic on $C$, i.e.\ $p$ is cyclic on $C$ and fulfils $W_C(p)=0$.
Then there exists a work process $\tilde p \in\mathcal{P}_S$ on $S$ alone such that 
$\lfloor \tilde p \rfloor_{S} = \lfloor p \rfloor_{S}$ and
$\lceil \tilde p \rceil_{S} = \lceil p \rceil_{S}$ as well as
$W_{A}(\tilde p) = W_{A}(p)$ for all $A\in\mathrm{Atom}(S)$.
\end{postulate}

This postulate automatically implies that $W_{C'}(\tilde p) = 0$ for all subsystems $C'\in\mathrm{Sub}(C)$ of $C$ in the new process, since neither of them is involved in $\tilde p$.
Likewise, $W_{S'}\tilde p) = W_{S'}(p)$ for al	l subsystems $S'\in\mathrm{Sub}(S)$ since their work costs on general subsystems are computed through the work costs on atomic subsystems.
Any other system that was not acted on by $p$ neither is by $\tilde p$.
Also, since the state of a composite system determines the states of all possible subsystems and vice versa, it holds that
$\lfloor \tilde p \rfloor_{S'} = \lfloor p \rfloor_{S'}$ and
$\lceil \tilde p \rceil_{S'} = \lceil p \rceil_{S'}$ for all $S'\in\mathrm{Sub}(S)$.\\

Postulate~\ref{post:freedom} is about where to draw the line between objects that thermodynamics explicitly describes and such that are not part of the theory but may nevertheless be used when executing a process. 
More specifically, whether a catalytic system $C$ is mentioned explicitly in the description of the process $p$, or whether it is suppressed, does not matter. The postulate states that if there is a process in which $C$ is made explicit, then there is also one in which it is not. The important point about this is that the thermodynamic properties of both $p$ and $\tilde p$ regarding the system $S$ are the same. \\

This section is concluded with the definition and the discussion of reversible work processes. 

\begin{definition}[Reversible processes]
\label{def:rev}
A work process $p\in\mathcal{P}_S$ on a system $S\in\mathcal{S}$ is called \emph{reversible} if there exists another work process $p^\mathrm{rev}\in\mathcal{P}_S$ on $S$, the \emph{reverse} work process, such that $p^\mathrm{rev}\circ p$ is an identity process. 
\end{definition}

As the definition suggests, given a specific reversible work processes $p\in\mathcal{P}_S$, there may be more than one reverse work process. Nevertheless, their thermodynamic effect on the system $S$ are all the same.
In the following we may also just write \emph{reversible process} and \emph{reverse process} instead of emphasizing that we talk about \emph{work} processes on specified systems. 
In general, whenever it is clear from the context whether the object is a work process or a thermodynamic process, we will just write process.\\

Reversibility of arbitrary thermodynamic processes does not require a separate definition. Since every thermodynamic process is a work process on some large enough system, the thermodynamic process is called reversible if it is reversible as a work process.\\

While identity processes have a zero work cost on any involved system, the work cost of a reverse work process is the negative forward process (Lemma~\ref{lemma:reverseworkapp}).
More precisely, if $p\in\mathcal{P}_S$ is a reversible work process on the system $S$ with a reverse process $p^\mathrm{rev}\in\mathcal{P}_S$, then 
$W_A(p^\mathrm{rev}) = -W_A(p)$ for all atomic subsystems $A\in\mathrm{Atom}(S)$.
This is easily seen since $p^\mathrm{rev}\circ p$ is an identity process and thus fulfils
$0=W_A(p^\mathrm{rev}\circ p) = W_A(p)+W_A(p^\mathrm{rev})$.\\

In Appendix~\ref{app:work} it is shown in detail that if a concatenated process $p = p_2\circ p_1$ is reversible, then so are the processes $p_1$ and $p_2$ (Proposition~\ref{prop:p1p2revapp}).
This result holds for any processes $p_1,p_2\in\mathcal{P}$ as long as their concatenation is defined. In particular, they need not be work processes on the same or on disjoint systems.

\newpage
\section{The first law}
\label{sec:firstlaw}

\begin{center}
\fcolorbox{OliveGreen}{white}{\begin{minipage}{0.98\textwidth}
\centering
\begin{minipage}{0.95\textwidth}
\ \\
\textcolor{black}{\textbf{Postulates:}}
the first law\\
\ \\
\textcolor{black}{\textbf{New notions:}} preorder on state space, internal energy\\
\ \\
\textcolor{black}{\textbf{Technical results}} 
for this section can be found in Appendix~\ref{app:firstlaw}.

\paragraph{\textcolor{black}{Summary:}}
The {first law} is stated as a postulate. It makes sure that for any two states there is a work process transforming one into the other. Furthermore, the total work cost of work process must not depend on anything except for its input and output states.
Through the existence of work processes connecting two states a {preorder} on the state space is introduced, which will come in handy later.
Due to this preordered structure it follows that identity processes exist for any state on any system.
Finally, the first law and its immediate consequences are used to defined the {internal energy} of a system, which is then shown to be additive. 

\vspace{.1cm}

\end{minipage}
\end{minipage}}
\end{center}

\subsection{The formal statement and its immediate consequences}

The first law will imply that every change in ``internal energy'' of a system is equal to the ``sum of work and heat''. However, in order to state this properly we first need definitions saying what ``internal energy'' and ``heat'' is. 
The first law lays the basics for these definitions.

\begin{postulate}[The first law]
\label{post:first}
For any system $S\in\mathcal{S}$ the following two statements hold:
\begin{itemize}
	\item [(i)]
	For any pair of states $\sigma_1,\sigma_2\in\Sigma_S$ there is a work process 
	$p\in\mathcal{P}_S$ on $S$ with $\lfloor p \rfloor_S = \sigma_1$ and 
	$\lceil p \rceil_S = \sigma_2$ or there is a work process $p'\in\mathcal{P}_S$ on $S$ 
	with $\lfloor p' \rfloor_S = \sigma_2$ and $\lfloor p' \rfloor_S = \sigma_1$.
	\item [(ii)] The total work cost of a work process $p\in\mathcal{P}_S$ on $S$, $W_S(p)$, 
	only depends on $\lfloor p \rfloor_S$ and $\lceil p \rceil_S$ 
	and not on any other details of the process. 
\end{itemize}
\end{postulate}

By the first law, the set of states of a system is preordered.
A \emph{preordered set} is a set $\mathcal{M}$ together with a relation $\rightarrow$ such that 	
the relation is (i) reflexive, i.e.\ $\forall m\in\mathcal{M}:\, m\rightarrow m$, 
and (ii) transitive, i.e.\ if both $m\rightarrow m'$ and $m'\rightarrow m''$, then $m\rightarrow m''$.


\begin{definition}[Preordered states]
\label{def:preorder}
For any system $S\in\mathcal{S}$ the \emph{preorder} $\rightarrow$ on $\Sigma_S$ is established by the reachability via a work process, i.e.\ for $\sigma,\sigma'\in\Sigma_S$ define 
\begin{align}
\label{eq:preorder}
\sigma\rightarrow\sigma'\, :\Leftrightarrow\, \exists p\in\mathcal{P}_S \ \mathrm{s.t.}\ 
\lfloor p \rfloor_S = \sigma,\ \lceil p \rceil_S = \sigma'\,.
\end{align}
\end{definition}

Processes $p\in\mathcal{P}_S$ can be seen as labels of the preordered pairs. In this sense, if one wants to precisely state which work process is responsible for the preordering of two states, one can write 
$\sigma\stackrel{p}{\rightarrow}\sigma'$ if the work process $p$ is such that $\lfloor p \rfloor_S = \sigma$ and $\lceil p \rceil_S = \sigma'$.
As mentioned before, there may be more than one label for a preordered pair.

The relation introduced in Eq.~\ref{eq:preorder} is reflexive since for all systems $S$ and all states $\sigma\in\Sigma_S$ there exists work process $p\in\mathcal{P}_S$ such that 
$\sigma\stackrel{p}{\rightarrow}\sigma$.
This is a consequence of Postulate~\ref{post:first} (i).
Furthermore, if $\sigma\stackrel{p}{\rightarrow}\sigma'$ and 
$\sigma'\stackrel{p'}{\rightarrow}\sigma''$ then 
$\sigma\stackrel{p'\circ p}{\longrightarrow}\sigma''$ is preordered too by means of the concatenated process $p'\circ p$.
Hence the relation is also transitive, which makes it a preorder.\\
%
%

Restricting our considerations to atomic systems $A\in\mathcal{A}$ for the moment, the reflexivity of the preorder $\rightarrow$ offers further insights. 
Let $q\in\mathcal{P}_A$ be a process which relates $\sigma\stackrel{q}{\longrightarrow}\sigma$. Since this process only acts on $A$ and initial and final states match, it can be concatenated with itself. The state change of the process in which $q$ is applied twice is obviously the same as the state change under $q$ itself. Therefore Postulate~\ref{post:first} (ii) requires 
\begin{align}
W_A(q) \stackrel{\mathrm{(ii)}}{=}
W_A(q \circ q) = 
W_A(q) + W_A(q) = 2 W_A(q)\,,
\end{align}
which is to say that the work cost of such processes is zero. Obviously, this makes it an identity process on $A$ for the state $\sigma$ (Lemma~\ref{lemma:idexistatomicapp}).

According to the first law, such processes exists for any state on any atomic system. We deduce that identity processes exist for all atomic subsystems and all states.
Using the decomposition of an arbitrary system into its atomic subsystems, it is then possible to construct identity process for all states of arbitrary systems.
To see this, let $S=A_1\vee \cdots \vee A_n$ with different atomic systems $A_i\neq A_j$ and consider an arbitrary joint state $\sigma = \sigma_1\vee \cdots \vee \sigma_n$. For the atomic states we know that corresponding identity processes $\mathrm{id}_{A_i}^{\sigma_i}$ exist. Therefore, the joint process $\mathrm{id}_S^\sigma := \mathrm{id}_{A_1}^{\sigma_1}\vee \mathrm{id}_{A_n}^{\sigma_n}$ is a cyclic work process on $S$. But this process also fulfils $W_{A_i}(\mathrm{id}_S^\sigma) = 0$ for all $i=1,\dots,n$ by construction, thus it is an identity process.
We conclude that identity processes exist for all states of all thermodynamic systems (Lemma~\ref{lemma:idexistapp}).\\

The existence of identity processes for any state in the theory suggests that what is called a state here can intuitively be seen as an ``equilibrium state''. It implies that it is possible to essentially ``do nothing'' while the state does not change. 
Certainly, this is a property one would expect of an object called equilibrium state. 
Notice that this is a consequence of the first law rather than an additional postulate on or definition of equilibrium. 
Hence, we remain in the position to claim that no further assumptions are made regarding what can be called a thermodynamic state. 

As an example of what is usually considered a non-equilibrium state, take the current position and momentum of an oscillating pendulum as its state. Since both position and momentum are parameters evolving in time even when no action is taken from the outside, such a definition of the state would not work here intuitively. Even if we ``do nothing'' the state would oscillate, i.e.\ change, and would thus not be considered an ``equilibrium state''.\\

In the previous section we mentioned the embedding of $\mathcal{P}_{S_1}\times\mathcal{P}_{S_2}$ for two disjoint systems $S_1$ and $S_2$ in the set of work processes $\mathcal{P}_{S_1\vee S_2}$.
Together with the existence of identity processes for all states this observation is strengthened. We can now conclude that not only the Cartesian product but also the individual sets $\mathcal{P}_{S_1}$ and $\mathcal{P}_{S_2}$ are represented in the set of work processes on the composite system.
For any state $\sigma_2\in\Sigma_{S_2}$ and any work process $p_1\in\mathcal{P}_{S_1}$ the embedding $(p_1, \mathrm{id}_{S_2}^{\sigma_2})\mapsto p_1 \vee \mathrm{id}_{S_2}^{\sigma_2}$ allows for a representation of $p_1$ in $\mathcal{P}_{S_1\vee S_2}$.\\

We also mentioned that the set of work processes of a disjointly composite system $S_1\vee S_2$ contains more than just joint work processes. 
With the first law we can now argue for this statement more precisely.
Consider two systems $S_1$ and $S_2$ with states $\sigma_1,\sigma_1'\in\Sigma_{S_1}$ 
which are preordered such that $\sigma_1 \rightarrow \sigma_1'$ but 
$\sigma_1' \nrightarrow\sigma_1$ and likewise for $\sigma_2,\sigma_2'\in\Sigma_{S_2}$.
I.e.\ the work processes that label the preorderings of $\sigma_i\rightarrow\sigma_i'$ are irreversible.
Then the joint work processes on $S_1\vee S_2$ will make sure that the joint states
$(\sigma_1,\sigma_2) \rightarrow (\sigma_1',\sigma_2)$, 
$(\sigma_1,\sigma_2') \rightarrow (\sigma_1',\sigma_2')$,
$(\sigma_1,\sigma_2) \rightarrow (\sigma_1,\sigma_2')$ and 
$(\sigma_1',\sigma_2) \rightarrow (\sigma_1',\sigma_2')$ are all preordered, in particular comparable.
But then, since the work processes labelling these preorderings are irreversible, the joint states
$(\sigma_1,\sigma_2')$ and $(\sigma_1',\sigma_2)$ cannot be preordered in either direction with joint work processes. Therefore, if any two states of any system must be comparable, as the first law requires, there must be more than just joint work processes in $\mathcal{P}_{S_1\vee S_2}$.


\subsection{Internal energy}

The first law eventually guarantees that every system $S$ has a well-defined internal energy function. The existence of such a function relies on the fact that the total work cost of a work process on $S$ only depends on the initial and final state of the process, and not on any other property of it. 
This is not true for arbitrary thermodynamic processes on $S$, which allow for different work costs on $S$ even though they induce the same state transfer. 
In general, the work cost depends on how exactly the process is carried out.\\

We are now in the position to define the internal energy function of a system.

\begin{definition}[Internal energy]
\label{def:Ustrong}
For a system $S\in\mathcal{S}$ fix an arbitrary reference state $\sigma_0\in\Sigma_S$ and an arbitrary reference energy $U_S^0\in\mathbbm{R}$.
The \emph{internal energy of a state $\sigma\in\Sigma_S$} is defined as
\begin{align}
\label{eq:Ustrong1}
&U_S(\sigma) := U_S^0 + W_S(p) \,, \qquad  \,
\text{where $p\in\mathcal{P}_S$ is s.t.\ 
$\lfloor p \rfloor_S =\sigma_0$ and $\lceil p \rceil_S = \sigma$.}\\
\label{eq:Ustrong2}
&U_S(\sigma) := U_S^0 - W_S(p') \,, \qquad  
\text{where $p'\in\mathcal{P}_S$ is s.t.\ 
$\lfloor p' \rfloor_S =\sigma$ and $\lceil p' \rceil_S = \sigma_0$.}
\end{align} 
\end{definition}

Since only differences $\Delta U_S$ of internal energies physically matter, the choice of $U_S^0$ is arbitrary for now.
Likewise, the reference state $\sigma_0\in\Sigma_S$ is arbitrary independently for each system $S$. 

\begin{definition}[State function]
\label{def:statefunc}
A \emph{state function on a system $S\in\mathcal{S}$} (also \emph{state variable}) is a function $Z:\Sigma_S\rightarrow\mathcal{Z}$ from the state space $\Sigma_S$ to a target space $\mathcal{Z}$. The co-domain $\mathcal{Z}$ is typically $\mathbbm{R}^n$, most often $n=1$. 
\end{definition}

When a system $S$ undergoes a process $p\in\mathcal{P}$ we denote the change in any state function $Z_S$ using an abbreviated notation by $\Delta Z_S(p) := Z_S(\lceil p \rceil_S) - Z_S(\lfloor p \rfloor_S)$.\footnote{This of course only works if a ``minus operation`'' is defined on the co-domain $\mathcal{Z}$. For all practical purposes considered here this is the case.} 
On the left hand side the dependence of $\Delta Z_S$ on $p$ may be omitted if the context makes clear which process is meant.\\

The internal energy function $U_S$ is a well-defined state function on $S$ for any system $S\in\mathcal{S}$ (Lemma~\ref{lemma:Uwelldefapp}).
Investigating this function further, it follows in Proposition~\ref{prop:Uadd} 
that it inherits the additivity property for composite systems from the work functions.
That is, for two disjoint systems $S_1\wedge S_2=\emptyset$ we always have that $\Delta U_S = \Delta U_{S_1} + \Delta U_{S_2}$.
Notice that we only talk about differences in internal energies here since the absolute values depend on the constants $U_{S_i}^0$ of the two systems.
In general, if $Z_A$ is a state variable but $A\notin\mathcal{A}_p$ is not involved in the process $p$, we write $\Delta Z_A(p)=0$. \\


\newpage
\section{Equivalent systems}
\label{sec:equivalentsys}

\begin{center}
\fcolorbox{OliveGreen}{white}{\begin{minipage}{0.98\textwidth}
\centering
\begin{minipage}{0.95\textwidth}
\ \\
\textcolor{black}{\textbf{Postulates:}}
arbitrarily many copies of atomic systems\\
\ \\
\textcolor{black}{\textbf{New notions:}}
thermodynamic isomorphism, equivalence of atomic and arbitrary systems\\
\ \\
\textcolor{black}{\textbf{Technical results}} 
for this section can be found in Appendix~\ref{app:equivalentsys}.

\paragraph{\textcolor{black}{Summary:}}
The notion of {equivalent systems} ({copies of systems}) is introduced and shown to be sensible. 
The intuitive results one would expect from such a notion are explained here and justified with technical results in Appendix~\ref{app:equivalentsys}. 
Finally, we postulate the {existence of arbitrarily many copies of atomic systems}.
\vspace{.1cm}

\end{minipage}
\end{minipage}}
\end{center}

\ \\

In the beginning we have emphasized that elements of $\mathcal{S}$ are seen as specific physical instances rather than types of systems. Nevertheless, it will be crucial to be able to talk about copies of systems or, in other words, systems of the same type. 
These are systems that can be interchanged without any noticeable thermodynamic differences, even though they are different systems. 
It is possible to turn this intuition into a mathematical concept through the notion of a {thermodynamic isomorphism}.
Some ideas presented here are inspired by T. Kriv\'{a}chy's master's thesis \cite{Krivachy17} which the authors supervised.

\subsection{Thermodynamic isomorphisms}

\begin{definition}[Thermodynamic isomorphism]
\label{def:isomorphism}
The pair of \emph{bijective} maps $\varphi: \mathcal{P}\rightarrow\mathcal{P}$, 
$\varphi_\mathcal{A}: \mathcal{A}\rightarrow\mathcal{A}$
is called a \emph{thermodynamic isomorphism} if for any thermodynamic processes $p,p'\in\mathcal{P}$ and any atomic system $A\in\mathcal{A}$ it holds
\begin{itemize}
	\item [(i)] $\varphi(p'\circ p) = \varphi(p')\circ\varphi(p)$ 
	whenever the concatenation $p'\circ p$ or $\varphi(p')\circ\varphi(p)$
	is defined,
	\item [(ii)] $A$ is involved in $p$ if and only if $\varphi_\mathcal{A}(A)$ is involved in $\varphi(p)$, and
	\item [(iii)] $W_{\varphi_\mathcal{A}(A)}(\varphi(p)) = W_{A}(p)$.
\end{itemize}
\end{definition}

The requirements on an isomorphism reveal the fundamental structures behind the discussed thermodynamic concepts. Namely these are (i) the \emph{thermodynamic processes} with \emph{concatenation}, (ii) \emph{atomic systems}, linked to processes through the notion of \emph{an atomic system being involved in a process}, and (iii) \emph{work}.

Remember that in the beginning we already established that input and output states are always either both defined or undefined. Hence (ii) is equivalent to saying that:
$\lfloor \varphi(p)\rfloor_{\varphi_\mathcal{A}(A)}$ is defined $\Leftrightarrow$ 
$\lfloor p \rfloor_{A}$ is defined.\\

Typically when defining isomorphisms of algebraic structures one first introduces the notion of a homomorphism. 
One can do so, but it would go beyond the purposes presented here. 


Both mappings, $\varphi$ and $\varphi_\mathcal{A}$, are part of the definition of a thermodynamic isomorphism. 
However, they are not independent degrees of freedom. In Appendix~\ref{app:equivalentsys}, where the details of this section are discussed, we prove in Lemma~\ref{lemma:isomdofapp} that if both $\varphi,\varphi_\mathcal{A}$ and $\varphi, \varphi_\mathcal{A}'$ are thermodynamic isomorphisms, then $\varphi_\mathcal{A} = \varphi_\mathcal{A}'$.  
In this sense, $\varphi_\mathcal{A}$ is determined by $\varphi$ and one could think of coming up with a Definition~\ref{def:isomorphism} such that it only talks about $\varphi$, while $\varphi_\mathcal{A}$ is derived from it afterwards. 
Even though this is possible it would make the definition much less readable and intuitive. Therefore we do not go further into this.\\	

The mapping of atomic systems can naturally be extended to arbitrary systems $S\in\mathcal{S}$ by defining 
\begin{align}
\varphi_\mathcal{S} (S) := \bigvee_{A\in\mathrm{Atom}(S)} \varphi_\mathcal{A} (A)\,.
\end{align}
A thermodynamic isomorphism maps all thermodynamic processes and systems to possibly other thermodynamic processes  and systems, while preserving the structure of these sets.
Consequently, the two mappings $\varphi$ and $\varphi_\mathcal{S}$ (or $\varphi_\mathrm{A}$, to use the primitive) induce a mapping of states $\varphi_\Sigma : \bigcup_{S\in\mathcal{S}} \Sigma_S \rightarrow \bigcup_{S\in\mathcal{S}} \Sigma_S$ by means of 
\begin{align}
\varphi_\Sigma (\lfloor p \rfloor_S) := \lfloor \varphi(p) \rfloor_{\varphi_\mathcal{S}(S)}\,,
\end{align}
which can be shown to be bijective, too (Lemma~\ref{lemma:isomstatebijapp}). 

As it turns out, an isomorphism $\varphi$ maps work processes on a system $S$ to work processes on the system $\varphi_\mathcal{S}(S)$, the properties of being reversible and being an identity process are preserved under $\varphi$, and internal energy changes fulfil $U_S(\sigma') - U_{S}(\sigma) = U_{\varphi_\mathcal{S}(S)}(\varphi_\Sigma(\sigma')) - U_{\varphi_\mathcal{S}(S)}(\varphi_\Sigma(\sigma))$. 
In short, it can be proved that any thermodynamic structure introduced so far is preserved by  isomorphisms, showing that it is indeed a sensible way of defining the concept. 

\subsection{From isomorphisms to equivalent systems}

It is now possible to consider special isomorphisms to define when two atomic systems $A_1,A_2\in\mathcal{A}$ are copies of each other. 
Intuitively, this is the case when it is possible to swap $A_1$ with $A_2$ by means of an isomorphism such that nothing else changes in the theory.
The notion of isomorphisms allows us to to make precise what is meant by ``nothing else changes''.

\begin{definition}[Equivalence of atomic systems]
\label{def:equiatom}
Two atomic systems $A_1,A_2\in\mathcal{A}$ are called \emph{equivalent}, and we write $A_1 \hat = A_2$, if there exists a thermodynamic isomorphism $\varphi, \varphi_\mathcal{A}$ which additionally fulfils 

\begin{itemize}
	\item [(iv)] $\varphi_\mathcal{A}(A_1)=A_2, \varphi_\mathcal{A}(A_2)=A_1$ and 
	$\varphi_\mathcal{A}(A)=A$ for all $A\in\mathcal{A}\smallsetminus\{A_1,A_2\}$, and

	\item [(v)] for $A\in\mathcal{A} \smallsetminus\{A_1,A_2\}$ it holds
	$\lfloor \varphi(p)\rfloor_A=\lfloor p \rfloor_A$ and 
	$\lceil \varphi(p)\rceil_A=\lceil p \rceil_A$.
\end{itemize}
\end{definition}
%
%
%

The points (iv) and (v) in this definition formally capture the additional intuitive requirements on an isomorphism that only swaps $A_1$ with $A_2$. 
The part of the isomorphism acting on the set of atomic systems shall only swap the two, and all processes shall induce the same state change as before on any atomic system apart from $A_1$ and $A_2$.\\

Using the results already known from isomorphisms in general, it can then be shown that $\hat=$ is indeed an \emph{equivalence relation} on $\mathcal{A}$. 
Having introduced a notion of equivalence for atomic systems, we use it to extend the concept of equivalence to arbitrary thermodynamic systems. 

\begin{definition}[Equivalence of systems]
\label{def:equisys}
Let $S_1,S_2\in\mathcal{S}$ be two arbitrary systems. They are \emph{equivalent}, and we write $S_1\hat=S_2$, if 
there exists a bijection between $\mathrm{Atom}(S_1)$ and $\mathrm{Atom}(S_2)$ which respects the equivalence classes of $\hat=$ for atomic systems. \\

Definition~\ref{def:equisys} can be rephrased as: The two systems are equivalent if 
$|\mathrm{Atom}(S_1)| = |\mathrm{Atom}(S_2)| =: n$ and there exists a labelling 
$\{A_k^{(i)}\}_{i=1,..,n} = \mathrm{Atom}(S_k)$ for $k=1,2$ such that
\begin{align}
A_1^{(i)} \hat= A_2^{(i)} \quad \text{for } i=1,\dots,n \,.
\end{align}
\end{definition}

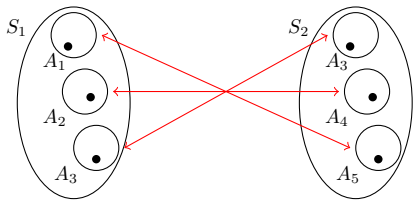
\begin{figure}
\hskip-1cm
\begin{center}
	\begin{tikzpicture}[scale=0.75, every node/.style={transform shape}]	

	\draw[] (-1,1) ellipse (1cm and 1.7cm);
	\node at (-1.7,2.6) [below left]{$S_{1}$};
		
	\draw[] (-1,2.2) ellipse (.4cm and .4cm);
	\node at (-1.1,2) [circle,fill,inner sep=1.5pt]{};
	\node at (-1,2) [below left]{$A_1$};

	\draw[] (-.8,1.2) ellipse (.4cm and .4cm);
	\node at (-.7,1.1) [circle,fill,inner sep=1.5pt]{};
	\node at (-1,1) [below left]{$A_2$};
	
	\draw[] (-.6,.2) ellipse (.4cm and .4cm);
	\node at (-.6,0) [circle,fill,inner sep=1.5pt]{};
	\node at (-.8,0) [below left]{$A_3$};
	
	\begin{scope}[xshift = 5cm]
		\draw[] (-1,1) ellipse (1cm and 1.7cm);
		\node at (-1.7,2.6) [below left]{$S_{2}$};
		
		\draw[] (-1,2.2) ellipse (.4cm and .4cm);
		\node at (-1.1,2) [circle,fill,inner sep=1.5pt]{};
		\node at (-1,2) [below left]{$A_3$};

		\draw[] (-.8,1.2) ellipse (.4cm and .4cm);
		\node at (-.7,1.1) [circle,fill,inner sep=1.5pt]{};
		\node at (-1,1) [below left]{$A_4$};
	
		\draw[] (-.6,.2) ellipse (.4cm and .4cm);
		\node at (-.6,0) [circle,fill,inner sep=1.5pt]{};
		\node at (-.8,0) [below left]{$A_5$};
	\end{scope}
	
	\draw[<->, red] (-.5,2.2) -- (-1.1+5,.2);
	\draw[<->, red] (-.3,1.2) -- (-1.3+5,1.2);
	\draw[<->, red] (-.1,0.2) -- (-1.5+5,2.2);
		
	\end{tikzpicture}
\end{center}
\caption{An example for two arbitrary thermodynamic systems $S_1\hat=S_2$ with $n=3$ (Definition~\ref{def:equisys}).
The equivalent atomic systems are connected with red arrows. They are
$A_1\hat=A_5$, $A_2\hat=A_4$ and the trivial $A_3\hat=A_3$.
}
\label{fig:equisystems} 
\end{figure}

Also for equivalent thermodynamic systems an isomorphism can be defined which fulfils analogous properties to the ones in Definition~\ref{def:isomorphism} and Definition~\ref{def:equiatom}. 
This isomorphism is essentially the concatenation of the individual isomorphisms for the atomic equivalences $A_1^{(i)} \hat= A_2^{(i)}$.
Just like before, one has to check that this definition preserves all kinds of thermodynamic properties of systems in order to justify calling the related systems \emph{equivalent}. 
In particular, now one also has to take composition of systems into account, which can be done.
These technical results are not tricky to derive but cumbersome at times. \\

Having done so in Appendix~\ref{app:equivalentsys} we can move on to the main postulate, for which the machinery of equivalent systems has been introduced. As it turns out, the power of the laws of thermodynamics partly relies on the assumption that in every situation it is possible to extend any particular setting of systems by ``duplicating'' certain subsystems, i.e.\ by adding some other copies of systems to the setting. For instance, this is important in the proof of Carnot's Theorem. The proof of Carnot's Theorem is constructive and uses the fact that systems of the type of those already considered can be added to the setting.  

%

\begin{postulate}[Arbitrarily many copies of (atomic) systems]
\label{post:copies}
Given an atomic system $A\in\mathcal{A}$ and $n\in\mathbbm{N}$ one may always assume that there exist $n$ different equivalent atomic system $\{A_i\}_{i=1}^n\subset\mathcal{A}$, $A_1\hat=A_2\hat=\cdots\hat=A_n\hat=A$.
\end{postulate}


This implies that there are in principle infinitely many copies available of any type of atomic system.\footnote{
However, it does not mean that infinitely many such systems must be present. Rather, it should be interpreted as saying that it is thinkable to have as many copies of an atomic system as one wants.
In this sense, there is no upper limit to the number of thermodynamic systems that are thinkable.} 
As a consequence, the same must hold for arbitrary systems because they consist of atomic systems:
For any $S\in\mathcal{S}$ there exist arbitrarily many different systems $S'\in\mathcal{S}$ 
with $S\wedge S' = \emptyset$ such that $S\hat=S'$. 

\newpage
\section{Heat and heat reservoirs}
\label{sec:heat}

\begin{center}
\fcolorbox{OliveGreen}{white}{\begin{minipage}{0.98\textwidth}
\centering
\begin{minipage}{0.95\textwidth}
\ \\
\textcolor{black}{\textbf{Postulates:}}
-\\
\ \\
\textcolor{black}{\textbf{New notions:}}
heat, heat reservoirs\\
\ \\
\textcolor{black}{\textbf{Technical results}} 
for this section can be found in Appendix~\ref{app:heat}.

\paragraph{\textcolor{black}{Summary:}}
The notion of {heat} is defined based on the previously introduced notions of work and internal energy. Heat is shown to be {additive} in the same way as work is, i.e.\ under concatenation and under composition.
As a preparation for the second law, {heat reservoirs} are defined as thermodynamic systems with special properties.
\vspace{.1cm}

\end{minipage}
\end{minipage}}
\end{center}

\ \\

In the previous sections we have introduced work and derived the internal energy function from it. Having these quantities it is possible to define heat.

\subsection{Heat}

As the first law (Postulate~\ref{post:first}) requires, the total work cost of a work processes on a system $S$ may only depend on initial and final state of the process. A general thermodynamic process acting on $S$ may affect more systems than just $S$ and the work cost on $S$ does \emph{not only} depend on initial and final states. 
It depends on \emph{how exactly} the state change is induced. 
In a typical formulation of thermodynamics this statement amounts to saying that the differential $\delta W_{S}$ is not exact. 
On the other hand, we have seen that the internal energy $U_{S}$ of a system $S$ is a state function, which would be the same as saying that the differential $\mathrm{d}U_{S}$ is exact. Hence, the difference of the two, internal energy minus work, gives rise to another quantity that generally depends on the executed process and not just on the input and output states. This is what we call heat.
It represents the change in internal energy that is \emph{not} a consequence of work done on or drawn from the system, i.e.\ the change in internal energy that is due to less controlled energy flows.\\



\begin{definition}[Heat]
\label{def:heat}
For an arbitrary thermodynamic process $p\in\mathcal{P}$ the \emph{heat flowing to a system $S\in\mathcal{S}$ under $p$} is defined as
\begin{align}
\label{eq:defheat}
Q_{S}(p) := \Delta U_{S}(p) - W_{S}(p) \,.
\end{align}
Here, $\Delta U_S(p)$ is defined as
\begin{align}
\label{eq:deltaUp}
\Delta U_S(p) := \sum_{\mathrm{Atom}(S)} \Delta U_A(p)\,
\end{align}
which is a natural extension to the usual definition of the additive function $\Delta U_S$ for processes in which not necessarily all atomic subsystems of $S$ are involved.\footnote{
Recall that if $A$ is not involved in $p$ we denote $\Delta U_A(p)=0$.}
\end{definition}

It follows automatically that if no atomic subsystem of $S$ is involved in $p$, then $Q_S(p)=0$. 
Also, $Q_S$ inherits additivity under composition (for disjoint systems) from $\Delta U_S$ and $W_S$. The same holds for additivity under concatenation, which is shown by considering atomic systems under a concatenated process $p'\circ p\in\mathcal{P}$ and distinguishing the three cases where (i) $A$ is neither involved in $p$ nor $p'$, (ii) $A$ is involved in one of them, and (iii) $A$ is involved in both of them. 
Furthermore, heat flows in reverse processes simply change their signs, just like work and internal energy do. This is a direct consequence of the definition of heat. 
And finally, heat flows of equivalent systems in equivalent processes are identical, as is shown in Lemma~\ref{lemma:equiheatapp} in Appendix~\ref{app:heat}, together with technical proofs of the other non-trivial statements of this paragraph.\\

If we suppress $p$ in this notation and rearrange Eq.~(\ref{eq:defheat}) we immediately obtain a standard form of the first law \cite{Clausius50} stating that the internal energy of a system $S$ may change in a thermodynamic process due to work and heat only by means of
\begin{align}
\label{eq:usualfirstlaw}
\Delta U_{S} = W_{S} + Q_{S}\,.
\end{align}
Hence, the way the first law is introduced in this work through Postulate~\ref{post:first} implies the common first law as in Eq.~(\ref{eq:usualfirstlaw}). 
In addition, following the arguments in this framework provides significant advantages over the standard statements. 
It does not suffer from undefined terms, which would be the case when stating Eq.~(\ref{eq:usualfirstlaw}) directly. 
Furthermore, the different roles of work and heat as energy contributions to the change in internal energy is not just stated but imprinted in the previous postulates, definitions and derived results. Work must be the energy the user of the theory controls, as it must be known how to compute it and how to reuse it \emph{a priori}. 
Heat, on the other hand, is the energy the user is aware of, but he does not control it directly -- hence it is defined as ``everything else''. \\

Consider now a work process $p\in\mathcal{P}_{S_1\vee S_2}$ on a disjointly composite system.
Instead of saying that the heat $Q_{S_1}(p)$ flows into system $S_1$ during the process $p$ one can also say that $Q_{S_1}(p)$ flows out of $S_2$ or, in other words, that $-Q_{S_1}(p)$ flows into $S_2$. For this statement to be compatible with Definition~\ref{def:heat} it is necessary that $Q_{S_2}(p) = -Q_{S_1}(p)$, otherwise the definition applied to the system $S_2$ would lead to a different notion of heat than the intuition explained above suggests. This is in fact the case, since for a work process $p\in\mathcal{P}_{S_1\vee S_2}$ it holds that 
\begin{align}
W_{S_1}(p) + W_{S_2}(p) = W_{S_1\vee S_2}(p) = \Delta U_{S_1\vee S_2} = \Delta U_{S_1} + \Delta U_{S_2} \,.
\end{align}
Using Definition~\ref{def:heat} for system $S_2$ now implies
\begin{align}
Q_{S_2}(p) = \Delta U_{S_2} - W_{S_2}(p) = -\Delta U_{S_1} + W_{S_1}(p) = -Q_{S_1}(p) \,.
\end{align}

It is worth emphasizing that Definition~\ref{def:heat} tells us what amount of heat flows into $S$, but not necessarily from which other system.
Only in a bipartite setting such as the one discussed in the previous paragraph a statement about where the heat is coming from is possible. 
As a consequence, in a more complex composite system it does not make sense to ask what amount of heat flows from one specific (atomic) subsystem to another, unless these are the only involved ones in the process.
Nevertheless, the intuition of heat flows between specific subsystems can and will be used, first and foremost when illustrating processes such as the Carnot process, which acts on at least three subsystems -- two reservoirs and a cyclic machine.
Whenever this intuition is used it will be made clear in what sense. The first example of this shows up in Figure~\ref{fig:carnot}, where internal arrows are used to mark heat flows in a composite system with up to three subsystems. The final paragraphs of this section discuss how these arrows must be interpreted.\\

In order to clarify the distinction between work and heat we here continue and extend Example~\ref{ex:A1A2}. 

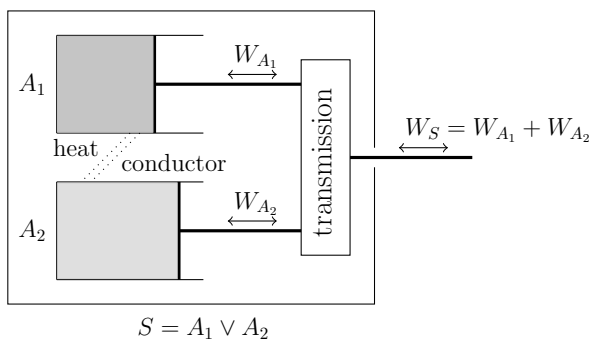
\begin{figure}
\hskip-1cm
\begin{center}
	\begin{tikzpicture}[scale=0.65, every node/.style={transform shape}]	


	\draw[-] (0,0) -- (3,0) ;
	\draw[-] (0,0) -- (0,2) ;
	\draw[-] (0,2) -- (3,2) ;
	\draw[] (-.5,1) node[] {\Large{$A_2$}};
%
	\draw[-, very thick] (2.5,0) -- (2.5,2);
	\draw[-, very thick] (2.5,1) -- (5,1);
	\draw[<->] (3.5,1.2) node[above right] {\Large{$W_{A_2}$}} -- (4.5,1.2);

	\draw[-,yshift=3cm] (0,0) -- (3,0) ;
	\draw[-,yshift=3cm] (0,0) -- (0,2) ;
	\draw[-,yshift=3cm] (0,2) -- (3,2) ;
	\draw[yshift=3cm] (-.5,1) node[] {\Large{$A_1$}};

	\draw[-, very thick,yshift=3cm] (2,0) -- (2,2);
	\draw[-, very thick,yshift=3cm] (2,1) -- (5,1);
	\draw[<->,yshift=3cm] (3.5,1.2) node[above right] {\Large{$W_{A_1}$}} -- (4.5,1.2);

	\fill[draw=black,color=lightgray,opacity=0.5] (0.01,0.02) -- (2.47,0.02) -- (2.47,1.98) -- (0.01,1.98) ;
	\fill[yshift=3cm,draw=black,color=lightgray,opacity=0.9] (0.01,0.02) -- (1.97,0.02) -- (1.97,1.98) -- (0.01,1.98) ;

	\draw[-,xshift=-1cm] (6,2.4) node[rotate=90, below, yshift = -.1cm] {\LARGE{transmission}} -- (6,.5) -- (7,.5) -- (7,4.5) -- (6,4.5) -- (6,2.4);
	\draw[-, very thick] (6,2.5) -- (8.5,2.5);
	\draw[<->] (7,2.7) node[above right] {\Large{$W_S = W_{A_1} + W_{A_2}$}} -- (8,2.7);

	\draw[-] (6.5,2.7) -- (6.5,5.5) -- (-1,5.5) -- (-1,-.5) -- (6.5,-.5) -- (6.5,2.3);
	\draw[] (3,-1) node[] {\Large{$S = A_1\vee A_2$}};
	
\begin{scope}[xshift=-.3cm]
	\draw[-, dotted] (1.8,3) -- (.8,2);
	\draw[-, dotted] (2,3) -- (1,2);
	\draw[] (0.1,2.7) node[right] {\Large{heat}};
	\draw[] (1.5,2.4) node[right] {\Large{conductor}};
\end{scope}

	\end{tikzpicture}
\end{center}
\caption{Similar to Figure~\ref{fig:exA1A2} the composite system $S$ consists of two subsystems $A_1$ and $A_2$. 
The total work done on $S$ is distributed to $A_1$ and $A_2$ through a transmission ratio specified by the thermodynamic process that is executed. 
The process also determines whether and when a thermal contact is established between the subsystems.
}
\label{fig:exWQ} 
\end{figure}

\begin{example}[Work and heat flows in a bipartite system]
Consider Figure~\ref{fig:exWQ}, where a composite system $S=A_1\vee A_2$ is shown. The subsystems are depicted as cylinders filled with gases and the total work $W_S$ is distributed over $A_1$ and $A_2$ via the transmission mechanism. Furthermore, a heat conductor can be put in place between the gases such that they can directly exchange energy.
A thermodynamic process involving $S$ must specify what part of the total work $W_S$ goes into which subsystem, e.g. in terms of determining the transmission ratio. 
It also specifies whether, when, and how a thermal contact is established between $A_1$ and $A_2$.

Suppose work $W>0$ is used to heat up $A_1$, e.g.\ by moving the piston back and forth very fast, and part of this energy is ``passed on'' to $A_2$. Does the energy flow to $A_2$ count as work or heat?

The answer depends on what is meant by ``passed on'' and can be illustrated by considering two different ways of implementing such a state change.
Consider the process $p'\in\mathcal{P}_S$ in which the work $W_S(p') = W_{A_1 }(p') = W$ is done on $A_1$ followed by the process $p\in\mathcal{P}_S$ transferring part of this work via the transmission to $A_2$, i.e.\ $W_S(p) = 0$ and $W_{A_2}(p) = -W_{A_1}(p)>0$.
In this case, the energy exchanged between $A_1$ and $A_2$ is called work and for the total process we find $W_{A_1}(p\circ p') = W - W_{A_2}(p)$ and 
$W_{A_2}(p\circ p') = 0 + W_{A_2}(p)$.

On the other hand, if after $p'$ the process $q\in\mathcal{P}_S$ is applied, in which the energy flow from $A_1$ to $A_2$ happens through the heat conductor (in particular it does not use the controlling mechanism of the transmission), then $W_{A_1}(q) = W_{A_2}(q) = 0$ and we find a non-zero heat flow $Q_{A_2}(q) = \Delta U_{A_2}(q) - W_{A_2}(q) = \Delta U_{A_2}(q) - 0 >0$ .

Notice that, depending on the initial states of the gases, it may be possible to attain exactly the same state changes through either of the two variants presented above. This can of course only happen, if $W_{A_2}(p) = Q_{A_2}(q)$.
We conclude that it depends on the actual process whether energy exchanged between thermodynamic systems is termed heat or work.
This is in accordance with the first law (Postulate~\ref{post:first}) which only requires the \emph{total} work $W_S$ to be independent of anything except the state changes. And indeed, in the case when both $p\circ p'$ and $q\circ p'$ induce the same state change we find
\begin{align}
\begin{split}
W_S(p\circ p') &= W_{A_1}(p\circ p') + W_{A_2}(p\circ p') \\
&= W_{A_1}(p') + W_{A_1}(p) + W_{A_2}(p') + W_{A_2}(p) \\
&= W - W_{A_2}(p) + 0 + W_{A_2}(p)  \\
&= W \,,\ \text{and} \\
W_S(q\circ p') &= W_{A_1}(p') + W_{A_1}(q) + W_{A_2}(p') + W_{A_2}(q) \\
&= W + 0 + 0 + 0 \\
&= W\,.
\end{split}
\end{align}
\end{example}

\subsection{Heat reservoirs}

Definition~\ref{def:heat} can be applied to an arbitrary system. However, in thermodynamics one often makes use of special systems providing or taking up heat, namely \emph{heat reservoirs} (also \emph{heat baths}). These systems play a central role, not least in the second law and Carnot's Theorem. 

A heat reservoir is thought of as a large but simple system. 
Large here stands for the fact that its behaviour does not change significantly when moderate amounts of energy are drawn from or given to it. 
Alternatively, this means that it essentially does not matter whether a finite amount of energy is supplied by one or more copies of the same reservoir. 
In this sense heat reservoirs are regarded as infinite systems. 

On the other hand, a heat reservoir is simple to the extent that there is only one macroscopic parameter that defines its state.\footnote{
Since heat reservoirs are large systems they are of course not simple in a microscopic sense. On the contrary, we know that the larger the system, the more complex it gets when one tries to describe the interplay between its microscopic degrees of freedom. However, when saying simple, we here mean that the used \emph{thermodynamic description} is simple.}
A heat reservoir's only purpose is to provide or take up heat. In particular, one should not be able to extract work from such a system in any process. \\

Formally, the set of heat reservoirs of a thermodynamic theory is characterized as follows.

\begin{definition}[Heat reservoirs]
\label{def:heatreservoir}
An atomic system $R\in\mathcal{A}$ is called a \emph{heat reservoir} if:
\begin{itemize}
	\item [(i)]
	For all $R\in\mathcal{R}$ the internal energy function $U_R$ is injective. 
	\item [(ii)] 
	For all $R\in\mathcal{R}$ and all $p\in\mathcal{P}$ it holds $W_R(p)\geq0$, i.e.\ there is no 
	thermodynamic process that extracts work from a reservoir.
	\item [(iii)] 
	For any thermodynamic process $p\in\mathcal{P}$ that acts on a reservoir $R\in\mathcal{R}$,
	and for any energy difference $\Delta U \in\mathbbm{R}$,
	there exists a corresponding process $p'\in\mathcal{P}$ acting on the states of $R$ shifted by $\Delta U$.
	That is,  
	$W_A(p')=W_A(p)$ for all $A\in\mathcal{A}$, 
	$\lfloor p'\rfloor_A = \lfloor p\rfloor_A$ and $\lceil p'\rceil_A = \lceil p \rceil_A$ 
	for all $A\in\mathcal{A}\smallsetminus\{R\}$, 
	and $U_R(\lfloor p' \rfloor_R) = U_R(\lfloor p \rfloor_R) + \Delta U$
	as well as $U_R(\lceil p' \rceil_R) = U_R(\lceil p \rceil_R) + \Delta U$.

\end{itemize}
The set $\mathcal{R} := \{R\in\mathcal{A} \,|\, R \text{ is heat reservoir}  \} \subset\mathcal{A}$ is called \emph{set of heat reservoirs}.
\end{definition}

It follows that if $R$ is a heat reservoir and $R\hat=R'$, then $R'$ is a heat reservoir, too.
All requirements (i)-(iii) hold either for both $R$ and $R'$ or neither of them since they are statements about the existence or inexistence of some process and processes of equivalent systems are in 1-1 correspondence.
%
\\

The first and the second points state that the description of thermodynamic states of heat reservoirs are simple.
The states must be in one-to-one correspondence with the internal energy,
i.e.\ no other macroscopic quantity is needed to describe it, and one cannot extract work from a reservoir in any thermodynamic process.
In some introductions to thermodynamics this is captured by saying that heat reservoirs have no ``work coordinates'' \cite{LY99,Thess11} (or that these stay constant, for that matter).
This does not yet exclude that one could use a cyclic machine to extract work indirectly from a reservoir by using a heat flow coming from the reservoir. Only the second law guarantees that even a more sophisticated setting does not allow one to extract work from a single reservoir. 

The fact that a heat reservoir cannot produce positive work is an important difference in comparison with work reservoirs. Work reservoirs are sometimes used in traditional approaches when a system is needed to explicitly model work and work flows.
In the case of heat reservoirs it is excluded that energy done on it as work can later again be used as such.

Definition~\ref{def:heatreservoir} (ii) also guarantees that for any two states 
$\sigma,\sigma'\in\Sigma_R$ with $U_R(\sigma')\geq U_R(\sigma)$ there is a work process 
$p\in\mathcal{P}_R$ with $\lfloor p \rfloor_R = \sigma$ and $\lceil p \rceil_R = \sigma'$. 
This can be seen by Postulate~\ref{post:first} (i), stating that either $\sigma \rightarrow \sigma'$ or $\sigma'\rightarrow\sigma$ or both. If $\sigma'\rightarrow\sigma$, the work process on $R$ inducing this state change would extract work, which is forbidden.
Hence the work process with positive work must exist.

We note that for any reversible process $p\in\mathcal{P}$ it must hold $W_R(p)=0$ for all $R\in\mathcal{R}$.
This follows due to the fact that in the reverse process the work done is the negative of the one from the forward process. Hence a non-zero work cost for reservoirs in a reversible process would have to violate (ii) either in the forward or the reverse process. 

Point (iii) formally captures the statement that a heat reservoir is invariant under translations of its internal energy.\footnote{
	In our previous paper \cite{Kammerlander18}, which focussed on the zeroth law, this point was phrased in very different terms. In the remainder of this work it will become clear that the previous requirement can be derived from this less complicated and more intuitive one.}
Keeping in mind that internal energy is injective for heat reservoirs, (iii) can be read as ``What is possible with some initial state is possible with any.''

An obvious one is that reservoirs are infinitely large systems in the sense that their spectrum is $(-\infty,\infty)$. 
Even though this seems unphysical at first sight, this is how we think of heat reservoirs. They are seen as infinitely big systems that do not change their behaviour under finite changes of energy.
Of course, for all practical purposes, a system with approximate characteristics (i)-(iii) which is much larger than all other involved systems can serve as a heat reservoir. 
But in the theoretical modelling the rigorous treatment asks for an ``infinite'' system. 
%

A further important observation is that requirement (iii) for heat reservoirs asks for the existence of processes independent of the initial states, i.e.\ for all initial states. This means that the actual state of a reservoir, and hence its actual internal energy, does not have an influence on what can be done with it. The reservoir's characteristics are in this sense independent of its current state. 
This is the reason why, from now on when using reservoirs, we will not discuss their states in any more depth.\\

Arguably, the formulated assumptions on heat reservoirs are neither new nor surprising. They are mostly standard assumptions, see e.g.\ \cite{LY99}, that may not even be spelled out in certain texts on thermodynamics. 
However, here they are central for the precise arguments given in the coming sections.

In particular, stating the requirements on heat reservoirs explicitly is also necessary in order to formulate the second law according to Kelvin \cite{Kelvin51}, as is done in the next section.\\

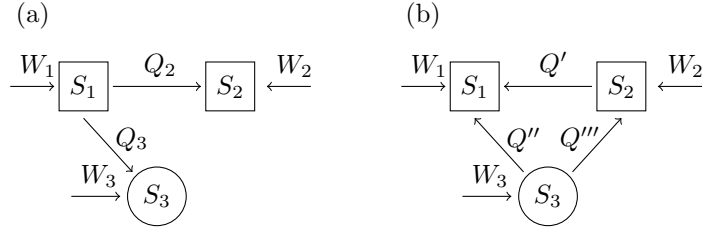
\begin{figure}
\begin{center} 
	\begin{tikzpicture}[scale=.65]
	\draw[] (-3,4.2) node[left] {(a)};	
	 
	\draw[] (-2.5,2.25) node[above, yshift = .05cm] {$S_1$} -- (-2,2.25) -- (-2,3.25) -- (-3,3.25)
	-- (-3,2.25) -- (-2.5,2.25);
	\draw[] (.5,2.25) node[above, yshift = .05cm] {$S_2$} -- (1,2.25) -- (1,3.25) -- (0,3.25) 
	-- (0,2.25) -- (.5,2.25);
	\draw[] (-1,.5) node[] {$S_3$} circle (.6cm);
	
	\draw[->] (-4,2.75) node[above right] {$W_{1}$} -- (-3.1,2.75);
	\draw[->] (2.15,2.75) -- (1.25,2.75) node[above right] {$W_2$};
	\draw[->] (-2.75,.5) node[above right] {$W_{3}$} -- (-1.75,.5);

	\draw[->] (-1.9,2.75) -- (-.1,2.75) node[above left, xshift = -.2cm] {$Q_2$};
	\draw[<-] (-1.5,1) -- (-2.5,2.1) node [below right, xshift = .3cm] {$Q_{3}$};
	
	\begin{scope}[xshift = 8cm]
	\draw[] (-3,4.2) node[left] {(b)};	
	 
	\draw[] (-2.5,2.25) node[above, yshift = .05cm] {$S_1$} -- (-2,2.25) -- (-2,3.25) -- (-3,3.25)
	-- (-3,2.25) -- (-2.5,2.25);
	\draw[] (.5,2.25) node[above, yshift = .05cm] {$S_2$} -- (1,2.25) -- (1,3.25) -- (0,3.25) 
	-- (0,2.25) -- (.5,2.25);
	\draw[] (-1,.5) node[] {$S_3$} circle (.6cm);
	
	\draw[->] (-4,2.75) node[above right] {$W_{1}$} -- (-3.1,2.75);
	\draw[->] (2.15,2.75) -- (1.25,2.75) node[above right] {$W_2$};
	\draw[->] (-2.75,.5) node[above right] {$W_{3}$} -- (-1.75,.5);

	\draw[<-] (-1.9,2.75) -- (-.1,2.75) node[above left, xshift = -.2cm] {$Q'$};
	\draw[->] (-1.5,1) -- (-2.5,2.1) node [below right, xshift = .3cm] {$Q''$};
	\draw[->] (-.5,1) node [above, xshift = .1cm, yshift = .15cm] {$Q'''$} -- (.5,2.1);
	\end{scope}
	
	\end{tikzpicture}
\end{center}
\caption{The three systems $S_1,S_2,S_3\in\mathcal{S}$ undergo a process $p\in\mathcal{P}_{S_1\vee S_2\vee S_3}$ with work flows $W_{S_i}(p) =: W_i$ and heat flows $Q_{S_i}(p) =: Q_i$.
Since the internal energy is a state variable, the sum of work and heat flows into a cyclic systems must be zero, i.e.\ $W_3+Q_3 = 0$.
While external work arrows are not arbitrary since they are determined by the work functions $W_{S_i}$, internal heat arrows have a freedom.
This is illustrated by the difference between (a) and (b). 
(a) suggests that the heat flows into $S_2$ and $S_3$ are coming from $S_1$ directly, while (b) says that the heat $Q'''$ flows between $S_3$ and $S_2$.
Both representations are valid as long as the sums of the internally exchanged heat flows satisfy
$Q'+Q'' = Q_1 \equiv -Q_2-Q_3$, $Q'''-Q' = Q_2$, $-Q''-Q''' = Q_3$.
Internal arrows could only be argued to be unique in this setting if $p$ was seen as a concatenated process $p=p_2\circ p_1$ with e.g.\ $p_i\in\mathcal{P}_{S_i\vee S_3}$. In this case, the heat flows could be split up into the ones flowing during the individual processes $p_i$. 
}
\label{fig:pictorialheat}
\end{figure}

We close this section with a comment on the pictorial representation of heat flows.
When illustrating a process on a composite system we use thick lines to border reservoirs and thin lines for other systems (see e.g.\ Figure~\ref{fig:carnot}). 
Circles border cyclic systems while squares leave open whether the system involved undergoes cyclic evolution or not. 
Directed arrows mark positive work (external) and heat (internal) flows. 

Importantly, arrows showing heat flows of single processes in more complex structures have no direct mathematical meaning since heat flows are not uniquely defined except when exchanged between exactly two subsystems. Only the sum of all internal arrows associated to a subsystem does have a mathematical meaning. The internal arrows nevertheless help to map the abstract processes to well-known situations such as Carnot engines. 
For an example, see Figure \ref{fig:pictorialheat}. Both illustrations show the same process 
$p\in\mathcal{P}_{S_1\vee S_2\vee S_3}$ on the composite system $S_1\vee S_2\vee S_3$. 
While (a) suggests that only the heat $Q_i$ flows from $S_1$ to $S_i$, (b) says that the heat $Q'''$ from $S_3$ to $S_2$.
The mathematical formalism leaves open which of the possibilities actually happen -- they are considered thermodynamically equivalent as long as for each system $\Delta U_i = W_i + Q_i$ is satisfied.
 
This ambiguity is not a problem. On the contrary, it is an asset of our framework that it is possible to make the usual thermodynamic statements without ever having to refer to a technical notion of a ``heat flow from $S_1$ to $S_2$'' in a composite system involving other systems ($S_3$) as well. 

When we use such a wording nevertheless, it is either a non-technical statement appealing to the reader's intuition, or we talk about concatenated processes.
In the latter case, heat flows between subsystems may have a mathematical meaning even though other subsystems are present as well. This is the case if the process can be split up into parts, each of which involves only two subsystems.
Thus the framework is not restricted by the fact that we cannot in general give a mathematical meaning to the heat flow arrows. If one wants to say that heat flows from or to a specific system, one can (at least sometimes) do it by splitting the process into appropriate parts. 
We will explicitly comment on this when the situation shows up.



\newpage
\section{The second law}
\label{sec:secondlaw}

\begin{center}
\fcolorbox{OliveGreen}{white}{\begin{minipage}{0.98\textwidth}
\centering
\begin{minipage}{0.95\textwidth}
\ \\
\textcolor{black}{\textbf{Postulates:}}
the second law\\
\ \\
\textcolor{black}{\textbf{New notions:}}
-

\paragraph{\textcolor{black}{Summary:}}
The {second law} is postulated in the form of the Kelvin-Planck statement and its immediate consequences are discussed. 
\vspace{.1cm}

\end{minipage}
\end{minipage}}
\end{center}

\ \\

Considered to be the core postulate of phenomenological thermodynamics by many, the second law certainly takes a central role in a theoretical introduction to the theory.
The most prominent formulations are due to Carnot \cite{Carnot24}, 
Clausius \cite{Clausius54}, 
Kelvin \cite{Kelvin51} and Planck \cite{Planck97}.
The similar statements by Kelvin and Planck, sometimes called the Kelvin-Planck statement, can be summarized by
``It is impossible to devise a cyclically operating device, the sole effect of which is to absorb energy in the form of heat from a single thermal reservoir and to deliver an equivalent amount of work.'' 
It shares the problem with the other versions by Clausius and Carnot that some terms used (e.g.\ ``sole effect'', ``heat'' or ``reservoir'') are ambiguous if no further specifications are made.
With the notions and definitions we have made up until here we are now able to formally state the Kelvin-Planck version of the second law.

\begin{postulate}[The second law]
\label{post:sec}
Consider a heat reservoir $R\in\mathcal{R}$ together with an arbitrary system $S\in\mathcal{S}$ and a work process $p\in\mathcal{P}_{R\vee S}$ on the composite system.
If $p$ is cyclic on $S$, then $W_S(p)\geq0$, i.e.\ no work can be drawn from $S$ in such a process.
\end{postulate}

The setting of the second law is shown in Figure \ref{fig:2lawcarnot}. 
In the formulation of Postulate~\ref{post:sec} the problems of the above mentioned versions of the second law do not occur.
The ``sole effect'' is captured in the fact that the process $p\in\mathcal{P}_{S\vee R}$ we talk about is a work process on $S\vee R$.
In addition, the terms ``heat'' and ``reservoir'' have been properly defined beforehand and can therefore be used now without danger of confusion. \\

A total amount of work $W_{R\vee S}(p) = W_R(p) + W_S(p)$ is done on the composite system 
$R\vee S$, where $W_R(p)\geq0$ has to be non-negative due to Definition~\ref{def:heatreservoir} (ii). 
Since $S$ is assumed to be cyclic under $p$, i.e.\ $\lceil p \rceil_S = \lfloor p \rfloor_S$, the internal energy of $S$ does not change, $\Delta U_S (p)=0$.
Hence the work done on $S$ is equal to the heat flowing from $S$ to $R$, which is denoted by
$Q_R(p)$.
Therefore, the second law can be rephrased as ``In a process on $R\vee S$ that is cyclic on $S$ heat can only flow from $S$ to $R$ and not in the other direction'', which is precisely the Kelvin-Planck statement.


 

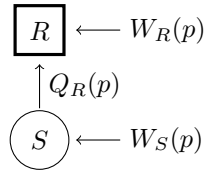
\begin{figure}[h]
\begin{center} 
	\begin{tikzpicture}[scale=.65]


	\draw[very thick] (.5,4.5) node[above, yshift=.05cm] {$R$}-- (1,4.5) -- (1,5.5) -- (0,5.5) -- (0,4.5) -- (.5,4.5); 
	\draw[] (.5,2.75) node[] {$S$} circle (.6cm);

	\draw[->] (2.15,5) node[right] {$W_R(p)$} -- (1.25,5);
	\draw[<-] (.5,4.3) -- (.5,3.4) node [above right] {$Q_R(p)$};
	\draw[->] (2.15,2.75) node[right] {$W_S(p)$} -- (1.25,2.75);
		
	\end{tikzpicture}
\end{center}
\caption{The setting for the Kelvin-Planck statement of the second law, Postulate~\ref{post:sec}.}
\label{fig:2lawcarnot}
\end{figure}

We consider a special case of the setting by assuming for the moment that the process $p$ is reversible in addition. Let $p^\mathrm{rev}\in\mathcal{P}_{R\vee S}$ be a reverse process of $p$. 
Both $p$ and $p^\mathrm{rev}$ fulfil the requirements stated in Postulate~\ref{post:sec}, that is, both are cyclic on $S$. Hence both must be such that no positive heat flows from $R$ to $S$, which is equivalent to $W_S(p)\geq0$ and $W_S(p^\mathrm{rev})\geq0$.
However, this implies that for reversible such processes
\begin{align}
W_{R\vee S}(p^\mathrm{rev})
= \underbrace{W_R(p^\mathrm{rev})}_{\geq0} + \underbrace{W_S(p^\mathrm{rev})}_{\geq0} 
= -W_{R\vee S}(p) 
=-\underbrace{W_R(p)}_{\geq0} - \underbrace{W_S(p)}_{\geq0} 
\end{align}
since the total work cost of a reverse process is the negative of the forward process.
This equality can only be fulfilled if all terms are zero, which implies that also 
$Q_R(p)=-Q_S(p)=0$ and $Q_R(p^\mathrm{rev})=-Q_S(p^\mathrm{rev})=0$.

We conclude that reversible processes on a composite system $R\vee S$ that are in addition cyclic on $S$ must have a trivial heat flow between $R$ and $S$, and that if the processes ought to be reversible, no work can be done on a reservoir. \\

The observation that doing work on a reservoirs is irreversible holds for arbitrary processes, not just the ones that fit the setting of the second law.
This meets our intuition, which says that reservoirs provide or take up heat, but not more than that. 
In fact, after having proved Carnot's Theorem in the next section, we will be able to make the even stronger statement that doing work in a heat reservoir is not only non-reversible but inefficient. 

\newpage
\section{Carnot's Theorem}
\label{sec:carnot}

\begin{center}
\fcolorbox{OliveGreen}{white}{\begin{minipage}{0.98\textwidth}
\centering
\begin{minipage}{0.95\textwidth}
\ \\
\textcolor{black}{\textbf{Postulates:}}
existence of reversible Carnot engines\\
\ \\
\textcolor{black}{\textbf{New notions:}}
Carnot's Theorem\\
\ \\
\textcolor{black}{\textbf{Technical results}} for this section can be found in Appendix~\ref{app:carnot}.

\paragraph{\textcolor{black}{Summary:}}
{Carnot's Theorem} states that any {machine operating between two heat reservoirs has a maximal efficiency} which only depends on the reservoirs, and that reversible processes are optimal.
This theorem will be explained and proved in this section.
\vspace{.1cm}

\end{minipage}
\end{minipage}}
\end{center}

\subsection{The setting}

This section is concerned with settings as 
depicted in Figure~\ref{fig:carnot}. Two reservoirs interact with a system under a work process 
$p\in\mathcal{P}_{R_1\vee S \vee R_2}$ such that the system undergoes cyclic evolution. 
The goal is to understand what values the work and heat flows can attain within the boundaries set by the laws of thermodynamics. Therefore we assume that $W_{R_i}(p)=0$ because we already know that heat can only be don on but never be drawn from a reservoir. 
Special interest will be given to the most efficient setting, in which heat is taken from one reservoir and given to the other, while doing as little work as necessary -- or put differently, while drawing as much work as possible.
The system $S$ is often called (Carnot) engine or machine.
Clearly, here it is assumed that all systems are pairwise disjoint, 
$R_1 \wedge S = \emptyset = R_2\wedge S = \emptyset = R_1\wedge R_2$. 
If they were not, the figure as well as the following discussion would make no sense.

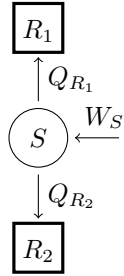
\begin{figure}[h]
\begin{center} 
	\begin{tikzpicture}[scale=.65]

		\draw[very thick] (.5,4.5) node[above] {$R_1$}-- (1,4.5) -- (1,5.5) -- (0,5.5) -- (0,4.5) -- (.5,4.5); 
		\draw[very thick] (.5,0) node[above] {$R_2$} -- (1,0) -- (1,1) -- (0,1) -- (0,0) -- (.5,0); 
		\draw[] (.5,2.75) node[] {$S$} circle (.6cm);

		\draw[<-] (.5,4.4) -- (.5,3.5) node [above right] {$Q_{R_1}$};
		\draw[<-] (.5,1.1) -- (.5,2) node[below right] {$Q_{R_2}$};
		\draw[->] (2.15,2.75) -- (1.25,2.75) node[above right] {$W_S$};
	
	\end{tikzpicture}
\end{center}
\caption{
A typical Carnot engine.
Two reservoirs $R_1$ and $R_2$ (not necessarily copies of each other) interact with another system $S$ through a (not necessarily reversible) work process $p\in\mathcal{P}_{R_1\vee S \vee R_2}$ that is cyclic on $S$.
The process dependence of the work and heat flows is omitted by writing 
$W_S := W_{S}(p)$, and $Q_{R_i} := Q_{R_i}(p)$.
For reversible processes $W_{R_i}(p)=0$ must hold.
As is shown in Lemma~\ref{lemma:carnotapp} in a reversible setting the heat flows always fulfil $Q_{R_2}(p)>0>Q_{R_1}(p)$ or $Q_{R_1}(p)>0>Q_{R_2}(p)$ (or $Q_{R_1}(p) = Q_{R_2}(p)=0$, but this is the trivial case).
}
\label{fig:carnot}
\end{figure}

The heat flows
$Q_{R_1}$ and $Q_{R_2}$ are defined such that they are positive when positive heat flows into the corresponding reservoir. Hence a negative heat flow means that positive heat is flowing into the cyclic system $S$. 
Consequently, the internal energy changes in terms of the quantities defined in Figure~\ref{fig:carnot} read
$\Delta U_{R_i} = Q_{R_i}$ and $0=\Delta U_S = W_S -Q_{R_1}-Q_{R_2}$.\\
%

Since Carnot engines, and in particular reversible ones, are central for the development of the notions of absolute temperature and entropy we postulate that between any two reservoirs there exists an engine and a non-trivial reversible process on the three systems.

\begin{postulate}[Existence of reversible Carnot engines]
\label{post:Carnotexist}
	Let $R_1,R_2\in\mathcal{R}$ be two reservoirs and $Q\in\mathbbm{R}$.
	Then there exists a system $S\in\mathcal{S}$ 
	and a reversible work process $p\in\mathcal{P}_{R_1\vee S\vee R_2}$, cyclic on $S$,
	with $Q_{R_1}(p)=Q$.
\end{postulate}

The reading of this statement is similar to the one of Postulate~\ref{post:copies} on the existence of arbitrarily many copies of any thermodynamic system.
It might be very difficult or only approximately possible to build a perfect reversible Carnot engine.
However, in principle the existence of such a machine is thinkable and this is what the theoretic considerations to come are based on. 

In traditional texts this postulate is usually not spelled out but implicitly taken for granted in the construction of the proof of Carnot's theorem. \\

Postulate~\ref{post:Carnotexist} will have important consequences when discussing absolute temperature as it will be necessary in order to define the temperature of \emph{all} heat reservoirs.
%
It turns out that without this, multiple incomparable definitions of absolute temperature can exist in parallel.
In the postulate it is not only asked for the existence of a reversible machine but also for a heat flow that one can choose to match $Q\in\mathbbm{R}$. 
Due to this requirement it is possible that the reversible heat flow $Q_{R_1}(p)$ (or $Q_{R_2}(p)$ due to the symmetry of the labels $1\leftrightarrow2$, but not both at the same time) can be chosen at will, which will become important for technical reasons.

Even though at this point it might appear that Postulate~\ref{post:Carnotexist} is very strong the discussion of ideal gases in Section~\ref{sec:exideal}, and in particular Figure~\ref{fig:pV}, show that if the theory of phenomenological thermodynamic is powerful enough to describe ideal gases with quasistatic processes (Section~\ref{sec:quasistatic}) then the postulate is easily satisfied.\\

Lemma~\ref{lemma:carnotapp} prepares the stage for Carnot's Theorem by showing that in any setting fulfilling the conditions from Figure~\ref{fig:carnot} at least one of the heat flows $Q_{R_i}$ is positive, or the process is \emph{trivial}, by which $Q_{R_1}=Q_{R_2}=0$ is meant.
This result holds independent of whether the process is reversible or irreversible.
A reversible process fulfils the stronger property that one of the heat flows is strictly positive while the other one is strictly negative. 
The proof makes direct use of the second law (Postulate~\ref{post:sec}) and the definition of heat reservoirs (Definition~\ref{def:heatreservoir}).\\


Even though this result basically only talks about positive and negative heat flows, it already says a lot about Carnot engines, in particular about reversible ones. This will become clear not least in the proof of Carnot's Theorem. 
In some sense, it strengthens the second law by saying that it is impossible to have two reservoirs pumping heat into a cyclic machine that generates work out of it. This is even true if the reservoirs are not copies of each other. 
On the other hand it is a precursor to Carnot's Theorem, which is based on Lemma~\ref{lemma:carnotapp} and goes beyond it by making a quantitative statement. 
Carnot's Theorem can be read as a statement about the efficiency of such machines.

\subsection{Carnot's Theorem}

\begin{thm}[Carnot's Theorem]
\label{thm:carnot}
Consider a machine $S\in\mathcal{S}$ and two reservoirs $R_1,R_2\in\mathcal{R}$ undergoing a reversible work process $p\in\mathcal{P}_{R_1\vee S \vee R_2}$ which is cyclic on $S$.
Let $S'\in\mathcal{S}$ be an additional machine operating between the same reservoirs under the work process $p'\in\mathcal{P}_{R_1\vee S \vee R_2}$, which is cyclic on $S'$, fulfils $W_{R_1}(p)=W_{R_2}(p)=0$, but is not necessarily reversible.
W.l.o.g.\ $Q_{R_2}(p)>0$ and $Q_{R_2}(p')>0$.\footnote{I.e.\ the processes are non-trivial and if the signs of $Q_{R_2}$ are not positive, swap the labels $1,2$ s.t.\ $Q_{R_2}(p')>0$ and then chose the reversible process $p$ to work in the direction in which $Q_{R_2}(p)>0$, too.}
Then
\begin{itemize}
	\item [(i)] the ratio of heat flows satisfies the inequality
	\begin{align}
	- \frac{Q_{R_1}(p')}{Q_{R_2}(p')} \leq - \frac{Q_{R_1}(p)}{Q_{R_2}(p)}\,, 
	\end{align}
	\item [(ii)] and the positively valued ratio $-\tfrac{Q_{R_1}(p)}{Q_{R_2}(p)}$ for reversible processes only depends on the reservoirs $R_1$ and $R_2$ and not on the machine $S$ nor the details of the process. It is \emph{universal} in this sense.
\end{itemize}

\end{thm}

The proof of Carnot's Theorem follows standard proofs from textbooks of undergraduate courses, see e.g.\ \cite{Feynman63}. 
The goal of the proof is to reduce the comparison of a reversible machine with a generic one to the situation discussed in Lemma~\ref{lemma:carnotapp} and thus to the setting of the second law. 
%

\begin{figure}[h]
\begin{center} 
	\begin{tikzpicture}[scale=.65]

	\draw[xshift = -4cm] (-1.5,5.5) node[left] {(a)};	
	
\begin{scope}[xshift=-4cm]
	\draw[very thick, xshift=-1cm] (.5,4.5) node[above] {$R_1$}-- (1,4.5) -- (1,5.5) -- (0,5.5) -- (0,4.5) -- (.5,4.5); 
	\draw[very thick, xshift=-1cm] (.5,0) node[above] {$R_2$} -- (1,0) -- (1,1) -- (0,1) -- (0,0) -- (.5,0); 

	\draw[] (-.5,2.75) node[] {$S$} circle (.6cm);
	\draw[<-, xshift=-2cm] (1.5,4.35) -- (1.5,3.45) node [above right] {$Q_{R_1}(p)$};
	\draw[<-, xshift=-2cm] (1.5,1.15) -- (1.5,2.05) node[below right] {$Q_{R_2}(p)$};
	\draw[->, xshift=-2cm] (3.15,2.75) node[right] {$W_{S}(p)$} -- (2.25,2.75);
\end{scope}		
	
	\draw[xshift = 2cm] (-1.5,5.5) node[left] {(b)};	

	\draw[very thick, xshift=1cm] (.5,4.5) node[above] {$R_1$}-- (1,4.5) -- (1,5.5) -- (0,5.5) -- (0,4.5) -- (.5,4.5); 
	\draw[very thick, xshift=1cm] (.5,0) node[above] {$R_2$} -- (1,0) -- (1,1) -- (0,1) -- (0,0) -- (.5,0); 
 
	\draw[] (1.5,2.75) node[] {$S'$} circle (.6cm);
	\draw[<-] (1.5,4.35) -- (1.5,3.45) node [above right] {$Q_{R_1}(p')$};
	\draw[<-] (1.5,1.15) -- (1.5,2.05) node[below right] {$Q_{R_2}(p')$};
	\draw[->] (3.15,2.75) node[right] {$W_{S'}(p')$} -- (2.25,2.75);
	
	\draw[xshift = 8cm] (-1.5,5.5) node[left] {(c)};	
	
\begin{scope}[xshift=8cm]
	\draw[very thick, xshift=-1cm] (.5,4.5) node[above] {$R_1$}-- (1,4.5) -- (1,5.5) -- (0,5.5) -- (0,4.5) -- (.5,4.5); 
	\draw[very thick, xshift=-1cm] (.5,0) node[above] {$R_2$} -- (1,0) -- (1,1) -- (0,1) -- (0,0) -- (.5,0); 

	\draw[] (-.5,2.75) node[] {$S\vee S'$} ellipse (1.2 cm and .6cm);
	\draw[<-, xshift=-2cm] (1.5,4.35) -- (1.5,3.45) node [above right] {$Q_{R_1}^\mathrm{tot}$};
	\draw[<-, xshift=-2cm] (1.5,1.15) -- (1.5,2.05) node[below right] {$Q_{R_2}^\mathrm{tot}$};
	\draw[->, xshift=-1.4cm] (3.15,2.75) node[right] {$W_{S}^\mathrm{tot}$} -- (2.25,2.75);
\end{scope}
	
	\end{tikzpicture}
\end{center}
\caption{We compare a reversible machine $S$ in (a) with an arbitrary one $S'$ in (b) operating between the same reservoirs. 
By letting the cycles run many times, we construct a process depicted in (c), for which we analyse the total work and heat flows.
Depending on the ratios of the heat flows in $p$ and $p'$ the constructed process in (c) may violate Lemma~\ref{lemma:carnotapp} and thus the second law (Postulate~\ref{post:sec}).
}
\label{fig:carnotproof}
\end{figure}
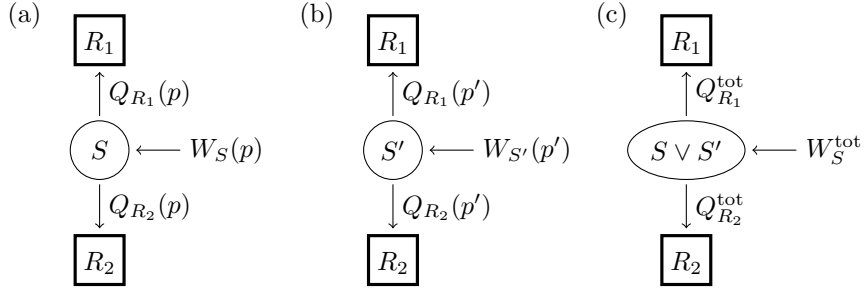

\begin{proof}
We prove (i) by contradiction. Figure~\ref{fig:carnotproof} shows the discussed situations.
Later on we argue that (i) implies (ii).
Suppose the ratios fulfilled 
\begin{align}
\label{eq:carnot1}
	- \frac{Q_{R_1}(p')}{Q_{R_2}(p')} > - \frac{Q_{R_1}(p)}{Q_{R_2}(p)}\,.
\end{align}
Since $Q_{R_1}(p)$ and $Q_{R_2}(p)$ must have opposite signs (Lemma~\ref{lemma:carnotapp}) the right hand side of Eq.~(\ref{eq:carnot1}) is strictly positive. Hence this case can only occur for $Q_{R_1}(p')<0$ and Eq.~(\ref{eq:carnot1}) is equivalent to $\tfrac{Q_{R_2}(p)}{Q_{R_2}(p')} > \tfrac{Q_{R_1}(p)}{Q_{R_1}(p')}$, which again compares positive ratios.
Choose positive integers $k,l\in\mathbbm{N}$ such that 
\begin{align}
\frac{Q_{R_2}(p)}{Q_{R_2}(p')} > \frac{k}{l} > \frac{Q_{R_1}(p)}{Q_{R_1}(p')}\,.
\end{align}
The existence of $p$ and $p'$ implies together with Definition~\ref{def:heatreservoir} (iii) that one can apply the respective process as many times as one wants in a row, in the sense that there exist corresponding processes that do the same and can be concatenated with $p$ and $p'$ respectively. 
Thus, apply the reverse process of $p$ on $R_1\vee S\vee R_2$ now $l$ times followed by $k$ applications of $p'$ on $R_1\vee S'\vee R_2$.
The so constructed process is cyclic on $S\vee S'$ and has total heat flows to $R_1$ and $R_2$ of
\begin{align}
Q_{R_1}^\mathrm{tot} = -l\, Q_{R_1}(p) + k\, Q_{R_1}(p') 
= \underbrace{\left(\frac{k}{l}-\frac{Q_{R_1}(p)}{Q_{R_1}(p')}\right)}_{> 0} 
  \cdot \underbrace{l\, Q_{R_1}(p')}_{< 0} < 0\,, \\
Q_{R_2}^\mathrm{tot} = -l\, Q_{R_2}(p) + k\, Q_{R_2}(p') 
= \underbrace{\left(\frac{k}{l}-\frac{Q_{R_2}(p)}{Q_{R_2}(p')}\right)}_{< 0} 
  \cdot \underbrace{l\, Q_{R_2}(p')}_{> 0} < 0\,.
\end{align}
This contradicts Lemma~\ref{lemma:carnotapp}, hence the ratios must fulfil (i).

If $p'$ is reversible too, the argument also works with exchanged roles of $p$ and $p'$ and we obtain the inequality in the other direction. Therefore, if both $p$ and $p'$ are reversible, the ratios must be equal. 
Obviously the ratios for reversible machines must be positive as the signs of the heat flows are always different.
This proves (ii).
\end{proof}

Carnot's Theorem implies that reversible engines are the most efficient ones.
Suppose the engine is such that work is extracted, which means that the total work done should be negative. 
The internal energy of the cyclic engine undergoes no net change, which is why the extracted work from $S$ reads
\begin{align}
\label{eq:carnotwork}
-W_S = -Q_{R_1}  -Q_{R_2} 
= Q_{R_2} \left( -\frac{Q_{R_1}}{Q_{R_2}}-1\right) \,.
\end{align}
Remember that $Q_{R_2}$ is positive, which is why the expression for the \emph{extracted} work $-W_S$ is \emph{maximised} when the term in brackets is maximal. This, however, is the case for reversible machines, as Carnot's Theorem states. 

The argument is also valid if positive work is used such that the machine pumps heat from one reservoir to the other. With the same calculation it follows that the work \emph{done} is \emph{minimal} when the ratio of interest is maximal.

Eq.~(\ref{eq:carnotwork}) also shows how the theorem limits the maximal efficiency of reversible machines. 
Given two reservoirs, the work extractable in relation to the heat given to a reservoir is always upper bounded by $-\tfrac{Q_{R_1}}{Q_{R_2}}-1$. Whatever machine one can come up with, this finite bound cannot be surpassed for given reservoirs. 
We conclude that Carnot's Theorem sets limits to what can be achieved with cyclic machines operating between two reservoirs. \\

A further peculiarity of the theorem as stated in this work is the fact that it could be derived without referring to anything similar to the zeroth law of thermodynamics nor to an a priori notion of thermal equilibrium. 
This is the topic of an earlier paper \cite{Kammerlander18}.
In short, the zeroth law requires that the relation ``being in thermal equilibrium with'' is transitive. In many introductions to thermodynamics it is taken to be a vital ingredient to the foundations of thermodynamics as it paves the way for the notion of an empirical temperature, which relies on the relation ``being in thermal equilibrium with'' to be an equivalence relation. 
In this work neither of these notions ever had to be used up to this point. Consequently, we could not even think of formulating (and making use of) the zeroth law. 
Therefore, up to here (as well as also for the remainder of this work, as will become clear) the zeroth law is redundant. \\

The most relevant implication of the theorem, however, is the fact that for reversible processes, the ratio is universal. This is the crucial statement which is used next when defining absolute temperature for heat reservoirs, and later in the definition of thermodynamic entropy. 
It is the ingredient necessary to make statements for the behaviour of classes of systems rather than individual ones, as it says that the ratio is ``system-independent''.

\newpage
\section{Absolute temperature}
\label{sec:abstemp}

\begin{center}
\fcolorbox{OliveGreen}{white}{\begin{minipage}{0.98\textwidth}
\centering
\begin{minipage}{0.95\textwidth}
\ \\
\textcolor{black}{\textbf{Postulates:}}
-\\
\ \\
\textcolor{black}{\textbf{New notions:}}
temperature ratio, absolute temperature, equivalence relation $\sim$ on $\mathcal{R}$\\
\ \\
\textcolor{black}{\textbf{Technical results}} for this section can be found in Appendix~\ref{app:abstemp}.

\paragraph{\textcolor{black}{Summary:}}
Based on Carnot's Theorem the {absolute temperature of a reservoir} is defined.
From this follows an equivalence relation on the set of heat reservoirs which shows that the {zeroth law}, which was not postulated in this framework, can be derived and is hence {redundant as a postulate}.
\vspace{.1cm}

\end{minipage}
\end{minipage}}
\end{center}

\subsection{Preparatory remarks}

As one would expect, the ratio $-\tfrac{Q_{R_1}}{Q_{R_2}}$ is not only independent of the actual reversible machine but also of the representative reservoirs of the equivalence class of $\hat=$. This is argued at the beginning of Appendix~\ref{app:abstemp}, where one can also find the proofs of the other non-trivial technical statements of this section. 

In addition, a reversible Carnot engine operating between reservoirs $R_1\hat= R_2$ fulfils $-\tfrac{Q_{R_1}}{Q_{R_2}}=1$, as is proved in Lemma~\ref{lemma:ratioequivapp}. 
The following definition makes use of the invariance of the ratio of reversible heat flows under equivalences. 

\begin{definition}[Temperature ratio]
\label{def:tau}
The \emph{temperature ratio $\tau$ of two equivalence classes} $[R_1],[R_2]\in\sfrac{\mathcal{R}}{\hat=}$ is
\begin{align}
\label{eq:taudef}
\begin{split}
\tau: \ \sfrac{\mathcal{R}}{\hat=}\times\sfrac{\mathcal{R}}{\hat=}  &\longrightarrow \mathbbm{R}_{>0}\\
([R_1],[R_2]) &\longmapsto -\tfrac{Q_{R_1}}{Q_{R_2}}\,
\end{split}
\end{align}
where $R_1$ and $R_2$ are two different\footnote{
A Carnot engine operates between two \emph{different} heat reservoirs. If $R_1=R_2$, then $R_1\vee S \vee R_2 = R_1\vee S$ and there would be no two heat flows to compare.}
heat reservoirs and $Q_{R_1}$ and $Q_{R_2}$ are the heat flows of a (non-trivial) reversible Carnot engine operating between them such that $Q_{R_2}>0$. 
Extending this definition to the set of pairs of heat reservoirs, we use the same symbol $\tau$ and write for the \emph{temperature ratio of two heat reservoirs}
\begin{align}
\label{eq:taudef2}
\begin{split}
\tau(R_1,R_2):=\tau([R_1],[R_2])\,.
\end{split}
\end{align}
\end{definition}

Definition~\ref{def:tau} is well-defined since the ratio of heat flows only depends on the equivalence classes of $\hat=$ and not on the specific representative.
Furthermore, if the equivalence classes $[R_1]=[R_2]$ are equal, there always exist different representatives. This is a consequence of the postulate on the existence of arbitrarily many copies of any system, Postulate~\ref{post:copies}.

Carnot's Theorem~\ref{thm:carnot} implies that it is irrelevant which reversible engine is used to determine the value of $\tau$ and that $\tau$ is a strictly positive function.
Furthermore, due to Postulate~\ref{post:Carnotexist} on the existence of reversible Carnot engines this definition can be applied to an arbitrary pair of reservoirs.\\

In principle, it is thinkable to formulate a theory of thermodynamics without the assumption of always having a reversible machine operating non-trivially between any two reservoirs. 
Without it, it could happen that there exist two reservoirs that are incomparable. 
Consequently, Definition~\ref{def:tau} would have to be phrased independently for the two (or more) classes of reservoirs that are comparable with each other. 
This is not a fundamental problem for thermodynamics, but it is not the common approach taught in introductions to the field. \\




The name temperature ratio of $\tau$ is suggestive. $\tau$ will be used to define the absolute temperature for heat reservoirs. Before introducing this notion we discuss some properties of $\tau$.
First, it holds $\tau([R],[R])=1$, which means that if one takes reservoirs of the same type, it is only possible to pump heat from one to the other, without investing or drawing work (Lemma~\ref{lemma:ratioequivapp}).

Second, we note that $\tau(R_2,R_1)=\tau(R_1,R_2)^{-1}$ by definition. Reversing the order in the argument of $\tau$ amounts to reversing the reversible Carnot process used in Definition~\ref{def:tau} to determine the value of $\tau$. 
Since reverse heat flows simply change their signs in the reverse process, the heat flows $Q_{R_1}<0$ and $Q_{R_2}>0$ in the forward process become $-Q_{R_1}>0$ and $-Q_{R_2}<0$ in the reverse process. Since the denominator must be the negative heat flow according to the definition, we obtain 
\begin{align}
\tau(R_1,R_2)^{-1} = \left(-\frac{Q_{R_1}}{Q_{R_2}}\right)^{-1} = -\frac{-Q_{R_2}}{-Q_{R_1}} = \tau(R_2,R_1)\,.
\end{align}
Finally, in Lemma~\ref{lemma:tautransapp}
we show that for three arbitrary heat reservoirs $R_1,R_2,R_3\in\mathcal{R}$ it always holds $\tau(R_1,R_2)\cdot\tau(R_2,R_3) = \tau(R_1,R_3)$. 

\subsection{Absolute temperature for heat reservoirs}

The properties of $\tau$ allow us to define the absolute temperature of an arbitrary heat reservoir. For this, choose an arbitrary but fixed reference heat reservoir $R_\mathrm{ref}\in\mathcal{R}$ and a reference temperature $T_\mathrm{ref}\in\mathbbm{R}_{>0}$.

\begin{definition}[Absolute temperature]
\label{def:abstempres}
The \emph{absolute temperature of a heat reservoir} $R\in\mathcal{R}$ is defined as
\begin{align}
T := \tau(R,R_\mathrm{ref}) \cdot T_\mathrm{ref}\,.
\end{align}
\end{definition}

The temperature of a reservoir is absolute up to the choice of the reference reservoir and the reference temperature. However, once this choice is made, any other reservoir with its temperature could serve as a reference, too. 
This is a consequence of the previous Lemma~\ref{lemma:tautransapp}, as for two reservoirs $R_1,R_2\in\mathcal{R}$ it holds that 
\begin{align}
T_2 = \tau(R_2,R_\mathrm{ref})\cdot T_\mathrm{ref} 
= \tau(R_2,R_1)\cdot\tau(R_1,R_\mathrm{ref}) \cdot T_\mathrm{ref}
= \tau(R_2,R_1) \cdot T_1\,.
\end{align}
That is what makes $\tau$ a temperature ratio, as it is called in Definition~\ref{def:tau}.\\

Typically physicists work with the absolute temperature scale such that a reservoir consisting of a big water tank has temperature $273.16\,\mathrm{K}$ at the tripe point of water. 
From a practical perspective it makes sense to fix a temperature scale once and for all so that when comparing temperatures no confusion can arise. 
Nevertheless, for developing the theory making specific choices according to some standard is not necessary. 
Therefore, we will not discuss this issue further and continue, knowing that both $R_\mathrm{ref}$ and $T_\mathrm{ref}$ have been fixed before defining absolute temperature for heat reservoirs. \\

Having defined absolute temperature for reservoirs, we can investigate the relation ``$R_1$ and $R_2$ have equal temperature'' on the set of heat reservoirs $\mathcal{R}$, denoted by 
$R_1\sim R_2$ (Definition~\ref{def:simapp}).
Intuitively, one can also call this relation ``being in thermal equilibrium with''.
By Definition~\ref{def:abstempres} $R_1$ and $R_2$ are at the same temperature if and only if $\tau(R_1,R_2)=1$. Several conclusions can be drawn from this observation.

First, the temperature of a reservoir is independent of its state. It is rather a property of the system. When thinking of general thermodynamic systems, this seems odd. But for reservoirs this is what one would expect, as the property of not changing its behaviour under finite changes of energy intuitively is a defining property of heat reservoirs.

Second, from the definition of absolute temperature and Lemma~\ref{lemma:ratioequivapp}
it follows that equivalent reservoirs must have the same absolute temperature. i.e.\ $R_1\hat=R_2 \ \Rightarrow \ R_1\sim R_2$.
Hence, the relation $\hat=$ restricted to the set of heat reservoirs $\mathcal{R}$ is a \emph{sub-relation} of $\sim$.

Third, 
it holds
\begin{itemize}
\item [(i)] $R_1\sim R_1$ (reflexive),
\item [(ii)] $R_1\sim R_2 \Rightarrow R_2 \sim R_1$ (symmetric), and
\item [(iii)] $R_1\sim R_2$ as well as $R_2\sim R_3 \Rightarrow R_1\sim R_3$ (transitive).
\end{itemize} 
Here, (i) and (ii) follow directly from Definition~\ref{def:abstempres}, while
(iii) is a consequence of Lemma~\ref{lemma:tautransapp}.
Points (i)-(iii) say that $\sim$ is an equivalence relation (Lemma~\ref{lemma:simequirelapp}).

Finally, heat reservoirs fulfilling $\tau(R_1,R_2)=1$ always allow for the exchange of an arbitrary amount of heat $Q$ between them during a reversible work process $p\in\mathcal{P}_{R_1\vee R_2}$ at zero work cost. 
This follows from the fact that any reversible machine operating between them has a heat flow ratio of $-\tfrac{Q_{R_1}}{Q_{R_2}}=1$ together with Postulate~\ref{post:freedom} on the freedom of description and Postulate~\ref{post:Carnotexist}.

\subsection{The zeroth law of thermodynamics is redundant} 
\label{sec:zerothlaw}

The fact that $\sim$ (``having equal temperature'') is an equivalence relation makes postulating the zeroth law redundant.
The zeroth law typically states that the relation ``being in thermal equilibrium with'' is transitive \cite{Maxwell71,Fowler39,Planck14,Buchdahl66}.
Notice that for a postulate stating this, a notion of ``thermal equilibrium'' must be introduced beforehand -- something we did not have to do either.
Together with (the usually implicitly assumed) reflexivity and symmetry of $\sim$, the zeroth law makes it an equivalence relation. 
Typically, this is then used to say that two systems are at equal temperature if and only if they are in thermal equilibrium relative to each other. 
Here we have derived such an equivalence relation on the set of reservoirs from the postulates without any reference to the zeroth law.\\

One may wonder why we have not yet discussed a definition of absolute temperature for arbitrary systems, not just heat reservoirs. 
As it turns out, we should not expect that temperature is a quantity that makes sense for an arbitrary system without further assumptions. Composite systems consisting of more than two different atomic subsystems are simple counter examples.
Temperature for arbitrary systems is thus not a fundamental concept of thermodynamics. 

Instead of a notion of temperature of arbitrary systems, we introduce the notion of the \emph{temperature of a heat flow} in the next section (Section~\ref{sec:theatflow}). This is the fundamental concept necessary to use the notion of temperature beyond the purpose of the efficiency of engines and the temperature of heat reservoirs.
In particular, the temperature of heat flows is the one which is used to define thermodynamic entropy based on Clausius' Theorem. 
\\

One might have hopes to be able to define absolute temperature for all atomic systems, since they are indivisible.
However, one should not take the statement about indivisibility of atomic systems as a statement about physical indivisibility. In particular, one should not confuse the term \emph{atomic} with the frequently used term ``simple system'' (see e.g.\ \cite{LY99}). Simple systems are usually considered to be those systems to which a meaningful notion of temperature can be assigned. This implies that they cannot be composed of two independent systems (otherwise at least two temperatures would be necessary in general). 
For us on the other hand, the term ``indivisible'' must be understood relative to the structure of thermodynamics systems with composition $\vee$, which can be defined in very abstract terms, not referring to any thermodynamic ideas. \\

The discussion on how to define a notion of absolute temperature for systems beyond heat reservoirs is continued in Section~\ref{sec:tsys}, when all other basic concepts have been introduced and a technical investigation based on them is possible.

\newpage
\section{The temperature of heat flows}
\label{sec:theatflow}

\begin{center}
\fcolorbox{OliveGreen}{white}{\begin{minipage}{0.98\textwidth}
\centering
\begin{minipage}{0.95\textwidth}
\ \\
\textcolor{black}{\textbf{Postulates:}}
-\\
\ \\
\textcolor{black}{\textbf{New notions:}}
temperature of heat flows\\
\ \\
\textcolor{black}{\textbf{Technical results}} for this section can be found in Appendix~\ref{app:theatflow}.

\paragraph{\textcolor{black}{Summary:}}
The absolute temperature for heat reservoirs gives rise to a definition of the {temperature of heat flows}. 
The {uniqueness of this temperature} is discussed, in particular with respect to reversible heat flows.
\vspace{.1cm}

\end{minipage}
\end{minipage}}
\end{center}

\ \\

Being able to talk about the temperature of heat reservoirs is important but not enough.
In particular when it comes to defining thermodynamic entropy the temperature of heat flows will be  essential. 

As will be explained and investigated in this section, not every heat flow occurs at some temperature. Some do, however, and the reversible ones will be of particular interest. 

\begin{definition}[Heat at temperature $T$]
\label{def:heatatt}
Let $S=S_1\vee S_2\in\mathcal{S}$ be composed of two disjoint subsystems and undergo an arbitrary work process $p\in\mathcal{P}_{S_1\vee S_2}$ with $Q:= Q_{S_2}(p)\neq0$.
We say that the \emph{heat $Q$ flows at temperature $T$} if
there exist two different reservoirs $R_1\sim R_2$ at temperature $T$ with processes $p_1\in\mathcal{P}_{S_1\vee R_1}$ and $p_2\in\mathcal{P}_{S_2\vee R_2}$ s.t.\ 
$W_{A}(p_i) = W_{A}(p)$ for all atomic systems 
$A\in\mathcal{A}\smallsetminus\mathrm{Atom}(S_{i+1})$
and the state changes on $S_i$ under $p_i$ are the same as under $p$, i.e.\ 
$\lfloor p_i \rfloor_{S_i} = \lfloor p \rfloor_{S_i}$ and 
$\lceil p_i \rceil_{S_i} = \lceil p \rceil_{S_i}$.
\end{definition}

Figure~\ref{fig:heatatt} (a), (b) and (c) illustrate the processes $p$, $p_1$ and $p_2$ schematically. By definition of $p_i$ as work processes on $S_i\vee R_i$ it is clear that $W_{S_i}(p_{i+1})=0$.
Furthermore, it holds $W_{R_i}(p_i)=0$ by assumption in the definition. 

The idea behind Definition~\ref{def:heatatt} is to refer to the already defined concept of the temperature of a reservoir to define the temperature of a heat flow. 
A heat flow occurs at temperature $T$ if it could also be exchanged with a reservoir at temperature $T$ and no thermodynamic properties change. 
Even though this may make sense intuitively it is \emph{a priori} not clear that this definition eventually leads to the correct notion of temperature that is needed for the definition of thermodynamic entropy. However, as we will show in the following, it does so.\\

It is possible to `reconstruct' the process $p$ from the $p_i$ in case $p$ satisfies Definition~\ref{def:heatatt} for some $T$.
Together with Definition~\ref{def:heatreservoir} (i) for heat reservoirs and Postulate~\ref{post:Carnotexist} on the existence of reversible Carnot engines it follows that if both $p_1$ and $p_2$ exist, then the initial states of the two reservoirs $R_1$ and $R_2$ can be restored by a reversible Carnot engine transferring the heat $Q$ from $R_1$ to $R_2$. 
The effect of the proper concatenations of these processes is depicted in Figure~\ref{fig:heatatt} (d).
The work cost of reversibly transferring the heat $Q$ from $R_1$ to $R_2$ by means of a cyclic engine $C$ is zero. This follows from Carnot's Theorem and Lemma~\ref{lemma:ratioequivapp}.\footnote{
Using the tuning of $Q_{R_1}$ in reversible Carnot processes as required by Postulate~\ref{post:Carnotexist} it becomes obvious that for two equivalent reservoirs $R_1\hat= R_2$ and $Q\in\mathbbm{R}$ there is always a reversible process $p\in\mathcal{P}_{R_1\vee S\vee R_2}$, involving an adequate machine $S$, that transfers the heat $Q_{R_1}(p) = Q = -Q_{R_2}(p)$ from $R_2$ to $R_1$.
In such a situation we have $W_S(p)=0$ because of the exactly opposite heat flows to the reservoirs. Since $S$ is cyclic in any case, in this special situation $S$ is even \emph{catalytic} and can be left out of the description (Postulate~\ref{post:freedom}). 
That is, there exists a process $\tilde p\in\mathcal{P}_{R_1\vee R_2}$ with the same heat flows and state changes. 

We conclude that for two equivalent reservoirs there is always a reversible process transferring an arbitrary amount of heat from to the other at no work cost. 
This observation will be relevant later, when we discuss the temperature of heat flows.}
%
We have thereby constructed a process from $p_1$ and $p_2$ which thermodynamically does the same as $p$ did on all involved systems. \\

\begin{figure}
\begin{center} 
	\begin{tikzpicture}[scale=.65]
	
	\draw[](-1.8,0) node[] {(a)};
	 
	\draw[xshift=.1cm] (.5,-.5) node[above] {$S_1$} -- (1,-.5) -- (1,.5) -- (0,.5) -- (0,-.5) -- (.5,-.5);
	\draw[] (3,-.5) node[above] {$S_2$} -- (3.5,-.5) -- (3.5,.5) -- (2.5,.5) -- (2.5,-.5) -- (3,-.5);
	\draw[->] (1.3,0) node[above right] {$Q$} -- (2.3,0);
	\draw[->] (-1,0) node[above right] {$W_1$} -- (-.1,0);
	\draw[->] (4.6,0) node[above left] {$W_2$} -- (3.7,0);

	\begin{scope} [xshift = 9cm]
	
		\draw[](-1.8,0) node[] {(b)};
		
		\draw[xshift=.1cm] (.5,-.5) node[above] {$S_1$} -- (1,-.5) -- (1,.5) -- (0,.5) -- (0,-.5) -- (.5,-.5);
		\draw[very thick] (3,-.5) node[above] {$R_1$} -- (3.5,-.5) -- (3.5,.5) -- (2.5,.5) -- (2.5,-.5) -- (3,-.5);
		\draw[->] (1.3,0) node[above right] {$Q$} -- (2.3,0);
		\draw[->] (-1,0) node[above right] {$W_1$} -- (-.1,0);
		
		\begin{scope}[xshift = 6cm]

		\draw[](-.6,0) node[] {(c)};

		\draw[xshift=.1cm, very thick] (.5,-.5) node[above] {$R_2$} -- (1,-.5) -- (1,.5) -- (0,.5) -- (0,-.5) -- (.5,-.5);
		\draw[] (3,-.5) node[above] {$S_2$} -- (3.5,-.5) -- (3.5,.5) -- (2.5,.5) -- (2.5,-.5) -- (3,-.5);
		\draw[->] (1.3,0) node[above right] {$Q$} -- (2.3,0);
		\draw[->] (4.6,0) node[above left] {$W_2$} -- (3.7,0);
		\end{scope}
		
		\begin{scope}[xshift = -5cm, yshift = -3cm]

		\draw[](-1.8,0) node[] {(d)};
		
		\draw[xshift=.1cm] (.5,-.5) node[above] {$S_1$} -- (1,-.5) -- (1,.5) -- (0,.5) -- (0,-.5) -- (.5,-.5);
		\draw[very thick] (3,0) node[] {$R_1$} circle (.6cm);
		\draw[] (5.4,0) node[] {$C$} circle (.6cm);
		\draw[->] (1.3,0) node[above right] {$Q$} -- (2.3,0);
		\draw[->] (-1,0) node[above right] {$W_1$} -- (-.1,0);
		\draw[->] (3.8,0) node[above right]{$Q$} -- (4.7,0);
		\draw[->] (6.1,0) node[above right]{$Q$} -- (7,0);		
		
		\begin{scope}[xshift = 7.2cm]
		\draw[very thick] (.5,0) node[] {$R_2$} circle (.6cm);
		\draw[] (3,-.5) node[above] {$S_2$} -- (3.5,-.5) -- (3.5,.5) -- (2.5,.5) -- (2.5,-.5) -- (3,-.5);
		\draw[->] (1.3,0) node[above right] {$Q$} -- (2.3,0);
		\draw[->] (4.6,0) node[above left] {$W_2$} -- (3.7,0);
		\end{scope}
		
		\end{scope}
	\end{scope}
	\end{tikzpicture}
\end{center}
\caption{(a) Under the process $p$ work $W_i := W_{S_i}(p)$ is done on $S_i$ and the heat 
$Q := Q_{S_2}(p)$ flows from $S_1$ to $S_2$.
(b) \& (c) In the divided processes $p_1$ and $p_2$ the same state changes are induced on $S_1$ and $S_2$, respectively, but the heat flow is exchanged with reservoirs. The work flows 
$W_i = W_{S_i}(p_i)$ also stay invariant.
(d) With the use of a reversible Carnot process between $R_1$ and $R_2$ it is possible to transfer the heat $Q$ from $R_1$ to $R_2$ at zero work cost since $R_1\sim R_2$. The initial states of $R_1$ and $R_2$ will then be restored as they only depend on the internal energies and the engine itself is cyclic by design. Together with the observation that $R_1\vee C\vee R_2$ is catalytic this essentially reconstructs $p$ due to Postulate~\ref{post:freedom}. 
}
\label{fig:heatatt}
\end{figure}
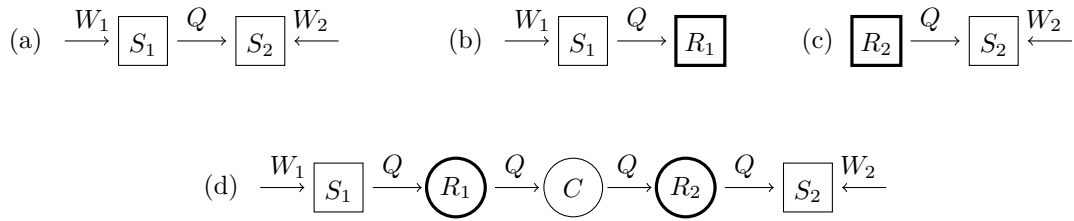

We describe some examples to give an intuition for this definition.

\begin{example}[Heat flow exchanged with a reservoir]
\label{ex:heatflowres}
The most obvious example is the one in which an arbitrary system $S$ exchanges heat $Q\neq 0$ with a reservoir $R$.
So let $p\in\mathcal{P}_{S\vee R}$ be such that $Q=Q_R(p)\neq0$ and think of $S_1=S$ and $S_2=R$. 
Then, the heat $Q$ flows at temperature $T$, where $T$ is the temperature of the reservoir $R$. 
To see this, observe that what is referred to $p_1$ in Definition~\ref{def:heatatt} can now be chosen as $p$ itself (or an equivalent process carried out on $S$ and a copy of $R$). 
On the other hand, Postulate~\ref{post:Carnotexist} says that a process that transfers heat $Q$ between $R$ and a copy of it always exists. By choosing such a process as $p_2$ we have found both $p_1$ and $p_2$ with reservoirs that are copies of $R$ and thus have the same temperature. Hence the heat $Q$ flows at temperature $T$ according to the definition.
\end{example}


Importantly, Definition~\ref{def:heatatt} on the temperature of heat flows principally allows that a general heat flow can be assigned more than one temperature, see the next Example~\ref{ex:heatdifft}. 
Even when heat is exchanged with a reservoir directly, the temperature of the reservoir does not have to be the unique temperature which can be assigned to the heat flow. 

\begin{example}[More than one temperature for the same heat flow]
\label{ex:heatdifft}
Let $S_1$ and $S_2$ be two ideal gases with the same amount of substance in states $(p_i,V_i)$, respectively, such that $p_1V_1 > p_2V_2$. This essentially means that their gas temperatures $T_i$ are different, where the gas temperature can in this case be defined according to the equation of state 
$pV=nRT$ with $R$ the universal gas constant and $n$ the amount of substance.
When connecting the two gases thermally, e.g.\ by putting them in contact with a metal rod, one can observe a positive heat flow from $S_1$ to $S_2$ (provided that the reference temperature $T_\mathrm{ref}$ for the definition of absolute temperature has be chosen accordingly).
It will now be observable that reservoirs with temperatures between $T_1>T_2$ will fulfil the above definition. In particular, more than one temperature can be assigned to the same heat flow. 
A more precise discussion of the concepts mentioned here (equation of state, temperature of the gas) is given in Section~\ref{sec:exideal}.
\end{example}

It is also possible that no temperature can be assigned to a heat flow at all, as the next example shows. 

\begin{example}[Heat flows without a temperature]
\label{ex:heatrevdifft}
Consider again an ideal gas denoted by $S_1$ in state $(p,V)$. This time, let $S_2=R_1 \vee R_2$ be a system composed of two reservoirs at temperatures $T_1<T_2$. 
Assume in addition, that $pV=nRT_1$, i.e.\ that the initial state of the ideal gas is such that its gas temperature matches the absolute temperature of reservoir $R_1$. 
This means that a reversible \emph{isothermal} compression, say from $V$ to $\tfrac{V}{2}$, can be achieved by thermally connecting the ideal gas to $R_1$ and slowly compressing the gas. 
Denote 
the heat flow from $S_1$ to $R_1$ by $Q_1$.
Checking Definition~\ref{def:heatatt} it follows that the heat $Q_1$ flows at temperature $T_1$.

Now, the same state change on the ideal gas $S_1$ could be achieved reversibly in a different way, by exchanging a different amount of heat with the other reservoir $R_2$. 
We first reversibly compress the ideal gas \emph{adiabatically}, thereby increasing the product $pV$ to the value $nRT_2$. This is followed by a (reversible) isothermal compression (or an expansion, if $T_2$ is so much hotter than $T_1$) in which $S_1$ isothermally exchanges the heat $Q_2$ with $R_2$. The isothermal compression (expansion) is done such that in the final state, there exists the reversible adiabatic process that brings the gas temperature back to $T_1$ while the volume has been brought to $\tfrac{V}{2}$ such that the net state change on $S_1$ is the same as under process described in the previous paragraph. 

By computing the amount of heat exchanged using the standard equations for the ideal gas, it is easy to see that for $T_1<T_2$ it follows $Q_1<Q_2$.
From this observation we learn several things. First, it is possible to have the same state change on a system (here $S_1$) having exchanged strictly different reversible heat flows. Second, the temperature of the involved heat flows do not have to match either. Third, since the two processes are reversible, we can concatenate the one with a reverse of the other, which will generally not lead to a net heat flow of zero. Neither will this non-zero heat flow allow for an assigned temperature. On the contrary, since it is the result of two heat flows at different temperatures it is obvious that it should be impossible to assign a temperature. Hence there exist reversible heat flows without a proper temperature. 
\end{example}

As is shown 
in Lemma~\ref{lemma:heatratioQTapp}, the different reversible heat flows $Q_i$ at temperature $T_i$ from the previous example, which induce the same state change on $S_1$, fulfil $\tfrac{Q_1}{T_1} = \tfrac{Q_2}{T_2}$.
This observation will be crucial for the definition of thermodynamic entropy.\\

Having discussed these cases of Definition~\ref{def:heatatt} we come to the more important case of a general reversible heat flow which can be assigned a temperature. In this case it holds that the temperature is unique, as is shown in Appendix~\ref{app:theatflow}. 
This is done by first showing that for a reversible process $p\in\mathcal{P}_{S_1\vee S_2}$ inducing a heat flow $Q\neq0$ at temperature $T$ according to Definition~\ref{def:heatatt} the two processes $p_i\in\mathcal{P}_{S_i\vee R_i}$ are reversible, too (Lemma~\ref{lemma:revpiapp}). 
Based on this, it is then possible to show that the assigned temperature $T$ must be unique (Lemma~\ref{lemma:revheatuniquetapp}).
Hence, non-zero reversible heat flows in a bipartite thermodynamic system can be assigned either a unique temperature or no temperature at all.

\newpage
\section{Clausius' Theorem and thermodynamic entropy}
\label{sec:clausius}

\begin{center}
\fcolorbox{OliveGreen}{white}{\begin{minipage}{0.98\textwidth}
\centering
\begin{minipage}{0.95\textwidth}
\ \\
\textcolor{black}{\textbf{Postulates:}}
existence of reversible processes with heat flows at well-defined temperatures\\
\ \\
\textcolor{black}{\textbf{New notions:}}
Clausius' Theorem, entropy, Entropy Theorem\\
\ \\
\textcolor{black}{\textbf{Technical results}} for this section can be found in Appendix~\ref{app:clausius}.

\paragraph{\textcolor{black}{Summary:}}
The derivation of {Clausius' Theorem} allows one to define {thermodynamic entropy} and show that it is a well-defined state function, which is additive under composition. 
A postulate guarantees that the entropy difference of any two states of a system can be computed.
The {Entropy Theorem} establishes that entropy is a monotone for the preorder $\rightarrow$ (as defined after the first law).
\vspace{.1cm}

\end{minipage}
\end{minipage}}
\end{center}

\subsection{Clausius' Theorem}

So far we have used the second law to prove Carnot's Theorem which was the starting point to define the temperature of reservoirs and heat flows. 
Next, we use these insights to prove Clausius' Theorem, which is the basis for the definition of thermodynamic entropy as a state variable.


\begin{thm}[Clausius' Theorem]
\label{thm:clausius}
Let $S\in\mathcal{S}$ be an arbitrary system and $\{R_i\}_{i=1}^N$ a set of reservoirs such that the temperature of $R_i$ is $T_i$.
For each $i=1,\dots,N$ let $p_i\in\mathcal{P}_{S\vee R_i}$ be a work process on $S$ and the reservoir $R_i$ with $W_{R_i}(p)=0$ such that the concatenated process 
$p:= p_N\circ\cdots\circ p_1$ is defined and in total cyclic on $S$.\\
Then:
\begin{itemize}
\item [(i)] $\sum_i \tfrac{Q_S(p_i)}{T_i} \leq 0$,
\item [(ii)] and if $p$ is reversible, then $\sum_i \tfrac{Q_S(p_i)}{T_i} = 0$.
\end{itemize}
\end{thm}


\begin{proof}\ \\
\vspace{-4mm}
\begin{itemize}
\item [(i)] Let $R_0$ be another reservoir at temperature $T_0$ and let 
$\{C_i\}_{i=1}^N\subset\mathcal{S}$ be machines operating cyclically between $R_0$ and $R_i$ under a reversible process $q_i\in\mathcal{P}_{C_i\vee R_i\vee R_0}$.
Let the machines $C_i$ and processes $q_i$ be such that the heat flows $Q_{R_i}(q_i) = Q_S(p_i)$ are provided to $R_i$ per cycle, and define $Q_{R_0}(q_i) =: Q_0^i$.
By Definition \ref{def:heatatt} we know that the heat flows $Q_S(p_i)$ between $R_i$ and $S$ are at temperature $T_i$. From this, we construct a cyclic machine $S_0$ as depicted in Figure~\ref{fig:clausius}. 

\begin{figure}

\begin{center} 
	\begin{tikzpicture}[scale=.65]
	\draw[very thick] (7,6.5) node[above] {$R_0$}-- (7.5,6.5) -- (7.5,7.5) -- (6.5,7.5) -- (6.5,6.5) -- (7,6.5); 
	\draw[<-] (7,6.4) -- (7,5.5) node[above right] {$Q_0$};
	\draw[] (7,5.4) -- (.5,4.4);
	\draw[] (7,5.4) -- (4,4.4);
	\draw[] (7,5.4) -- (11,4.4);
	\draw[<-, dashed] (7,5.4) -- (7,3.5) node[below] {$\cdots$};
	\draw[<-, dashed] (7,5.4) -- (9,4.4) -- (9,3.5) node[below] {$\cdots$};
	
	\draw[] (7,.5) node[below] {$\cdots$};
	\draw[] (9,.5) node[below] {$\cdots$};	
	\draw[] (7,2) node[below] {$\cdots$};
	\draw[] (9,2) node[below] {$\cdots$};
	
	\draw[very thick] (.5,.5) node[] {$R_1$} circle (.6cm);
	\draw[<-] (.5,1.2) -- (.5,2) node[below right] {$Q_S(p_1)$};
	\draw[->] (.5,-.2) -- (.5,-1) node[above right] {$Q_S(p_1)$};
	\draw[->] (2.15,2.75) -- (1.25,2.75) node[above right] {$W_{C_1}$};
	\draw[dashed] (2.15,2.75) -- (3,2.55);
	\draw[] (.5,2.75) node[] {$C_1$} circle (.6cm);
	\draw[<-] (.5,4.4) -- (.5,3.5) node [above right] {$Q_0^1$};

	\begin{scope}[xshift = 3.5cm, yshift = 0cm]
	\draw[very thick] (.5,.5) node[] {$R_2$} circle (.6cm);
	\draw[<-] (.5,1.2) -- (.5,2) node[below right] {$Q_S(p_2)$};
	\draw[->] (.5,-.2) -- (.5,-1) node[above right] {$Q_S(p_2)$};
	\draw[->] (2.15,2.75) -- (1.25,2.75) node[above right] {$W_{C_2}$};
	\draw[dashed] (2.15,2.75) -- (3,2.55);
	\draw[] (.5,2.75) node[] {$C_2$} circle (.6cm);
	\draw[<-] (.5,4.4) -- (.5,3.5) node [above right] {$Q_0^2$};
	\end{scope}

	\begin{scope}[xshift = 10.5cm, yshift = 0cm]
	\draw[very thick] (.5,.5) node[] {$R_N$} circle (.6cm);
	\draw[<-] (.5,1.2) -- (.5,2) node[below right] {$Q_S(p_N)$};
	\draw[->] (.5,-.2) -- (.5,-1) node[above right] {$Q_S(p_N)$};
	\draw[->] (2.45,2.75) -- (1.25,2.75) node[above right] {$W_{C_N}$};
	\draw[] (.5,2.75) node[] {$C_N$} circle (.6cm);
	\draw[<-] (.5,4.4) -- (.5,3.5) node [above right] {$Q_0^N$};
	\end{scope}	
	
	\draw[] (7,-3) node[] {$S$} circle (.6cm);
	\draw[->] (8.95,-3) -- (7.75,-3) node[above right] {$W_S$};
	
	\draw[] (.5,-1) -- (7,-2.3);
	\draw[] (4,-1) -- (7,-2.3);
	\draw[] (11,-1) -- (7,-2.3);
	\draw[->, dashed] (7,-.1) -- (7,-1);
	\draw[dashed] (7,-1) -- (7,-2.3);
	\draw[->, dashed] (9,-.1) -- (9,-1);
	\draw[dashed] (9,-1) -- (7,-2.3);
	
	\draw[->] (16,-1) node[above] {$W_{S_0}$} -- (15,-1);
	\draw[] (14.9,-1) -- (12.95,2.75);
	\draw[] (14.9,-1) -- (8.95,-3);
	\draw[dashed] (14.9,-1) -- (13.5,-.4);
	\draw[dashed] (14.9,-1) -- (13.5,-.7);
	
	\draw[dotted, rounded corners] (-1,5.45) -- (-1,-4) -- (14.95,-4) node[above right] {$S_0$} -- (14.95,5.45) -- (-1,5.5);
	
	\end{tikzpicture}
\end{center}
\caption{Extending the system $S\vee R_1\vee \cdots \vee R_N$ with the cyclic machines $\{C_i\}_{i=1}^N$ and the additional reservoir $R_0$ one can construct a cyclic machine $S_0$ interacting with a single reservoir $R_0$. The second law (Postulate~\ref{post:sec}) can be applied to this situation. }
\label{fig:clausius}
\end{figure}
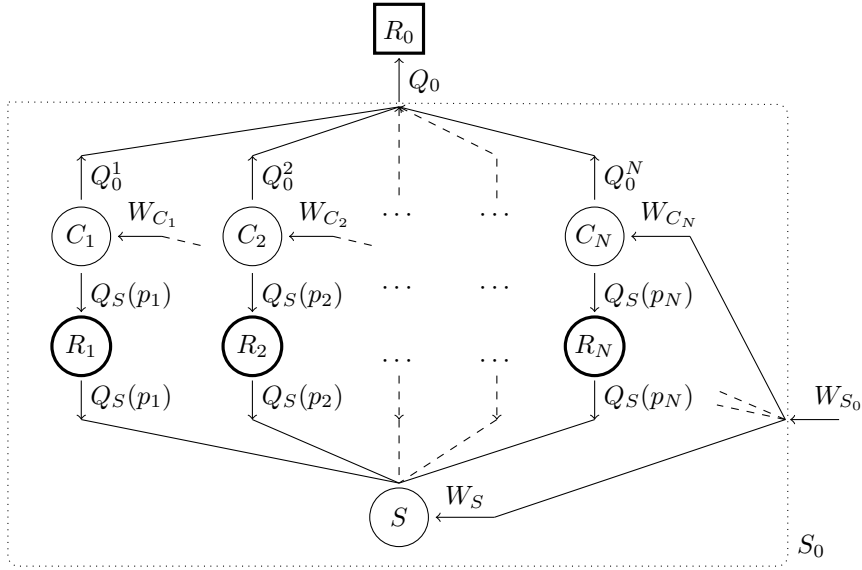

The machines $C_i$ together with the reservoir $R_0$ are used to provide the heat flows $Q_S(p_i)$ to the reservoirs $R_i$ such that, under this extension, all reservoirs except for $R_0$ become cyclic. 
For this, the heat flows $Q_0^i$, exchanged between the machines $C_i$ and $R_0$, as well as the work flows $W_{C_i}$ are needed. The heat flows $Q_0^i$ occur at temperature $T_0$.
In total, the system $S_0$ summarized in the dotted box is then cyclic. 

The second law (Postulate~\ref{post:sec}) can now be applied to this situation, which yields $Q_0 := \sum_i Q_0^i \geq 0$.
Together with Carnot's Theorem (ii) and the definition of absolute temperature, which says that for reversible machines
\begin{align}
-\frac{Q_0^i}{T_0} = \frac{Q_S(p_i)}{T_i}
\end{align}
holds, it follows with $T_0>0$ immediately that
\begin{align}
\label{eq:clausius}
0 \stackrel{2^\mathrm{nd}}{\geq} -\frac{Q_0}{T_0}
\stackrel{\mathrm{Def.} Q_0}{=} -\frac{Q_0^1+\cdots+Q_0^N}{T_0} 
\stackrel{\mathrm{Carnot}}{=} \frac{Q_S(p_1)}{T_1} + \cdots + \frac{Q_S(p_N)}{T_N} 
= \sum_{i=1}^N \frac{Q_S(p_i)}{T_i}\,.
\end{align}

\item [(ii)] If all processes $p_i$ are reversible Lemma~\ref{lemma:revheatuniquetapp} says that the temperatures of the associated heat flows are unique and equal to the temperatures in the reverse processes.\footnote{If the heat flow was $Q=0$ this statement does not make any sense. However, in this case the heat does not contribute either in Eq.~(\ref{eq:clausius}).}
In the reverse processes the heat flows go in the opposite direction.
Therefore, with the same argument as in (i), we obtain the opposite inequality, which together with (i) implies
\begin{align}
\sum_{i=1}^N\frac{Q_S(p_i)}{T_i}=0 \,.
\end{align}

\end{itemize}
\end{proof}

Even though Clausius' Theorem is formulated only for heat flows between a system and a set of reservoirs, it also makes a statement about other heat flows. Due to Definition~\ref{def:heatatt} any heat flow exchanged between two arbitrary systems to which a temperature can be assigned, can be reduced to the situation in which this heat flows between one of the systems and a reservoir. Therefore, if a system undergoes a sequential process in which all heat flows to the system can be assigned a temperature, the theorem holds too.\\

The fact that it is possible that more than one temperature can be assigned to a specific heat flow does not lead to problems. The inequality of Clausius' Theorem (i) holds for all temperatures that fulfil Definition~\ref{def:heatatt}. 
In the case of reversible heat flows, where the inequality works in both directions and thus equality holds, we are guaranteed that the assigned temperatures are unique. 


\subsection{Entropy and the Entropy Theorem}

The quantity $\sum_{i=1}^N\frac{Q_S(p_i)}{T_i}=0$ discussed in Clausius' Theorem (Theorem~\ref{thm:clausius}) applied to a sequence of processes which does not have to be cyclic will serve as the definition of thermodynamic entropy. 
More precisely, the entropy difference between two states will be computed by means of such a sequential process that starts from one and ends at the other state.
In order to be able to compute the entropy difference for any pair of states of any system, it is thus necessary that such a sequence always exists. 
This is made sure by the coming postulate. 

\begin{postulate}[Reversible processes with heat flows at well-defined temperatures]
\label{post:entropy}
Given any System $S\in\mathcal{S}$ and any two states $\sigma_1,\sigma_2\in\Sigma_S$ there always exists a finite sequence of reversible processes 
$\{p_i\}_{i=1}^N$, $p_i\in\mathcal{P}_{S\vee R_i}$,
with $R_i\in\mathcal{R}$
such that $p:=p_N\circ\cdots\circ p_1$ is well-defined and 
transforms $\lfloor p \rfloor_S = \sigma_1$ into $\lceil p \rceil_S = \sigma_2$.
\end{postulate}

The existence of a connecting sequence of processes as asked for by the postulate for arbitrary choices of $\sigma_1$ and $\sigma_1$ does not follow from the first law (Postulate~\ref{post:first}), which only asks for a work process connecting the two states.
The processes in the sequence, however, are not work processes on $S$. Instead, they need to be reversible and the heat flows to $S$ must have a well-defined temperature for all members of the sequence.\\

We are now in the position to define thermodynamic entropy, which is a state variable due to Clausius' theorem.

\begin{definition}[Entropy]
\label{def:entropy}
Let $S\in\mathcal{S}$ be a system and $\sigma_0\in\Sigma_S$ an arbitrary but fixed state.\footnote{One can choose this reference state to be the same as for the internal energy but this is not mandatory.}
Let $S_S^0\in\mathbbm{R}$ be an arbitrary real constant.
For a state $\sigma\in\Sigma_S$ we define its entropy as
\begin{align}
\label{entropydef}
S_S(\sigma) := \sum _i \frac{Q_S(p_i)}{T_i} + S_S^0\,,
\end{align}
where the sum goes over a sequence of reversible concatenable processes $\{p_i\}$, 
each of which fulfils $p_i\in\mathcal{P}_{S\vee R_i}$ with $W_{R_i}(p)=0$,
with total initial state $\sigma_0$ and final state $\sigma$.
\end{definition}


Lemma~\ref{lemma:entropyaddapp}
proves that entropy is additive under composition for disjoint systems.
It is thinkable that one works with a theory that does not fulfil Postulate~\ref{post:entropy}, i.e.\ there would exist pairs of states for which there is no such sequence of processes connecting them. 

\begin{example}[A system that does not satisfy Postulate~\ref{post:entropy}]
\label{ex:entropypost}
As an example, consider a mixture of 1 mole of oxygen and 2 moles of hydrogen undergoing the oxyhydrogen reaction, which might be known, whereas the opposite process, electrolysis, is unknown and thus not a process in this theory. In this case, the state of the initial mixture is connected by the work process ``oxyhydrogen reaction'' with the final state, 1 mole of water, and thus satisfying the first law. On the other hand, this process is irreversible and there is no other known process that could bring the water back to its initial state, the mixture. Hence, there is no sequence of processes that could serve for the definition of entropy and consequently, one has to define two measures of entropy on the different ``connected subsets'' of the state space which are incomparable.
The two reference entropies could be chosen arbitrarily.
\end{example}

Such a case is thinkable and (so far) allowed by the presented framework. We exclude if from now on by postulating that for any two states a sequence of processes as described in Postulate~\ref{post:entropy} exists. Hence, for a given system one needs to define only \emph{one} entropy measure for all of its state space. \footnote{This discussion is reminiscent of the one after Postulate~\ref{post:Carnotexist} asking for the comparability of any two reservoirs, which lead to a unique absolute temperature, instead of different incomparable ones. 
As was the case for absolute temperature, also here we aim for accordance of our framework with standard phenomenological thermodynamics and solve this issue with the postulate on the existence of reversible processes.
}\\

Entropy is a very helpful quantity. This is owed to the so called \emph{Entropy Theorem} which says that entropy is a monotone for the preorder $\rightarrow$. 
That is, if two states of a system are preordered (in the sense established by Definition~\ref{def:preorder}), i.e.\ there exists a work process that transforms one into the other, then the entropy of the final state cannot be smaller than the entropy of the initial state.

\begin{thm}[Entropy Theorem]
\label{thm:entropythm}
Let $S\in\mathcal{S}$ be a thermodynamic system and $p\in\mathcal{P}_S$ a work process on $S$. 
Then
\begin{align}
S_S(\lfloor p \rfloor_S) \leq S_S(\lceil p \rceil_S)\,.
\end{align}
with equality if $p$ is reversible.
In particular, since entropy is a state variable, any two states that are preordered in this sense (Definition~\ref{def:preorder}) fulfil this inequality, which is to say that entropy is a monotone for $\rightarrow$.
\end{thm}

\begin{proof}
Let $\{q_i\}_{i=1}^N\subset\mathcal{P}$ be a valid sequence to determine the entropy difference between the output and input state of $p$. That is, for each $i$ the process $q_i\in\mathcal{P}_{S_\vee R_i}$ is reversible and concatenable with its predecessor and successor, such that $\lfloor q_1\rfloor_S = \lceil p \rceil_S$ and $\lceil q_N \rceil_S = \lfloor p \rfloor_S$.
In each process $q_i$ the heat $Q_S(q_i)$ is exchanged at temperature $T_i$, or $Q_S(q_i)=0$.

Extend $p$ with an arbitrary identity process to $p\vee\mathrm{id}_R$ for an arbitrary $R\in\mathcal{R}\smallsetminus\{ R_1,\dots,R_N \}$.
Then it is possible to concatenate $q_N\circ \cdots \circ q_1\circ (p\vee\mathrm{id}_R)$ and
$Q_S(p\vee\mathrm{id}_R)=0$. 
The total process $q_N\circ \cdots \circ q_1\circ (p\vee\mathrm{id}_R)$ is cyclic on $S$ and fulfils the conditions in order to apply Clausius' Theorem, which then states
\begin{align}
0 \geq \sum_{i=1}^N \frac{Q_S(q_i)}{T_i} + 0 
= S_S(\lceil q_N\rceil_S) - S_S(\lfloor q_1\rfloor_S) 
= S_S(\lfloor p \rfloor_S) - S_S (\lceil p \rceil_S)\,.
\end{align}
This proves the inequality. 
In the case of a reversible $p$, Clausius' Theorem requires equality, which then implies
$S_S(\lfloor p \rfloor_S) = S_S (\lceil p \rceil_S)$.
\end{proof}

The Entropy Theorem establishes the parallels between what we call \emph{work processes} and what is usually called \emph{adiabatic processes}.
In standard phenomenological thermodynamics the Entropy Theorem is usually formulated for adiabatic processes, which are defined to be processes in which no heat is exchanged with other systems. 
The fact that our work processes serve as the traditional adiabatic processes was already suggested by the similarity of Postulate~\ref{post:first} (i) to the ``adiabatic accessibility'' statement of Lieb and Yngvason \cite{LY99}.
However, they use very different axioms and techniques to derive the corresponding monotone, the entropy.
In particular, since the term adiabatic is usually taken to mean ``no heat flows'', in our way of introducing thermodynamics we could not even use this concept as heat is only introduced after the first law is stated and internal energy is derived.
Furthermore, as is explained in an example at the end of Appendix~\ref{app:clausius}, in our framework reversible processes with ``no heat flow'' can still change the entropy of a system. Hence the traditional definition of an adiabatic process would fail to fulfil the Entropy Theorem.\\

The Entropy Theorem does not strictly hold in the other direction, i.e.\ one cannot always deduce the existence of a work process transforming $\sigma_2$ into $\sigma_1$ from 
$S_S(\sigma_2)\geq S_S(\sigma_1)$. 
However, it almost holds. Starting from the \emph{strict} inequality $S_S(\sigma_2) > S_S(\sigma_1)$ it follows due to the Entropy Theorem that there cannot be a work process on $S$ transforming $\sigma_1$ into $\sigma_2$.
Together with the first law it then follows that a work process in the other direction, i.e.\ taking $\sigma_2$ as input and giving $\sigma_1$ as output, must exist.
Hence, in this sense, the Entropy Theorem is almost invertible.

\newpage
\section{Quasistatic processes}
\label{sec:quasistatic}

\begin{center}
\fcolorbox{OliveGreen}{white}{\begin{minipage}{0.98\textwidth}
\centering
\begin{minipage}{0.95\textwidth}
\ \\
\textcolor{black}{\textbf{Postulates:}}
quasistatic first law (quasistatic Postulate~\ref{post:first}), quasistatic reversible processes with heat flows at well-defined temperatures (quasistatic Postulate~\ref{post:entropy})\\
\ \\
\textcolor{black}{\textbf{New notions:}}
quasistatic thermodynamic processes, differential work and heat

\paragraph{\textcolor{black}{Summary:}}
Based on the existing notions of states and processes we define quasistatic processes which leads to the notions of differential work and heat. In order to formally use the differential quantities instead of the discrete ones discussed so far, the Postulates~\ref{post:first} and~\ref{post:entropy} must be adjusted. The implications of these adjustments are then discussed.
\vspace{.1cm}

\end{minipage}
\end{minipage}}
\end{center}

\ \\

In this section we extend our framework, which talked about discrete quantities up to this point, to allow for continuous processes and state spaces. 
This is coupled to additional assumptions on the technical structure of state spaces and the set of thermodynamic processes. 
Arguably, many of the additional assumptions are not on an equal footing with the physical assumptions discussed in the main text. They are technical in nature and not necessary to be fulfilled if one wants to work with the minimal assumptions from the main text to obtain the standard results in their discrete form. 
However, they are necessary in order to have notions like the ones of a \emph{quasistatic process} or \emph{differential work and heat}, which are about continuous and differentiable quantities. \\

\subsection{General Definition}

In order to define a quasistatic process on $S$ the minimal requirement on $\Sigma_S$ is that it is equipped with a topology. 

\begin{definition}[Quasistatic process]
\label{def:quasistaticapp}
A \emph{quasistatic work process on a system $S\in\mathcal{S}$} is a two-parameter family of work processes
$\{p(\lambda,\lambda')\}_{0\leq\lambda\leq\lambda'\leq1}\subset\mathcal{P}_S$ 
such that:
\begin{itemize}
\item [(i)] For $\lambda\leq\lambda'\leq\lambda''\in[0,1]$ it holds 
$p(\lambda',\lambda'') \circ p(\lambda,\lambda') = p(\lambda,\lambda'')$.

\item [(ii)] The curve $\gamma: [0,1] \rightarrow \Sigma_S$, $\gamma(\lambda) := \lfloor p(\lambda,1)\rfloor_S$ is continuous in the topology of $\Sigma_S$.

\item [(iii)] For all $A\in\mathrm{Atom}(S)$ the work function $W_A(p(\lambda,\lambda'))$ is continuous in both $\lambda$ and $\lambda'$ on the respective domain.
\end{itemize}
\end{definition}

A \emph{quasistatic process on $S$} is then just a quasistatic work process on a larger system $S'\in\mathcal{S}$ such that $S\in\mathrm{Sub}(S')$ is a subsystem.
Even if we talk about a quasistatic work process on $S$ we sometimes call it a quasistatic process on $S$ when it helps to improve the readability of the text and the context makes clear what is meant. 

\begin{lemma}[Properties of quasistatic processes]
\label{lemma:quasistaticapp}
Let $\{p(\lambda,\lambda')\}_{\lambda,\lambda'}$ be a quasistatic process on $S$. 
For all $\lambda, \lambda',\lambda''\in[0,1]$ with $\lambda\leq \lambda'\leq\lambda''$ it holds:
\begin{itemize}
\item[(i)] $\lceil p(\lambda,\lambda')\rceil_A = \lfloor p(\lambda',\lambda'')\rfloor_A$ for all 
$A\in\mathrm{Atom}(S)$.

\item [(ii)] $\lfloor p(\lambda,\lambda')\rfloor_S$ is independent of $\lambda'$ and
$\lceil p(\lambda,\lambda')\rceil_S$ is independent of $\lambda$. 
In particular, the curve $\gamma$ from Definition~\ref{def:quasistaticapp} (ii) satisfies 
$\gamma(\lambda) = \lfloor p(\lambda,\lambda')\rfloor_S	= \lceil p(\lambda'',\lambda)\rceil_S$ for all $\lambda'\geq\lambda$ and $\lambda''\leq\lambda$. 

\item [(iii)] $p(\lambda,\lambda)$ is an identity process on $S$.

\item [(iv)] For all $\lambda,\lambda'\in[0,1]$ and all partitions $\lambda=\lambda_0<\lambda_1<\cdots<\lambda_m=\lambda'$ it holds
$W_A(p(\lambda,\lambda')) = \sum_{i=1}^m W_A(p(\lambda_{i-1},\lambda_i))$.

\end{itemize}
\end{lemma}

\begin{proof}
Lemma~\ref{lemma:quasistaticapp} (i) follows directly from Definition~\ref{def:quasistaticapp} (i).
Claim (ii) is then a consequence of (i).
For $p(\lambda,\lambda)$ we find according to Definition~\ref{def:quasistaticapp} (i) that it can always be concatenated with itself, and $p(\lambda,\lambda)\circ p(\lambda,\lambda) = p(\lambda,\lambda)$. 
Due to the additivity of work under concatenation, it follows that $W_A(p(\lambda,\lambda))=0$ for all $A\in\mathrm{Atom}(S)$. Hence $p(\lambda,\lambda)$ is an identity process on $S$.
Finally, claim (iv) is again a simple consequence of Definition~\ref{def:quasistaticapp} and the additivity of work under concatenation.
\end{proof}

Quasistatic processes are those which allow for a continuous partition into ``smaller'' processes, which can be concatenated to from the whole process again. 
The process $p(\lambda,\lambda')$ connects the state $\gamma(\lambda)$ with $\gamma(\lambda')$, as Lemma~\ref{lemma:quasistaticapp} shows.
The family of processes thereby defines a continuous curve in the state space and the work functions evaluated on the members of the family must be continuous as well in the curve parameter. This has to hold for each atomic subsystem of $S$ individually. The naive approach asking only for continuity of $W_S$ would allow for an arbitrary chaotic behaviour of the $W_A$. 
This is something that should not be the case, as one expects that $p(\lambda,\lambda')$ for very close $\lambda$ and $\lambda'$ corresponds to a small ``step''.\\

The family of processes defining a quasistatic process on $S$ can also be called \emph{quasistatic on a subsystem} $S'\in\mathrm{Sub}(S)$ of $S$.
Consequently, a family of processes which are not necessarily work processes on a system $S$ are called \emph{quasistatic} if they form a quasistatic process on the larger system on which they are work processes. \\

It is easy to see that two quasistatic processes $\{p(\lambda,\lambda')\}_{\lambda,\lambda'}$ and $\{q(\lambda,\lambda')\}_{\lambda,\lambda'}$ can be concatenated to a new quasistatic process $\{(q\circ p)(\lambda,\lambda')\}_{\lambda,\lambda'}$ whenever $\lceil p(0,1) \rceil_S = \lfloor q(0,1) \rfloor_S$. 
The concatenated quasistatic process is then defined as 
\begin{align}
(q\circ p) (\lambda,\lambda') = \begin{cases}
p(2\lambda,2\lambda')\,, \text{ for } \lambda\leq \lambda' < \frac{1}{2}\,, \\
q(0,2\lambda')\circ p(2\lambda,1)\,, \text{ for } \lambda < \frac{1}{2}\leq \lambda'\,, \\
q(2\lambda,2\lambda')\,, \text{ for } \frac{1}{2}\leq \lambda\leq \lambda'\,.
\end{cases}
\end{align}

In practice one often calls $p(0,1)$ the ``quasistatic process''. However, in our framework this process alone does not say anything about the continuous curve in the state space. Hence the mathematical structure of a quasistatic process must be richer than that. 
Nevertheless, we will sometimes call a process $p\in\mathcal{P}_S$ quasistatic on $S$ without explicitly mentioning that formally the family $\{p(\lambda,\lambda')\}_{\lambda,\lambda'}$ with $p=p(0,1)$ is meant.

\begin{example}[Identity processes]
An identity process $\mathrm{id}_S^\sigma\in\mathcal{P}_S$ on a system $S$, which, concatenated with itself yields itself again, can always be seen as quasistatic processes (with any topology on $\Sigma_S$).
To see this, consider the defining family $p(\lambda,\lambda') = \mathrm{id}_S^\sigma$ 
for all $\lambda\leq\lambda'\in[0,1]$.
This family of processes obviously fulfils Definition~\ref{def:quasistaticapp} (i).
The curve $\gamma(\lambda)=\sigma$ for all $\lambda\in[0,1]$ is constant and thus continuous in any topology, which shows (ii). 
In addition, (iii) is fulfilled since identity processes have zero work cost on any atomic subsystem and thus $W_A(p(\lambda,\lambda'))=0$ constant. 

This example also works for discrete spaces $\Sigma_S$ equipped with the discrete topology. 
In the discrete topology, such quasistatic processes are the only possible ones owed to the fact that every continuous curve in a discrete space is constant. 
\end{example}

We next comment on the continuity of the internal energy of a system, on which quasistatic processes exist. 
For a quasistatic process $\{p(\lambda,\lambda)\}_{\lambda,\lambda'}$ on $S$ the internal energy
of $S$ is continuous on the curve in the curve parameter.
That is, $U_S(\lambda) := U_S(\gamma(\lambda))$ is continuous in $\lambda$ and can be seen as follows.

Due to Definition~\ref{def:quasistaticapp} (iii) together with the additivity of work under composition and the fact that the sum of continuous functions is continuous, we find
\begin{align}
U_S(\lambda') - U_S(\lambda) 
= U_S(\gamma(\lambda')) - U_S(\gamma(\lambda)) 
= W_S(p(\lambda,\lambda') )
\stackrel{\lambda\rightarrow\lambda'}{\longrightarrow} 0 \,,
\end{align}
for $\lambda\leq\lambda'\in[0,1]$.
We used that the change in internal energy $U_S$ can always be expressed as the work cost of an appropriate work process on $S$.
The proof does not directly generalize to the internal energies of subsystems of $S$ because the second equality only holds when the change in internal energy can be expressed by a step in a quasistatic process. The existence of a quasistatic process on $S$ does not imply this.

%

If for any continuous curve in $\Sigma_S$ there exist quasistatic processes on $S$ that, when putting the quasistatic curves together, pass through this curve (not necessarily in one direction, as the different quasistatic curves may have different directions in which they are gone through) it follows that $U_S$ is continuous on $\Sigma_S$, now in the topology of $\Sigma_S$.
This holds because a function that is continuous on any continuous path is continuous on the topological space.
Again, the same only follows for subsystems $S'\in\mathrm{Sub}(S)$ and $U_{S'}$ if also in $\Sigma_{S'}$ there exist quasistatic processes for any continuous curve. 

Systems with this property fulfil a stronger first law asking not only for the existence of work processes between any two states, but that these are quasistatic in addition. 
We now know that for such systems the internal energy function is continuous. 
Furthermore, the state space of such a system is necessarily path-connected.
We do not formalize this strengthened first law here, as it would only be an intermediate step towards the ``quasistatic first law'' below, which relies on even richer structure on the state space.

\begin{example}[Continuous internal energy of a gas]
\label{ex:contUapp}
A gas described by $\sigma = (p,V)\in\Sigma_S = \mathbbm{R}_{>0}^2$ is again a good example for a system which allows for quasistatic processes between any two states. 
The state space $\Sigma_S=\mathbbm{R}_{>0}^2$ shall be equipped with the usual topology on $\mathbbm{R}^2$.
For gases adiabatic compression and expansion and isochoric warming (i.e.\ heating) can be considered quasistatic processes if they are carried out slow enough.
By considering those types of processes any continuous path in $\Sigma_S$ can be seen as a union of quasistatic curves, where the directions of the curves play no role for the purposes of determining the internal energy. 
We conclude from this, that the internal energy of a gas is continuous in $\sigma$.
\end{example}

\subsection{Differential work and heat}

We now investigate the differential work and heat, which can be defined if further structure is present on the state spaces of the involved systems. Specifically, $\Sigma_S$ must be at least a $C^1$ manifold. The state spaces $\Sigma_A$ of atomic subsystems $A\in\mathrm{Atom}(S)$ are then submanifolds of $\Sigma_S$ and thus also $C^1$. The same holds for arbitrary subsystems of $S$.

\begin{definition}[$C^1$-quasistatic processes and differential work]
\label{def:diffWapp}
A \emph{(piecewise) $C^1$-quasistatic process on $S\in\mathcal{S}$} is a quasistatic process $\{p(\lambda,\lambda')\}_{\lambda,\lambda'}\subset\mathcal{P}_S$ on $S$ with a (piecewise) $C^1$-curve $\gamma$ in $\Sigma_S$
together with continuous $1$-forms $\delta^{(p)} W_A$ defined on the curve for all $A\in\mathrm{Atom}(S)$ such that
for all $\lambda\leq\lambda'\in [0,1]$ it holds
\begin{align}
\int_{\gamma|_{[\lambda,\lambda']}} \delta^{(p)} W_A = W_A(p(\lambda,\lambda'))\,.
\end{align}
\end{definition}

Here, the left hand side is a part of the path integral over the quasistatic path.
The continuous $1$-form $\delta^{(p)} W_A$ is called \emph{differential work of $A$} and may depend on the quasistatic process, which is indicated by the superscript $(p)$ standing for the whole family $\{p(\lambda,\lambda')\}_{\lambda,\lambda'}$.
Through additivity of work under composition we obtain corresponding continuous $1$-forms for all subsystems of $S$ and of course $S$ itself.

We are not only interested in $C^1$-curves but also piecewise $C^1$-curves as this is what one generically obtains when concatenating two $C^1$-quasistatic processes. 
In order to be able to integrate a piecewise $C^1$-curve is sufficient.\\

At this point it makes sense to strengthen the fist law (Postulate~\ref{post:first}) so it also talks about piecewise $C^1$-quasistatic processes. 
The new version is thought for thermodynamic theories in which the systems' state spaces are $C^1$ manifolds.
In addition to the previous first law the ``quasistatic version'' will also make sure that the internal energy $U_S$ of any system is a differentiable function and hence the differential d$U_S$ exists. 
Section~\ref{sec:exideal} discusses the example of the ideal gas in this context in more depth.

\begin{postulateqsfirst}[The quasistatic first law]
\label{post:qsfirst}
For any system $S\in\mathcal{S}$ the following three statements hold:
\begin{itemize}
	\item [(o)] For all states $\sigma\in\Sigma_S$ there exist $m:=\mathrm{dim}\Sigma_S$ 
	$C^1$-quasistatic processes with curves $\gamma_1,\dots,\gamma_m$ passing through $\sigma$
	such that the derivatives of these curves at $\sigma$ are linearly independent.
	That is, there exist $\lambda_i\in(0,1)$ with $\gamma_i(\lambda_i)=\sigma$ for $i=1,\dots,m$,
	such that $\gamma_1'(\lambda_1),\dots,\gamma_m'(\lambda_m)$ are linearly independent. 
	
	\item [(i)]
	For any pair of states $\sigma_1,\sigma_2\in\Sigma_S$ there exists a piecewise $C^1$-quasistatic 
	process 
	$\{p(\lambda,\lambda')\}_{\lambda,\lambda'}\subset\mathcal{P}_S$ on $S$ 
	with $\lfloor p(0,1) \rfloor_S = \sigma_1$ and 	$\lceil p(0,1) \rceil_S = \sigma_2$ 
	or in the other direction, or both.
	
	\item [(ii)] The total work cost of a work process $p\in\mathcal{P}_S$ on $S$, $W_S(p)$, 
	only depends on $\lfloor p \rfloor_S$ and $\lceil p \rceil_S$ 
	and not on any other details of the process. 
\end{itemize}
\end{postulateqsfirst}

Postulate ~\ref{post:qsfirst} (i) obviously implies point (i) of the previous first law, Postulate~\ref{post:first}, while (ii) has not changed at all. Therefore, all the machinery relying on the first law can as well be based on the quasistatic first law.
The additional point (o) may seem unnecessary at first, since it requires more quasistatic processes passing through a given state than just the one that is asked for in (i). 
However, this point will be crucial in order to show that the internal energy of systems fulfilling the quasistatic first law is differentiable.

One may wonder whether (o) already implies (i) in the quasistatic formulation. In fact it almost does so, but not quite. 
One could think of constructing the quasistatic processes asked for in (i) by means of the ones asked for in (o) step by step.
Even if this construction lead to a piecewise $C^1$-curve from one state to the other (which may depend on further assumptions on the manifold $\Sigma_S$), it would have the problem that the direction of the quasistatic steps might change throughout the curve. 
If this is the case, the quasistatic processes cannot be composed to a piecewise $C^1$-quasistatic process connecting the states.\footnote{
Thinking of the states of system as the vertices and the work processes connecting them as the edges of a graph, (i) says that this graph must be complete, i.e.\ there is an edge between two arbitrary vertices. In fact, it possible to formulate thermodynamics with a weaker first law in which (i) only asks for this graph to be connected. In this formulation, the existence of an (undirected) $C^1$-curve connecting any two states as in the above construction is enough to derive (i), which may follow from (o). 
However, we do not go deeper into this weaker first law here, as it makes proofs and technical derivations much more cumbersome while there is no gain over the strong formulation except for a weaker postulate.}\\

Example~\ref{ex:contUapp} may also serve as an example of a system which fulfils the quasistatic first law. As is argued before the example, one can show that the internal energy of the gas is continuous in this case.
We now prove that for any atomic system $A\in\mathcal{A}$ satisfying the quasistatic first law (Postulate~\ref{post:qsfirst}) the internal energy is actually differentiable. 
The extension of this result to arbitrary systems $S\in\mathcal{S}$ follows by the additivity of the work functions. 

According to the Definition~\ref{def:Ustrong} of internal energy for any state $\sigma\in\Sigma_A$ of an atomic system $A\in\mathcal{A}$ we can write 
\begin{align}
\label{eq:qsUapp}
U_A(\sigma) = \int_{\sigma_0}^\sigma \delta W_A + U_A^0
\end{align}
where the integral goes over an appropriate piecewise $C^1$-curve arising from a $C^1$-quasistatic process on $A$ transforming $\sigma_0$ to $\sigma$ or the other way around.
Notice that Equation~(\ref{eq:qsUapp}) holds in both cases, i.e.\ it also holds when the piecewise $C^1$-quasistatic process on $A$ transforms $\sigma$ to $\sigma_0$, and is independent of the choice of quasistatic process used to compute the integral, which here in particular means that it is independent of the path over which the integral is evaluated.
In this case, one integrates in the opposite direction of the given curve, thereby introducing an additional minus sign, and thus the equality still holds.

The existence of such an appropriate curve is guaranteed by Postulate~\ref{post:qsfirst} (i) while the fact that the continuous $1$-form $\delta W_A$ is independent of the actual process is the statement of (ii). This is why the superscript $(p)$ is neglected in the notation. 
With (o) it is now possible to argue that $U_A$ is differentiable in $\sigma$.
Since the $m$ $C^1$-quasistatic processes passing through $\sigma$ must exist and their derivatives at $\sigma$ are linearly independent, these derivatives build a basis of the tangent space $T_\sigma\Sigma_A$. These quasistatic processes together with their $C^1$-curve can be used to take the directional derivative of $U_A$ in the $m$ ``different directions'' of the tangent space.
Since $U_A$ is defined as the path integral of a continuous $1$-form over a piecewise $C^1$-curve, these derivatives all exist and are continuous.
This is sufficient to deduce that $U_A$ is differentiable in $\sigma$ since a continuous function on a manifold for which all partial derivatives at a point exist and are continuous is differentiable in this point (see e.g.\ Theorem 9.21 in \cite{Rudin76}).
Since this can be deduced for an arbitrary $\sigma\in\Sigma_A$, it also shows differentiability of $U_A$ on $\Sigma_A$.\\

For any differentiable function $f$ the corresponding differential d$f$ can be defined.
This also holds for the now differentiable internal energy $U_A$. 
Relying on this, the \emph{differential heat of $A$} in a quasistatic process on a possibly larger system $S$ with $A\in\mathrm{Atom}(S)$ is defined as 
\begin{align}
\delta^{(p)} Q_A := \mathrm{d}U_A - \delta^{(p)} W_A\,.
\end{align}
For arbitrary systems the differential heat is defined via additivity under composition.\\

Again, the superscript $(p)$ is in general necessary to make explicit that the $1$-form may depend on the quasistatic process.
In particular, the $1$-forms for differential work and heat are not exact, while d$U$ as the differential of a differentiable function is. 
This is manifest in the notation of the inexact differentials with $\delta$ and the the exact differential with $d$. 
Generally, path integrals over exact differentials are independent of the path, while integrals over inexact differentials depend on the path.
In our case the integrals over differential work and heat are in general not only path-dependent but \emph{process-dependent}, i.e.\ they depend on the family of processes constituting the quasistatic process on which they are defined. 

\begin{example}[Differential work and heat are generally process-dependent]
\label{ex:diffWQdependentapp}
As an example we consider the isothermal expansion of a gas.
If the constant temperature of the gas is achieved through thermal contact with a reservoir, as is normally the case when one talks about isothermal expansion, the differential work is $-p\,\mathrm{d}V$. 
On the other hand, one could achieve a constant temperature throughout the process by continuously applying friction, i.e.\ bringing up the proper amount of work. The energy that keeps the temperature constant must then be counted as work instead of heat in the previous case. 
In this case, the differential work will be different and the total amount of work done as well.
However, the paths corresponding to the quasistatic processes in the state space will be exactly the same. 
\end{example}

From now on, whenever we talk about quasistatic processes and differential work and heat, we implicitly assume that the basic structure on the state spaces is given and that the quasistatic first law, Postulate~\ref{post:qsfirst}, is fulfilled.
In particular, by a quasistatic process from now on we mean a (piecewise) $C^1$-quasistatic one. 
Also, as is usually done in the literature, we will only write $\delta W_S$ and $\delta Q_S$, without manifesting the process-dependence of these forms in the notation explicitly. \\

\subsection{Entropy through quasistatic processes}

When using quasistatic processes for computing entropy differences we have to consider reversible ones. 
A quasistatic process $\{p(\lambda,\lambda')\}_{\lambda,\lambda'}$ on a system is called \emph{reversible} if all members $p(\lambda,\lambda')$ of the family are reversible work processes on this system. 

Suppose now a system $S\in\mathcal{S}$ undergoes a quasistatic process $\{p(\lambda,\lambda')\}_{\lambda,\lambda'}\subset\mathcal{P}_{S\vee R_1\vee\cdots\vee R_N}$, 
where $\{R_i\}_i\subset\mathcal{R}$ are reservoirs, 
such that there is a partition $0=\lambda_0 < \lambda_1 < \cdots < \lambda_N=1$ 
and for $\lambda_{i-1}<\lambda\leq\lambda'<\lambda_i$ the process $p(\lambda,\lambda')$ only acts on $S\vee R_i$, i.e.\ it can be written as a work process on $S\vee R_i$ composed with an appropriate identity on the other reservoirs.
Denote the initial state of this process on $S$ by $\sigma_1$ and the final state by $\sigma_2$.
We find
\begin{align}
\label{eq:qsentrpyapp}
S_S(\sigma_2) - S_S(\sigma_1) = \sum_{i=1}^N \frac{Q_{S,i}}{T_i}
= \sum_{i=1}^N \int_{\gamma|_{[\lambda_{i-1}, \lambda_i]}}  \frac{\delta Q_S}{T_i}
= \int_{\gamma}  \frac{\delta Q_S}{T} \,.
\end{align}
The first equality is just Definition~\ref{def:entropy}, where $Q_{S,i}$ is the heat that $S$ exchanges at temperature $T_i$ with reservoir $R_i$ between parameters $\lambda_{i-1}$ and $\lambda_i$. 
Therefore, $Q_{S,i} = \int_{\gamma|_{[\lambda_{i-1}, \lambda_i]}} \delta Q_S$ and the second equality follows.
For the third equality $T = T(\lambda)$ is introduced as a (piecewise constant) function on the curve such that $T(\lambda) = T_i$ for $\lambda_{i-1}<\lambda<\lambda_i$.\\

Starting from this, one can interpret the integral expression from Equation~(\ref{eq:qsentrpyapp}) with a continuous function $T$ on the path as the limit of reversible quasistatic process
with finer and finer partitions such that in the limit the piecewise constant temperature functions converge pointwise to the continuous one and the piecewise $C^1$-curves converge to $\gamma$. 
The \emph{interpretation} that the heat 
$\delta Q_S(\gamma(\lambda))\tfrac{\mathrm{d}\gamma(\lambda)}{\mathrm{d}\lambda}$
flows at temperature $T(\lambda)$ is valid in this limit. 
However, for the operational meaning of a continuous temperature function it is important to remember that it arises from a limit of piecewise constant functions. Only then the temperature can be put into relation with Definition~\ref{def:heatatt} for the temperature of heat flows.\\

Just like we strengthened the first law in the presence of quasistatic processes, we can also strengthen the postulate on the existence of reversible processes with heat flows at well-defined temperatures, Postulate~\ref{post:entropy}, which is relevant for defining entropy on all of $\Sigma_S$.

\begin{postulateqsentropy}[Quasistatic reversible processes with heat flows at well-defined temperatures]
\label{post:qsentropy}
For any system $S\in\mathcal{S}$ it holds that
\begin{itemize}
	\item [(o)] For all states $\sigma\in\Sigma_S$ there exist $m:=\mathrm{dim}\Sigma_S$ 
	reversible
	$C^1$-quasistatic processes on $S\vee R_i$, $R_i\in\mathcal{R}$ for $i=1,\dots,m$,
	with curves $\gamma_1,\dots,\gamma_m$ passing through $\sigma$
	such that the derivatives of these curves at $\sigma$ are linearly independent.
	That is, there exist $\lambda_i\in(0,1)$ with $\gamma_i(\lambda_i)=\sigma$ for $i=1,\dots,m$,
	such that $\gamma_1'(\lambda_1),\dots,\gamma_m'(\lambda_m)$ are linearly independent. 
	
	\item [(i)] For any two states $\sigma_1,\sigma_2\in\Sigma_S$ 
	there exists a piecewise $C^1$-quasistatic process 
	$\{p(\lambda,\lambda')\}_{\lambda,\lambda'}\subset\mathcal{P}_{S\vee R_1\vee\cdots\vee R_N}$
	with $\{R_i\}_{i=1}^N\subset\mathcal{R}$ such that 
	there is a partition $\{\lambda_i\}_{i=0}^N$ of $[0,1]$ 
	and in each segment $(\lambda_{i-1},\lambda_i)$ $S$ interacts with a single reservoir, 
	and in total it transforms $\lfloor p(0,1)\rfloor_S=\sigma_1$ into 
	$\lceil p(0,1)\rceil_S=\sigma_2$.
	
\end{itemize}
\end{postulateqsentropy}
%

In this more restrictive version, (i) corresponds to the previous statement with the additional requirement that the process connecting the two states is quasistatic. This again makes sure that the entropy can be computed for any state in $\Sigma_S$, for any system $S$ now with the integral expression from Eq.~(\ref{eq:qsentrpyapp}). Of course, also here the result from Clausius' Theorem (Theorem~\ref{thm:clausius}) is important in order to make sure that this way of computing the entropy difference does not depend on the used process.

The purpose of the additional point (o) is the same as the corresponding point (o) in Postulate~\ref{post:qsfirst}. 
And again, one might wonder whether (o) already implies (i).
In fact, it looks even more plausible here than before since the argument of the direction of the processes can no longer pose a problem for reversible processes. 
Nevertheless, the existence of finitely many curves in a neighbourhood of a point $\sigma$ with linearly independent derivatives do not immediately imply the existence of a curve (coming from a quasistatic process) connecting any two points in the manifold.
We do not discuss this technical topic any further here as it does not play a major role from a conceptual viewpoint. It is at this point good enough to have sufficient criteria in order to deduce differentiability of the entropy.
This is what we do next.\\

In practice it is often the case that the differential work for \emph{reversible} quasistatic processes is independent of the actual process. In this case the notation $\delta W_S$ without the superscript $(p)$ is not just convenient but also formally correct, and the $1$-form is defined on all of the state space $\Sigma_S$.
For gases for instance, one can find $\delta W_S = -p{\rm d}V$ during reversible quasistatic processes. 
This is not to say that the $1$-form is exact, as one can see from the gas example.

As a derived $1$-form from other $1$-forms, which are independent of the process, the differential heat $\delta Q_S$ then also shares this independence.

If now Postulate~\ref{post:qsentropy} holds,
it follows that entropy is differentiable in the same way as differentiability of $U_S$ was derived.
In the argument for the differentiability of the entropy Clausius' Theorem~\ref{thm:clausius} plays an important role, as it is necessary to be able to express entropy in terms of a path-independent integral.
Again it is important that the $1$-form over which is integrated is continuous. 
This may not be the case in all quasistatic processes. Namely, at the points where the piecewise constant function $T$ jumps, the $1$-form $\tfrac{\delta Q_S}{T}$ is not continuous. 
However, with (o) we know that there are quasistatic paths through any state during which $T$ is constant, in particular continuous. This is guaranteed by the condition that the $m$ reversible quasistatic processes are acting on $S$ and a \emph{single} reservoir. Hence, when showing differentiability of the entropy, one can just use these paths.

As a consequence, the differential ${\rm d}S_S$ exists and is exact. 
In the context of a reversible quasistatic process in which heat is exchanged at temperature $T$, the differential reads ${\rm d}S_S = \tfrac{\delta Q_S}{T}$.  
Also the proof showing that entropy is differentiable is worked out in more detail in Section~\ref{sec:exideal} using the example of the ideal gas. \\

With this we have shown that the presented framework, which was formulated in a discrete manner, can be extended to a continuous (even differentiable) formulation such that the usual differentiability properties of the internal energy and the entropy hold.
This makes life a lot easier when following the standard arguments from phenomenological thermodynamics, e.g.\ when deriving different thermodynamic potentials or computing thermodynamic quantities such as the compressibility or the heat capacity.

\newpage
\section{The temperature of arbitrary systems}
\label{sec:tsys}

\begin{center}
\fcolorbox{OliveGreen}{white}{\begin{minipage}{0.98\textwidth}
\centering
\begin{minipage}{0.95\textwidth}
\ \\
\textcolor{black}{\textbf{Postulates:}}
the temperature of a system\\
\ \\
\textcolor{black}{\textbf{New notions:}}
quasistatic thermodynamic processes, differential work and heat

\paragraph{\textcolor{black}{Summary:}}
We discuss how to extend the notion of absolute temperature for reservoirs (Section~\ref{sec:abstemp}) and heat flows (Section~\ref{sec:theatflow}) to the temperature of a system. In particular, we give insights into when this is possible and when one cannot expect that such a notion can be meaningfully defined.
\vspace{.1cm}

\end{minipage}
\end{minipage}}
\end{center}

\ \\

The concepts of temperature introduced so far include the \emph{temperature of a reservoir} and the \emph{temperature of a heat flow}. 
However, we have not made a general definition of the \emph{temperature of a system} $S\in\mathcal{S}$ nor the \emph{temperature of a state} $\sigma\in\Sigma_S$.
The reason for this is that it is not always possible.
As an example, it already suffices to consider a composite system $S=S_1\vee S_2$ in a generic state $\sigma_1\vee\sigma_2$.
Even if both $S_1$ in state $\sigma_1$ an $S_2$ in state $\sigma_2$ can be assigned individual temperatures $T_1$ and $T_2$, these will in general not be equal. 
Therefore, there is no unique ``combined'' temperature of the composed system $S$ -- one would intuitively need two ``temperatures'' in such a setting.
Also for atomic systems it is not clear that a single temperature can always be assigned to each state. Even though atomic systems are considered ``indivisible'', this is only introduced in abstract terms of the composition $\vee$, thereby not excluding the possibility that the atomic system could still consist of two compartments with independent properties, just not in the sense composition speaks about it.\\

Nevertheless, we know from the standard theory of thermodynamics that there are many systems for which a temperature (of the system) can be defined in any state. 
We here investigate the case of a system $S$ with a state space $\Sigma_S\subset\mathbbm{R}^2$ that is an open subset of $\mathbbm{R}^2$, hence the states are determined by two real parameters $x,Y$.
We further assume that this system fulfils the quasistatic first law as well as the quasistatic version of the postulate for defining entropy and that in reversible quasistatic processes the differential work can be written as 
$\delta W = x{\rm d}Y$. 
These properties have been introduced and discussed in the previous Section~\ref{sec:quasistatic}.
They are formulated in a generic way, but we do not claim that there exists no setting with weaker assumptions in which nevertheless a temperature can be assigned to the system.

We will now define the temperature of such a system as the temperature of a reversible heat flow that is exchanged with the system. This is the basic compatibility requirement we have in mind here.
Obviously, in the most general case there might be different possible reversible heat flows with different temperatures coming from a system in a given state. In this case the temperature could not be uniquely defined. As argued already in the beginning of this section, this construction of temperature cannot always be applied.
With the assumptions made, however, it is possible, and eventually allows us to define the temperature of a system via the temperature of heat flows, which in turn relies on the definition of the absolute temperature of reservoirs.\\

In general, the differentials of work and heat sum up to ${\rm d}U = \delta W + \delta Q$.
Using this for reversible quasistatic processes in which heat is exchanged at a well-defined temperature it holds $\delta Q = T{\rm d}S$ and with the above assumptions it follows
\begin{align}
{\rm d}U = x{\rm d}Y + T{\rm d}S\,.
\end{align}
Considering now processes during which $Y$ is kept constant, we find 
\begin{align}
T = \left.\frac{\partial U}{\partial S}\right|_{Y}\,.
\end{align}
This equation needs an explanation. 
First, the internal energy is \emph{a priori} a function of $x$ and $Y$. However, the entropy $S$ is another state variable and one can consider $U$ as a function of $S$ and $Y$ instead.
Second, due to the differentiability of $U$ and $S$, the partial derivative exists. 
We conclude that the \emph{temperature of the (infinitesimal) heat flow} in a process, in which $Y$ is kept constant, can be written as the partial derivative of internal energy $U$ w.r.t.\ entropy $S$.
But the right hand side only depends on state variables, hence $T$ is also a state variable in this particular case.
Thus we call $T$ the \emph{temperature of the system in state $(S,Y)$}. \\

Thinking of a gas and identifying $x$ with the pressure $p$ and $Y$ with the volume $V$ (up to a minus sign) it follows that one can generally define the absolute temperature of a gas in the way sketched here. 
This example is discussed in Section~\ref{sec:exideal} in more depth.

This notion of temperature, just like the temperature of a heat flow, eventually relies on the absolute temperature of reservoirs and is defined as the temperature of the reversible heat flow exchanged with the system in the state of interest. \\

We close this section with the consideration of two gases with state variables $p_1,V_1$ and $p_2,V_2$, internal energies $U_1,U_2$ and entropies $S_1,S_2$. 
The internal energy of the composed system of the two gases, $U_{12} := U_1 + U_2$ then fulfils
\begin{align}
{\rm d}U_{12} = -p_1{\rm d}V_1 -p_2{\rm d}V_2 +T_1{\rm d}S_1 +T_2{\rm d}S_2\,.
\end{align}
In the same spirit as done above for a single system one can now define the temperature of gas $1$ as 
\begin{align}
T_1 = \left.\frac{\partial U_{12}}{\partial S_1}\right|_{V_1,V_2,S_2}\,.
\end{align}
As before, it will indeed be a state variable, but it would not be a good idea to call $T_1$ the ``temperature of the (total) system'', as it is not the only temperature one can assign to the system. One could have done the same with exchanged indices $1\leftrightarrow2$ and obtain $T_2$, which will in general not yield the same result. \\

We conclude that calling $T$ the ``temperature of a system'' is always a matter of interpretation. 
The good thing about this is that one does not have to call anything a temperature of a system in order to apply the thermodynamic theory. It is sufficient to think of $T$ as the temperature of a heat flow, which is a well-defined concept without space for ambiguities.

\newpage
\section{Example: The ideal gas}
\label{sec:exideal}

\begin{center}
\fcolorbox{OliveGreen}{white}{\begin{minipage}{0.98\textwidth}
\centering
\begin{minipage}{0.95\textwidth}
\ \\
\textcolor{black}{\textbf{Postulates:}}
-\\
\ \\
\textcolor{black}{\textbf{New notions:}}
-

\paragraph{\textcolor{black}{Summary:}}
This section treats the example of the ideal gas in full detail. Starting with the postulates and empirical observations the internal energy as well as the entropy function of the ideal gas is formally derived. 
From there one can define the temperature of an ideal gas, a concept that is eventually used to illustrate what happens in an irreversible heat flow between two ideal gases of different temperatures. 
\vspace{.1cm}

\end{minipage}
\end{minipage}}
\end{center}

\ \\

The example of a gas, in particular of the ideal gas, has been mentioned in several instances in this work. 
In all cases so far, the details were only introduced to the level at which they needed to be known. In this section we want to consider the example of the ideal gas from the beginning to the end, i.e.\ discuss the postulates and how they are reflected in this example in detail and full length. 
This will culminate in the derivation of the ideal gas law, which connects the state variable pressure and volume, which are initially assumed to describe the states, with the absolute temperature $T$.

\subsection{From thermodynamic processes to internal energy and entropy}

In order to do so we first need to make clear what can be experimentally observed when dealing with an ideal gas. This tells us which processes exist and how to compute their work cost.
As emphasized in the first sections of this paper, these are the necessary inputs to the theory of phenomenological thermodynamics. 
Once these inputs are formulated, one can check whether the postulates are satisfied and, if yes, the framework can be used to derive the other concepts of interest like thermodynamic entropy.  \\

Denote by $A$ the system of the ideal gas. 
We choose the states to be described by pressure $p$ and volume $V$. This means that an arbitrary state $\sigma$ can always be written as 
\begin{align}
\sigma = ( p,V)\in\left( \mathbbm{R}_{>0} \right)^2 =: \Sigma_A  \,.
\end{align}
Notice that for this treatment the amount of substance, usually denoted by $n$ (the number of moles) or $N$ (the number of particles), is here not a variable but a characteristic constant for the system, i.e.\ $n=n_A$. 
One could also consider $n$ to be an additional variable in the tuple describing states and carry out the steps that follow below.
However, in order to see a proper application of the theory to the example this is not necessary and would make it unnecessarily complicated. For the sake of an accessible example we do not do this here. \\

Now, when working with this system, one can find that any state of the form $(p_1,V)$ can be transformed into $(p_2,V)$ with $p_2\geq p_1$ by a $C^1$-quasistatic work process on $A$ with the curve $\gamma(\lambda) = ((1-\lambda)p_1 + \lambda p_2, V)$. 
This can for instance be done by slowly applying friction on the container from all sides, thereby heating up the gas (in an intuitive sense). 
Such processes are here said to be of \emph{type $1$}.
During this quasistatic process, in which the volume is kept constant, the differential work on $A$ is found to be
\begin{align}
\label{eq:irrevdiffW}
\delta W_A = \frac{3}{2} V \mathrm{d}p\,.
\end{align}

Next, one can observe that whenever two states $(p_1,V_1)$ and $(p_2,V_2)$ fulfil 
$p_1V_1^\frac{5}{3} = p_2V_2^\frac{5}{3}$ then there exists a reversible $C^1$-quasistatic work process on $A$ transforming one into the other along the curve with $pV^\frac{5}{3}=\mathrm{const.}$\footnote{We tacitly assume here that the ideal gas is monoatomic. This implies that the number of degrees of freedom is $f=3$, which yields the prefactor $\tfrac{f}{2}=\tfrac{3}{2}$ for the differential work in constant volume processes. Likewise, this sets the isentropic exponent to $\kappa = \tfrac{5}{3}$. In a more general consideration, one could work with $f$ and $\kappa$ as variables that satisfy $\kappa=\tfrac{f+2}{f}$. }
This is done by compressing or expanding the gas slowly while keeping it isolated from its surroundings.
These processes are said to be of \emph{type $2$}.
During reversible $C^1$-quasistatic processes on $A$, i.e.\ processes of type $2$, the differential work can be computed as 
\begin{align}
\label{eq:revdiffW}
\delta W_A = -p\,\mathrm{d}V\,.
\end{align}

The typical set of work processes $\mathcal{P}_A$ associated with the ideal gas contains all finite alternating sequences of processes of type 1 and type 2. We here work with exactly these work processes, i.e.\ we say that $\mathcal{P}_A$ contains all these sequences and no other processes. As a consequence, any work process can be written as a finite alternating sequence of processes of type 1 and type 2. \\

Third, it turns out that for all states $(p_1,V_1)$ and $(p_2,V_2)$ with $p_1V_1=p_2V_2$ there exists a heat reservoir $R\in\mathcal{R}$ such that they can be reversibly transformed into each other through a reversible $C^1$-quasistatic work process on $A\vee R$ along the curve $pV=\mathrm{const.}$ in $\Sigma_A$.
Even more, the temperature $T$ of the involved reservoir is found to depend linearly on $pV$, i.e.\ there exists a (system-specific) constant $n_A$ such that $pV=n_ART$, where $R$ is the ideal gas constant.\footnote{We all know that $n_A$ will turn out to be the amount of substance (in moles) of the ideal gas $A$. However, for the purposes presented here this is not important.}
Such processes are \emph{type $3$}. 
Also for this type of processes, Equation~(\ref{eq:revdiffW}) gives the differential work. In fact, it turns out that the differential work of any reversible quasistatic process on $A$ has this form. \\

We summarize the necessary assumptions outside of the framework that must be given as an input to the theory in this example: 
\begin{itemize}
	\item the description of states as $\sigma = (p,V)\in\left(\mathbbm{R}_{>0}\right)^2 = \Sigma_A$
	
	\item quasistatic work process on $A$ of type 1 along constant $V$, pressure increasing, 
	with $\delta W_A = \tfrac{3}{2}V \mathrm{d}p$
	
	\item quasistatic reversible work processes on $A$ of type 2 along constant $pV^\frac{5}{3}$
	
	\item work processes on $A$ are exactly those which can be written as a finite alternating
	sequence of processes of type $1$ and type $2$
	
	\item  quasistatic reversible work processes on $A\vee R$ for $R\in\mathcal{R}$ of type 3 
	along constant $pV$ with $pV = n_ART$ for the temperature $T$ of the reservoir $R$
	
	\item for all reversible quasistatic process involving $A$ (in particular type 2 and 3)
	 the differential work is given by $\delta W_A = -p\, \mathrm{d}V$
	
\end{itemize}

The choice of these three type of processes as our ``basic processes'' on which the thermodynamic description of the system is based is somewhat arbitrary. One could come up with different types of processes and work processes that allow us to describe state changes on $A$ and their thermodynamic properties. However, the choice is here made such that the calculations of the internal energy and the entropy can be carried out easily. \\

\begin{figure}[]
\begin{center}
	\begin{tikzpicture}[scale=2]
    \draw[<->] (2.5,.5) node(xline)[right] {$V$} -|
                    (0,2.5) node(yline)[above] {$p$};
    \draw[dashed,samples=100,domain=0.6:2.3] 
          plot(\x,{1/(\x)^(1.66)+0.5});   
    \draw[dashed,samples=100,domain=0.45:2.3]
          plot(\x,{.6/(\x)^(1.66)+0.5});   
    \draw[dashed,samples=100,domain=0.3:2.3] 
          plot(\x,{.3/(\x)^(1.66)+0.5});   
    \draw[dashed,samples=100,domain=0.155:2.3] 
          plot(\x,{.1/(\x)^(1.66)+0.5});   

	\draw[very thick,samples=100,domain=0.6:1.6] 
		plot(\x,{.3/(\x)^(1.66)+0.5});     
	\draw[very thick] (.6,1.2) -- (.6,1.9); 

	\draw[] (1.6,.45) node[below] {$V_1$} -- (1.6,.55);
	\draw[dotted] (1.6,.55) -- (1.6,.64);
	\draw[] (.6,.45) node[below] {$V_2$} -- (.6,.55);
	\draw[dotted] (.6,.55) -- (.6,1.2);

	\draw[] (-.05,.64) node[left] {$p_1$} -- (.05,.64);
	\draw[dotted] (.05,.64) -- (1.6,.64);
	\draw[] (-.05,1.9) node[left] {$p_2$} -- (.05,1.9);
	\draw[dotted] (.05,1.9)-- (.6,1.9);	
	
	\draw[] (.5,1.1) node[] {$\sigma'$};

\begin{scope}[xshift=4cm]
    
    \draw[<->] (2.5,.5) node(xline)[right] {$V$} -|
                    (0,2.5) node(yline)[above] {$p$};
	\draw[smooth,samples=100,domain=0.07:2.3] 
          plot(\x,{0.15/\x+0.5}) ;	
	\draw[smooth,samples=100,domain=0.18:2.3] 
          plot(\x,{0.4/\x+0.5}) ;   
	\draw[smooth,samples=100,domain=0.35:2.3] 
          plot(\x,{0.8/\x+0.5}) ;

	\draw[dashed,samples=100,domain=0.6:2.3] 
		plot(\x,{1/(\x)^(1.66)+0.5});   
    \draw[dashed,samples=100,domain=0.45:2.3] 
          plot(\x,{.6/(\x)^(1.66)+0.5});   
    \draw[dashed,samples=100,domain=0.3:2.3] 
          plot(\x,{.3/(\x)^(1.66)+0.5});   

	\draw[very thick,samples=100,domain=0.65:1.3]
		plot(\x,{.6/(\x)^(1.66)+0.5});     
	\draw[very thick,samples=100,domain=0.65:1.4] 
          plot(\x,{0.8/\x+0.5}) ;
	\draw[very thick,samples=100,domain=0.85:1.4]
		plot(\x,{1/(\x)^(1.66)+0.5});   
		
	\draw[] (1.3,.45) node[below] {$V_1$} -- (1.3,.55);
	\draw[dotted] (1.3,.55) -- (1.3,.89);
	\draw[] (.85,.45) node[below] {$V_2$} -- (.85,.55);
	\draw[dotted] (.85,.55) -- (.85,1.81);

	\draw[] (-.05,.89) node[left] {$p_1$} -- (.05,.89);
	\draw[dotted] (.05,.89) -- (1.3,.89);
	\draw[] (-.05,1.81) node[left] {$p_2$} -- (.05,1.81);
	\draw[dotted] (.05,1.81)-- (.85,1.81);
	
	\draw[] (.55,1.7) node[] {$\sigma'$};
	\draw[] (1.5,1.2) node[] {$\sigma''$};

\end{scope}

	\end{tikzpicture}
\caption{The $(p,V)$-diagram of the processes discussed in the proof showing that Postulates~\ref{post:qsfirst} and~\ref{post:qsentropy} are fulfilled. 
The quasistatic curves $\gamma$ are marked with thick lines. 
(a) A piecewise $C^1$-quasistatic work process on $A$ consisting of the concatenation of a process of type $2$ followed by one of type $1$ transforms $\sigma_1 = (p_1,V_1)$ into $\sigma_2 = (p_2,V_2)$.
The dashed lines are graphs of constant $pV^\frac{5}{3}$.
(b) The reversible piecewise $C^1$-quasistatic process on $A\vee R'$ is a concatenation of a process of type $2$ transforming $\sigma_1 = (p_1,V_1)$ into $\sigma' = (p',V')$ followed by a process of type $3$ turning the state into $\sigma'' = (p'',V'')$, followed by another process of type $2$ transforming this into the final state $\sigma_2 = (p_2,V_2)$.
The dashed lines are as in (a), the continuous lines are graphs of constant $pV$. 
}
\label{fig:pV}
\end{center}
\end{figure}
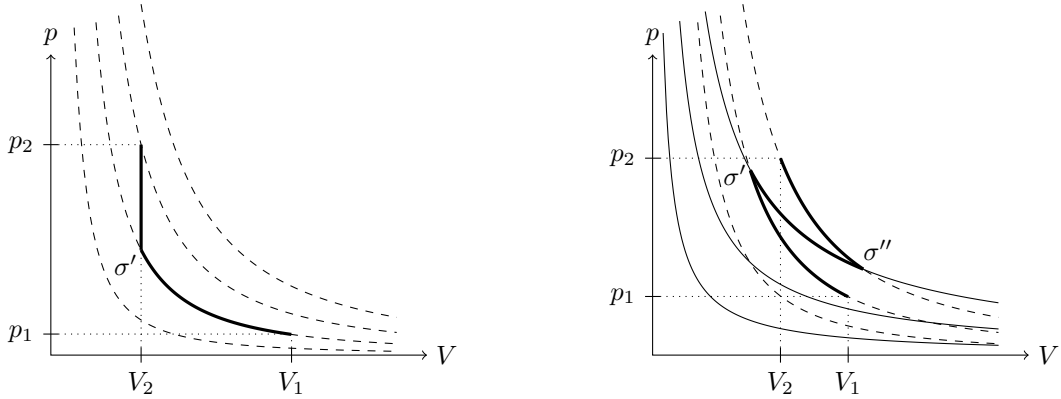

In addition, the existence and the form of the processes of type $1$, $2$, and $3$ allow us to conclude without much effort that the ideal gas $A$ satisfies Postulate~\ref{post:qsfirst} and Postulate~\ref{post:qsentropy}.
Regarding Postulate~\ref{post:qsfirst} consider Figure~\ref{fig:pV} (a):
\begin{itemize}
	\item [(o)] Obviously, $\dim \Sigma_A=2$. The work processes of types $1$ and $2$ can pass through
	any state. Since the $C^1$-curves $V=\mathrm{const.}$ and $pV^\frac{5}{3}=\mathrm{const.}$ 
	pass through any state s.t.\ the derivatives of the curves are linearly independent, 
	it follows that (o) is satisfied.
	\item [(i)] For two given states, choose the numbering $(p_1,V_1)$ and $(p_2,V_2)$ s.t.\ 
	$p_1V_1^\frac{5}{3} \leq p_2V_2^\frac{5}{3}$.
	Then the state change $(p_1,V_1) \rightarrow (p_2,V_2)$ can always be obtained by process
	of type $1$, transforming $(p_1,V_1) \rightarrow (p',V_2)$, followed by a process of type $1$
	raising $p'$ to $p_2$. In total this is a piecewise $C^1$-quasistatic work process on $A$.
	\item [(ii)] Since any work process on $A$ can be written as a sequence of
	 alternating processes of types $1$ and $2$, it can easily be
	 seen that the total work done on $A$ only depends on the initial and final state.
\end{itemize}

\noindent In a similar spirit, it follows for Postulate~\ref{post:qsentropy}, see Figure~\ref{fig:pV} (b):
\begin{itemize}
	\item [(o)] Reversible $C^1$-quasistatic processes of type $2$ and $3$ can pass through 
	any state and their curves $pV^\frac{5}{3}=\mathrm{const.}$ and $pV=\mathrm{const.}$
	have linearly independent first derivatives in any point. 
	\item [(i)] Any two states $(p_1,V_1)$ and $(p_2,V_2)$ can be reversibly transformed into
	each other by the concatenation of processes of type $2$, followed by type $3$, 
	and again followed by type $2$.
	In fact, the reservoir $R'$ needed for executing the process of type $3$ can be chosen 
	arbitrarily, in particular, it can have any temperature $T'$.\footnote{
	This is a consequence of the fact that any line of constant non-zero $pV$ has exactly one intersection with any other line of constant non-zero $pV^\frac{5}{3}$.}
	In total, this is a reversible piecewise $C^1$-quasistatic process on $A\vee R'$.
\end{itemize}

The other postulates do not only depend on the system $A$ but on the whole structure of systems and processes. We do not check them explicitly here but with a reasonable choice of $\Omega$ (the thermodynamic world, see Postulate~\ref{post:systems}) and thus $\mathcal{S}$ (the set of thermodynamic systems) as well as $\mathcal{P}$ they are fulfilled as one intuitively expects. \\

We can now compute the internal energy difference (and hence the internal energy function up to the constant $U_A^0$) of two arbitrary states $\sigma_1 = (p_1,V_1)$ and $\sigma_2 = (p_2,V_2)$ according to Definition~\ref{def:Ustrong}. This can be done via the process suggested in the discussion of Postulate~\ref{post:qsfirst} (i) above and shown in Figure~\ref{fig:pV} (a).  The intermediate state $\sigma'=(p',V_2)$ fulfils $p'V_2^\frac{5}{3} = p_1V_1^\frac{5}{3}$ as it is obtained from $\sigma_1$ through a process of type $2$.
We find:
\begin{align}
\begin{split}
\Delta U_A &= U_A(\sigma_2) - U_A(\sigma_1) \\
&= \int_{\sigma_1}^{\sigma'} -p\,\mathrm{d}V + \int_{\sigma'}^{\sigma_2} \frac{3}{2}V\mathrm{d}p\\
&= -p_1V_1^\frac{5}{3} \left[ -\frac{3}{2}V^{-\frac{2}{3}} \right]_{V_1}^{V_2} + \frac{3}{2} V_2 \big[ p \big]_{p'}^{p_2}\\
&= \frac{3}{2} p_1V_1^\frac{5}{3}\left( V_2^{-\frac{2}{3}} - V_1^{-\frac{2}{3}} \right) +
\frac{3}{2} V_2\left( p_2-p_1\left(\frac{V_1}{V_2}\right)^\frac{5}{3}\right)\\
&= \frac{3}{2} (p_2V_2 - p_1V_1)\,.
\end{split}
\end{align}
With the reference energy $U_A^0$ we can conclude that for any state $\sigma = (p,V)\in\Sigma_A$ the internal energy is
\begin{align}
\label{eq:pVint}
U_A(\sigma) = \frac{3}{2}pV + U_A^0\,.
\end{align}
The internal energy function is obviously $C^\infty$ on $\Sigma_A$, in particular $C^1$. As proved in Section~\ref{sec:quasistatic}, the latter must be the case.\\

Using the process suggested in the discussion of Postulate~\ref{post:qsentropy} above we can compute the entropy difference of the two states $\sigma_1$ and $\sigma_2$ according to Definition~\ref{def:entropy}.
The intermediate states $\sigma'=(p',V')$ and $\sigma''=(p'',V'')$ fulfil 
$p'V' = n_ART' = p''V''$ as well as $p'(V')^\frac{5}{3} = p_1V_1^\frac{5}{3}$ and 
$p''(V'')^\frac{5}{3} = p_2V_2^\frac{5}{3}$, as can easily be seen from Figure~\ref{fig:pV} (b).
Also, the only non-trivial heat flow appears in the type 3 process transforming $(p',V')$ into $(p'',V'')$ at constant temperature $T'$. Hence:
\begin{align}
\begin{split}
\Delta S_A &= S_A(\sigma_2) - S_A(\sigma_1) = \int_{\sigma_1}^{\sigma_2} \mathrm{d}S_A
= \int_{\sigma_1}^{\sigma_2} \frac{\delta Q_A}{T}  
= \int_{\sigma'}^{\sigma''} \frac{\delta Q_A}{T'} \\
&= \int_{\sigma'}^{\sigma''} \frac{\mathrm{d}U_A +p\,\mathrm{d}V}{T'} \\
&= \frac{3}{2T'}(p''V'' - p'V') + \frac{1}{T'} \int_{V'}^{V''} p\,\mathrm{d}V \\
&= 0 + \frac{p'V'}{T'} \int_{V'}^{V''} \frac{1}{V}\,\mathrm{d}V \\
&= n_AR \ln\left( \frac{V''}{V'} \right) \\
&= n_AR \left(\frac{3}{2} \ln\left(\frac{p_2}{p_1}\right) + \frac{5}{2} \ln\left(\frac{V_2}{V_1}\right)\right)\,.
\end{split}
\end{align}
As before, with the reference entropy $S_A^0$ it follows that the entropy of any state $\sigma = (p,V)$ can be written as
\begin{align}
\label{eq:pVentropy}
S_A(\sigma) = n_AR \left(\frac{3}{2} \ln\left(\frac{p}{p_0}\right) + \frac{5}{2} \ln\left(\frac{V}{V_0}\right)\right) + S_A^0 \,,
\end{align}
where $(p_0,V_0) = \sigma_0$ is the reference state for computing the entropy.
As was the case before for the internal energy, it immediately follows that the entropy is $C^\infty(\Sigma_A)$, in particular $C^1(\Sigma_A)$.
Notice that Equations~(\ref{eq:pVint}) and~(\ref{eq:pVentropy}) are the standard expressions for the internal energy and the entropy of an ideal gas in the variables pressure $p$ and volume $V$. \\

We now want to compute the internal energy $U(S,V)$ as a function of the entropy and the volume. This will allow us to determine an expression for the \emph{temperature of the ideal gas}, as is discussed in Section~\ref{sec:tsys}.
With Equations~(\ref{eq:pVint}) and~(\ref{eq:pVentropy}) it follows that 
\begin{align}
\label{eq:USV}
U(S,V) = \frac{3}{2}p_0V_0 \left(\frac{V}{V_0}\right)^{-\frac{2}{3}} \cdot
\exp\left( \frac{2}{3} \frac{S-S_A^0}{n_AR} \right) +U_A^0 \,.
\end{align}
Taking the derivative of $U(S,V)$ at constant $V$ with respect to $S$ we obtain 
the temperature of the ideal gas:
\begin{align}
\label{eq:pVT}
T = \left.\frac{\partial U}{\partial S}\right|_V 
=\frac{3}{2}p_0V_0 \left(\frac{V}{V_0}\right)^{-\frac{2}{3}} \cdot \frac{2}{3}\frac{1}{n_AR}
\exp\left( \frac{2}{3} \frac{S-S_A^0}{n_AR} \right)
= \frac{2}{3}\frac{1}{n_AR} (U_A-U_A^0)\,.
\end{align}
There are important conclusions one can draw from this.
First, the internal energy of an ideal gas can be written as $U_A(\sigma)-U_A^0 = \tfrac{3}{2}n_ART$. Hence, the temperature of the ideal gas in state $\sigma$, $T$, in a state variable and $U_A$ can be written as a function of $T$ only.
The fact the $T$ defined according to Equation~(\ref{eq:pVT}) is a state variable in this setting has already been argued in Section~\ref{sec:tsys} on more general assumptions and can be observed in the explicit example here.\\

Second, combining this result with Equation~(\ref{eq:pVint}), we see that the temperature of a state $\sigma = (p,V)$ is always $T=\tfrac{pV}{n_AR}$, which is the ideal gas law.
This is remarkable, as this expression for the temperature of the gas is exactly the same as we had for the \emph{temperature of the heat reservoir $R$} with the help of which the state $\sigma$ can be transformed reversibly (through a quasistatic process on $A\vee R$) to any other state with the same temperature. 
Hence, what we call the \emph{temperature of the system} (or the state of the system, for that matter) is equal to the \emph{temperature of a heat reservoir} with which the system can reversibly interact while exchanging heat in this state. 
Again, the fact that both notions of temperature are the same has been shown more generally in Section~\ref{sec:tsys}.\\


\subsection{A heat flow between two ideal gases of different temperatures}

\begin{figure}
\begin{center} 
	\begin{tikzpicture}[scale=.65]
	
	\begin{scope}[xshift=1cm]
	\draw[](-.8,0) node[] {(a)};
	 
	\draw[xshift=.1cm] (.5,-.5) node[above] {$A_1$} -- (1,-.5) -- (1,.5) -- (0,.5) -- (0,-.5) -- (.5,-.5);
	\draw[] (3,-.5) node[above] {$A_2$} -- (3.5,-.5) -- (3.5,.5) -- (2.5,.5) -- (2.5,-.5) -- (3,-.5);
	\draw[->] (1.3,0) node[above right] {$Q$} -- (2.3,0);

	\end{scope}
	
	\begin{scope} [xshift = 7cm]
	
		\draw[](-.8,0) node[] {(b)};
		
		\draw[xshift=.1cm] (.5,-.5) node[above] {$A_1$} -- (1,-.5) -- (1,.5) -- (0,.5) -- (0,-.5) -- (.5,-.5);
		\draw[very thick] (3,-.5) node[above] {$R_1$} -- (3.5,-.5) -- (3.5,.5) -- (2.5,.5) -- (2.5,-.5) -- (3,-.5);
		\draw[->] (1.3,0) node[above right] {$Q$} -- (2.3,0);
		
		\begin{scope}[xshift = 6cm]

		\draw[](-.6,0) node[] {(c)};

		\draw[xshift=.1cm, very thick] (.5,-.5) node[above] {$R_2$} -- (1,-.5) -- (1,.5) -- (0,.5) -- (0,-.5) -- (.5,-.5);
		\draw[] (3,-.5) node[above] {$A_2$} -- (3.5,-.5) -- (3.5,.5) -- (2.5,.5) -- (2.5,-.5) -- (3,-.5);
		\draw[->] (1.3,0) node[above right] {$Q$} -- (2.3,0);
		\end{scope}
		
		\begin{scope}[xshift = -3.6cm, yshift = -3cm]

		\draw[](-.8,0) node[] {(d)};
		
		\draw[xshift=.1cm] (.5,-.5) node[above] {$A_1$} -- (1,-.5) -- (1,.5) -- (0,.5) -- (0,-.5) -- (.5,-.5);
		\draw[very thick] (3,0) node[] {$R'$} circle (.6cm);
		\draw[very thick] (5.4,0) node[] {$R''$} circle (.6cm);
		\draw[->] (1.3,0) node[above right] {$Q$} -- (2.3,0);
		\draw[->] (3.8,0) node[above right]{$Q$} -- (4.7,0);
		\draw[->] (6.1,0) node[above right]{$Q$} -- (7,0);		
		
		\begin{scope}[xshift = 7.2cm]
		\draw[very thick] (.5,0) node[] {$R'''$} circle (.6cm);
		\draw[] (3,-.5) node[above] {$A_2$} -- (3.5,-.5) -- (3.5,.5) -- (2.5,.5) -- (2.5,-.5) -- (3,-.5);
		\draw[->] (1.3,0) node[above right] {$Q$} -- (2.3,0);
		\end{scope}
		
		\end{scope}
	\end{scope}
	\end{tikzpicture}
\end{center}
\caption{(a) When thermally connecting two ideal gases at different temperatures a non-zero heat flow  $Q$ can be observed to flow from the warmer to the colder gas. Call this process $p\in\mathcal{P}_{A_1\vee A_2}$.
(b) \& (c) According to Definition~\ref{def:heatatt} the heat $Q$ flows at temperature $T$ if there exist two reservoirs $R_1\sim R_2$, both at temperature $T$, and the process can be split up into two processes $p_i\in\mathcal{P}_{A_i \vee R_i}$ such that for $A_i$ the thermodynamic properties of $p$ and $p_i$ are the same. 
(d) Thinking of a rod between the two gases has the advantage that it becomes manifest that the rod could be intercepted at any point with a heat reservoir at a certain temperature. Since the decrease of temperature over the rod from $A_1$ to $A_2$ will be continuous, any temperature in the interval $[T_2,T_1]$ is attained at one point.
In the figure, three different reservoirs at temperatures $T'\geq T''\geq T''$ are depicted.
}
\label{fig:idgasheatapp}
\end{figure}
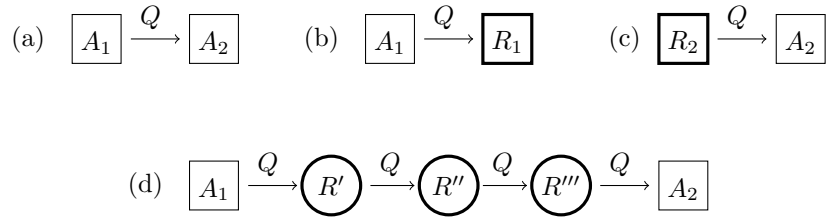

We apply the gained knowledge about the temperature of an ideal gas to Definition~\ref{def:heatatt}, the temperature of a heat flow. 
We know from Lemma~\ref{lemma:revheatuniquetapp} that the temperature of a reversible heat flow $Q\neq0$ is unique if it exists. 
In Example~\ref{ex:heatdifft} we have also seen that if the heat flow is not reversible it is very well possible that different temperatures can be assigned to it. 
With the concept of the temperature of an ideal gas it is now possible to refine this argument. 

Consider two ideal gases $A_1,A_2\in\mathcal{A}$ in states $\sigma_1\in\Sigma_{A_1}$ and 
$\sigma_2\in\Sigma_{A_2}$, respectively. Suppose that the temperatures according to Equation~(\ref{eq:pVT}) are $T_1>T_2$.
When thermally connecting the two gases, e.g.\ by bringing them into physical contact with the same thin metallic rod, we know from experience that a non-zero heat flow $Q$ from the hotter gas to the colder one can be observed. What temperature(s) can be assigned to this heat flow?

Using Definition~\ref{def:heatatt} we see that this question can be rephrased as: what heat reservoirs can be put in between the two systems without changing the	 thermodynamic properties of the process on the systems $A_1$ and $A_2$?
Figure~\ref{fig:idgasheatapp} depicts this interpretation of the question.
Thinking of a long rod used to establish the contact between the systems, one can in principle cut the rod in half at any point and find a reservoir with the right temperature to put in between. This will not change anything with respect to the thermodynamic properties of $A_1$ and $A_2$ under the process. 
Since the temperature decrease over the rod is continuous\footnote{This is again a statement that is supported by experience, i.e.\ experiment, rather than the theoretical framework.}
one can find places to put reservoirs in between at all temperatures in the interval $[T_2,T_1]$.
Hence, any temperature $T$ with $T_2\leq T \leq T_1$ can be assigned to the heat flow $Q$ and our framework provides a precise language to talk about these intuitive but usually only vaguely described concepts. 


\newpage

\section{Scaling}
\label{sec:scaling}

\begin{center}
\fcolorbox{OliveGreen}{white}{\begin{minipage}{0.98\textwidth}
\centering
\begin{minipage}{0.95\textwidth}
\ \\
\textcolor{black}{\textbf{Postulates:}}
-\\
\ \\
\textcolor{black}{\textbf{New notions:}}
scaled system, intensive and extensive state variables, principle of maximum entropy

\paragraph{\textcolor{black}{Summary:}}
The scaling of thermodynamic systems is introduced and used to prove the principle of maximum entropy.
\vspace{.1cm}

\end{minipage}
\end{minipage}}
\end{center}

\ \\

Often one wants to be able to talk about the same type of systems in different sizes, i.e.\ about scaled systems. 
Intuitively, the properties of a system $2S$ that is ``twice as big'' as $S$ are clear: 
the state spaces have the same structure with the difference that ``extensive'' (i.e.\ homogeneous) variables of corresponding states are scaled by a factor of $2$; likewise the thermodynamic processes that can be carried out on $2S$ are the ``scaled'' processes $p$ that can be carried out on $S$.\\

We immediately notice that these intuitive ideas lack a formal clarification of the notions ``extensive variable'' and ``scaled process''.
In the example of an ideal gas as discussed in Section~\ref{sec:exideal} this is rather easy. 
Let $2A$ be the ideal gas from the previous sections scaled by a factor of $2$.
We know that the variable that should be scaled is the volume, while the pressure is ``intensive'', i.e.\ is not scaled with the system size. 
Also, the system specific constant $n_A$ (corresponding to the amount of substance) scales as $n_{2A}=2n_A$ like the volume. 
On the other hand the pressure does not change under the scaling.
It is important to realize that the knowledge about the different behaviours of the state variables under scaling cannot come out of the framework but must be seen as an input to it. 

With this, states $2\sigma = (p, 2V) \in\Sigma_{2A} \leftrightarrow \sigma = (p,V)\in\Sigma_A$ are in a 1-1 correspondence. Likewise, the (quasistatic) processes connecting two states $2\sigma_1$ and $2\sigma_2$ fulfil $\delta W_{2A} = 2\delta W_A$, as can easily be seen from Eq.~(\ref{eq:irrevdiffW}) and (\ref{eq:revdiffW}), and are in a 1-1 correspondence with the (quasistatic) processes connecting $\sigma_1$ and $\sigma_2$.
It follows that, when choosing\footnote{Even if one does not make this choice, the physically relevant internal energy \emph{differences} still scale linearly with the factor $2$. The analogous statement holds for the entropy.} $U_{2A}^0 = 2U_A^0$, the internal energy must satisfy $U_{2A} = 2U_A$, which is indeed the case as can easily be checked in Eq.~(\ref{eq:pVint}). 
In order to obtain the analogous result for entropy it suffices to choose $S_{2A}^0 = 2S_A^0$. Again, we see that Eq.~(\ref{eq:pVentropy} satisfies $S_{2A} = 2S_A$.\footnote{Here it is important not to forget that the reference state $\sigma_0 = (p_0,V_0)\in\Sigma_A$ from which the entropy is computed also scales to $2\sigma_0 = (p_0,2V_0)\in\Sigma_{2A}$. Otherwise one could come to the wrong conclusion that the entropy does not scale linearly, i.e.\ is not extensive.}
The derived concept of the temperature of an ideal gas in state $\sigma\in\Sigma_A$, Eq.~(\ref{eq:pVT}), does not scale with the system size, as the temperature of $2\sigma$ is the same as the temperature of the state $\sigma$.

\subsection{General definition of scaled systems}

More generally, we consider the case of a scaling factor $\lambda = n \in\mathbbm{N}_{>0}$ and explain the proper way of talking about the scaled system $nS$. 
Let $S^{(2)} \hat= \cdots \hat= S^{(n)} \hat= S$ be copies of $S$ and let $\varphi^{(i)}$ be the thermodynamic isomorphism describing the equivalences. 
For processes $p,p^{(2)},\dots,p^{(n)}\in\mathcal{P}$ and states $\sigma\in\Sigma_S, \sigma^{(i)}\in\Sigma_{S^{(i)}}$ with 
$\varphi^{(i)}(p)=p^{(i)}$ and $\varphi^{(i)}_\Sigma(\sigma)=\sigma^{(i)}$
we write $p^{(i)}\hat=p$ and $\sigma^{(i)}\hat=\sigma$.

\begin{definition}[Scaled system by a factor of $\lambda=n\in\mathbbm{N}_{>0}$]
\label{def:scaling}
The system $nS\in\mathcal{S}$ is the \emph{scaled system $S\in\mathcal{S}$ by a factor of $n$} if there exists an \emph{embedding} $\psi:\mathcal{P} \rightarrow \mathcal{P}$ of $nS$ into 
$S\vee S^{(2)}\vee\cdots\vee S^{(n)}$, where $S\hat=S^{(i)}$ are copies of each other. 
The embedding $\psi$ is an injective mapping with the properties 
\begin{itemize}
	\item [(i)] for $\tilde p\in\mathcal{P}$ acting on $nS$ it holds 
	$\psi(\tilde p)= p\vee p^{(2)}\vee \cdots \vee p^{(n)}$ 
	for some $p\in\mathcal{P}$ acting on $S$ and $p^{(i)}\hat=p$,\footnote{Clearly, if $p$ acts on $S$, then its copy $p^{(i)}$ acts on the copy $S^{(i)}$ of $S$.}
	
	\item [(ii)] $W_{nS}(\tilde p) = W_{S\vee S^{(2)} \vee \cdots \vee S^{(n)}}(\psi(\tilde p))$.
	
\end{itemize}
\end{definition}
The embedding can be understood as the $nS = S\vee S^{(2)}\vee\cdots\vee S^{(n)}$ \emph{with restrictions}, the restriction being that the thermodynamic processes that can be applied to $nS$ are exactly those of the form $ p\vee p^{(2)}\vee \cdots \vee p^{(n)} $ with $p^{(i)}\hat=p$.
As a consequence, the states $n\sigma\in\Sigma_{nS}$ correspond to states of the form $\sigma\vee \sigma^{(2)}\vee \cdots\vee \sigma^{(n)}$ with $\sigma^{(i)}\hat=\sigma$ for $\sigma\in\Sigma_S$.\\

If a process $\tilde p\in\mathcal{P}$ acts on $nS$ we write $\tilde{p} =:np$, where $p\in\mathcal{P}$ is the process acting on $S$ such that $\psi(2p)=p\vee p^{(2)}\vee \cdots \vee p^{(n)}$. 
This notation is well-defined since $\psi$ is injective and it can be applied to any $\tilde p \in\mathcal{P}$ acting on $nS$ due to (i). \\


The work function $W_{nS}$ scales with the system's size. This follows from Definition~\ref{def:scaling} (ii),
\begin{align}
W_{nS}(np) \stackrel{\mathrm{(ii)}}{=} 
W_{S\vee S^{(2)} \vee \cdots \vee S^{(n)}}(\psi(\tilde np)) 
= W_S(p) + \sum_{i=2}^n W_{S^{(i)}}(p^{(i)}
= nW_S(p)\,.
\end{align}
In the second equality we used the additivity of work, while the third equality holds because $S^{(i)}\hat=S$ are copies of each other with $p^{(i)}\hat=p$. 
Consequently, the internal energy and the entropy also scale with the system size. 
They inherit the scaling property of the work function since any entropic and energetic quantity homogeneously depends on the work functions. \\


If $\lambda = \tfrac{1}{\nu}$ for some $\nu\in\mathbbm{N}_{>0}$ the scaled system $\lambda S $ is defined as the system satisfying 
\begin{align}
\underbrace{\lambda S \vee \cdots \vee \lambda S}_{\nu \text{ times}} \hat= S\,,
\end{align}
which allows us to extend the concept of scaling to at least\footnote{Being able to scale with rational factors, since the rational numbers are dense in $\mathbbm{R}$, we can ``approximate'' $\lambda S$ with $\lambda\in\mathbbm{R}$ arbitrarily well. However, making the notion ``approximate'' precise requires at least a topology on $\mathcal{S}$. We do not go further into this technical discussion and are content with rational factors for now.} the positive rational numbers $\lambda = \tfrac{\mu}{\nu}\in\mathbbm{Q}_{>0}$ by defining 
\begin{align}
\lambda S = \mu\left(\tfrac{1}{\nu} S \right)\,.
\end{align}

Generally, a state variable $X_S$ on the system $S$ is called \emph{extensive} if it satisfies 
\begin{align}
X_{\lambda S} (\lambda\sigma) = \lambda X_S(\sigma)\,,
\end{align}
where $X_{\lambda S}$ is the corresponding state variable on $\lambda S$. Similarly, it is called \emph{intensive} if 
\begin{align}
X_{\lambda S} (\lambda\sigma) = X_S(\sigma)\,.
\end{align}
With volume, internal energy and entropy (and the amount of substance, for that matter, even though we here treat it as a constant state variable) we have already encountered several extensive state variables of the ideal gas. 
The intensive state variables were pressure and temperature. \\

Clearly, in practice the scaling factor $\lambda$ cannot become arbitrarily small or large without surpassing the region of validity of the theory. 

For instance, if $1\,$l of an ideal gas $A$ at atmospheric pressure and temperature and is scaled with a factor $\lambda = 10^{-25}$ we know that the scaled system $\lambda A$ will contain less than a particle.
There is an obvious problem with the interpretation of $\lambda A $ as a physical system. 
Even if there were still a few particles in $\lambda A$ quantum effects may become dominant in the description of the system that were averaged out in the description of $A$. Hence, $\lambda A$ and $A$ will not be in correspondence with each other as was required above. 
More specifically, the differential work for reversible quasistatic process will most likely no longer read $\delta W_{\lambda A} = -\lambda p \, \mathrm{d}V$. 


On the other hand, when scaling an ideal gas $A$ with a factor $\lambda = 10^{30}$, the mass of the system $\lambda A$ will be of the order of a solar mass and effects like self-gravitation will become dominant.

Notice that the invalidity of thermodynamics at the very large or very small scale is similar to the one encountered in classical mechanics. In both cases it is not a shortcoming of the framework itself but rather the fact that the underlying theory (here: the one to define systems, states, processes, and to compute work) is only an approximation. 
Even though we know about the different effects at the very small or the very large scale, we keep considering general rational scaling factors $\lambda\in\mathbbm{Q}\smallsetminus\{0\}$ since it is impossible to decide on the region of validity on an abstract level. The general mathematical way of formalizing the scaling of systems cannot take these effects into account.

\subsection{Principle of maximum entropy}

In this section we apply the scaling of the ideal gas to prove the principle of maximum entropy for ideal gases. This principle is a consequence of the Entropy Theorem~\ref{thm:entropythm}. As will become evident during the developments in this section, the considerations presented here can be extended to other systems that share basic properties with the ideal gas. These properties are discussed at the end of this section.\\

From Eq.~(\ref{eq:pVint}) we immediately see that instead of describing states with the tuple $(p,V)$ we could also choose $(U,V)$ as the basic variables. In this section we will take this description,
\begin{align}
\sigma = (U,V)
\end{align}
due to the nice property that under scaling with the factor $\lambda$ one can write  
$\lambda \sigma = (\lambda U,\lambda V)$. 
Consider now the composition of two scaled ideal gases 
$A'=\lambda A$ and $A''=(1-\lambda)A$ for $\lambda\in(0,1)$ and denote 
$S=A'\vee A''$. 
States of $S$ are described as $\sigma = \sigma'\vee\sigma''$ 
for $\sigma'\in\Sigma_{A'}$ and $\sigma''\in\Sigma_{A''}$ and we may write $\sigma = (U',V',U'',V'')$. 
As shown in Lemma~\ref{lemma:entropyaddapp} the entropy of $S$ can be written as the sum of the entropies of its subsystems,
\begin{align}
S_S(\sigma) = S_{A'} (\sigma') + S_{A''}(\sigma'')\,.
\end{align}
It was emphasized in Section~\ref{sec:systemsprocessesstates} that the composition operation does not imply that the composed subsystems are in physical contact. They are simply described as one system. 
The state $\sigma'\vee \sigma''$ is thus sometimes called \emph{constrained} in the sense that the two subsystems (ideal gases here) cannot physically interact. 

Typically there exists a work process $q\in\mathcal{P}_S$ on $S$ that lets the subsystem exchange energy and volume for a sufficient amount of time\footnote{
In practice this means that the subsystems may exchange energy and volume (e.g.\ by means of a freely moving piston between them) for a time $t$ that is longer than the typical equilibration time of the system. 
}
such that $q$ transforms $\lfloor q \rfloor_S  = (U',V',U'',V'') =\sigma'\vee\sigma''$ into the final state
\begin{align}
\lceil q \rceil_S 
= (\lambda ( U'+U''), \lambda (V'+V''), (1- \lambda) ( U'+U''), (1- \lambda) (V'+V'')) \,,
\end{align}
i.e.\ the internal energies and the volumes have changed such that each subsystem has a share proportional to its size (in terms of $\lambda$). 
This process has work cost $W_{A'}(q)=W_{A''}(q)=0$ on both subsystems.
We call the final state \emph{unconstrained} and write $\lceil q \rceil_S =: \sigma'+\sigma'' \equiv (U'+U'', V'+V'')$. This notation is motivated by our choice of extensive variables and has to be explained.

First, we profit from the choice of the extensive variables $U$ and $V$ which give manifest meaning to the addition ``$+$'' of two states. Even though this is convenient, it is not a necessary choice in order to define this addition operation. One could also define it in terms of the variables pressure $p$ and volume $V$ such that when computing the internal energies they satisfy $\lambda(U'+U'')$ for $A'$ and $(1-\lambda)(U'+U'')$ for $A''$, for instance. 

Second, technically the $2$-tuple $(U'+U'', V'+V'')$ is not a proper description of a state on $S$, since the states of $S$ are described by $4$-tuples. However, by the definition of $\sigma'+\sigma''$ as the final state of $q$ on $S$, the states of the subsystems $A'=\lambda A$ and $A''=(1-\lambda)A$ can always be obtained from the notation $(U'+U'',V'+V'')$ as the part of it proportional to $\lambda$ and $(1-\lambda)$, respectively.
We will use this notation with its implicit meaning from now on.

Third, the ``decomposition'' of an unconstrained state $\sigma'+\sigma''$ into summands is not unique. That is, there exists arbitrarily many other states $\tilde \sigma'\in\Sigma_{A'}$ and 
$\tilde \sigma''\in\Sigma_{A''}$ with $ \tilde \sigma'+\tilde \sigma'' = \sigma' +\sigma''$. 

Finally, not every state $\sigma\in\Sigma_S$ is an unconstrained state. This is already clear from the construction of the unconstrained states as the final states of special processes.\\

Since $q\in\mathcal{P}_S$ is a work process on $S$ we can use the Entropy Theorem~\ref{thm:entropythm} to bound the entropy of the unconstrained state $\sigma'+\sigma''$ from below.
It states that 
\begin{align}
\label{eq:unconstrained}
S_{A'}(\sigma') + S_{A''}(\sigma'') \equiv S_S(\sigma'\vee\sigma'') \leq S_S(\sigma'+\sigma'')
\end{align}
with equality if $q$ is reversible. Typically, $q$ is reversible if and only if it is an identity process, i.e.\ if and only if the constrained state was already unconstrained, which means that $\lambda U'' = (1-\lambda)U'$ and $\lambda V'' = (1-\lambda)V'$. That is, the states $\tilde \sigma' = (\lambda(U'+U''), \lambda(V'+V'')) \in\Sigma_{A'}$ and $\tilde \sigma'' = ((1-\lambda)(U'+U''),(1- \lambda)(V'+V''))\in\Sigma_{A''}$ satisfy $\tilde \sigma' \vee \tilde \sigma'' = \tilde \sigma' + \tilde\sigma''$. 
For these states the bound in Eq.~(\ref{eq:unconstrained}) is saturated. 

With this the derivation of the principle of maximum entropy for ideal gases is done.

\begin{thm}[Principle of maximum entropy]
\label{thm:maxentropy}
For an ideal gas $A\in\mathcal{S}$ and a scaling parameter $\lambda\in(0,1)$ let $S=A'\vee A''$ be the composed system of $A' = \lambda A$ and $A'' = (1-\lambda)A$. Furthermore, let $(U,V)\in\Sigma_S$ be an unconstrained state. Then
\begin{align}
\label{eq:maxentropy}
S_S(U,V) 
= \max_{\substack{U'+U'' = U \\ V'+V''=V}} S_{A'} (U',V') + S_{A''}(U'',V'')\,.
\end{align}
\end{thm}

Informally one can interpret the statement of this theorem as:
in a closed system the entropy is maximal for unconstrained states. Instead of ``unconstrained states'' the notion of ``equilibrium states'' is used more often.
However, removing the constraint is what happens operationally, which is why we stick to the former notion. \\

Now we can provide a good reason for choosing $U$ and $V$ as the variable to describe the states of the ideal gas. The principle of maximum entropy states that the entropy function for $S$, $S_S$, restricted to unconstrained states $(U,V)$, is \emph{concave}. 
This can be seen as follows. 
The $S=A'\vee A'' \equiv \lambda A \vee (1-\lambda)A$ with the restricted state space containing only unconstrained states (and thus only processes respecting this restriction can be applied) is equivalent to the system $A$, since $\lambda + (1-\lambda)=1$.\footnote{For this argument it is again crucial that the states in the statement of the principle of maximum entropy are unconstrained ones only.}
This means in particular, that the principle of maximum entropy can be written as
\begin{align}
S_A(U,V) 
&= \max_{\substack{U'+U'' = U \\ V'+V''=V}} S_{\lambda A} (U',V') + S_{(1-\lambda)A}(U'',V'') \,,
\end{align} 
for any state $\sigma = (U,V)$ on $A$. Notice that now on the left hand side te entropy does no longer have to be evaluated for a composed system.

For arbitrary internal energies $U_1,U_2$, volumes $V_1,V_2$ and any parameter $\lambda\in [0,1]$ it then holds
\begin{align}
\begin{split}
S_A(\lambda U_1 +(1-\lambda)U_2,\lambda V_1 +(1-\lambda)V_2) 
&\geq S_{\lambda A} (\lambda U_1, \lambda V_1) + S_{(1-\lambda)A}((1-\lambda)U_2,(1-\lambda)V_2) \\
&= \lambda S_A(U_1,V_1) + (1-\lambda) S_A(U_2,V_2)\,.
\end{split}
\end{align}
This is the definition of a concave function. 
For the second equality it is crucial that all three quantities entropy, internal energy and volume, are extensive, which explains the choice of variables. \\

Notice that Eq.~(\ref{eq:pVentropy}) giving an explicit formula for the entropy of an ideal gas in terms of internal energy $U$ and volume $V$ must fulfil the principle of maximum entropy according to our derivations. Hence it must be concave.
Indeed $S(U,V)$ fulfils both statements.
One could now ask why we did not just state Theorem~\ref{thm:maxentropy} as a mathematical statement and proved it by direct calculation using the explicit formula for $S(U,V)$ from Eq.~(\ref{eq:pVentropy}).
Indeed, we could have done so, if the ideal gas was the only case in which the principle of maximum entropy held. However, with the detailed derivation in terms of concepts of the framework we can now discuss extensions of this theorem to other systems than just ideal gases. \\

In short, any type of system for which the relevant concepts in the derivation can be defined meaningfully, fulfil a maximum entropy principle. In particular, the fact the the gas was ideal has not been used anywhere in this derivation. 
More generally, the relevant concepts we needed were:
\begin{itemize}
	\item The total system $S$ can be written as the disjoint composition of two other systems 
	$S'$ and $S''$, $S=S'\vee S''$.
	
	\item Arbitrary states $\sigma'$ and $\sigma''$ of $S'$ and $S''$, respectively, 
	can be ``added'' such that $\sigma'+\sigma''\in\Sigma_S$. This is the state
	the system attains after the constraints are removed, 
	i.e.\ the two subsystems can interact and exchange e.g.\ energy, volume, or even particles.
	Notice that ``$+$'' does not have to be component wise addition of the entries of a 
	tuple denoting the state. Also, as is always the case with states of a composite system, 
	the individual states of the subsystems can be retrieved from the joint state.
	
	\item There are work processes on $S$ with zero work cost on any subsystem that remove
	the constraints, i.e.\ that transform $\sigma'\vee \sigma''$ to $\sigma' + \sigma''$
	for any choice of $\sigma'\in\Sigma_{S'}$ and $\sigma''\in\Sigma_{S''}$.
	\end{itemize}

The principle of maximum entropy in the more general case then reads says that for any unconstrained state $\sigma\in\Sigma_S$ (i.e.\ any state which is the final state of a process that removes the constraint) it holds
\begin{align}
\label{eq:maxentropygen}
S_S(\sigma) = \max_{\sigma'+\sigma''=\sigma} S_{S'} (\sigma') + S_{S''}(\sigma'')\,,
\end{align}
where the maximum goes over all $\sigma'\in\Sigma_{S'}$ and $\sigma''\in\Sigma_{S''}$ that satisfy $\sigma' + \sigma''=\sigma$.

Turning to concavity, an additional property of the system becomes relevant:
\begin{itemize}
	\item The system $S$ restricted to unconstrained states is equivalent to both 
	$\tfrac{1}{\lambda}\tilde{S}$ and $\tfrac{1}{1-\lambda}S''$ for some $\lambda\in(0,1)$ 
	and the description of the states of all system happens in corresponding extensive variables
	(which are the ones that are relaxed when removing the constraint discussed to define ``$+$'').
	 
\end{itemize}
If this is the case, then concavity is implied also in the more general case.
For concavity, the states need to be described in terms of extensive variable such that the addition ``$+$'' is just the component-wise addition of the entries of the tuple describing thermodynamic states. \\

Even though concavity of a function is a \emph{formal} property a given function can or can not have, the argument for a concave entropy function in the general case goes over an \emph{operational} line of reasoning involving many of the concepts introduced throughout this paper. 
In this context the statement that entropy is concave in two extensive variables (e.g.\ U and V) is a statement about what happens when one relaxes a constraint between two subsystems by coupling them (e.g.\ by letting them exchange energy).\\

The general principle of maximum entropy is the basis to investigate further thermodynamic potentials such as the free energy or the Gibbs free energy. These will fulfil similar principles of extremal values with different constraints. 
It will also be relevant when deriving stability conditions which for instance imply that the specific heat capacities 
\begin{align}
C_V := \left. \frac{\partial U}{\partial T}\right|_V \geq 0 
\quad \text{and} \quad 
C_p := \left. \frac{\partial U}{\partial T}\right|_p \geq 0
\end{align}
are non-negative and additionally fulfil $C_p-C_V\geq0$.

\newpage
\part*{Discussion and conclusion}
\addcontentsline{toc}{part}{Discussion and conclusion}

\subsection*{Summary}
\addcontentsline{toc}{subsection}{Summary}

In this paper we have introduced the basic Postulates~\ref{post:systems}-\ref{post:entropy} in order to derive the foundations of phenomenological thermodynamics.
From the notion of thermodynamic systems, processes, states and work (Postulates~\ref{post:systems}-\ref{post:freedom}), we formulated the first law (Postulate~\ref{post:firstapp}) and derived the internal energy function.  
After introducing the notion of equivalent systems, culminating in Postulate~\ref{post:copies} requiring that there is an arbitrary number of copies of any system, we stated the definition of heat reservoirs and formulated the second law (Postulate~\ref{post:sec}). From there, we were able to rigorously follow the standard lines to prove Carnot's Theorem and introduce absolute temperature. 
Using the notion of absolute temperature of heat reservoirs it was possible to define the concept of the temperature of heat flows, which was then used to prove Clausius' Theorem~\ref{thm:clausius} and the Entropy Theorem~\ref{thm:entropythm}.
Along the way we defined entropy, for which it was relevant that there exist reversible processes over which the entropy difference between any two states can be computed (Postulate~\ref{post:entropy}).\\

We then used the introduced concepts to extend the framework to quasistatic processes. This is a ``continuous'' version of the \emph{a priori} discrete concept of a thermodynamic process. 
With quasistatic processes it became possible to show that whenever the generalized Postulates~\ref{post:qsfirst} and~\ref{post:qsentropy} hold, then the internal energy function as well as the entropy are continuously differentiable state variables. 
With this it was possible to discuss the absolute temperature of arbitrary systems. 

The worked out example of the ideal gas gives insights into how the theory can be applied to traditional settings. The assumptions that have to be made prior to applying the framework for phenomenological thermodynamics in this example are spelled out and discussed. 
Also, the validity of the postulates is checked and specific expressions for the internal energy as well as the entropy function are derived.
This allowed us to revisit the example of an (irreversible) heat flow that can be assigned a spectrum of temperatures more formally. 

A further exploration into the formalism for describing the scaling of systems led to the derivation of the maximum entropy principle. 
Again, the ideal gas was a very helpful example. However, the result holds also for the entropy functions of more general systems. The specific assumption for the maximum entropy principle to hold have been discussed.\\


This concluded the comprehensive discussion of ``the laws of thermodynamics'' in the traditional sense contrasted with the basic postulates of the framework of phenomenological thermodynamics presented in this paper.

\subsection*{Discussion of main contributions}
\addcontentsline{toc}{subsection}{Main contributions}

It may be surprising how much technical work was necessary to make these intuitively straightforward basic foundations precise. This tells us two things: it is possible to equip the intuitively clear physical theory of phenomenological thermodynamics with a mathematically rigorous and precise background; but it seems that this can only be done with considerable effort on the technical side.

After working on this project for a long time, we have good reason to believe that the complicated derivations are intrinsically necessary. 
Having a closer look at the sections which required more technical proofs than others, we see that most of these talk about very intuitive concepts while the less technical sections are conceptually more relevant. For example, coming up with a suitable definition for when two systems are considered to be copies of each other took several pages of definitions and lemmas with proofs while, on the other end of the spectrum, the proof of Carnot's Theorem could be distilled to half a page. \\

It is worth mentioning here that the intuition of all concepts can be made precise enough to introduce thermodynamics according to this framework in a undergraduate course without too many technicalities. The authors of this paper have reworked the script of the course \emph{Theory of Heat} held by Gian Michele Graf at ETH in 2005 for the same course in 2019 taught by Renato Renner \cite{script19}. In the script from 2019, the main ideas of this paper are presented on a more intuitive level. \\

Nevertheless, we believe it is one of the main achievements of this paper that we handle terms such as equivalent systems, heat reservoirs or the temperature of a heat flow with great care, thereby crystallising the definitions to the essence.
Especially the definition of heat reservoirs is something the authors have not seen elsewhere in any introduction to phenomenological thermodynamics. 
Even though this term is omnipresent in thermodynamics it seems rather difficult to formalise it. Heat reservoirs should be infinite systems, without using the term ``infinite`'' in the definition. They should be very simple in terms of what one can do with them. Basically they should just provide or take up heat. On the other hand, they should be translation invariant in their internal energy so the current energy does not dramatically affect what one can do with them.
The three points in our Definition~\ref{def:heatreservoir} for heat reservoirs carefully capture these characteristics in mathematical statements that are sufficient to formally work with the defined notion. 
Each point could be justified on an intuitive basis and they are simple and clear enough in order to see that they are necessary conditions.
As the derivations using the term heat reservoirs and the second law show, they are also sufficient for introducing all basic concepts of phenomenological thermodynamics, in particular thermodynamic entropy. \\

Due to the systematic introduction of the eleven postulates, from which all results are derived, it is now possible to investigate which of the usually made assumptions are unnecessary and what may be missing in traditional texts.

An example for a postulate that is usually missing is, as we would argue, the assumption that any two states of any system can be transformed into each other by means of reversible processes with heat flows at well-defined temperatures, Postulate~\ref{post:entropy}. 
Example~\ref{ex:entropypost} emphasizes that, without this postulate, it is thinkable to have systems in the set of thermodynamic systems in a general formulation of phenomenological thermodynamics that do not fulfil it. This shows that the postulate is necessary. 

Thermodynamic equilibrium, on the other hand, is a notion that we did not have to use formally. We did not have to define in the beginning what we meant by this. Neither did we require that the states we worked with are thermodynamic equilibrium states. 
Arguably, the result of Lemma~\ref{lemma:idexistapp} (for any state of any system there exists an identity process) captures already a great deal of what is usually considered an equilibrium state.
But this lemma came as a consequence of the first law (Postulate~\ref{post:first}) rather than any requirement on what is called a state. 
Hence we show that phenomenological thermodynamics can be phrased as a physical theory without this term.\\

We consider the zeroth law as the most prominent victim of the postulates that are usually unnecessarily made. This was the main topic of our previous paper \cite{Kammerlander18}, but nevertheless deserves mentioning again. 
The zeroth law can be phrased in many different ways \cite{Maxwell71, Fowler39, Planck14, Buchdahl66}. 
Almost all variations share the core statement that ``being in two thermal equilibrium with'' is a transitive relation on the set of systems.\footnote{Other notions instead of systems used in \cite{Buchdahl66} are bodies \cite{Planck14} or assemblies \cite{Fowler39}.}
That is, if $A$ and $B$ are in thermal equilibrium with each other and likewise $B$ and $C$, then so are $A$ and $C$. 
Together with the usually implicitly assumed reflexivity and symmetry, the zeroth law then implies that this relation is an equivalence relation. 

This equivalence relation is then relevant for defining a sensible notion of empirical temperature, long before the absolute temperature is introduced with the help of Carnot's Theorem. 
More precisely, one can say that two systems have the same empirical temperature if they are in equilibrium with each other. 

The zeroth law is usually postulated on the same level as the first law and second law. 
Our work now shows that this is in fact not necessary. While the first and second law are core postulates also in our considerations, the zeroth law was never assumed. 
Even more, we have shown that the notion of absolute temperature can nevertheless be defined and the corresponding relation ``being in thermal equilibrium with'' (Definition~\ref{def:simapp}) on the set of heat reservoirs is an equivalence relation, thereby deriving something one could call the zeroth law. 
Hence the zeroth law as a postulate is redundant. 

One could now argue that the zeroth law is still a necessary postulate when introducing the relation ``being in thermal equilibrium with'' before the second law, in particular when using an empirical temperature scale such as the (ideal) gas temperature scale, as discussed in \cite{Zemansky68}. 
We counter that even if this is true, eventually the absolute temperature scale will be the one that is used. In other words: two notions of temperature are one too many.
After Carnot's Theorem the absolute temperature can be introduced even if the zeroth law has not been postulated and will no longer (and hence not at all) be needed. \\

We briefly discuss previous work on thermodynamics without the zeroth law. 
Based on Cara\-th\'eo\-dory's version of the second law \cite{Caratheodory09} a series of papers \cite{Turner61, Turner63, Ehrlich81, Buchdahl86} argued that is not a necessary postulate in thermodynamics. 
A few years later these results were challenged and controversially discussed \cite{Miller52, Walter89, Buchdahl89, Turner05, Helsdon82, Clayton82}. 
The important difference of the approach via Carath\'eodory's second law to ours is the order in which the central concepts entropy and temperature are introduced.
Cara\-th\'eo\-dory-like approaches derive temperature from the entropy measure, which is defined as a monotone of an order relation. In our framework entropy is defined very differently and follows after establishing absolute temperature. 
Figure~\ref{fig:LYcomparison}, where our approach is contrasted with the modern prototype of Cara\-th\'eo\-dory \cite{LY99}, discusses similarities and differences between the two views in more detail. 

Another paper discussing the redundancy of the zeroth law that deserves special emphasis is \cite{Home77}. In a rather informal manner an argument is presented that resembles ours. 
In contrast to the framework presented here, the assumptions in this paper are not spelled out explicitly and remain therefore somewhat unclear. 
Also, the notion of thermal equilibrium is used before defining absolute temperature.
Nevertheless, the paper sketches conclusions that are based on similar ideas. \\

Our framework benefits from the systematic introduction of the postulates in another way. 
It puts us in the position to ask what happens if Postulate \textcolor{orange}{X} is replaced by another Postulate \textcolor{orange}{Y}, or what if it is left away completely.
For instance, one might be interested in a thermodynamic theory in which work is not additive under concatenation. What could then be said about an internal energy function? What about heat? What about entropy? What weaker postulate could replace the additivity postulate (Postulate~\ref{post:addwork}) so one can still talk about these concepts?
Such investigations are the topic of future work. \\

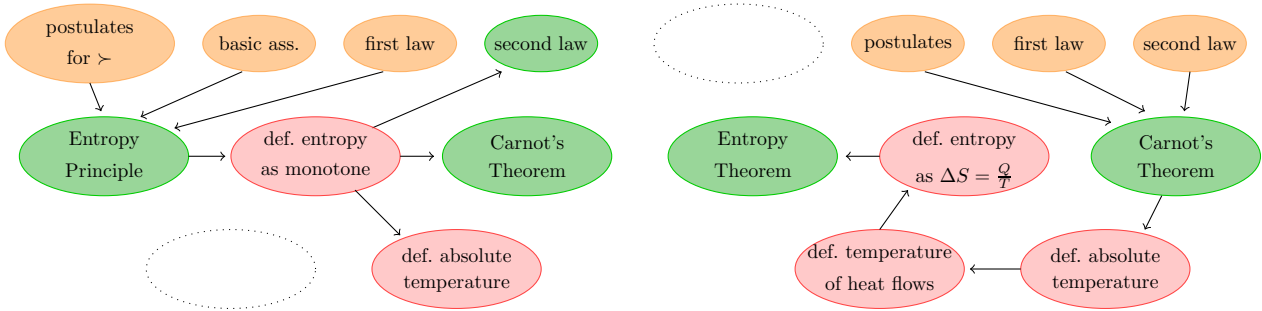
\begin{figure}[]
\hskip-1cm
	\begin{tikzpicture}[scale=0.75, every node/.style={transform shape}]	
	
	\draw[orange, fill=orange!80!white, opacity=0.5] 
	(-5.5,3) node[above, black, opacity=1] {postulates}
	node[below, black, opacity=1] {for $\succ$} ellipse (1.5cm and .7cm);
	\draw[orange, fill=orange!80!white, opacity=0.5] 
	(-2.5,3) node[black, opacity=1] {basic ass.} ellipse (1cm and .5cm);
	\draw[orange, fill=orange!80!white, opacity=0.5] 
	(0,3) node[black, opacity=1] {first law} ellipse (1cm and .5cm);
	\draw[green!80!black,  fill=green!60!black, fill opacity=0.4] 
	(2.5,3) node[black, opacity=1] {second law} ellipse (1cm and .5cm);
	
	\draw[->] (-2.8,2.5) -- (-4.6,1.7);
	\draw[->] (-0.3,2.5) -- (-4,1.5);
	\draw[->] (-5.5,2.3) -- (-5.3,1.8);
	\draw[->] (-.5,1.5) -- (1.8,2.5);

	\draw[green!80!black,  fill=green!60!black, fill opacity=0.4] 
	(-5.25,1) node[above, black, opacity=1] {Entropy} 
	node[below, black, opacity=1] {Principle} ellipse (1.5cm and .7cm);
		
	\draw[red, fill=red!30!white, opacity=0.7] 
	(-1.5,1) node[above, black, opacity=1] {def.\ entropy} 
	node[below, black, opacity=1] {as monotone} ellipse (1.5cm and .7cm);
	
	\draw[green!80!black,  fill=green!60!black, fill opacity=0.4] 
	(2.25,1) node[above, black, opacity=1] {Carnot's} 
	node[below, black, opacity=1] {Theorem} ellipse (1.5cm and .7cm);
	
	\draw[<-] (-3.1,1) -- (-3.75,1);
	
	\draw[->] (0,1) -- (.6,1);
		
	\draw[->] (-.8,.4) -- (0,-.4);
	
	\draw[dotted] 
	(-3,-1) node[above, black, opacity=1] {} 
	node[below, black, opacity=1] {} ellipse (1.5cm and .7cm);
	\draw[red, fill=red!30!white, opacity=0.7] 
	(1,-1) node[above, black, opacity=1] {def.\ absolute} 
	node[below, black, opacity=1] {temperature} ellipse (1.5cm and .7cm);
	
	\begin{scope}[xshift = 11.5cm]
		\draw[dotted] 
		(-5.5,3) node[above, black, opacity=1] {}
		node[below, black, opacity=1] {} ellipse (1.5cm and .7cm);
		\draw[orange, fill=orange!80!white, opacity=0.5] 
		(-2.5,3) node[black, opacity=1] {postulates} ellipse (1cm and .5cm);
		\draw[orange, fill=orange!80!white, opacity=0.5] 
		(0,3) node[black, opacity=1] {first law} ellipse (1cm and .5cm);
		\draw[orange, fill=orange!80!white, opacity=0.5] 
		(2.5,3) node[black, opacity=1] {second law} ellipse (1cm and .5cm);
	
		\draw[->] (-2.2,2.5) -- (1.1,1.6);
		\draw[->] (0.3,2.5) -- (1.7,1.8);
		\draw[->] (2.5,2.5) -- (2.4,1.8);

		\draw[green!80!black,  fill=green!60!black, fill opacity=0.4] 
		(-5.25,1) node[above, black, opacity=1] {Entropy} 
		node[below, black, opacity=1] {Theorem} ellipse (1.5cm and .7cm);
		
		\draw[red, fill=red!30!white, opacity=0.7] 
		(-1.5,1) node[above, black, opacity=1] {def.\ entropy} 
		node[below, black, opacity=1] {as $\Delta S = \tfrac{Q}{T}$} ellipse (1.5cm and .7cm);
	
		\draw[green!80!black, fill=green!60!black, fill opacity=0.4] 
		(2.25,1) node[above, black, opacity=1] {Carnot's} 
		node[below, black, opacity=1] {Theorem} ellipse (1.5cm and .7cm);
	
		\draw[->] (-3,1) -- (-3.6,1);
	
		\draw[->] (2,.3) -- (1.7,-.3);
	
		\draw[red, fill=red!30!white, opacity=0.7] 
		(-3,-1) node[above, black, opacity=1] {def.\ temperature} 
		node[below, black, opacity=1] {of heat flows} ellipse (1.5cm and .7cm);
		\draw[red, fill=red!30!white, opacity=0.7] 
		(1,-1) node[above, black, opacity=1] {def.\ absolute} 
		node[below, black, opacity=1] {temperature} ellipse (1.5cm and .7cm);
	
		\draw[->] (-.5,-1) -- (-1.4,-1);
	
		\draw[->] (-3,-.3) -- (-2.5,.4);
	\end{scope}

	\end{tikzpicture}
\caption{We contrast the main line of argument presented in the modern prototype of Cara\-th\'eo\-dory-based introductions to thermodynamics \cite{LY99} (left) with the paper at hand (right). The color code is \textcolor{red}{definitions} in red, \textcolor{orange}{assumptions and postulates} in orange, and \textcolor{OliveGreen}{derived implications} in green. Left: Lieb and Yngvason define an order relation $\succ$ and postulate several properties of it. With this, they are able to prove a theorem called the Entropy Principle, from which entropy is obtained as the monotone of $\succ$.
Based on this definition of entropy they show that the second law holds as well as Carnot's theorem. 
Other approaches in the spirit of Carath\'eodoy's \cite{Caratheodory09} proceed in a similar way. 
Right: In this paper we capture all basic assumptions in postulates, from which the first and second law are two (they are depicted separately picture the comparison more clearly). After introducing absolute temperature via Carnot's Theorem the temperature of heat flows between arbitrary systems is defined. The usual definition of thermodynamic entropy can then be used and results such as Clausius' Theorem or the Entropy Theorem can be proved. 
}
\label{fig:LYcomparison}
\end{figure}

We conclude with a comparison of the work in this paper with the view proposed by Lieb and Yngvason in \cite{LY99}.
Certainly the definition of entropy is at the heart of any thermodynamic theory. 
It is illuminating to compare the way this is done here with the work by Lieb and Yngvason, see Figure~\ref{fig:LYcomparison}. This will reveal one of the fundamental differences between the two approaches.

In \cite{LY99} entropy is defined as the unique monotone of the order relation $\succ$. This order relation and its properties are introduced axiomatically in precise mathematical terms. Entropy is then the result of the Entropy Principle, a theorem they prove within the axiomatic basis that guarantees the existence, uniqueness and monotonicity and additivity of the entropy function. 
From there all other concepts follow. For instance, absolute temperature is derived from the entropy function.\footnote{As Lieb and Yngvason put it: ``Temperature [...] is a corollary of entropy; it is epilogue rather than prologue.'' \cite{LY99}.} Likewise, different versions of the second law as well as Carnot's Theorem are results from these considerations.

In our framework the definition of thermodynamic entropy relies on the traditional ratio of reversible heat divided by temperature. This way of introducing entropy requires that the notion of absolute temperature and thus also the proof of Carnot's theorem are there beforehand. 
On the other hand, our Entropy Theorem, which is similar albeit not equivalent to Lieb and Yngvason's Entropy Principle, is derived from the thermodynamic definition of entropy and thus eventually relies on the concept of absolute temperature and the second law.
We conclude that which concepts are considered ``more fundamental'' than others heavily depends on the basis one formulates. 
And the method proposed by Lieb and Yngvason, as beautiful as it is, is not the only way to formalize the main concepts of phenomenological thermodynamics. 


\section*{Acknowledgements}
\addcontentsline{toc}{section}{Acknowledgements}

A particular source of inspiration for this work was the script to the course \emph{Theory of heat} taught by Gian Michele Graf at the ETH in 2005 \cite{Graf05}. Many of the ideas presented here arose from discussions about this script.
We thank 
Jakob Yngvason and 
David Jennings for fruitful discussions on our previous paper \cite{Kammerlander18} and 
Stefan Wolf and Tam\'{a}s Kriv\'{a}chy for comments on an earlier version of this framework. 
Both authors acknowledge support from the
``COST Action MP1209'' as well as the
Swiss National Science Foundation through SNSF project No.\ 200020\_165843 and through the the National Centre of Competence in Research \emph{Quantum Science and Technology} (QSIT).

\newpage

\appendix
\addcontentsline{toc}{part}{Appendices}

\part*{Appendices}

\section{Systems}
\label{app:systems}

For completeness and an easier reading we repeat the postulate introducing thermodynamic systems. 

\begin{postulateapp}[Thermodynamic systems]The \emph{thermodynamic world} is a 
set $\Omega$ 
and the \emph{the set of thermodynamic systems} is consists of finite non-empty subsets of the thermodynamic world, $\mathcal{S}:= \{ S\subset\Omega \,|\, 0<|S|<\infty \}$.
\end{postulateapp}

In other words, the structure of thermodynamic systems is given by finite set theory. Hence the following intuitive definition.

\begin{definition}[Composition and intersection]
\label{def:compapp}
For two systems $S_1,S_2\in\mathcal{S}$ we define their \emph{composition} as the union,
\begin{align}
S_1\vee S_2 := S_1\cup S_2 \,,
\end{align}
and their \emph{intersection} as 
\begin{align}
S_1\wedge S_2 := S_1\cap S_2\,,
\end{align}
where in the latter definition the case of \emph{disjoint systems} must be seen as notation only, since $\emptyset$ is not a system.
\end{definition}

Composing and intersecting $n>2$ systems is done by applying $\vee$ and $\wedge$ $n$ times. 
The set of thermodynamic systems is obviously closed under composition of finitely many systems.
It is also closed under intersection except if the systems are disjoint.
When two disjoint systems are composed we sometimes use the term \emph{disjoint composition}.

\begin{definition}[Atomic systems]
\label{def:atomicsysapp}
An \emph{atomic system }$A\in\mathcal{S}$ is a thermodynamic system which is represented by a singleton, $|A|=1$. 
The set of atomic systems is denoted by $\mathcal{A}$.
\end{definition}

Atomic systems are indivisible, i.e.\ they cannot be written as a composition of two different thermodynamic systems. 
Non-atomic systems can always be seen as compositions of finitely many atomic systems. 

\begin{definition}[Subsystems and atomic subsystems] 
\label{def:subsysapp}
Given a system $S\in\mathcal{S}$ the \emph{set of subsystems of $S$} is defined as the set of non-empty subsets of $S$,
\begin{align}
\mathrm{Sub}(S) := \{ S'\in\mathcal{S} \ | \ S' \subset S \}\,.
\end{align} 
The \emph{set of atomic subsystems} is denoted by 
\begin{align}
\mathrm{Atom}(S) := \{ A\in\mathcal{A} \, | \, A \in\mathrm{Sub}(S) \}\,.
\end{align} 
\end{definition}

As such, the set of subsystems does not only contain proper subsystems but also $S$ itself. 
Both sets, $\mathrm{Sub}(S)$ and $\mathrm{Atom}(S)$, are always non-empty. 

By the definitions of composition, subsystems and atomic subsystems, both
\begin{align}
S = \bigvee \mathrm{Sub}(S) \quad \text{and} \quad S = \bigvee \mathrm{Atom}(S)
\end{align}
hold for any thermodynamic system $S\in\mathcal{S}$. 
However, only in the second equality we have disjoint composition. 
Every system can be uniquely composed into its (different, i.e.\ disjoint) atomic subsystems.

\begin{definition}[Disjoint complement]
\label{def:discomplementapp}
Given a system $S\in\mathcal{S}$ and a proper subsystem $S'\in\mathrm{Sub}(S)$, $S'\neq S$, the \emph{disjoint complement of $S'$ w.r.t.\ $S$} is the unique system $S''\in\mathcal{S}$ which fulfils
$S'\wedge S'' = \emptyset$ and $S'\vee S'' = S$.
It is denoted by $S'' := S\smallsetminus S'$. 
\end{definition}

The fact the such a unique disjoint complement always exists is a simple consequence from set theory.

\section{Processes and states}
\label{app:processesstates}

Again, we repeat the postulates in which the main concepts, here states and processes, are introduced.

\begin{postulateapp}[Thermodynamic processes and states]
The non-empty \emph{set of thermodynamic processes} (also simply \emph{processes}) that the theory allows for is denoted by $\mathcal{P}$.
A thermodynamic process $p\in\mathcal{P}$ specifies the \emph{initial and final states} of a finite and non-zero number of
\emph{involved} atomic systems $A\in\mathcal{A}$ by means of 
the functions 
$\lfloor\cdot\rfloor_A: \mathcal{P}\rightarrow\Sigma_A$ 
and $\lceil\cdot\rceil_A: \mathcal{P}\rightarrow\Sigma_A$.  
If $A$ is \emph{not involved} the two functions are undefined.
\end{postulateapp}

The notion of \emph{being involved} is sometimes also phrased the other way around. That is, instead of saying that $A\in\mathcal{A}$ is involved in $p\in\mathcal{P}$, we may say that $p$ \emph{acts on} $A$.\\

The co-domain $\Sigma_A$ of the two function $\lfloor \cdot \rfloor_A$ and $\lceil\cdot\rceil_A$ is called the \emph{state space} or \emph{set of states} of $A$.
For two unequal atomic systems $A_1\neq A_2$ we assume w.l.o.g. that the corresponding state spaces are disjoint, $\Sigma_{A_1}\cap\Sigma_{A_2}=\emptyset$.
This assumption basically says that a state contains a label that indicates to which system is belongs.
It will allow a convenient notation for the states of general thermodynamic systems and simplifies the analysis of state spaces of equivalent systems in Section~\ref{sec:equivalentsys}. \\

\begin{definition}[Set of involved atomic systems]
\label{def:involvedatomicapp}
For a thermodynamic process $p\in\mathcal{P}$ the \emph{set of involved atomic systems} is defined as
\begin{align}
\mathcal{A}_p := \bigvee \{ A\in\mathcal{A} \,|\, A \text{ involved in }p \}.
\end{align}
\end{definition}

As stated in the postulate, in any process $p\in\mathcal{P}$ at least one and at most finitely many atomic systems are involved.
In terms of $\mathcal{A}_p$ this reads
$0<|\mathcal{A}_p|<\infty$ for all $p\in\mathcal{P}$.\\
%
%
%

\begin{definition}[State space]
\label{def:statespaceapp}
For an arbitrary thermodynamic system $S=A_1\vee\cdots\vee A_n$, where $\{A_i\}_{i=1}^n$ are pairwise disjoint atomic systems, the state change of $S$ under a thermodynamic process $p\in\mathcal{P}$ is given by the state changes on its atomic subsystems in terms of the functions
\begin{align}
\begin{split}
\lfloor \cdot \rfloor_S :\,\, &\mathcal{P} \rightarrow \Sigma_S\\
&p \, \mapsto \lfloor p \rfloor_S := \{ \lfloor p \rfloor_{A_1},\dots,\lfloor p \rfloor_{A_n} \}
\end{split}
\end{align}
and likewise for $\lceil \cdot \rceil_S$.
States of a non-trivially composite system are called \emph{joint states} and $\Sigma_S$ is called the \emph{set of states of $S$}.
\end{definition}

The input (or output) state of an arbitrary process on a general thermodynamic system $S$ is defined if and only if the function $\lfloor \cdot \rfloor_A$ are defined for all $A\in\mathrm{Atom}(S)$. Otherwise it is undefined. \\


For a simple notation we write 
$\lfloor p \rfloor_S \equiv \bigvee_{A\in\mathrm{Atom}(S)} \lfloor p \rfloor_A \equiv 
\lfloor p \rfloor_{A_1}\vee\cdots\vee\lfloor p \rfloor_{A_n}$ 
for $S=A_1\vee\cdots\vee A_n$, 
thereby exchanging the symbol $\cup$ for set union with the composition symbol $\vee$. 
This will not lead to confusion in the use of the notation as the arguments used with the symbol $\vee$ will make clear what is meant exactly. 

For two disjoint systems $S_1\wedge S_2=\emptyset$ with $\sigma_1\in\Sigma_{S_1}$ and $\sigma_2\in\Sigma_{S_2}$ this notation reads 
$\sigma_1\vee\sigma_2 = \sigma_2\vee\sigma_1\in\Sigma_{S_1\vee S_2}$.
The fact that $\vee$ is commutative when composing states is no issue as state spaces of different systems are disjoint, i.e.\ it is always clear which symbol ($\sigma_1$ or $\sigma_2$) describes the state of which subsystem ($S_1$ or $S_2$).

The subsystems' states 
$\sigma_1\in\Sigma_{S_1}$ and $\sigma_2\in\Sigma_{S_2}$ can be extracted as the corresponding entries of the ``tuple'' $\sigma_1\vee \sigma_2 \in\Sigma_{S_1\vee S_2}$.
In particular, this implies that if $\lfloor p \rfloor_S$ is defined for a (composite) system $S$, then $\lfloor p \rfloor_{S'}$ is automatically well-defined for all subsystems $S'\in\mathrm{Sub}(S)$ as well. 
Likewise, we already know that if all subsystems of $S$ have a well-defined state change in some thermodynamic process, then so does $S$. 

On the other hand, the structure of state spaces also makes clear that if a proper subsystem $S'\in\mathrm{Sub}(S)$ has well-defined state changes in a thermodynamic process $p\in\mathcal{P}$ this does not imply that the same holds for $S$. 
There might exist other subsystems of $S$ with undefined state change.

This is the reason why the terms \emph{involved} and \emph{not involved} cannot be directly extended to arbitrary systems. These consist of atomic systems, but it could happen that an atomic subsystem is involved in a process, while another one is not. An arbitrary system containing the two would then intuitively be involved. However, its state change is undefined, as there is at least one atomic subsystem whose state change is undefined.\\

Turning to the postulate introducing the partially defined concatenation operation for thermodynamic processes, $\circ: \mathcal{P}\times\mathcal{P} \rightarrow\mathcal{P}$, we note that it implies that $\mathcal{P}$ is closed under $\circ$.

\begin{postulateapp}[Concatenation of processes]
Let $p,p'\in\mathcal{P}$ such that for all $A\in\mathcal{A}_p\cap\mathcal{A}_{p'}$ it holds
$\lceil p \rceil_A = \lfloor p' \rfloor_A$.
Then $p$ and $p'$ can be \emph{concatenated} to form a new process denoted by $p'\circ p\in\mathcal{P}$,
which represents the consecutive execution of $p$ followed by $p'$. 
If $\mathcal{A}_p\cap\mathcal{A}_{p'} = \emptyset$,
then in addition $p'\circ p = p \circ p'$, i.e.\ concatenation commutes.
An atomic system $A\in\mathcal{A}$ is involved in the concatenated process $p'\circ p$ if and only if $A\in\mathcal{A}_p\cup\mathcal{A}_{p'}$.
For the involved atomic systems the initial and final states are
\begin{align}
\label{eq:inputconcapp}
\lfloor p'\circ p \rfloor_A = 
\begin{cases}
\lfloor p \rfloor_A\,, \quad &\text{if } A \in\mathcal{A}_p \\
\lfloor p' \rfloor_A\,, \quad &\text{otherwise} \\
\end{cases}
\end{align}
and the final state follows the same rules with swapped roles for $p$ and $p'$.
\end{postulateapp}


This postulate introduces a technical basis to work with consecutively applied thermodynamic processes and captures the minimal intuitive standards such a composition operation should have. 
The technical use of this postulate will become evident in the coming section, in particular when work is discussed in detail, as is done in Appendix~\ref{app:work}. \\

To make the technical requirements more accessible we provide a generic example of two processes which can be concatenated.

\begin{example}[Concatenation.]
\label{ex:concapp}
Let $S_1=A_1\vee A_2$ and $S_2=A_2\vee A_3$ be two systems, where $A_1,A_2,A_3\in\mathcal{A}$ are different atomic systems. 
Consider two processes $p_1,p_2\in\mathcal{P}$ with $\mathcal{A}_{p_i} = \mathrm{Atom}(S_i)$ for $i=1,2$ and $\lceil p_1 \rceil_{A_2} = \lfloor p_2 \rfloor_{A_2}$. 
Then we know that $p_2\circ p_1$ is defined and the state changes on the involved subsystems are

\begin{align}
\label{eq:exconcapp1}
\lfloor p_2\circ p_1 \rfloor_{A_i} = 
\begin{cases}
\lfloor p_1 \rfloor_{A_1}\,, \quad &i=1\,, \\
\lfloor p_1 \rfloor_{A_2}\,, \quad &i=2\,, \\
\lfloor p_2 \rfloor_{A_3}\,, \quad &i=3\,, 
\end{cases}
\end{align}
and 
\begin{align}
\label{eq:exconcapp2}
\lceil p_2\circ p_1 \rceil_{A_i} = 
\begin{cases}
\lceil p_1 \rceil_{A_1}\,, \quad &i=1\,, \\
\lceil p_2 \rceil_{A_2}\,, \quad &i=2\,, \\
\lceil p_2 \rceil_{A_3}\,, \quad &i=3\,. 
\end{cases}
\end{align}
The atomic system $A_2$ is involved in both processes, i.e.\ it changes its state under both $p_1$ and $p_2$. Thus the total state change on $A_2$ under $p_2\circ p_1$ is from the initial state of $p_1$ to the final state of $p_2$.
The other two atomic systems are only involved in one of the two processes. Thus the initial state on $A_3$ is taken to be $\lfloor p_2 \rfloor_{A_3}$ and the final state on $A_1$ to be $\lceil p_1 \rceil_{A_1}$.
\end{example}

\section{Work and work processes}
\label{app:work}

The postulate introducing \emph{work} states the following.

\begin{postulateapp}[Work]
For any atomic system $A\in\mathcal{A}$ exists a function 
$W_A:\mathcal{P}\rightarrow\mathbbm{R}$ that maps a thermodynamic process $p$ to $W_A(p)$,
the \emph{work done on system $A$} by performing $p$. 
The value $W_A(p)$ is positive 
whenever positive work is done on $A$ while executing $p$.
If system $A$ is not involved in $p$, $A\notin\mathcal{A}_p$, then $W_A(p)=0$ necessarily.
\end{postulateapp}

Technically, the work function is simply a function assigning a real number for any atomic system to any processes such that this number is zero whenever the atomic system is not involved in the process.

This function does not \emph{follow} from a thermodynamic theory but is an \emph{input} the user has to decide on before formulating the theory.\\

\begin{postulateapp}[Additivity of work under concatenation]
If for two processes $p,p'\in\mathcal{P}$ the concatenation $p'\circ p$ is well-defined, then the work cost of the concatenated process equals the sum of the work costs of the individual processes.
That is, for all atomic systems $A\in\mathcal{A}$ it holds that 
$W_A(p'\circ p) = W_A(p) + W_A(p')$ is \emph{additive}.
\end{postulateapp}

For processes which can be concatenated, the total work cost of the sequential execution of the processes must be equal to the sum of the individual work costs. 

Based on the work cost function for atomic systems it is possible to define the work cost function for arbitrary systems in $\mathcal{S}$.

\begin{definition}[Work function for arbitrary systems]
\label{def:workapp}
Let $S\in\mathcal{S}$ be an arbitrary thermodynamic system. We define its \emph{work cost function} (also simply \emph{work function}) $W_S:\mathcal{P}\rightarrow\mathbbm{R}$ by
\begin{align}
W_S = \sum_{A\in\mathrm{Atom}(S)} W_A\,.
\end{align}
\end{definition}

The following two lemmas state that the work functions of any system fulfil two different types of additivity, namely additivity under composition and additivity under concatenation.

\begin{lemma}[Additivity under composition.]
\label{lemma:addworkcompapp}
Let $S_1,S_2\in\mathcal{S}$ be two disjoint systems, $S_1\wedge S_2 = \emptyset$. 
Then the work function $W_{S_1\vee S_2}:\mathcal{P}\rightarrow\mathbbm{R}$ of the composite system $S_1\vee S_2$ is \emph{additive}: for all $p\in\mathcal{P}$ it holds 
$W_{S_1\vee S_2}(p) = W_{S_1}(p) + W_{S_2}(p)$. 
\end{lemma}

\begin{proof}
Disjoint systems $S_1\wedge S_2 = \emptyset$ fulfil
$\mathrm{Atom}(S_1) \cap \mathrm{Atom}(S_2) = \emptyset$. 
On the other hand, by Definition~\ref{def:compapp} we have that
$\mathrm{Atom}(S_1 \vee S_2) = \mathrm{Atom}(S_1) \cup \mathrm{Atom}(S_2)$.
For any $p\in\mathcal{P}$ this implies
\begin{align}
\begin{split}
W_{S_1\vee S_2}(p) 
&= \sum_{A\in\mathrm{Atom}(S_1\vee S_2)} W_A(p) 
= \sum_{A\in\mathrm{Atom}(S_1)} W_A(p) + \sum_{A\in\mathrm{Atom}(S_2)} W_A(p)  \\
&= W_{S_1}(p) + W_{S_2}(p)\,.
\end{split}
\end{align}
\end{proof}

With this, it follows that additivity under concatenation naturally extends from atomic systems, for which it is postulated, to arbitrary systems.


\begin{lemma}[Additivity under concatenation for arbitrary systems]
\label{lemma:addworkconcapp}
If for two processes $p,p'\in\mathcal{P}$ the concatenation $p'\circ p$ is defined, then for all $S\in\mathcal{S}$ additivity under concatenation holds:
\begin{align}
W_S(p'\circ p) = W_S(p)+W_S(p')\,.
\end{align}
\end{lemma}

\begin{proof}
We compute
\begin{align}
\begin{split}
W_S(p'\circ p) 
&= \sum_{A\in\mathrm{Atom}(S)} W_A(p'\circ p)
= \sum_{A\in\mathrm{Atom}(S)} W_A(p) + W_A(p')\\
&= \sum_{A\in\mathrm{Atom}(S)} W_A(p) + \sum_{A\in\mathrm{Atom}(S)} W_A(p') \\
&= W_S(p) + W_S(p')\,.
\end{split}
\end{align}
\end{proof}

The sets of thermodynamic processes which act exactly on the system $S$ deserve special attention. They are addressed in the first law.

\begin{definition}[Work process]
\label{def:wpapp}
For $S\in\mathcal{S}$ a process $p\in\mathcal{P}$ is a \emph{work process on $S$} 
if all its atomic subsystems are involved in $p$ and no other atomic systems are.
That is, $p$ is a work process on $S$ if $S=\bigvee \mathcal{A}_p$.
The \emph{set of work processes on $S$} is denoted by $\mathcal{P}_S$.
\end{definition}

Any process is a work process on some system. More precisely, let $p\in\mathcal{P}$. Then $p$ is a work process on the system $S=\bigvee \mathcal{A}_p\in\mathcal{S}$.
In this sense the set of thermodynamic processes $\mathcal{P}$ is in fact the set of work processes \emph{on some system}, while the set $\mathcal{P}_S$ is the set of work processes \emph{on the system $S$}. 
The latter set is closed under $\circ$ in the sense specified by the next lemma.

\begin{lemma}[$\mathcal{P}_S$ closed under $\circ$]
\label{lemma:closedapp}
For any system $S\in\mathcal{S}$ the set of work processes $\mathcal{P}_S$ is closed under concatenation. That is, if $p'\circ p$ is defined for $p,p'\in\mathcal{P}_S$, then $p'\circ p\in\mathcal{P}_S$.
\end{lemma}

\begin{proof}
The fact that $p,p'\in\mathcal{P}_S$ means that 
$\lfloor p \rfloor_{S'}$, $\lceil p \rceil_{S'}$ and 
$\lfloor p' \rfloor_{S'}$, $\lceil p' \rceil_{S'}$
are defined if and only if $S'\in\mathrm{Sub}(S)$. 
Since by assumption $p'\circ p$ is defined, 
$\lfloor p'\circ p \rfloor_{S'} = \lfloor p \rfloor_{S'}$ and $\lceil p'\circ p \rceil_{S'} = \lceil p' \rceil_{S'}$
for all $S'\in\mathrm{Sub}(S)$ by the postulate on concatenation. For systems $S'\in\mathcal{S}\smallsetminus \mathrm{Sub}(S)$ the input and output states $\lfloor p'\circ p \rfloor_{S'}$ and $\lceil p'\circ p \rceil_{S'}$ are undefined. \\
Hence, $\lfloor p'\circ p \rfloor_{S'}$ and $\lceil p'\circ p \rceil_{S'}$
are defined if and only if $S'\in\mathrm{Sub}(S)$ and thus $p'\circ p\in\mathcal{P}_S$ is a work process on $S$, too.
\end{proof}

\begin{definition}[Joint work processes.]
\label{def:jointwpapp}
For two disjoint systems $S_1 \wedge S_2=\emptyset$ and two work processes $p_i\in\mathcal{P}_{S_i}$ for $i=1,2$
we call their concatenation $p_1\circ p_2$ \emph{the joint work process of $p_1$ and $p_2$} and denote it by
\begin{align}
p_1\vee p_2 := p_1\circ p_2 \in\mathcal{P}_{S_1\vee S_2}\,.
\end{align}
\end{definition}

This definition is well-defined since for two work processes on disjoint system their concatenation is always defined. Furthermore, according to the postulate on concatenation in this case $p_1 \circ p_2 = p_2 \circ p_1$ since $\mathcal{A}_{p_1}\cap\mathcal{A}_{p_2}=\emptyset$. 
Also, the joint work process of $p_1$ and $p_2$ as in the above definition is again a work process on $S_1\vee S_2$ since the involved atomic systems are exactly $\mathcal{A}_{p_1}\cup \mathcal{A}_{p_2}$.\\

As a consequence, the input and output states of a joint work process 
$p_1\vee p_2$ are given by 
$\lfloor p_1\vee p_2 \rfloor_{S_1\vee S_2}  
= \lfloor p_1\rfloor_{S_1} \vee \lfloor p_2 \rfloor_{S_2}
\in\Sigma_{S_1\vee S_2} $
and likewise for $\lceil\cdot\rceil$.

Such joint work processes imply an embedding (an injective mapping) from $\mathcal{P}_{S_1}\times\mathcal{P}_{S_2}$ to $\mathcal{P}_{S_1\vee S_2}$ and images of this mapping $(p_1,p_2)\mapsto p_1\vee p_2$ stand for the parallel execution of the work process $p_1$ on subsystem $S_1$ and of $p_2$ on $S_2$.
Hence, what was achievable by means of work processes on the individual systems $S_1$ and $S_2$ can still be realized as work processes on the composite system $S_1\vee S_2$.
In this sense, composition respects work processes.\\

While the state space of a composite system consists of joint states only, this is not the case with work processes.
In general, the set $\mathcal{P}_{S_1\vee S_2}$ contains more work processes than just the joint work processes of its subsystems. An example of a more general work process on a composite system is thermally connecting two subsystems and letting them exchange energy.

\begin{lemma}[Work cost of a joint work process]
\label{lemma:wjointapp}
Let $S=S_1\vee S_2$ be a composite system with disjoint subsystems $S_1$ and $S_2$ and let $p_1\in\mathcal{P}_{S_1}$ and $p_2\in\mathcal{P}_{S_2}$ be work processes on the subsystems. 
Then the total work cost of the joint work process $p_1\vee p_2\in\mathcal{P}_S$ on $S$ is given by the sum of the local work costs,
\begin{align}
W_{S}(p_1\vee p_2) = W_{S_1}(p_1) + W_{S_2}(p_2)\,.
\end{align}
\end{lemma}

\begin{proof}
Additivity of the work cost functions under concatenation and composition imply
\begin{align}
W_{S}(p_1\vee p_2) &\stackrel{\mathrm{def.}}{=} W_{S}(p_1\circ p_2)\\
&\stackrel{\mathrm{conc.}}{=} W_S(p_1) + W_S(p_2) \\
&\stackrel{\mathrm{comp.}}{=} W_{S_1}(p_1) + W_{S_2}(p_1) + W_{S_1}(p_2) + W_{S_2}(p_2) \\
&= W_{S_1}(p_1) + W_{S_2}(p_2) \,,
\end{align}
for any two work processes $p_1\in\mathcal{P}_{S_1}$ and $p_2\in\mathcal{P}_{S_2}$.
In the last equality it was used that if a system $S_i$ is not involved in a thermodynamic process $p_{i+1}$, then $W_{S_i}(p_{i+1})=0$.
\end{proof}

Another special kind of work processes on a system are \emph{identity processes}. 

\begin{definition}[Identity process]
\label{def:idapp}
An \emph{identity process on $S$}, where $S\in\mathcal{S}$ is an arbitrary thermodynamic system, is a work process $\mathrm{id}_S\in\mathcal{P}_S$ on $S$ with $\lfloor \mathrm{id}_S \rfloor_S = \lceil \mathrm{id}_S \rceil_S$ and zero work cost on all atomic subsystems of $S$, $W_A(\mathrm{id}_S)=0$ for all
$A\in\mathrm{Atom}(S)$. 
For an identity process on $S\in\mathcal{S}$ acting on the state $\sigma\in\Sigma_S$ the notation 
$\mathrm{id}_S^\sigma$ is used.
\end{definition}

A thermodynamic process can act trivially on a system without being a work process on that system, too. This is captured by the notions introduced next. 

\begin{definition}[Cyclic and catalytic process]
\label{def:cyccatapp}
Given a system $C\in\mathcal{S}$ an arbitrary thermodynamic process $p\in\mathcal{P}$ is called \emph{cyclic on $C$} if 
$\lceil p \rceil_C = \lfloor p \rfloor_C$.\\
The process is called \emph{catalytic on $C$} if 
it is cyclic on $C$ and in addition $W_C(p)=0$.\footnote{Notice that the work costs of $p$ for subsystems of $C$ do not have to be zero, only the total work done on $C$ does.}
\end{definition}

The definition of an cyclic process on $S$ which is in addition a work process on $S$ differs from a identity process on $S$ by the missing requirement on the work costs on atomic subsystems.\\

Incorporating a catalytic system explicitly in the thermodynamic description of a process is possible. However, our theory should be such that it is not mandatory, i.e.\ that the explicit mentioning of the catalytic system could also be left out. 
Technically this is captured by the final postulate in this section.

\begin{postulateapp}[Freedom of description]
\label{post:freedomapp}
For $S,C\in\mathcal{S}$ disjoint, let $p\in\mathcal{P}_{S\vee C}$ be such that $p$ is catalytic on $C$, i.e.\ $p$ is cyclic on $C$ and fulfils $W_C(p)=0$.
Then there exists a work process $\tilde p \in\mathcal{P}_S$ on $S$ alone such that 
$\lfloor \tilde p \rfloor_{S} = \lfloor p \rfloor_{S}$ and
$\lceil \tilde p \rceil_{S} = \lceil p \rceil_{S}$ as well as
$W_{A}(\tilde p) = W_{A}(p)$ for all $A\in\mathrm{Atom}(S)$.
\end{postulateapp}


Since $\tilde p \in\mathcal{P}_S$ and $S\wedge C = \emptyset$ it holds automatically that $W_{C'}(\tilde p) = 0$ for all subsystems $C'\in\mathrm{Sub}(C)$ of $C$ in the new process.
Likewise, $W_{S'}\tilde p) = W_{S'}(p)$ for all subsystems $S'\in\mathrm{Sub}(S)$ since their work costs on general subsystems are computed through the work costs on atomic subsystems.

Postulate~\ref{post:freedomapp} is about where to draw the line between objects that thermodynamics explicitly describes and such that are not part of the theory but may nevertheless be used when executing a process. \\

In the final part of this section we discuss the definition of and results on reversible work processes.

\begin{definition}[Reversible processes]
\label{def:revapp}
A work process $p\in\mathcal{P}_S$ on a system $S\in\mathcal{S}$ is called \emph{reversible} if there exists another work process $p^\mathrm{rev}\in\mathcal{P}_S$ on $S$, the \emph{reverse} work process, such that $p^\mathrm{rev}\circ p$ is an identity process. 
\end{definition}

Considering the work costs of reverse processes we find the intuitive result that they simply change their signs relative to the forward process.

\begin{lemma}[Work cost of reverse work processes]
\label{lemma:reverseworkapp}
Let $p\in\mathcal{P}_S$ be a reversible work process on $S\in\mathcal{S}$ with reverse process 
$p^\mathrm{rev}\in\mathcal{P}_S$. Then
\begin{align}
W_A(p^\mathrm{rev}) = -W_A(p)
\end{align}
for all $A\in\mathrm{Atom}(S)$. 
\end{lemma}

\begin{proof}
By definition, $p^\mathrm{rev}\in\mathcal{P}_S$ is a reverse process for $p\in\mathcal{P}_S$ if and only if $p^\mathrm{rev}\circ p$ is an identity process. 
Therefore, by additivity of the work functions under concatenation and the Definition~\ref{def:idapp} of identity processes 
\begin{align}
0 = W_A(p^\mathrm{rev}\circ p) = W_A(p)+W_A(p^\mathrm{rev})
\end{align}
for all atomic subsystems $A\in\mathrm{Atom}(S)$, and the claim follows.
\end{proof}

We next investigate the reversibility of a concatenated process $q\circ p$ in relation to the reversibility of $p$ and $q$. 
Clearly, if both $p$ and $q$ are reversible, then so is $q\circ p$ whenever the concatenation is defined. 
This follows from the fact that the reverse processes $p^\mathrm{rev}$ and $q^\mathrm{rev}$ can be concatenated to a reverse process $p^\mathrm{rev}\circ q^\mathrm{rev}$. 
But the implication in the opposite direction, i.e.\ that if $q\circ p$ is reversible, then so are $p$ and $q$, needs a bit more work to obtain.

\begin{lemma}[Reversibility of $p,q\in\mathcal{P}_S$ relative to $q\circ p$]
\label{lemma:pqrevapp}
Let $p,q\in\mathcal{P}_S$ for some $S\in\mathcal{S}$ such that $r := q\circ p$ is defined. Then, if $r$ is reversible, both $p$ and $q$ must also be reversible.
\end{lemma}

\begin{proof}
Let $r^\mathrm{rev}\in\mathcal{P}_S$ be a reverse process for $r$ and consider 
$r^\mathrm{rev}\circ q$. This process is well-defined, as $r^\mathrm{rev}\circ r$ is defined and the output state of $r$ is equal to the output state of $q$.
Also, it can be concatenated with $p$ from the left (since $q$ can be concatenated with $p$) and from the right (since $r^\mathrm{rev}$ is a reverse process for $q\circ p$).
Furthermore, it holds that $r^\mathrm{rev}\circ q \circ p$ is a cyclic work process on $S$ and thus so is $p\circ r^\mathrm{rev}\circ q$. 
Finally, we compute 
\begin{align}
W_A(r^\mathrm{rev}\circ q\circ p)
= W_A(r^\mathrm{rev}) + W_A(q) + W_A(p)
= W_A(r^\mathrm{rev}) + W_A(r)
= 0\,,
\end{align}
for all $A\in\mathrm{Atom}(S)$.

Together this implies that $r^\mathrm{rev}\circ q \circ p$ is an identity process and thus $p^\mathrm{rev} := r^\mathrm{rev}\circ q\in\mathcal{P}_S$ is indeed a reverse process for $p$, which means that $p$ is reversible. 
To show reversibility of $q$, we proceed analogously with $q^\mathrm{rev} := p\circ r^\mathrm{rev}$. 
\end{proof}

\begin{lemma}[Reversibility of $p_1\vee \mathrm{id_2}$]
\label{lemma:pidrevapp}
Let $S_1,S_2\in\mathcal{S}$ be two disjoint systems. Let further $p_1\in\mathcal{P}_{S_1}$ be an arbitrary work process on $S_1$ and $\mathrm{id}_2\in\mathcal{P}_{S_2}$ an arbitrary identity process on $S_2$.
Then, if $p:=p_1\vee \mathrm{id}_2$ is reversible, so is $p_1$.
\end{lemma}

\begin{proof}
Let $p^\mathrm{rev}\in\mathcal{P}_{S_1\vee S_2}$ be a reverse process.
Then $W_{S_2}(p^\mathrm{rev}) = - W_{S_2}(p) = -W_{S_2}(\mathrm{id}_2) = 0$ due to Lemma~\ref{lemma:reverseworkapp}.
On the other hand, $p^\mathrm{rev}$ is cyclic on $S_2$, since the state of $S_2$ did not change under $p$ and hence neither under its reverse process.

Therefore, $p^\mathrm{rev}$ is catalytic on $S_2$ which, together with Postulate~\ref{post:freedom} on the freedom of thermodynamic description, implies that there exists a work process $p_1^\mathrm{rev}\in\mathcal{P}_{S_1}$ on $S_1$ alone such that the state change as well as the work flows under $p_1^\mathrm{rev}$ on $S_1$ are the same as the ones change under $p^\mathrm{rev}$. Obviously, $p_1^\mathrm{rev}$ is a reverse process for $p_1$, which proves the claim.
\end{proof}

\begin{prop}[Reversibility of general $p_1\circ p_2$]
\label{prop:p1p2revapp}
Let $S_1,S_2\in\mathcal{S}$ be two arbitrary systems (not necessarily disjoint) and $p_i\in\mathcal{P}_{S_i}$ such that $p := p_2\circ p_1$ is defined. 
Then, if $p$ is reversible, so are $p_1$ and $p_2$.
\end{prop}

\begin{proof}
We first note that according to the postulate introducing concatenation, we know that $p\in\mathcal{P}_{S_1\vee S_2}$ is a work process on the composite systems $S_1\vee S_2$.
Also, we know that the disjoint complements $S_1\smallsetminus S_2$ and $S_1\smallsetminus S_2$ exist, and are unique systems whenever the one system is not completely contained in the other. Let us for the moment assume that this is the case.

Then, extend the processes $p_i$ with 
identities\footnote{The fact that such identity processes always exist has not yet been established. It is a consequence of Postulate~\ref{post:first} stated in Lemma~\ref{lemma:idexistapp} in Appendix~\ref{app:firstlaw}. We use this fact here already.}
so that they induce defined state changes on all of $S_1\vee S_2$, i.e.\ let $\mathrm{id}_{S_2\smallsetminus S_1}$ and $\mathrm{id}_{S_1\smallsetminus S_2}$ be identities such that 
$p':= p_2'\circ p_1'$ is defined, where $p_2' := \mathrm{id}_{S_1\smallsetminus S_2}\vee p_2$ and $p_1' := p_1\vee \mathrm{id}_{S_2\smallsetminus S_1}$.
For $i=1,2$ both $p_i'\in\mathcal{P}_{S_1\vee S_2}$ are work processes on $S_1\vee S_2$. 
Furthermore, $p'$ induces the same state change with the same work costs as $p$ on any involved atomic system. Thus $p'$ is also reversible and has the same reverse processes as $p$.

By Lemma~\ref{lemma:pqrevapp} we know that both $p_1'\in\mathcal{P}_{S_1\vee S_2}$ and $p_2'\in\mathcal{P}_{S_1\vee S_2}$ must be reversible.
But then, according to the previous Lemma~\ref{lemma:pidrevapp}, both $p_1\in\mathcal{P}_{S_1}$ and $p_2\in\mathcal{P}_{S_2}$ are reversible, which concludes the proof for the case when neither $S_1\in\mathrm{Sub}(S_2)$ nor vice versa.

The case $S_1\in\mathrm{Sub}(S_2)$ follows in the very same way, where we just define $p_2' := p_2$
and likewise, if $S_2\in\mathrm{Sub}(S_1)$, define $p_1' := p_1$.
\end{proof}

\section{The first law}
\label{app:firstlaw}

The first law of thermodynamics states the following.

\begin{postulateapp}[The first law]
\label{post:firstapp}
For any system $S\in\mathcal{S}$ the following two statements hold:
\begin{itemize}
	\item [(i)]
	For any pair of states $\sigma_1,\sigma_2\in\Sigma_S$ there is a work process 
	$p\in\mathcal{P}_S$ on $S$ with $\lfloor p \rfloor_S = \sigma_1$ and 
	$\lceil p \rceil_S = \sigma_2$ or there is a work process $p'\in\mathcal{P}_S$ on $S$ 
	with $\lfloor p' \rfloor_S = \sigma_2$ and $\lfloor p' \rfloor_S = \sigma_1$.
	\item [(ii)] The total work cost of a work process $p\in\mathcal{P}_S$ on $S$, $W_S(p)$, 
	only depends on $\lfloor p \rfloor_S$ and $\lceil p \rceil_S$ 
	and not on any other details of the process. 
\end{itemize}
\end{postulateapp}

In particular, (ii) implies that if $p'\in\mathcal{P}_S$ is another work process on $S$ with 
$\lfloor p \rfloor_S = \lfloor p' \rfloor_S$ and 
$\lceil p \rceil_S = \lceil p' \rceil_S$, then
$W_S(p)=W_S(p')$.\\

The first law implies a relation on the set of states of any system which will turn out to be a preorder.
A \emph{preordered set} is a set $\mathcal{M}$ together with a relation $\rightarrow$ such that 	
the relation is (i) reflexive, i.e.\ $\forall m\in\mathcal{M}:\, m\rightarrow m$, 
and (ii) transitive, i.e.\ if both $m\rightarrow m'$ and $m'\rightarrow m''$, then $m\rightarrow m''$.


\begin{definition}[Preordered states]
\label{def:preorderapp}
For any system $S\in\mathcal{S}$ the \emph{preorder} $\rightarrow$ on $\Sigma_S$ is established by the reachability via a work process, i.e.\ for $\sigma,\sigma'\in\Sigma_S$ define 
\begin{align}
\label{eq:preorderapp}
\sigma\rightarrow\sigma'\, :\Leftrightarrow\, \exists p\in\mathcal{P}_S \ \mathrm{s.t.}\ 
\lfloor p \rfloor_S = \sigma,\ \lceil p \rceil_S = \sigma'\,.
\end{align}
\end{definition}

Processes $p\in\mathcal{P}_S$ can be seen as labels of the preordered pairs. In this sense, if one wants to precisely state which work process is responsible for the preordering of two states, one can write 
$\sigma\stackrel{p}{\rightarrow}\sigma'$ if the work process $p$ on $S$ is such that $\lfloor p \rfloor_S = \sigma$ and $\lceil p \rceil_S = \sigma'$.
As mentioned before, there may be more than one label for a preordered pair.

\begin{lemma}[Preordered states]
\label{lemma:perorderapp}
The relation $\rightarrow$ in Definition~\ref{def:preorderapp} is a preorder.
\end{lemma}

\begin{proof}
The relation introduced in Eq.~\ref{eq:preorder} is reflexive since for all systems $S$ and all states $\sigma\in\Sigma_S$ there exists work process $p\in\mathcal{P}_S$ such that 
$\sigma\stackrel{p}{\rightarrow}\sigma$.
This is a consequence of Postulate~\ref{post:firstapp} (i).
Furthermore, if $\sigma\stackrel{p}{\rightarrow}\sigma'$ and 
$\sigma'\stackrel{p'}{\rightarrow}\sigma''$ then 
$\sigma\stackrel{p'\circ p}{\longrightarrow}\sigma''$ is preordered too by means of the concatenated process $p'\circ p$.
Hence the relation is also transitive, which makes it a preorder.
\end{proof}

Due to the reflexivity of the preorder induced by the first law it follows that for any state on an atomic system there is an identity process.

\begin{lemma}[Existence of identity process for all atomic states]
\label{lemma:idexistatomicapp}
For any atomic system $A\in\mathcal{A}$ and any state $\sigma\in\Sigma_A$ there exists an identity process $\mathrm{id}_A^\sigma\in\mathcal{P}_A$ with  
$\lfloor \mathrm{id}_A^\sigma \rfloor_A = \sigma 
= \lceil \mathrm{id}_A^\sigma \rceil_A$.
\end{lemma}

\begin{proof}
Since $\rightarrow$ is a preorder, in particular reflexive, we know that for any $\sigma\in\Sigma_A$ there is a work process $q\in\mathcal{P}_A$ on $A$ such that $\sigma\stackrel{q}{\longrightarrow}\sigma$.
This process acts on $A$ alone and initial and final states match. Thus it can be concatenated with itself. 
The state change of the process in which $q$ is applied twice is obviously the same as the state change under $q$ itself. 
Therefore Postulate~\ref{post:firstapp} (ii) requires 
\begin{align}
W_A(q) \stackrel{\mathrm{(ii)}}{=}
W_A(q \circ q) = 
W_A(q) + W_A(q) = 2 W_A(q)\,,
\end{align}
which is to say that the work cost of such processes is zero. Obviously, this makes it an identity process on $A$ for the state $\sigma$ which can rightfully be called $\mathrm{id}_S^\sigma \rceil_A$.
\end{proof}

Using Lemma~\ref{lemma:idexistatomicapp} it is then possible to show that identity processes exist for all states on any thermodynamic system $S\in\mathcal{S}$.

\begin{lemma}[Existence of identity process for all states]
\label{lemma:idexistapp}
For any atomic system $S\in\mathcal{S}$ and any state $\sigma\in\Sigma_S$ there exists an identity process $\mathrm{id}_S^\sigma\in\mathcal{P}_S$ with  
$\lfloor \mathrm{id}_S^\sigma \rfloor_S = \sigma 
= \lceil \mathrm{id}_S^\sigma \rceil_S$.
\end{lemma}

\begin{proof}
Using the decomposition of an arbitrary system into its atomic subsystems it is possible to construct identity process for all states of arbitrary thermodynamic systems.
To see this, let $S=A_1\vee \cdots \vee A_n$ with different atomic systems $A_i\neq A_j$ and consider an arbitrary joint state $\sigma = \sigma_1\vee \cdots \vee \sigma_n$. For the atomic states we know that corresponding identity processes $\mathrm{id}_{A_i}^{\sigma_i}$ exist. Therefore, the joint process $\mathrm{id}_S^\sigma := \mathrm{id}_{A_1}^{\sigma_1}\vee \mathrm{id}_{A_n}^{\sigma_n}$ is a cyclic work process on $S$. But this process also fulfils $W_{A_i}(\mathrm{id}_S^\sigma) = 0$ for all $i=1,\dots,n$ by construction, thus it is an identity process.
We conclude that identity processes exist for all states of all thermodynamic systems. 
\end{proof}

As was stated before the first law guarantees that every system $S$ has a well-defined internal energy function.
We are now in the position to define this function.

\begin{definition}[Internal energy]
\label{def:Ustrongapp}
For a system $S\in\mathcal{S}$ fix an arbitrary reference state $\sigma_0\in\Sigma_S$ and an arbitrary reference energy $U_S^0\in\mathbbm{R}$.
The \emph{internal energy of a state $\sigma\in\Sigma_S$} is defined as
\begin{align}
\label{eq:Ustrong1app}
&U_S(\sigma) := U_S^0 + W_S(p) \,, \qquad  \,
\text{where $p\in\mathcal{P}_S$ is s.t.\ 
$\lfloor p \rfloor_s =\sigma_0$ and $\lceil p \rceil_S = \sigma$.}\\
\label{eq:Ustrong2app}
&U_S(\sigma) := U_S^0 - W_S(p') \,, \qquad  
\text{where $p'\in\mathcal{P}_S$ is s.t.\ 
$\lfloor p' \rfloor_s =\sigma$ and $\lceil p' \rceil_S = \sigma_0$.}
\end{align} 
\end{definition}

Only differences $\Delta U_S$ of internal energies physically matter. Thus the choice of $U_S^0$ is arbitrary at this point.
Likewise, the reference state $\sigma_0\in\Sigma_S$ is arbitrary independently for each system $S$. 

\begin{definition}[State function]
\label{def:statefuncapp}
A \emph{state function on a system $S\in\mathcal{S}$} (also \emph{state variable}) is a function $Z:\Sigma_S\rightarrow\mathcal{Z}$ from the state space $\Sigma_S$ to a target space $\mathcal{Z}$. The co-domain $\mathcal{Z}$ is typically $\mathbbm{R}^n$, most often $n=1$. 
When a system $S$ undergoes a process $p\in\mathcal{P}$ we denote the change in any state function $Z_S$ using an abbreviated notation by $\Delta Z_S(p) := Z_S(\lceil p \rceil_S) - Z_S(\lfloor p \rfloor_S)$.\footnote{This of course only works if a ``minus operation`'' is defined on the co-domain $\mathcal{Z}$. For all practical purposes considered here this is the case.} On the left hand side the dependence of $\Delta Z_S$ on $p$ may be omitted if the context makes clear which process is meant.
 \end{definition}

It must now be proven that $U_S$ is a state function on $S$.

\begin{lemma}[Internal energy is a state function]
\label{lemma:Uwelldefapp}
Internal energy as defined in Definition \ref{def:Ustrongapp} 
is a well-defined state function.
\end{lemma}

\begin{proof}
By the definition and the fact that identity processes exist for any state and have a total work cost of zero it holds $U_S(\sigma_0)=U_S^0$.
This is in agreement with both Eq.~(\ref{eq:Ustrong1app}) and Eq.~(\ref{eq:Ustrong2app}).
For an arbitrary state $\sigma$, Postulate~\ref{post:firstapp} (i) guarantees the existence of at least one of the work processes $p$ and $p'$ used to define $U_S$, hence $U_S(\sigma)$ is never undefined. \\
Furthermore, if both $p$ in one direction and $p'$ in the other direction exist, then there is agreement between the two possibilities of computing $U_S$ due to the fact that the work cost of reverse processes are always the negative of the total work cost of the forward process.\\
Finally, if more than one work process exist with input state $\sigma_0$ and output state $\sigma$, then their work cost must be the same due to Postulate~\ref{post:firstapp} (ii). Hence, it does not matter which one is used to compute $U_S(\sigma)$ by Eq.~(\ref{eq:Ustrong1app}). The same argument also works if $\sigma_0$ and $\sigma$ play exchanged roles, as in Eq.~(\ref{eq:Ustrong2app}).
\end{proof}

To conclude this section, we show that the internal energy function is additive for any system. This is a consequence of the additivity of the work cost functions discussed in the last section.

\begin{prop} [Additivity of $U$]
\label{prop:Uadd}
For a disjointly composite system $S= S_1 \vee S_2$, $S_1\wedge S_2=\emptyset$, that undergoes a work process $p\in\mathcal{P}_S$ the change in internal energy of $S$ equals the sum of the changes of the individual subsystems,
\begin{align}
\label{eq:Uadd}
\Delta U_S = \Delta U_{S_1} + \Delta U_{S_2}\,.
\end{align}
\end{prop}

\begin{proof} 
We denote initial and final states of $p$ 
by $\sigma_1^{\mathrm{in}}\vee \sigma_2^\mathrm{in} \in \Sigma_{S_1 \vee S_2}$ and 
$\sigma_1^{\mathrm{out}} \vee \sigma_2^\mathrm{out} \in \Sigma_{S_1 \vee S_2}$.
Hence, the short notation from Eq.~\ref{eq:Uadd} can be extended as 
\begin{align}
U_S(\sigma_1^\mathrm{out}\vee\sigma_2^\mathrm{out}) 
- U_S(\sigma_1^{\mathrm{in}}\vee \sigma_2^\mathrm{in}) 
= \big( U_{S_1}(\sigma_1^\mathrm{out}) - U_{S_1}(\sigma_1^\mathrm{in}) \big) + 
\big( U_{S_2}(\sigma_2^\mathrm{out}) - U_{S_2}(\sigma_\mathrm{in}) \big) \,.
\end{align}

There are essentially two cases to distinguish: 
(i) there exist work processes $p_1\in \mathcal{P}_{S_1}$ and $p_2\in\mathcal{P}_{S_2}$ with 
$\lfloor p_1 \rfloor_{S_1} = \sigma_1^\mathrm{in}$, 
$\lceil p_1 \rceil_{S_1} = \sigma_1^\mathrm{out}$ 
and $\lfloor p_2 \rfloor_{S_2} = \sigma_2^\mathrm{in}$, 
$\lceil p_2 \rceil_{S_2} = \sigma_2^\mathrm{out}$, i.e.\ two work processes in the same direction;
(ii) there exist work processes $p_1\in \mathcal{P}_{S_1}$ and $p_2\in\mathcal{P}_{S_2}$ in different directions, w.l.o.g.\ with 
$\lfloor p_1 \rfloor_{S_1} = \sigma_1^\mathrm{out}$, 
$\lceil p_1 \rceil_{S_1} = \sigma_1^\mathrm{in}$ 
and $\lfloor p_2 \rfloor_{S_2}= \sigma_2^\mathrm{out}$, 
$\lceil p_2 \rceil_{S_2}= \sigma_2^\mathrm{in}$.

The case of two work processes both in the opposite direction of (i) can be treated just like (i) with $-\Delta U_S$ instead of $\Delta U_S$.
The first law guarantees that one of these cases always applies.

(i) In this case the joint work process $p_1\vee p_2\in\mathcal{P}_S$ gives rise to exactly the same state transfer as $p$. Hence the work done during each of these work processes must also be the same. By definition of the internal energy it follows
\begin{align}
\begin{split}
\Delta U_S 
&= U_S(\sigma_1^\mathrm{out}\vee\sigma_2^\mathrm{out}) 
- U_S(\sigma_1^{\mathrm{in}}\vee \sigma_2^\mathrm{in}) 
=W_S(p)
= W_S(p_1\vee p_2) \\
&= W_{S_1}(p_1) + W_{S_2}(p_2) 
= \big( U_{S_1}(\sigma_1^\mathrm{out}) - U_{S_1}(\sigma_1^\mathrm{in}) \big) + 
\big( U_{S_2}(\sigma_2^\mathrm{out}) - U_{S_2}(\sigma_2^\mathrm{in}) \big) \\
&=\Delta U_{S_1} + \Delta U_{S_2} \,,
\end{split}
\end{align}
where Lemma~\ref{lemma:wjointapp} was used to decompose the total work in the joint work process.


(ii) If $p_1\in\mathcal{P}_{S_1}$ exists with $\lfloor p_1 \rfloor_{S_1} = \sigma_1^\mathrm{in}$, 
$\lceil p_1 \rceil_{S_1} = \sigma_1^\mathrm{out}$ and $p_2\in\mathcal{P}_{S_"}$ with 
$\lfloor p_2 \rfloor_{S_2} = \sigma_2^\mathrm{out}$, $\lceil p_2 \rceil_{S_2} = \sigma_2^\mathrm{in}$, then $p$ can be concatenated with the joint work process 
$\mathrm{id}_{S_1}^{\sigma_1^\mathrm{in}} \vee p_2$. Furthermore, the concatenated work process 
$p\circ (\mathrm{id}_{S_1}^{\sigma_1^\mathrm{in}} \vee p_2)$ induces the same state transfer as the joint work process $p_1 \vee \mathrm{id}_{S_2}^{\sigma_2^\mathrm{in}}$. Hence, due to Postulate\ref{post:firstapp}, their work costs are equal and we obtain
\begin{align}
\begin{split}
\Delta U_S - \Delta U_{S_2} 
&= \big( U_S(\sigma_1^\mathrm{out} \vee \sigma_2^\mathrm{out}) 
- U_S(\sigma_1^{\mathrm{in}} \vee \sigma_2^\mathrm{in}) \big) 
- \big( U_{S_2}(\sigma_2^\mathrm{out}) - U_{S_2}(\sigma_2^\mathrm{in}) \big) \\
&= W_S(p) + W_{S_2}(p_2) = W_S(p) + W_S(\mathrm{id}_{S_1}^{\sigma_1^\mathrm{in}} \vee p_2) \\
&= W_S(p\circ (\mathrm{id}_{S_1}^{\sigma_1^\mathrm{in}} \vee p_2)) 
= W_S(p_1 \vee \mathrm{id}_{S_2}^{\sigma_2^\mathrm{in}}) = W_{S_1}(p_1) \\
&= U_{S_1}(\sigma_1^\mathrm{out}) - U_{S_1}(\sigma_1^\mathrm{in}) = \Delta U_{S_1}\,.
\end{split}
\end{align}
\end{proof}

We have proved that when a composite system undergoes a work process, then the total change in internal energy on the composite system is equal to the sum of the internal energy changes on the individual subsystems. However, this result automatically extends to arbitrary processes on a composite system, not just work processes, because internal energy $U_S$ is a state function.

\section{Equivalent systems}
\label{app:equivalentsys}

In this section we explain the technical background of the formalisms introducing thermodynamic isomorphisms and equivalent systems. 
These are important concepts in order to be able to talk about copies of systems.
Some ideas presented here are inspired by the Master's Thesis \cite{Krivachy17} which the authors supervised.

In the beginning we have emphasized that elements of $\mathcal{S}$ are seen as specific physical instances rather than types of systems. With a notion of copies (or equivalent systems) it will nevertheless be possible to talk about types of systems.
Two systems are copies of each other if they can be interchanged without any noticeable thermodynamic differences.
To make this more precise we introduce a type of map that called a ``thermodynamic isomorphism''. It maps thermodynamic processes to other thermodynamic processes while preserving the thermodynamic structure that has been introduced in the previous sections.
The mapping should be such that the two specific systems are swapped and basically nothing else happens.\\

Before defining equivalent systems we discuss the less restrictive notion of a thermodynamic isomorphism and its properties. 

\begin{definition}[Thermodynamic isomorphism]
\label{def:isomorphismapp}
The pair of \emph{bijective} maps $\varphi: \mathcal{P}\rightarrow\mathcal{P}$, 
$\varphi_\mathcal{A}: \mathcal{A}\rightarrow\mathcal{A}$
is called a \emph{thermodynamic isomorphism} if for any thermodynamic processes $p,p'\in\mathcal{P}$ and any atomic system $A\in\mathcal{A}$ it holds
\begin{itemize}
	\item [(i)] $\varphi(p'\circ p) = \varphi(p')\circ\varphi(p)$ 
	whenever the concatenation $p'\circ p$ or $\varphi(p')\circ\varphi(p)$
	is defined,
	\item [(ii)] $A$ is involved in $p$ if and only if $\varphi_\mathcal{A}(A)$ is involved in $\varphi(p)$, and
	\item [(iii)] $W_{\varphi_\mathcal{A}(A)}(\varphi(p)) = W_{A}(p)$.
\end{itemize}
\end{definition}

The requirements on an isomorphism emphasize the fundamental structure behind the basic thermodynamic concepts. These are (i) the \emph{thermodynamic processes} with \emph{concatenation}, (ii) \emph{atomic systems}, linked to processes through the notion of \emph{an atomic system being involved in a process}, and (iii) \emph{work}.

We have already established that input and output states are always either both defined or both undefined. Hence (ii) is equivalent to saying that:
$\lfloor \varphi(p)\rfloor_{\varphi_\mathcal{A}(A)}$ is defined $\Leftrightarrow$ 
$\lfloor p \rfloor_{A}$ is defined.

We directly start with the notion of an isomorphism without introducing homomorphisms first, as one might expect when discussing algebraic structures. This is because the notion of an isomorphisms is exactly what we will need for defining equivalences while a detailed discussion of homomorphisms would go beyond the scope of this work. 
Investigations of less restrictive mappings than an isomorphism may be the topic of future work.

Even though both mappings $\varphi$ and $\varphi_\mathcal{A}$ are part of the definition of a thermodynamic isomorphism they are not independent degrees of freedom, as the next lemma shows. 
In this sense, $\varphi_\mathcal{A}$ is determined by $\varphi$ and one could think of coming up with a more minimal Definition~\ref{def:isomorphismapp} such that it only talks about $\varphi$, while $\varphi_\mathcal{A}$ is derived from it afterwards. 
Even though this is possible it would make the definition much less readable and intuitive. Therefore we do not go further into this.
Nevertheless it is good to know that the fundamental mapping is the one on processes and the other concepts depend on this. 

\begin{lemma}[$\varphi$ and $\varphi_\mathcal{A}$ not independent]
\label{lemma:isomdofapp}
Let $\varphi,\varphi_\mathcal{A}$ and $\varphi,\varphi'_\mathcal{A}$ be thermodynamic isomorphisms.
Then $\varphi_\mathcal{A} = \varphi'_\mathcal{A}$. 
\end{lemma}

\begin{proof}
Consider an arbitrary atomic system $A\in\mathcal{A}$ and a work process $p\in\mathcal{P}_A$ on $A$. 
By Definition~\ref{def:isomorphismapp} (ii) we know that 
$\lfloor \varphi(p) \rfloor_{\varphi_\mathcal{A}(A')}$ is defined if and only if $A'=A$ and likewise, 
$\lfloor \varphi(p) \rfloor_{\varphi'_\mathcal{A}(A')}$ is defined if and only if $A'=A$. In fact, these statements are equivalent. 

Since the maps $\varphi_\mathcal{A}$ and $\varphi'_\mathcal{A}$ are bijective the statements can rephrased as
$\lfloor \varphi(p) \rfloor_{A''}$ is defined if and only if $A''=\varphi_\mathcal{A}(A)$ and
$\lfloor \varphi(p) \rfloor_{A''}$ is defined if and only if $A''=\varphi'_\mathcal{A}(A)$.
From this it immediately follows that $\varphi'_\mathcal{A}(A) = \varphi_\mathcal{A}(A)$.
\end{proof}

Notice that we only used (ii) of Definition~\ref{def:isomorphismapp} in this proof (together with bijectivity of the maps). This is also the point that would become more complicated in a version of the definition that does not make use of $\varphi_\mathcal{A}$ but allows one to derive it. 

The mapping of atomic systems under a thermodynamic isomorphism can be naturally extended to a mapping on all systems as
\begin{align}
\varphi_\mathcal{S} (S) := \bigvee_{A\in\mathrm{Atom}(S)} \varphi_\mathcal{A} (A)\,.
\end{align}
By construction this mapping is bijective, too.\\

We now aim at showing that thermodynamic isomorphisms preserve all thermodynamic properties that have been introduced so far. 

\begin{lemma}[Work processes under thermodynamic isomorphisms]
\label{lemma:isomwpapp}
Let $\varphi, \varphi_\mathcal{A}$ be a thermodynamic isomorphism 
and $p\in\mathcal{P}$. Then for all $A\in\mathcal{A}$ and all $S\in\mathcal{S}$:
\begin{itemize}
	\item [(i)] $p\in\mathcal{P}_{S} \ \Longleftrightarrow \ \varphi(p) \in\mathcal{P}_{\varphi_\mathcal{S}(S)}$,
	and in particular $p\in\mathcal{P}_A \ \Longleftrightarrow \ \varphi(p) \in\mathcal{P}_{\varphi_\mathcal{A}(A)}$,
	\item [(ii)] $p\in\mathcal{P}_{S}$ is an identity process 
	$\Longleftrightarrow \ \varphi(p) \in\mathcal{P}_{\varphi_\mathcal{S}(S)}$ 
	is an identity process,
	\item [(iii)] $p\in\mathcal{P}_S$ is reversible $\Longleftrightarrow \ \varphi(p)\in\mathcal{P}_{\varphi_\mathcal{S}(S)}$ is reversible
\end{itemize}
\end{lemma}

\begin{proof} 
Let $A'\in\mathcal{A}$ and $S'\in\mathcal{S}$ arbitrary.
\begin{itemize}
	\item [(i)] By Definition~\ref{def:wp} $p\in\mathcal{P}_{S}$ if and only if 
	$\lfloor p \rfloor_{A'}$ is defined for $A'\in\mathrm{Atom}(S)$ only.
	Using Definition~\ref{def:isomorphismapp} (ii) this implies
	\begin{align}
	\begin{split}
	\lfloor p \rfloor_{A'} \text{ def. iff } A'\in\mathrm{Atom}(S) 
	&\Leftrightarrow \lfloor \varphi(p) \rfloor_{\varphi_\mathcal{A}(A')} \text{ def. iff } A'\in\mathrm{Atom}(S) \\
	&\Leftrightarrow \lfloor \varphi(p) \rfloor_{\varphi_\mathcal{A}(A')} \text{ def. iff } \varphi_\mathcal{A}(A')\in\mathrm{Atom}(\varphi_\mathcal{S}(S)) \\
	&\Leftrightarrow \lfloor \varphi(p) \rfloor_{A''} \text{ def. iff } A''\in\mathrm{Atom}(\varphi_\mathcal{S}(S)) \\
	\end{split}
	\end{align}
	In the second line we used $A'\in\mathrm{Atom}(S) \Leftrightarrow \varphi_\mathcal{A}(A')\in\mathrm{Atom}(\varphi_\mathcal{S}(S)$), while the third line holds because $\varphi_\mathcal{A}$ is bijective. 
	The third line, however, is equivalent to $\lfloor \varphi(p) \rfloor_{A''}$ being defined for $A''\in\mathrm{Atom}(\varphi_\mathcal{S}(S))$ and thus true if and only if $\varphi(p)\in\mathcal{P}_{\varphi_\mathcal{S}(S)}$.
	\item [(ii)] Let $p\in\mathcal{P}_S$ be an identity process. By Definition~\ref{def:id} this is the case if and only if it is cyclic on $S$, i.e.\ if and only if $p$ can be concatenated with itself, and $W_A(p) = 0$ for all $A\in\mathrm{Atom}(S)$.
	By Definition~\ref{def:isomorphismapp} (i) it holds
	\begin{align}
	p\circ p \text{ def.} \Leftrightarrow \varphi(p) \circ \varphi(p) \text{ def.}\,,
	\end{align}
	which means that $\varphi(p)$ is a cyclic work process on $\varphi_\mathcal{S}(S)$ if and only
	if $p$ is a cyclic work process on $S$. Furthermore, for $A\in\mathrm{Atom}(S)$ we have 
	$\varphi_\mathcal{S}(A)\in\mathrm{Atom}(\varphi_\mathcal{S}(S))$ and 
	\begin{align}
	W_{\varphi_\mathcal{S}(A)}(\varphi(p)) = W_A(p)
	\end{align}		
	for all $A\in\mathrm{Atom}(S)$. Hence one side of this equation is zero if and only if the other is.
	\item [(iii)] By Definition~\ref{def:rev} $p\in\mathcal{P}_S$ is reversible 
	if there exists a reverse work process on the same system. Thus
	\begin{align}
	\begin{split}
	p\in\mathcal{P}_S \text{ is reversible}
	&\Leftrightarrow \exists p^\mathrm{rev}\in\mathcal{P}_S \text{ s.t. } 
	p^\mathrm{rev}\circ p \text{ is identity} \\
	&\Rightarrow \varphi(p^\mathrm{rev})\circ\varphi(p) \text{ is identity and } 
	\varphi(p^\mathrm{rev})\in\mathcal{P}_{\varphi_\mathcal{S}(S)} \\
	&\Rightarrow \varphi(p^\mathrm{rev})\in\mathcal{P}_{\varphi_\mathcal{S}(S)} 
	\text{ is reversible} \,.
	\end{split}
	\end{align}
	We used the previously proved (i) and (ii).
	Since $\varphi$ is bijective by Definition~\ref{def:isomorphismapp}, the argument also works 
	in the other direction. 
\end{itemize}
\end{proof}

\begin{lemma}[States under thermodynamic isomorphisms]
\label{lemma:isomstateapp}
Let $p,q\in\mathcal{P}$ and $A\in\mathcal{A}$. Furthermore, let $\varphi, \varphi_\mathcal{A}$ be a thermodynamic isomorphism. Then:
\begin{itemize}
	\item [(i)] $\lfloor p \rfloor_A = \lfloor q \rfloor_A \ \Longleftrightarrow \ 
	\lfloor \varphi(p) \rfloor_{\varphi_\mathcal{A}(A)} 
	= \lfloor \varphi(q) \rfloor_{\varphi_\mathcal{A}(A)}$ ,
	\item [(ii)] $\lfloor p \rfloor_A = \lceil q \rceil_A \ \Longleftrightarrow \ 
	\lfloor \varphi(p) \rfloor_{\varphi_\mathcal{A}(A)} 
	= \lceil \varphi(q) \rceil_{\varphi_\mathcal{A}(A)}$ ,
	\item [(iii)]$\lceil p \rceil_A = \lceil q \rceil_A \ \Longleftrightarrow \ 
	\lceil \varphi(p) \rceil_{\varphi_\mathcal{A}(A)} 
	= \lceil \varphi(q) \rceil_{\varphi_\mathcal{A}(A)}$ .
\end{itemize}
\end{lemma}

\begin{proof}
We prove (i). The remaining points (ii) and (iii) follow in the very same way. Again we make use of Definition~\ref{def:isomorphismapp} (i) and (ii), where in particular(i) is important. 
We also use Lemma~\ref{lemma:isomwpapp} (i). \\
Let $\sigma := \lfloor p \rfloor_A = \lfloor q \rfloor_A$ and consider
$\mathrm{id}_A^\sigma\in\mathcal{P}_A$. Both 
$p \circ \mathrm{id}_A^\sigma$ and $q \circ \mathrm{id}_A^\sigma$ are defined and hence by
Definition~\ref{def:isomorphismapp} (i) so are $\varphi(p) \circ \varphi(\mathrm{id}_A^\sigma)$ 
and $\varphi(q) \circ \varphi(\mathrm{id}_A^\sigma)$.
Furthermore, $\varphi(\mathrm{id}_A^\sigma)$ is an identity process on $\varphi_\mathcal{A}(A)$, has in particular defined input and output states on this system, 
and we know that by Definition~\ref{def:isomorphismapp} (ii) both 
$\lfloor \varphi(p) \rfloor_{\varphi_\mathcal{A}(A)}$ 
and $\lfloor \varphi(q) \rfloor_{\varphi_\mathcal{A}(A)}$ are defined.
Hence, they must be equal, as Postulate~\ref{post:conc} (concatenation of processes) requires
$\lfloor \varphi(p) \rfloor_{\varphi_\mathcal{A}(A)} 
= \lfloor \varphi(q) \rfloor_{\varphi_\mathcal{A}(A)}$.
The opposite direction can be argued in the same way with the inverse maps.
\end{proof}

With Lemma~\ref{lemma:isomstateapp} it is possible to define a mapping of the states induced by a thermodynamic isomorphism.

\begin{definition}[States under thermodynamic isomorphisms]
\label{def:isomstateapp}
Given a thermodynamic isomorphism $\varphi, \varphi_\mathcal{A}$ together with its corresponding map $\varphi_\mathcal{S}$ we define the corresponding mapping of atomic states
$\varphi_\Sigma: \bigcup_{A\in\mathcal{A}} \Sigma_A \rightarrow \bigcup_{A\in\mathcal{A}} \Sigma_A$
by
\begin{align}
\varphi_\Sigma(\lfloor p \rfloor_A) := \lfloor \varphi(p) \rfloor_{\varphi_\mathcal{A}(A)}
\end{align}
and its extension to arbitrary states 
$\varphi_\Sigma: \bigcup_{S\in\mathcal{S}} \Sigma_S \rightarrow \bigcup_{S\in\mathcal{S}} \Sigma_S$
by
\begin{align}
\varphi_\Sigma(\lfloor p \rfloor_S) := \lfloor \varphi(p) \rfloor_{\varphi_\mathcal{S}(S)} \,.
\end{align}
\end{definition}

Remember that we assumed that for two different systems, in particular for two different atomic systems, their state spaces are disjoint. This means that the unions
$\bigcup_{A\in\mathcal{A}} \Sigma_A$ and $\bigcup_{S\in\mathcal{S}} \Sigma_S$ are disjoint unions. 
It is immediate that the two functions agree on the intersection of their domains, namely on 
$\bigcup_{A\in\mathcal{A}} \Sigma_A$. This is a consequence of the fact that any system's state can be written as the composition of its atomic subsystems' states, i.e.\
$\lfloor p \rfloor_S = \bigvee_{A\in\mathrm{Atom}(S)} \lfloor p \rfloor_S$
for any $p\in\mathcal{P}$.
In particular, every property of $\varphi_\Sigma$ proved for atomic systems automatically extends to $\varphi_\Sigma$ in general.

Lemma~\ref{lemma:isomstateapp} proves that Definition~\ref{def:isomstateapp} is well-defined. It says that even for two different processes with equal input states (or output states, or input state on one process and output state on the other) the mapping of the states under $\varphi_\Sigma$ is unique. 
Hence the choice of defining $\varphi_\Sigma$ using the input state function $\lfloor \cdot \rfloor$ was arbitrary -- we could have done it with the output state function $\lceil \cdot \rceil$ as well. 
Furthermore, by Definition~\ref{def:isomorphismapp} the output of this function is defined if and only if the input is defined. \\

We remark that by definition $\varphi_\Sigma(\Sigma_{A}) \subset \Sigma_{\varphi_\mathcal{A}(A)}$ for any atomic system $A\in\mathcal{A}$. 
In fact, $\varphi_\Sigma$ is bijective, as the next lemma shows. In particular, this implies that 
$\varphi_\Sigma(\Sigma_{A}) = \Sigma_{\varphi_\mathcal{A}(A)}$ since the domain (and the codomain, which are equal sets here) is a disjoint union of such sets. 

\begin{lemma}[$\varphi_\Sigma$ is bijective]
\label{lemma:isomstatebijapp}
The mapping of states $\varphi_\Sigma$ associated with a thermodynamic isomorphisms $\varphi,\varphi_\mathcal{A}$ defined in Definition~\ref{def:isomstateapp} is bijective.
\end{lemma}

\begin{proof}
We show that the map $\varphi_\Sigma$ is surjective and injective.\\
For surjectivity consider $S\in\mathcal{S}$ and $\sigma\in\Sigma_S$ together with the corresponding identity process $\mathrm{id}_S^\sigma\in\mathcal{P}_S$. 
Since both $\varphi$ and $\varphi_\mathcal{S}$ are bijective, we know that there exist $S'\in\mathcal{S}$ such that $\varphi_\mathcal{S}(S')=S$ and $\mathrm{id}_{S'}^{\sigma'}\in\mathcal{P}_{S'}$ such that $\varphi(\mathrm{id}_{S'}^{\sigma'}) = \mathrm{id}_S^\sigma$.
Therefore, 
\begin{align}
\varphi_\Sigma(\sigma') 
= \varphi_\Sigma(\lfloor \mathrm{id}_{S'}^{\sigma'} \rfloor_{S'})
= \lfloor \varphi(\mathrm{id}_{S'}^{\sigma'}) \rfloor_{\varphi_\mathcal{S}(S')}
= \lfloor \mathrm{id}_S^\sigma \rfloor_S = \sigma\,.
\end{align}
For injectivity, let $\sigma_1\in\Sigma_{S_1}$ and $\sigma_2\in\Sigma_{S_2}$ such that $\sigma_1 \neq \sigma_2$.
Consider again the identity processes $\mathrm{id}_{S_i}^{\sigma_i}\in\mathcal{P}_{S_i}$ for $i=1,2$. 
Then we can write
\begin{align}
\varphi_\Sigma(\sigma_i) 
= \varphi_\Sigma(\lfloor \mathrm{id}_{S_i}^{\sigma_i} \rfloor_{S_i})
= \lfloor \varphi(\mathrm{id}_{S_i}^{\sigma_i}) \rfloor_{\varphi_\mathcal{S}(S_i)}\,.
\end{align}
If $S_1\neq S_2$, $\varphi_\Sigma(\sigma_i)\in\Sigma_{\varphi_\mathcal{S}(S_i)}$ are contained in disjoint sets (because $\varphi_\mathcal{S}$ is bijective and state spaces of different systems are disjoint). Hence $\varphi_\Sigma(\sigma_2) \neq \varphi_\Sigma(\sigma_2)$.\\
If $S_1=S_2$, $\sigma_1\neq\sigma_2$ implies that 
$\mathrm{id}_{S_1}^{\sigma_1} \circ \mathrm{id}_{S_2}^{\sigma_2}$ is not defined (Postulate~\ref{post:conc}). 
By Definition~\ref{def:isomorphismapp} (i) this implies that 
$\varphi(\mathrm{id}_{S_1}^{\sigma_1}) \circ \varphi(\mathrm{id}_{S_2}^{\sigma_2})$ is not defined either and hence
also in this case $\varphi_\Sigma(\sigma_2) \neq \varphi_\Sigma(\sigma_2)$.
\end{proof}

So far we have discussed the mapping of processes, systems and states with their properties. 
For all of this Definition~\ref{def:isomorphismapp} (iii) has not been touched.
When showing that the internal energy function also transforms naturally, this point will become important. 

\begin{lemma}[Internal energy under thermodynamic isomorphisms]
\label{lemma:isomUapp}
Let $\varphi, \varphi_\mathcal{A}$ be a thermodynamic isomorphism with associated mappings $\varphi_\mathcal{S}$ and $\varphi_\Sigma$ for arbitrary systems and states, respectively. 
Then for all $S\in\mathcal{S}$ and $p\in\mathcal{P}$ 
\begin{align}
\Delta U_{\varphi_\mathcal{S}(S)}(\varphi(p)) = \Delta U_S(p)\,.
\end{align}
\end{lemma}

\begin{proof}	
We use the previous results together with Definition~\ref{def:isomorphismapp} (iii), which is the only part of the definition on equivalent atomic systems that makes a statement about work costs. In addition, we directly apply the first law (Postulate~\ref{post:first}). 
Let $S=A\in\mathcal{A}$ be an atomic system for the moment.\\
If $A$ is not involved in $p$ the neither is $\varphi_\mathcal{A}(A)$ involved in $\varphi(p)$ (due to Definition~\ref{def:isomorphismapp} (ii)). Hence in this case
\begin{align}
\Delta U_A (p) = 0 = \Delta U_{\varphi_\mathcal{A}(A)}(\varphi(p))\,.
\end{align}
If $A$ is involved in $p$, suppose there exists a work process $q\in\mathcal{P}_{A}$ with 
$\lfloor q \rfloor_{A} = \lfloor p \rfloor_A$ and $\lceil q \rceil_{A} = \lceil p \rceil_A$.
By definition of $U_{A}$ it follows $\Delta U_{A}(p) = W_{A}(q)$. 
We now consider the image of $q$ under $\varphi$. It follows
$\lfloor \varphi(q) \rfloor_{\varphi_\mathcal{A}(A)} = \varphi_\Sigma(\lfloor q \rfloor_{A}) = \varphi_\Sigma(\lfloor p \rfloor_A)$ 
and
$\lceil \varphi(q) \rceil_{\varphi_\mathcal{A}(A)} = \varphi_\Sigma(\lceil q \rceil_{A}) = \varphi_\Sigma(\lceil p \rceil_A)$
together with $W_{A}(q) = W_{\varphi_\mathcal{A}(A)}(\varphi(q))$.
Hence, using again the definition of internal energy, 
\begin{align}
\Delta U_{\varphi_\mathcal{A}(A)}(\varphi(p)) = W_{\varphi_\mathcal{A}(A)}(\varphi(q)) = W_{A}(q) = \Delta U_{A}(p)\,.
\end{align}
If such a $q\in\mathcal{P}_A$ does not exist, then there exists one in the other direction, as the first law guarantees, and the proof works analogously (with a minus sign).\\
Finally, the proof extends to arbitrary $S\in\mathcal{A}$ due to the additivity of internal energy under composition.
\end{proof}

In Lemma~\ref{lemma:isomUapp} we only talk about differences of internal energies as the reference energies of $S$ and $\varphi_\mathcal{S}(S)$ can be chosen independently according to Definition~\ref{def:Ustrong}.
However, one can always adjust the reference energies of the two atomic systems such that
equality of internal energies for corresponding states
also holds for absolute values of internal energies.\\

We conclude that a thermodynamic isomorphism $\varphi, \varphi_\mathcal{A}$, gives rise to corresponding mappings $\varphi_\mathcal{S}$ on the set of systems and $\varphi_\Sigma$ on the set of all states. These mappings are compatible in a natural way and preserve all concepts introduced so far, as discussed in the results above. 
The definition of a thermodynamic isomorphism is hence a sensible way of introducing such a concept.\\

It is now possible to further restrict the attention to isomorphisms with special properties. 
In particular, if it only swaps two atomic system $A_1,A_2\in\mathcal{A}$ we use it to define the notion of equivalent atomic systems.

\begin{definition}[Equivalence of atomic systems]
\label{def:equiatomapp}
Two atomic systems $A_1,A_2\in\mathcal{A}$ are called \emph{equivalent}, and we write $A_1 \hat = A_2$, if there exists a thermodynamic isomorphism $\varphi, \varphi_\mathcal{A}$ which additionally fulfils 

\begin{itemize}
	\item [(iv)] $\varphi_\mathcal{A}(A_1)=A_2, \varphi_\mathcal{A}(A_2)=A_1$ and 
	$\varphi_\mathcal{A}(A)=A$ for all $A\in\mathcal{A}\smallsetminus\{A_1,A_2\}$, and

	\item [(v)] for $A\in\mathcal{A} \smallsetminus\{A_1,A_2\}$ it holds
	$\lfloor \varphi(p)\rfloor_A=\lfloor p \rfloor_A$ and 
	$\lceil \varphi(p)\rceil_A=\lceil p \rceil_A$.
\end{itemize}
We say that $\varphi, \varphi_\mathcal{A}$ is a \emph{thermodynamic isomorphism for $A_1\hat=A_2$}.
\end{definition}

Essentially, such an isomorphism swaps the thermodynamic roles of the two atomic systems with all their belongings such as states, work processes, work functions and so on, such that ``nothing else changes''. The notion of a thermodynamic isomorphisms allows us to make this precise.

Looking at the mapping of states $\varphi_\Sigma$ for an equivalence we see that (iv) and (v) imply the following.
For all atomic states $\sigma\in\bigcup_{A\in\mathcal{A}\smallsetminus\{A_1,A_2\}} \Sigma_A$
it holds $\varphi_\Sigma(\sigma) = \sigma$. That is, $\varphi_\Sigma$ is equal to 
the identity map outside $\Sigma_{A_1}\cup\Sigma_{A_2}$.
On the other hand, $\varphi_\Sigma |_{\Sigma_{A_1}}$ and $\varphi_\Sigma |_{\Sigma_{A_2}}$ are bijective and imply a 1-1 correspondence of states in $\Sigma_{A_1}$ with states in $\Sigma_{A_2}$.\\

In order to prove that $\hat=$ is an equivalence relation on $\mathcal{A}$ (i.e.\ in order to justify the name given to $A_1\hat=A_2$) we have to work a little more. 

\begin{lemma}[Inverse of thermodynamic isomorphism]
\label{lemma:equiinverseapp}
Let $\varphi, \varphi_\mathcal{A}$ be a thermodynamic isomorphism for the atomic systems $A_1\hat=A_2\in\mathcal{A}$. 
Then its inverse $\varphi^{-1}, \varphi_\mathcal{A}^{-1}$ is also a thermodynamic isomorphism for $A_1\hat=A_2$ and the corresponding maps for systems and states fulfil
\begin{align}
(\varphi^{-1})_\mathcal{S} = (\varphi_\mathcal{S})^{-1} \quad \text{and} \quad 
(\varphi^{-1})_\Sigma = (\varphi_\Sigma)^{-1} \,.
\end{align}
\end{lemma}

\begin{proof}
First of all, since $\varphi^{-1}$ and $\varphi_\mathcal{A}^{-1}$ are the inverses of bijective maps, They are bijective themselves. Let now $p,p'\in\mathcal{P}$ and define their images under $\varphi^{-1}$ to be $q:=\varphi^{-1}(p)$ and $q':=\varphi^{-1}(p')$.
Likewise, let $A\in\mathcal{A}$ and define $A' := \varphi^{-1}_\mathcal{A}(A)$.
We check Definitions~\ref{def:isomorphismapp} and~\ref{def:equiatomapp} point by point.
\begin{itemize}
	\item [(i)] 
	If $p'\circ p$ is defined, then 
	\begin{align}
	\begin{split}
	\varphi^{-1}(p'\circ p) 
	&= \varphi^{-1}\big(\varphi(q') \circ \varphi(q)\big)
	= \varphi^{-1}\big(\varphi(q'\circ q)\big) 
	= q'\circ q = \varphi^{-1}(p')\circ\varphi^{-1}(p)\,,
	\end{split}
	\end{align}
	where Definition~\ref{def:isomorphismapp} (i) was used for $\varphi$ in the second equality.
	Writing the equation in terms of $q$ and $q'$ it follows that $\varphi^{-1}$ satisfies (i). 
	\item [(ii)] Using (ii) for $\varphi,\varphi_\mathcal{A}$, we immediately obtain
	\begin{align}
	\lfloor p \rfloor_A \text{ def.} 
	\Leftrightarrow \lfloor \varphi(q) \rfloor_{\varphi_\mathcal{A}(A')} \text{ def.}
	\Leftrightarrow \lfloor q \rfloor_{A'} \text{ def.}
	\Leftrightarrow \lfloor \varphi^{-1}(p) \rfloor_{\varphi^{-1}_\mathcal{A}(A)} \text{ def.}\,.
	\end{align}
	\item [(iii)] We directly check
	$W_{\varphi^{-1}_\mathcal{A}(A)}(\varphi^{-1}(p)) 
	= W_{A'}(q) = W_{\varphi_\mathcal{A}(A')}(\varphi(q)) = W_{A}(p)$ .
	\item [(iv)] Since $\varphi_\mathcal{A}$ is the simple map that 
	swaps $A_1$ with $A_2$ it follows that its inverse does exactly the same.
	\item[(v)] For $A\neq A_1,A_2$ we compute 
	$\lfloor \varphi^{-1}(p)\rfloor_A = \lfloor q \rfloor_A = \lfloor \varphi(q) \rfloor_A = \lfloor p \rfloor_A$ and similarly for $\lceil \cdot \rceil_\cdot$.

\end{itemize}
Finally, we consider the corresponding maps on systems and states. Since $\varphi^{-1}, \varphi_\mathcal{A}^{-1}$ is a thermodynamic isomorphism for the same two systems as $\varphi$ we obtain $(\varphi^{-1})_\mathcal{S} = \varphi_\mathcal{S}$, and because $\varphi_\mathcal{S}(\varphi_\mathcal{S}(S)) = S$ for all $S\in\mathcal{S}$ this implies that $(\varphi^{-1})_\mathcal{S} = (\varphi_\mathcal{S})^{-1}$.\\
As for the mapping of states, by Definition~\ref{def:isomstateapp} for every $S\in\mathcal{S}$
and $p\in\mathcal{P}$ it holds
\begin{align}
\varphi_\Sigma(\lfloor p \rfloor_S) = \lfloor \varphi(p) \rfloor_{\varphi_\mathcal{S}(S)} \quad
\text{and} \quad (\varphi^{-1})_\Sigma(\lfloor p \rfloor_S) 
= \lfloor \varphi^{-1}(p)\rfloor_{\varphi^{-1}_\mathcal{S}(S)} \,.
\end{align}
Therefore 
\begin{align}
\varphi_\Sigma \left( (\varphi^{-1})_\Sigma(\lfloor p \rfloor_S) \right) 
&= \varphi_\Sigma\left( \lfloor \varphi^{-1}(p)\rfloor_{\varphi^{-1}_\mathcal{S}(S)} \right)
= \lfloor \varphi(\varphi^{-1}(p)) \rfloor_{\varphi_\mathcal{S}(\varphi^{-1}_\mathcal{S}(S))}
= \lfloor p \rfloor_S 
\end{align}
and likewise $\varphi^{-1}_\Sigma( \varphi_\Sigma(\lfloor p \rfloor_S)) = \lfloor p \rfloor_S$, 
which means that $(\varphi^{-1})_\Sigma = (\varphi_\Sigma)^{-1}$.
\end{proof}


\begin{lemma}[Conjugating thermodynamic isomorphisms]
\label{lemma:equiconjapp}
Let $A_1,A_2,A_3\in\mathcal{A}$ be three different atomic systems.
If $\varphi, \varphi_\mathcal{A}$ is a thermodynamic isomorphism for $A_1\hat=A_2$ and $\psi, \psi_\mathcal{A}$ is one for $A_2\hat=A_3$, then $\varphi^{-1}\circ\psi\circ\varphi, \varphi_\mathcal{A}^{-1}\circ\psi_\mathcal{A}\circ\varphi_\mathcal{A}$ is a thermodynamic isomorphism for $A_1\hat=A_3$.
Furthermore, the corresponding mapping of systems and states transform analogously, 
\begin{align}
(\varphi^{-1} \circ \psi \circ \varphi)_\mathcal{S} 
= \varphi^{-1}_\mathcal{S} \circ \psi_\mathcal{S} \circ \varphi_\mathcal{S} \quad
\text{and} \quad
(\varphi^{-1} \circ \psi \circ \varphi)_\Sigma
= \varphi^{-1}_\Sigma \circ \psi_\Sigma \circ \varphi_\Sigma \,.
\end{align}
\end{lemma}

\begin{proof}
As before, the fact that $\varphi^{-1}\circ\psi\circ\varphi$ and 
$\varphi^{-1}_\mathcal{A}\circ\psi_\mathcal{A}\circ\varphi_\mathcal{A}$ are bijective is trivial to see. 
Let $p,p'\in\mathcal{P}$. The proof uses several times the fact that $\varphi^{-1}$ is a thermodynamic isomorphism for $A_1\hat=A_2$, established by Lemma~\ref{lemma:equiinverseapp}.
It is important that the three atomic systems are all different. Otherwise the intuition that the conjugated isomorphism essentially swaps $A_1$ with $A_3$ is wrong.\footnote{If $A_2=A_3$, for instance, then $\psi$ would not do anything (except for maybe relabelling some processes). Hence the total map would not do anything either as everything done by $\varphi$ would be undone by $\varphi^{-1}$ and $A_1$ would not be swapped with $A_3=A_2$. }
We will write $\varphi^{-1}(\psi(\varphi(\cdot)))$ instead of $\varphi^{-1}\circ\psi\circ\varphi$ in order not to confuse the concatenation of these mapping with the concatenation of thermodynamic processes.\\
Let $p,p'\in\mathcal{P}$ and $A\in\mathcal{A}$.
\begin{itemize}
	\item [(i)] 
	Suppose $p'\circ p$ is defined, then:
	\begin{align}
	\begin{split}
	\varphi^{-1}(\psi(\varphi(p'\circ p))) 
	&= \varphi^{-1}(\psi(\varphi(p')\circ\varphi(p))) 
	= \varphi^{-1}(\psi(\varphi(p'))\circ\psi(\varphi(p)))  \\
	&= \varphi^{-1}(\psi(\varphi(p')))\circ\varphi^{-1}(\psi(\varphi(p))) \,.
	\end{split}
	\end{align}
	If the right hand side is defined, the equation must hold, too.
	\item [(ii)] Since all involved maps are thermodynamic isomorphisms and fulfil (ii)
	individually, it easy to see shell by shell that
	\begin{align}
	\lfloor p \rfloor_A \text{ def.} 
	\Leftrightarrow \cdots \Leftrightarrow 
	\lfloor \varphi^{-1}(\psi(\varphi(p)))) \rfloor_{\varphi^{-1}_\mathcal{A}(\psi_\mathcal{A}(\varphi_\mathcal{A}(A)))} \text{ def.}\,.
	\end{align}
	\item [(iii)] With the same logic as in (i) and (ii) it follows
	\begin{align}
	W_{\varphi^{-1}_\mathcal{A}(\psi_\mathcal{A}(\varphi_\mathcal{A}(A)))}(\varphi^{-1}(\psi(\varphi(p)))) = W_{\psi_\mathcal{A}(\varphi_\mathcal{A}(A))}(\psi(\varphi(p))) 
	= W_{\varphi_\mathcal{A}(A)}(\varphi(p)) = W_A(p)\,.
	\end{align}
	\item [(iv)] We check
	\begin{align}
	\begin{split}
	\varphi^{-1}_\mathcal{A}(\psi_\mathcal{A}(\varphi_\mathcal{A}(A_1)))
	&= \varphi^{-1}_\mathcal{A}(\psi_\mathcal{A}(A_2))
	= \varphi^{-1}_\mathcal{A}(A_3) 
	= A_3 \,, \\
	\varphi^{-1}_\mathcal{A}(\psi_\mathcal{A}(\varphi_\mathcal{A}(A_2)))
	&= \varphi^{-1}_\mathcal{A}(\psi_\mathcal{A}(A_1))
	= \varphi^{-1}_\mathcal{A}(A_1) 
	= A_2 \,, \\
	\varphi^{-1}_\mathcal{A}(\psi_\mathcal{A}(\varphi_\mathcal{A}(A_3)))
	&= \varphi^{-1}_\mathcal{A}(\psi_\mathcal{A}(A_3))
	= \varphi^{-1}_\mathcal{A}(A_2) 
	= A_1 \,.
	\end{split}
	\end{align}
	For all other atomic systems $A\in\mathcal{A}\smallsetminus\{A_1,A_2,A_3\}$ it follows
	$\varphi^{-1}_\mathcal{A}(\psi_\mathcal{A}(\varphi_\mathcal{A}(A)))=A$ since they are
	mapped to themselves under all three mappings.
	\item [(v)] For $A\in\mathcal{A}\smallsetminus\{A_1,A_2,A_3\}$ this follows shell by shell.
	For $A_2$ we shortly check
	\begin{align}
	\begin{split}
	\lfloor \varphi^{-1}(\psi(\varphi(p)))\rfloor_{A_2} 
	&= \lfloor \psi(\varphi(p))\rfloor_{\varphi_\mathcal{A}(A_2)} 
	= \lfloor \psi(\varphi(p))\rfloor_{A_1} 
	= \lfloor \varphi(p)\rfloor_{\psi^{-1}_\mathcal{A}(A_1)} \\
	&= \lfloor \varphi(p)\rfloor_{A_1} 
	= \lfloor p\rfloor_{\varphi^{-1}_\mathcal{A}(A_1)}
	= \lfloor p \rfloor_{A_2}\,.
	\end{split}
	\end{align}
	The same follows for $\lceil \cdot \rceil_\cdot$.
	
\end{itemize}
Checking $(\varphi^{-1} \circ \psi \circ \varphi)_\mathcal{S} 
= \varphi^{-1}_\mathcal{S} \circ \psi_\mathcal{S} \circ \varphi_\mathcal{S} $
is trivial after (iv). For the mapping of states let $S\in\mathcal{S}$ and $p\in\mathcal{P}$ and it follows
\begin{align}
\begin{split}
\varphi^{-1}_\Sigma(\psi_\Sigma(\varphi_\Sigma(\lfloor p \rfloor_S)))
&= \varphi^{-1}_\Sigma(\psi_\Sigma( \lfloor \varphi( p ) \rfloor_{\varphi_\mathcal{S}(S)}))
= \varphi^{-1}_\Sigma(\lfloor \psi(\varphi( p )) \rfloor_{\psi_{\mathcal{S}}(\varphi_\mathcal{S}(S))}) \\
&= \lfloor \varphi^{-1}(\psi(\varphi( p ))) \rfloor_{\varphi^{-1}_\mathcal{S}(\psi_{\mathcal{S}}(\varphi_\mathcal{S}(S)))}
= (\varphi^{-1}\circ\psi\circ\varphi)_\Sigma(\lfloor p \rfloor_S)\,.
\end{split}
\end{align}
\end{proof}

\begin{prop}[$\hat=$ is equivalence relation on $\mathcal{A}$]
\label{prop:equirelA}
The relation $\hat=$ is reflexive, symmetric and transitive, hence it is an equivalence relation on the set of atomic systems $\mathcal{A}$.
\end{prop}

\begin{proof}
Reflexive: 
Consider the identity $\varphi(p) = p$ for all $p\in\mathcal{P}$. This map clearly fulfils all requirements in Definitions~\ref{def:isomorphismapp} and~\ref{def:equiatomapp} for the atomic system(s) $A$ and $A$ and thus 
$A\hat=A$ always.\\
Symmetric:
By Definition~\ref{def:isomorphismapp} the roles of $A_1$ and $A_2$ are symmetric. Hence if $A_1\hat=A_2$ then also $A_2\hat=A_1$.\\
Transitive:
Lemma~\ref{lemma:equiconjapp} states that if $A_1\hat=A_2$ and $A_2\hat=A_3$ holds for three different atomic systems, then the two associated thermodynamic isomorphisms can be conjugated to a new thermodynamic isomorphism for $A_1\hat=A_3$, which proves that in this case also $A_1\hat=A_3$.
If two or all of the three systems are the same, then transitivity follows from reflexivity and symmetry.
\end{proof}

A comment on the repeated application of a thermodynamic isomorphism for $A_1\hat=A_2$: by definition of 
$\varphi_\mathcal{S}$ it is idempotent, i.e.\ applying this function twice yields the identity, 
$\varphi_\mathcal{S}(\varphi_\mathcal{S}(S))=S$ for all $S\in\mathcal{S}$. 
One could thus expect that something similar holds for $\varphi$ and $\varphi_\Sigma$. 
This is not the case in general.

For $\varphi$, the mapping of thermodynamic processes, idempotency is in general wrong because there may be two thermodynamic process with exactly the same thermodynamic properties, i.e.\ they act on the same systems, induce the same state changes, have the same work cost, and so on. Under $\varphi(\varphi(\cdot))$ these may be interchanged, which makes $\varphi(\varphi(\cdot))$
different from the identity mapping, even though thermodynamically nothing has changed.\footnote{
We do not use the $\circ$ notation for the concatenation of $\varphi$ with itself on purpose in order not to confuse the reader, as $\circ$ is already used for the concatenation of thermodynamic processes. Instead, we used the notation $\varphi(\cdot)$ to still be able to denote a concatenation of such functions.} 
As for the mapping of states $\varphi_\Sigma$, it does not have to be idempotent either. Finding an easy example where this is the case but intuition does not fail is a bit trickier than in the case of processes. However, the point is that $\varphi_\Sigma$ must map thermodynamic states to a ``thermodynamically equivalent'' states, i.e.\ a state with exactly the same thermodynamic properties. This can in principle also happen when $\varphi_\Sigma(\varphi_\Sigma(\cdot))$ is not an identity mapping. \\

These observation are the reason why we had go via Lemmas~\ref{lemma:equiinverseapp} and~\ref{lemma:equiconjapp} to prove that $\hat=$ is an equivalence relation.\\


We can now extend the definition of equivalent systems from atomic to arbitrary systems.
Namely, two arbitrary systems are equivalent if their atomic subsystems are pairwise equivalent.
The extended relation will keep the status of an equivalence relation, as we show below.

\begin{definition}[Equivalence of systems]
\label{def:equisysapp}
Let $S_1,S_2\in\mathcal{S}$ be two arbitrary systems. They are \emph{equivalent}, and we write $S_1\hat=S_2$, if 
there exists a bijection between $\mathrm{Atom}(S_1)$ and $\mathrm{Atom}(S_2)$ which respects the equivalence classes of $\hat=$ for atomic systems. 
\end{definition}

Definition~\ref{def:equisysapp} can be rephrased as: The two systems are equivalent if 
$|\mathrm{Atom}(S_1)| = |\mathrm{Atom}(S_2)| =: n$ and there exists a labelling 
$\{A_k^{(i)}\}_{i=1,..,n} = \mathrm{Atom}(S_k)$ for $k=1,2$ such that
\begin{align}
A_1^{(i)} \hat= A_2^{(i)} \quad \text{for } i=1,\dots,n \,.
\end{align}
Again we need to prove that $\hat=$ is an equivalence relation, now on $\mathcal{S}$.

\begin{prop}[$\hat=$ is equivalence relation on $\mathcal{S}$]
\label{prop:equirel}
The relation $\hat =$ on $\mathcal{S}$ is an equivalence relation.
\end{prop}

\begin{proof}
Reflexive: If $S_1=S_2$, i.e.\ $\mathrm{Atom}(S_1) = \mathrm{Atom}(S_2)$, then a possible labelling is $A_1^{(i)}=A_2^{(i)}$ for $i=1,\dots,|\mathrm{Atom}(S_1)| =: n$. \\
Symmetry: Again, the roles of $S_1$ and $S_2$ in Definition~\ref{def:equisysapp} are interchangeable. \\
Transitivity: Suppose $S_1\hat=S_2$ and $S_2\hat=S_3$. This means that there are two labellings 
$\{A_k^{(i)}\}_{i=1,..,n} = \mathrm{Atom}(S_k)$ for $k=1,2$ and 
$\{\tilde A_l^{(j)}\}_{j=1,..,n} = \mathrm{Atom}(S_l)$ for $l=2,3$ which fulfil
\begin{align}
\begin{split}
A_1^{(i)} &\hat= A_2^{(i)} \quad \text{for } i=1,\dots,n \,, \\
\tilde A_2^{(j)} &\hat= \tilde A_3^{(j)} \quad \text{for } j=1,\dots,n \,.
\end{split}
\end{align}
If the two labellings do not agree on $\mathrm{Atom}(S_2)$, one can relabel the $j$'s such that $\tilde A_2^{(j)} = A_2^{(j)}$ for $j=1,\dots,n$.
By transitivity of $\hat=$ for atomic systems, we then find
\begin{align}
A_1^{(j)} \hat= \tilde A_3^{(j)} \quad \text{for } j=1,\dots,n
\end{align}
and thus $S_1\hat=S_3$. 
\end{proof}

\begin{lemma}[Equivalence and composition]
\label{lemma:equicomp}
Let $A_1,A_2\in\mathcal{A}$ and $S\in\mathcal{S}$ an arbitrary system disjoint with $A_1$ and $A_2$, $\mathrm{Atom}(S) \cap \{A_1,A_2\} = \emptyset$.
Then:
$A_1\hat=A_2 \Longleftrightarrow A_1\vee S \hat=A_2\vee S$.
\end{lemma}

\begin{proof}
Let $\{A_S^{(1)},\dots,A_S^{(n)}\} = \mathrm{Atom}(S)$.
By assumption, $A_1$ and $A_2$ are not in this list, i.e.\ all atomic systems we deal with are different.
The direction $A_1\hat=A_2 \Rightarrow A_1\vee S \hat=A_2\vee S$ follows immediately with the labelling 
\begin{align}
\begin{split}
A_1 &\hat= A_2 \,,\\
A_S^{(j)} &\hat= \tilde A_S^{(j)} \quad \text{for } j=1,\dots,n \,.
\end{split}
\end{align}
Let now $A_1\vee S \hat=A_2\vee S$. 
Consider the labelling of atomic system in this equivalence. W.l.o.g.\ the numbering can be chosen such that $A_1 \hat= A_S^{(1)} \hat= \cdots \hat= A_S^{(k)}$.
It holds $k\geq1$ because by assumption $A_1$ is equivalent to at least one atomic system in 
$\mathrm{Atom}(S)\cup A_2$ and it is not $A_2$.
If $k=n$, it means that all atomic systems at hand are equivalent to each other and consequently $A_1\hat=A_2$ must hold.
If $k<n$ then the labelling must be such that members of $\{A_S^{(k+1)},\dots,A_S^{(n)}\}$ are equivalent to one or more other members of this set but not to any members of $\{A_S^{(1)},\dots,A_S^{(k)}\}$.
I.e.\ the labelling is such that the ``sublabelling'' of the subset 
$\{A_S^{(k+1)},\dots,A_S^{(n)}\}$ of atomic systems is closed. 
But this means that there is no space for $A_2$ in this part of the labelling. Hence $A_2$ is must be equivalent to at least one member of $\{A_S^{(1)},\dots,A_S^{(k)}\}$ and thus by transitivity also $A_2\hat=A_1$. 
\end{proof}

This section of the appendix was opened with a definition thermodynamic isomorphisms on which the concept of equivalent atomic is based.
Having now defined when two arbitrary systems are called equivalent it may strike us that for this definition an explicit notion of a more general thermodynamic isomorphism associated with $A_1\hat=S_2$ was not necessary. This is because the definition of equivalence of arbitrary systems relies on the equivalences of the atomic subsystems. However, if we want to generalize the specific results derived for atomic equivalences (e.g.\ Lemma~\ref{lemma:isomwpapp}, Lemma~\ref{lemma:isomstateapp}, Lemma~\ref{lemma:isomstatebijapp}, Lemma~\ref{lemma:isomUapp}) it is necessary to explicitly talk about the correspondence of processes.

For this, we construct a thermodynamic isomorphism for the equivalence of two arbitrary equivalent systems.
\begin{definition}[Thermodynamic isomorphism for $S_1\hat=S_2$]
\label{def:equisysisomapp}
Let $S_1\hat=S_2$ according to Definition~\ref{def:equisysapp}
and let $n:=|\mathrm{Atom}(S_1)| = |\mathrm{Atom}(S_2)|$. Furthermore, let 
$l:=|\mathrm{Atom}(S_1)\cap\mathrm{Atom}(S_1)|$ be the number of shared atomic subsystems. 
Choose a labelling of the atomic subsystems of each system such that 
\begin{align}
\label{eq:pairequal}
\begin{split}
A_1^{(1)} = A_2^{(1)}, \dots,
A_1^{(l)} = A_2^{(l)}, 
A_1^{(l+1)} \hat= A_2^{(l+1)}, \dots, 
A_1^{(n)} \hat= A_2^{(n)}
\end{split}
\end{align}
(this implies that we can choose the isomorphisms for the first $l$ atomic equivalences as identities, $\varphi^{(i)}(p)=p$ and $\varphi_\mathcal{A}(A) = A$ for $i=1,\dots,l$).
Denote the remaining (non-identity) thermodynamic isomorphisms by $\varphi^{(i)}, \varphi_\mathcal{A}^{(i)}$ for $i=l+1,\dots,n$.
Then define the \emph{thermodynamic isomorphism $\varphi, \varphi_\mathcal{A}$ for $S_1\hat=S_2$} by
\begin{align}
\label{eq:equisysisom}
\begin{split}
\varphi(p) &:=  \varphi^{(l+1)}(\varphi^{(l+2)}(\cdots\varphi^{(n)}(p)\cdots))\,, \\
\varphi_\mathcal{A}(A) &:=  \varphi_\mathcal{A}^{(l+1)}(\varphi_\mathcal{A}^{(l+2)}(\cdots\varphi_\mathcal{A}^{(n)}(A)\cdots))\,.
\end{split}
\end{align} 
\end{definition}

Having defined a thermodynamic isomorphism for two equivalent atomic systems we must now check whether this generalization fulfils the expected generalized properties (i)-(v) from Definitions~\ref{def:isomorphismapp} and~\ref{def:equiatomapp}.

\begin{prop}[]
\label{prop:equisysisomapp}
A thermodynamic isomorphism $\varphi,\varphi_\mathcal{A}$ for $S_1\hat=S_2$ as in Definition~\ref{def:equisysisomapp} is indeed a thermodynamic isomorphism (i.e.\ fulfils Definition~\ref{def:isomorphismapp} (i)-(iii)) and
\begin{itemize}
	\item [(iv)]  for 
	$A\in\mathcal{A}\smallsetminus\big( \mathrm{Atom}(S_1)\cup\mathrm{Atom}(S_2) \big)$
	we have $\varphi_\mathcal{A}(A) = A$, while 
	$\varphi_\mathcal{A}(A_1^{(i)}) = A_2^{(i)}$ and
	$\varphi_\mathcal{A}(A_2^{(i)}) = A_1^{(i)}$, and
	\item [(v)] for all $p\in\mathcal{P}$ and 
	$A\in\mathcal{A}\smallsetminus\big( \mathrm{Atom}(S_1)\cup\mathrm{Atom}(S_2) \big)$
	we have 
	$\lfloor \varphi(p) \rfloor_A = \lfloor p \rfloor_A$
	and the same for $\lceil \cdot \rceil_\cdot$.
\end{itemize}
\end{prop}

\begin{proof}
Any thermodynamic isomorphism for two equivalent atomic systems is bijective by definition, hence so is a finite subsequent application of such maps.
For the remainder of the proof we note that the construction of $\varphi$ is such that the non-identity thermodynamic isomorphisms act non-trivially on disjoint pairs $\{A_1^{(i)}, A_2^{(i)}\}$ of atomic systems. 
For instance, $\varphi^{(n)}$ acts non-trivially on $A_1^{(n)}$ and $A_2^{(n)}$, but any $A_k^{(i)}$ with $i\neq n$ is untouched by $\varphi^{(n)}$ according to Definition~\ref{def:equiatomapp} (same for $\varphi_\mathcal{A}^{(n)}$). 
This observation is key for proving the five points characterizing thermodynamic isomorphisms for equivalences.\\
Let $p,p'\in\mathcal{P}$ be thermodynamic processes and $A\in\mathcal{A}$ am atomic system.
\begin{itemize}
	\item [(i)] If $p'\circ p$ is defined we obtain
	\begin{align}
	\begin{split}
	\varphi(p'\circ p)& = \varphi^{(l+1)}(\cdots\varphi^{(n)}(p'\circ p)\cdots) \\
	&= \varphi^{(l+1)}(\cdots\varphi^{(n-1)}(
		\varphi^{(n)}(p')\circ\varphi^{(n)}(p))\cdots) \\
	&= \dots \\
	&= \varphi^{(l+1)}(\cdots\varphi^{(n)}(p')\cdots) \circ
		\varphi^{(l+1)}(\cdots\varphi^{(n)}( p)\cdots) \\
		&= \varphi(p')\circ\varphi(p) \,,
	\end{split}	
	\end{align}
	by using Definition~\ref{def:isomorphismapp} (i) for $\varphi^{(i)}$ for $i=l+1,\dots,n$ subsequently.
	If instead $\varphi(p')\circ\varphi(p)$ is defined the equation still holds, 
	thus concluding the proof of (i).

	\item [(ii)] We use that the $\varphi^{(i)}$ fulfil (ii) individually: 
	\begin{align}
	\begin{split}
	&\lfloor \varphi(p) \rfloor_{\varphi_\mathcal{A}(A)} 
	= \lfloor \varphi^{(l+1)}(\cdots\varphi^{(n)}(p)\cdots) \rfloor_{\varphi_\mathcal{A}^{(l+1)}(\cdots\varphi^{(n)}_\mathcal{A}(A)\cdots)} 
	 \text{ def.} \\
	&\Leftrightarrow  
	\lfloor \varphi^{(l+2)}(\cdots\varphi^{(n)}(p)\cdots) \rfloor_{\varphi_\mathcal{A}^{(l+2)}(\cdots\varphi^{(n)}_\mathcal{A}(A)\cdots)} 
	\text{ def.} \\
	&\Leftrightarrow \cdots\\
	&\Leftrightarrow \lfloor p \rfloor_A \text{ def.}
	\end{split}
	\end{align}
	
	\item [(iii)] Like in the previous points we compute step by step, using the fact that the 
	$\varphi^(i),\varphi_\mathcal{A}^{(i)}$ fulfil (iii) individually:
	\begin{align}
	\begin{split}
	W_{\varphi_\mathcal{A}(A)}(\varphi(p)) 
	&= W_{\varphi^{(l+1)}_\mathcal{A}(\cdots \varphi^{(n)}_\mathcal{A}((A)\cdots)}(\varphi^{(l+1)}(\cdots \varphi^{(n)}(p)\cdots)) \\
	&= W_{\varphi^{(l+2)}_\mathcal{A}(\cdots \varphi^{(n)}_\mathcal{A}((A)\cdots)}(\varphi^{(l+2)}(\cdots \varphi^{(n)}(p)\cdots)) \\
	&= \cdots \\
	&= W_A(p)\,.
	\end{split}
	\end{align}
	
	\item [(iv)] Since all atomic systems 
	$A\in\mathcal{A}\smallsetminus\big( \mathrm{Atom}(S_1)\cup\mathrm{Atom}(S_2) \big)$
	are untouched by the $\varphi_\mathcal{A}^{(i)}$, i.e.\ mapped to themselves, 
	it immediately follows that $\varphi_\mathcal{A}(A) = A$ for those.\\ 
	For $i=1,\dots,l$ we see with the same argument that 
	$\varphi_\mathcal{A}(A_1^{(i)}) = A_1^{(i)} = A_2^{(i)} = \varphi_\mathcal{A}(A_2^{(i)})$,
	where we used that the labelling is chosen such that the first $l$ atomic subsystems are
	pairwise equal. \\
	For $i=l+1,\dots,n$ on the other hand, we know that the the $A_k^{(i)}$ are mapped to 
	themselves by the $\varphi_\mathcal{A}^{(j)}$ except for $j=i$, in which case 
	$\varphi_\mathcal{A}^{(i)}(A_1^{(i)}) = A_2^{(i)}$ and vice versa. 
	Thus $\varphi_\mathcal{A}(A_1^{(i)}) = A_2^{(i)}$ 
	and $\varphi_\mathcal{A}(A_2^{(i)}) = A_1^{(i)}$. 
\end{itemize}
\end{proof}

Proposition~\ref{prop:equisysisomapp} establishes the natural generalizations of the properties listed in Definition~\ref{def:equiatomapp} for thermodynamic isomorphisms of atomic equivalences to thermodynamic isomorphisms for generic equivalent systems. 
Since $\varphi,\varphi_\mathcal{A}$ from Definition~\ref{def:equisysisomapp} is in particular a thermodynamic isomorphism we can define associated bijective mappings  $\varphi_\mathcal{S}$ and $\varphi_\Sigma$ of arbitrary systems and states (Definition~\ref{def:isomstateapp}, Lemma~\ref{lemma:isomstatebijapp}). 
Also, the results about properties of work processes (Lemma~\ref{lemma:isomwpapp}), state changes (Lemma~\ref{lemma:isomstateapp}), and internal energy (Lemma~\ref{lemma:isomUapp}) under thermodynamic isomorphisms can be directly applied. They yield invariance of the quantities and properties under $\varphi,\varphi_\mathcal{A}$.\\

These observations conclude the construction of thermodynamic isomorphisms describing equivalences for general systems. 
We have shown that all relevant concepts introduced so far are invariant under equivalences. 
The concepts that are introduced from here on will also have this property, as will be shown for each.

\section{Heat and heat reservoirs}
\label{app:heat}

The heat flow into a system $S$ due to a process $p\in\mathcal{P}$, $Q_S(p)$, denoted by $Q_s$ in the following, is defined as the difference of the change in internal energy and work done during a thermodynamic process (Definition~\ref{def:heat}).
%
%
Heat $Q_S$ defined as such inherits all thermodynamically relevant properties from the work function $W_S$ and the internal energy $U_S$. 
If no atomic subsystem of $S$ is involved in $p$, then $Q_S(p)=0$;
likewise, $Q_S(\mathrm{id}) = 0$ for identity processes;
$Q_S(p)=0$ for all work processes $p\in\mathcal{P}_S$ on $S$;
heat flows in reverse processes change their signs;
$Q_S$ is additive under composition (for disjoint systems);
it is additive under concatenation (Lemma~\ref{lemma:heatconccompapp});
and finally, heat flows are invariant under thermodynamic isomorphisms (Lemma~\ref{lemma:equiheatapp}).

While the first five statements are direct to see, the latter two need a bit more discussion. They are states in the following two lemmas.

\begin{lemma}[Heat under concatenation]
\label{lemma:heatconccompapp}
Let $p, p'\in\mathcal{P}$ be arbitrary processes such that $p'\circ p$ is defined 
and consider an atomic system $A\in\mathcal{A}$. Then
\begin{align}
Q_A(p'\circ p) = Q_A(p) + Q_A(p')\,.
\end{align}
This statement naturally extends to arbitrary systems by additivity under composition.
\end{lemma}

\begin{proof}
We distinguish three cases. 
First, assume that $A$ is neither involved in $p$ nor in $p'$. This implies that both 
$Q_A(p)=Q_A(p')=0$, and that $A$ is not involved in $p'\circ p$ either.
But then $Q_A(p'\circ p) = 0 = Q_A(p) + Q_A(p')$ is obviously fulfilled.

Second, if $A$ is involved in $p$ but not in $p'$, only $Q_A(p')=0$ and $W_A(p')=0$ necessarily. In this case, according to the postulate introducing concatenation, $A$ is involved in $p'\circ p$ and undergoes the same state change (thus also the same change in internal energy) as it does under $p$ alone. We find 
\begin{align}
Q_A(p'\circ p) = \Delta U_A (p'\circ p) - W_A(p'\circ p) 
= \Delta U_A (p) - W_A(p') - W_A(p)
= Q_A(p) = Q_A(p) + Q_A(p')\,.
\end{align}
If $A$ is involved in $p'$ but not in $p$ the argument works analogously.

Third, assume that $A$ is involved in both $p$ and $p'$. In this case we argue directly that both 
$\Delta U_A$ and $W_A$ are additive under concatenation, the former by the fact that
it is a state variable, the latter by the additivity postulate for the work function. 
With Definition~\ref{def:heat} for heat this implies that also $Q_A$ is additive under concatenation.
\end{proof}

\begin{lemma}[Heat under isomorphisms]
\label{lemma:equiheatapp}
Let $\varphi, \varphi_\mathcal{A}$ be a thermodynamic isomorphism. 
Then for any process $p\in\mathcal{P}$ it holds 
\begin{align}
Q_{\varphi_\mathcal{A}(A)}(\varphi(p)) = Q_{A}(p) \,.
\end{align}
\end{lemma}

\begin{proof}
If $A_1$ is involved in $p$ we can make use of Lemma~\ref{lemma:isomUapp} implying that 
$\Delta U_{\varphi_\mathcal{A}(A)}(\varphi(p)) = \Delta U_{A}(p)$.
Together with Definition~\ref{def:isomorphism} (iii) it immediately follows
\begin{align}
\begin{split}	
Q_{\varphi\mathcal{A}(A)}(\varphi(p)) 
&= \Delta U_{\varphi_\mathcal{A}(A)}(\varphi(p))- W_{\varphi_\mathcal{A}(A)}(\varphi(p)) \\
&= \Delta U_{A}(p) - W_{A}(p) \\
&= Q_{A}(p)\,.
\end{split}
\end{align}
By Definition~\ref{def:equiatom} (ii) the atomic system $A$ is involved in $p$ if and only if $\varphi_\mathcal{A}(A)$ is involved in $\varphi(p)$.
Hence, if $A$ is not involved in $p$, $Q_{A}(p)=0=Q_{\varphi_\mathcal{A}(A)}(\varphi(p))$.
\end{proof}


Definition~\ref{def:heat} says what amount of heat flows into a specific system $S$, but it does not specify where this heat comes from. 
In a bipartite setting such as the one discussed in the main text, where a work process $p\in\mathcal{P}_{S_1\vee S_2}$ is considered, it is possible to say that heat $Q_{S_1}(p)$ flows from $S_2$ to $S_1$. Likewise in the opposite view, one can say that the heat $Q_{S_2}(p) \equiv -Q_{S_1}(p)$ flows from $S_1$ to $S_2$. 
However, in a more complex composite system such a statement is not necessarily possible, unless one splits it up into two subsystems to be considered, in which case we end up with the bipartite setting again.
This becomes relevant in the pictorial representation of work and heat flows, as is discussed in the main text.


\section{Carnot's Theorem}
\label{app:carnot}

As a direct consequence of the second law (Postulate~\ref{post:sec}) and Definition~\ref{def:heatreservoir} for heat reservoirs, we prove the following result on the signs of heat flows to cyclic machines operating between two reservoirs.

\begin{lemma}[Direction of heat flows in Carnot engines]
\label{lemma:carnotapp}
Consider a (not necessarily reversible) process $p\in\mathcal{P}_{R_1\vee S\vee R_2}$ operating on a cyclic system $S$ and two reservoirs $R_1,R_2\in\mathcal{R}$.
If $p$ is non-trivial, i.e.\ if not both heat flows $Q_{R_1}(p)$ and $Q_{R_2}$ are zero, then
\begin{itemize}
\item [(i)] at least one of the heat flows is strictly positive, and
\item [(ii)] if $p$ is reversible, one of the heat flows is strictly positive and the other one strictly negative.
\end{itemize}
\end{lemma}

\begin{proof}
We prove (i) and then show that (ii) is a consequence of it. Suppose by contradiction that $Q_{R_1}(p) \leq 0$ and $Q_{R_2}(p) \leq 0$.
Excluding the trivial case, w.l.o.g.\ we can assume that $Q_{R_1}(p)<0$ (otherwise swap the labels $1\leftrightarrow2$. 
Since $\Delta U_{R_2}(p) = Q_{R_2}(p) \leq 0$ there exists a process $q\in\mathcal{P}_{R_2}$ with $W_{R_2}(q) = -\Delta U_{R_2}(p) \geq 0$ and well-defined concatenation $q\circ p$. This is a consequence of Definition~\ref{def:heatreservoir} (ii). 
By construction we find that under $q\circ p$
\begin{align}
\Delta U_{R_2}(q\circ p) = \Delta U_{R_2}(q) + \Delta U_{R_2}(p) = 0\,
\end{align}
implying that not only $S$ but $S\vee R_2$ is cyclic, and
\begin{align}
Q_{R_1}(q\circ p) = Q_{R_1}(p)<0\,.
\end{align}
This contradicts the second law, requiring $Q_{R_1}(q\circ p) \geq 0$, and proves (i).

For a reversible $p$ w.l.o.g.\ $Q_{R_2}(p)>0$. If $p^\mathrm{rev}$ is a reverse process, we know that $Q_{R_2}(p^\mathrm{rev}) = -Q_{R_2}(p) <0$ and due to (i) applied to $p^\mathrm{rev}$ it must be that $Q_{R_1}(p^\mathrm{rev})>0$.
Going back to $p$, this implies that $Q_{R_1}(p)<0$, which concludes the proof.
\end{proof}

\section{Absolute temperature}
\label{app:abstemp}


We investigate the reversible heat flows between equivalent reservoirs and a cyclic machine.
Lemma~\ref{lemma:ratioequivapp} proves that reversible heat flows to reversible engines from equivalent reservoirs are exactly opposite and hence $-\tfrac{Q_{R_1}}{Q_{R_2}}=1$.
In addition, we know from Lemma~\ref{lemma:equiheatapp} that swapping equivalent system leaves all heat flows invariant and thus the same holds for the ratio of heat flows in Carnot engines. 

\begin{lemma}[Reversible engine with equivalent reservoirs.]
\label{lemma:ratioequivapp}
Consider a reversible $p\in\mathcal{P}_{R_1\vee S\vee R_2}$, cyclic on $S$, with equivalent heat reservoirs $R_1\hat=R_2\in\mathcal{R}$. 
Then the heat flows to $R_1$ and $R_2$ must be exactly opposite, $Q_{R_2}(p) = -Q_{R_1}(p)$.
\end{lemma}

\begin{proof}
If $p$ is trivial, i.e.\ $Q_{R_1}(p)=Q_{R_2}(p)=0$, we are done. So we assume that it is non-trivial.
W.l.o.g.\ $Q_{R_2}(p)>0$, which implies that $Q_{R_1}(p)<0$ due to Lemma~\ref{lemma:carnotapp}.
For equivalent reservoirs $R_1\hat=R_2$ there exists a corresponding thermodynamic isomorphism $\varphi$.
Define $q:=\varphi(p)\in\mathcal{P}_{R_1\vee S \vee R_2}$. By definition, $q$ fulfils
$Q_{R_1}(q) = Q_{R_2}(p)$ and $Q_{R_2}(q) = Q_{R_1}(p)$.
Suppose now by contradiction that $-\tfrac{Q_{R_1}(p)}{Q_{R_2}(p)} \neq 1$.
W.l.o.g.\ $-\tfrac{Q_{R_1}(p)}{Q_{R_2}(p)} > 1$ (otherwise use the reverse process to $p$ and swap the labels $1\leftrightarrow2$).
Choose positive integers $k,l\in\mathbbm{N}$ such that
\begin{align}
-\frac{Q_{R_1}(p)}{Q_{R_2}(p)} > \frac{k}{l} > 1
\end{align}
and apply $p$ $l$ times followed by $k$ applications of $q$.\footnote{
As in the proof of Carnot's Theorem, by ``$k$ applications of $p$'' we mean that one applies $p$, then the corresponding process to $p$ for the new initial states of the reservoirs which exists due to Definition~\ref{def:heatreservoir} (iii) and so on, $k$ times in total. Then apply $q$ or, if the initial states of the reservoirs do not match, a translated version of $q$, $l$ times in the same manner.}
The total process is still cyclic on $S$ and the heat flows sum up to 
\begin{align}
Q_{R_1}^\mathrm{tot} = l\, Q_{R_1}(p) + k\, Q_{R_1}(q) 
= l\, Q_{R_1}(p) + k\, Q_{R_2}(p) 
= \underbrace{\left(\frac{k}{l}-\left(-\frac{Q_{R_1}(p)}{Q_{R_2}(p)}\right)\right)}_{< 0} 
  \cdot \underbrace{l\, Q_{R_2}(p)}_{> 0} < 0\,, \\
Q_{R_2}^\mathrm{tot} = l\, Q_{R_2}(p) + k\, Q_{R_2}(q)
l\, Q_{R_2}(p) + k\, Q_{R_1}(p) 
= \underbrace{\left(\frac{k}{l}-\left(-\frac{Q_{R_2}(p)}{Q_{R_1}(p)}\right)\right)}_{> 0} 
  \cdot \underbrace{l\, Q_{R_1}(p)}_{< 0} < 0\,,
\end{align}
where we used in the second line that $-\tfrac{Q_{R_2}(p)}{Q_{R_1}(p)}<1<\tfrac{k}{l}$.
This contradicts Lemma~\ref{lemma:carnotapp}, and thus concludes the proof.
\end{proof}

Postulate~\ref{post:Carnotexist} on the existence of reversible Carnot engines guarantees that the universality statement of Carnot's Theorem~\ref{thm:carnot} is meaningful for an arbitrary pair of heat reservoirs. 
As discussed, it is a comparability statement in this sense. 
Based on this we define the temperature ratio $\tau$ as follows.

\begin{definition}[Temperature ratio]
\label{def:tauapp}
The \emph{temperature ratio $\tau$ of two equivalence classes} $[R_1],[R_2]\in\sfrac{\mathcal{R}}{\hat=}$ is
\begin{align}
\label{eq:taudefapp}
\begin{split}
\tau: \ \sfrac{\mathcal{R}}{\hat=}\times\sfrac{\mathcal{R}}{\hat=}  &\longrightarrow \mathbbm{R}_{>0}\\
([R_1],[R_2]) &\longmapsto -\tfrac{Q_{R_1}}{Q_{R_2}}\,
\end{split}
\end{align}
where $R_1$ and $R_2$ are two different\footnote{
A Carnot engine operates between two \emph{different} heat reservoirs. If $R_1=R_2$, then $R_1\vee S \vee R_2 = R_1\vee S$ and there would be no two heat flows to compare.}
heat reservoirs and $Q_{R_1}$ and $Q_{R_2}$ are the heat flows of a (non-trivial) reversible Carnot engine operating between them such that $Q_{R_2}>0$. 
Extending this definition to the set of pairs of heat reservoirs, we use the same symbol $\tau$ and write for the \emph{temperature ratio of two heat reservoirs}
\begin{align}
\label{eq:taudef2app}
\begin{split}
\tau(R_1,R_2):=\tau([R_1],[R_2])\,.
\end{split}
\end{align}
\end{definition}

$\tau$ is well-defined, as discussed in the main text. 
Furthermore, it obviously fulfils $\tau([R],[R])=1$ and $\tau([R_2],[R_1])=\tau([R_1],[R_2])^{-1}$.

In the coming lemma about the temperature ratio $\tau$ it is proven that it actually behaves as a ratio, meaning that multiplying $\tau(R_1,R_2)$ with $\tau(R_2,R_3)$, where $R_2$ first shows up in the second argument and then in the first, is equal to $\tau(R_1,R_3)$, independently of the reservoir $R_2$.

\begin{lemma}[$\tau$ as a ratio]
\label{lemma:tautransapp}
Let $R_1,R_2,R_3\in\mathcal{R}$ be three arbitrary heat reservoirs. Then
%
$\tau(R_1,R_2)\cdot\tau(R_2,R_3) = \tau(R_1,R_3)$.
\end{lemma}
%

\begin{proof} 
W.l.o.g.\ the reservoirs $R_1\neq R_2\neq R_3\neq R_1$ are different representatives of their respective equivalence class. If this was not the case, take equivalent but different reservoirs (this is always possible due to the postulate on the existence of arbitrarily many copies).
For the constructive proof, we need two copies of $R_2$ in the beginning. Call them $R_2$ and $R_2'$ and choose them to be different from each other and all the others, too.

Let $S$ be a Carnot engine operating between $R_1$ and $R_2$ (all systems pairwise disjoint) through a reversible work process $p\in\mathcal{P}_{R_1\vee S \vee R_2}$ with $Q_{R_1}:= Q_{R_1}(p)<0$ and $Q_{R_2}:=Q_{R_2}(p)>0$. Likewise, let $S'$ and $p'\in\mathcal{P}_{R_1\vee S' \vee R_2}$ be an analogous machine and reversible process for the reservoirs $R_2'$ and $R_3$ with $Q_{R_2'}:=Q_{R_2'}(p')=-Q_{R_2}<0$ and $Q_{R_3}:=Q_{R_3}(p')>0$. Such a machine together with the reversible work process on $R_2'\vee S\vee R_3$ exists due to Postulate~\ref{post:Carnotexist}.\footnote{Notice that this is the first time where we use the fact that we can tune one of the heat flows. In all previous proofs we only used the fact that non-trivial machines exist between any two reservoirs.}
The described setting is depicted in Figure~\ref{fig:taulemma} (a).
It follows that 
\begin{align}
\tau(R_1,R_2)=-\frac{Q_{R_1}}{Q_{R_2}} \quad \text{and} \quad \tau(R_2',R_3)=-\frac{Q_{R_2'}}{Q_{R_3}}\,.
\end{align}

We now use an additional machine $S''$ with a reversible process $p''\in\mathcal{P}_{R_1\vee S'' \vee R_2'}$ that transfers the heat $Q_{R_2}$ from $R_2$ to $R_2'$. Such a machine exists due to Postulate~\ref{post:Carnotexist} and has no work cost according to Lemma~\ref{lemma:ratioequivapp}. 
Due to the translation invariance of reservoirs (Definition~\ref{def:heatreservoir} (iii)) it is possible to find $p, p'$ and $p''$ which can be concatenated to $p''\circ p' \circ p$. 

Under the concatenated process the reservoirs $R_2$ and $R_2'$ are cyclic by construction, while the machines $S$, $S'$ and $S''$ are also cyclic. 
%
%
Hence, under this extension the situation changes to the one depicted in Figure~\ref{fig:taulemma} (b), where the reservoirs $R_1$ and $R_3$ interact with a cyclic machine while taking up the heat flows $Q_{R_1}$ and $Q_{R_3}$, respectively. 
The construction could be made analogously for the reverse processes and as a consequence the constructed process is reversible. 
We conclude that $\tau(R_1,R_3)=-\tfrac{Q_{R_1}}{Q_{R_3}}$, which implies
\begin{align}
\tau(R_1,R_2)\cdot\tau(R_2',R_3)
= \left( -\frac{Q_{R_1}}{Q_{R_2}} \right) \cdot \left( -\frac{Q_{R_2'}}{Q_{R_3}} \right)
= \left( -\frac{Q_{R_1}}{Q_{R_2}} \right) \cdot \left( \frac{Q_{R_2}}{Q_{R_3}} \right)
= -\frac{Q_{R_1}}{Q_{R_3}} 
=\tau(R_1,R_3)\,.
\end{align}
Together with the observation that heat flows do not change when exchanging equivalent systems (Lemma~\ref{lemma:equiheatapp}), and hence $\tau(R_2',R_3) = \tau(R_2,R_3)$, this concludes the proof.

\begin{figure}

\begin{center} 
\begin{tikzpicture}[scale=.65]
	\draw[] (-.8,5) node[] {(a)};

	\begin{scope}[yshift=-1.1cm]
	\draw[very thick] (.5,4.5) node[above] {$R_1$}-- (1,4.5) -- (1,5.5) -- (0,5.5) -- (0,4.5) -- (.5,4.5); 
	\draw[very thick] (.5,0) node[above] {$R_2$} -- (1,0) -- (1,1) -- (0,1) -- (0,0) -- (.5,0); 
	\draw[] (.5,2.75) node[] {$S$} circle (.6cm);

	\draw[<-] (.5,4.4) -- (.5,3.5) node [above right] {$Q_{R_1}$};
	\draw[<-] (.5,1.1) -- (.5,2) node[below right] {$Q_{R_2} $};
	\draw[->] (2.15,2.75) -- (1.25,2.75) node[above right] {$W_S$};
	\end{scope}
	
	\begin{scope}[xshift = 4cm, yshift = -1.1cm]
		\draw[very thick] (.5,4.5) node[above] {$R_2'$}-- (1,4.5) -- (1,5.5) -- (0,5.5) -- (0,4.5) -- (.5,4.5); 
		\draw[very thick] (.5,0) node[above] {$R_3$} -- (1,0) -- (1,1) -- (0,1) -- (0,0) -- (.5,0); 
		\draw[] (.5,2.75) node[] {$S'$} circle (.6cm);
	
		\draw[<-] (.5,4.4) -- (.5,3.5) node [above right] {$Q_{R_2'} = -Q_{R_2}$};
		\draw[<-] (.5,1.1) -- (.5,2) node[below right] {$Q_{R_3}$};
		\draw[->] (2.15,2.75) -- (1.25,2.75) node[above right] {$W_{S'}$};
	\end{scope}
	
	\begin{scope}[xshift=10cm]
	\draw[] (-1,5) node[left] {(b)};	
	\draw[very thick] (.5,4.5) node[above] {$R_1$}-- (1,4.5) -- (1,5.5) -- (0,5.5) -- (0,4.5) -- (.5,4.5); 
	\draw[<-] (.5,4.4) -- (.5,3.5) node [above right] {$Q_{R_1}$};
	\draw[] (.5,2.35) node[] {$S \vee$};
	\draw[] (.5,1.65) node[] {$R_2\vee R_2'\vee$} circle (1.7cm);
	\draw[] (.5,.95) node[] {$S' \vee S''$};
	\draw[->] (3.2,1.65) -- (2.3,1.65) node[above right] {$W_{S}+W_{S'}$}; 
	\draw[very thick, yshift=-2.2cm] (.5,0) node[above] {$R_3$} -- (1,0) -- (1,1) -- (0,1) -- (0,0) -- (.5,0); 
	\draw[->, yshift=-4.6cm] (.5,4.4) -- (.5,3.5) node [above right] {$Q_{R_3}$};
	\end{scope}

	\end{tikzpicture}
\end{center}

\caption{(a) The two reversibly operating cyclic machines $S$ and $S'$ between $R_1$, $R_2$ and $R_2'$, $R_3$, respectively, are such that $Q_{R_2}=-Q_{R_2'}>0$ and consequently $Q_{R_1}<0$, $Q_{R_3}>0$. 
(b) After making use of Postulate~\ref{post:Carnotexist} the copies $R_2$ and $R_2'$ we obtain a reversibly operating cyclic machine $S\vee R_2\vee R_2' \vee S' \vee S''$ between $R_1$ and $R_3$.
}
\label{fig:taulemma}

\end{figure}
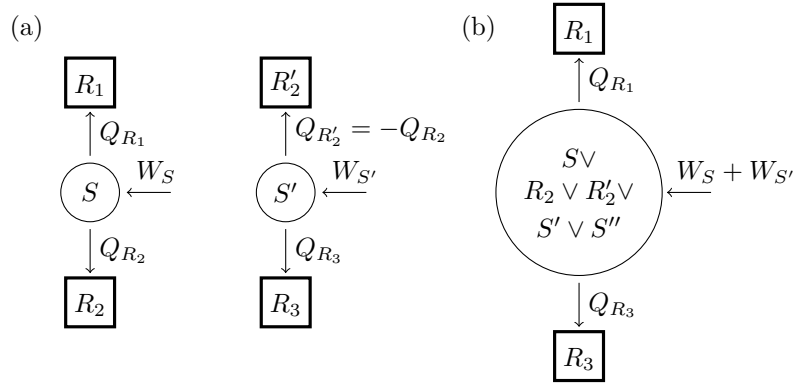

\end{proof}

It is then possible to define the absolute temperature of a heat reservoir up to the free choice of a reference temperature $T_\mathrm{ref}$ and a reference reservoir $R_\mathrm{ref}$.

\begin{definition}[Absolute temperature]
\label{def:abstempresapp}
The \emph{absolute temperature of a heat reservoir} $R\in\mathcal{R}$ is defined as
\begin{align}
T := \tau(R,R_\mathrm{ref}) \cdot T_\mathrm{ref}\,.
\end{align}
\end{definition}

It follows that absolute temperature, just like the heat flows in reversible Carnot engines, is a property of the reservoir as a whole and independent of its state.
This is what one expects for reservoirs, which are systems that should not change its properties and how they interact with other systems, after finite amounts of energy have been exchanged.\\

Systems at equal temperature are considered to be in thermal equilibrium. For reservoirs we can now make this definition precise.

\begin{definition}[Thermal equilibrium for reservoirs, $\sim$]
\label{def:simapp}
Let $R_1,R_2\in\mathcal{R}$ be heat reservoirs.
We say that they are in \emph{thermal equilibrium} if $\tau(R_1,R_2)=1$ and write $R_1\sim R_2$.
\end{definition}

$\hat=$, restricted to the set of heat reservoirs $\mathcal{R}$, is a sub-relation of $\sim$.
In addition, just like $\hat=$, $\sim$ is an equivalence relation.

\begin{lemma}[$\sim$ is equivalence relation]
\label{lemma:simequirelapp}
The relation $\sim$ defined in Definition~\ref{def:simapp} is an equivalence relation.
\end{lemma}

\begin{proof}
We need to show that it is reflexive, symmetric, and transitive. Reflexivity follows directly from Lemma~\ref{lemma:ratioequivapp}, while symmetry holds due to $\tau([R_2],[R_1])=\tau([R_1],[R_2])^{-1}$ for all $R_1,R_2\in\mathcal{R}$. 
Finally, transitivity is a consequence of Lemma~\ref{lemma:tautransapp}. 
\end{proof}

A thorough discussion of absolute temperature, the relation for thermal equilibrium and it implication is done in the main text.

\section{The temperature of heat flows}
\label{app:theatflow}

We shortly repeat the definition of the temperature of a heat flows. 

\begin{definition}[Heat at temperature $T$]
\label{def:heatattapp}
Let $S=S_1\vee S_2\in\mathcal{S}$ be composed of two disjoint subsystems and undergo an arbitrary work process $p\in\mathcal{P}_{S_1\vee S_2}$ with $Q:= Q_{S_2}(p)\neq0$.
We say that the \emph{heat $Q$ flows at temperature $T$} if
there exist two different reservoirs $R_1\sim R_2$ at temperature $T$ with processes $p_1\in\mathcal{P}_{S_1\vee R_1}$ and $p_2\in\mathcal{P}_{S_2\vee R_2}$ s.t.\ 
$W_{A}(p_i) = W_{A}(p)$ for all atomic systems 
$A\in\mathcal{A}\smallsetminus\mathrm{Atom}(S_{i+1})$
and the state changes on $S_i$ under $p_i$ are the same as under $p$, i.e.\ 
$\lfloor p_i \rfloor_{S_i} = \lfloor p \rfloor_{S_i}$ and 
$\lceil p_i \rceil_{S_i} = \lceil p \rceil_{S_i}$.
\end{definition}

In this section we technically investigate this definition.
The case of reversible heat flows will be of particular interest. 

\begin{lemma}[Reversible $p_i$ for reversible $p$]
\label{lemma:revpiapp}
Consider the setting described in Definition~\ref{def:heatattapp}.
If $p$ is reversible, then so are the $p_i$.
\end{lemma}

\begin{proof}
Let $p^\mathrm{rev}\in\mathcal{P}_{S_1\vee S_2}$ be a reverse process for $p$.
Furthermore let $p_C\in\mathcal{P}_{R_1\vee R_2}$ be the process that reversibly transports the heat $Q$ from $R_1$ to $R_2$ starting from the states $\lfloor p_1\rfloor_{R_1}$ and $\lceil p_2\rceil_{R_2}$. The process $p_C$ exists according to Postulate~\ref{post:Carnotexist}.
We then claim that the process 
$p_2^\mathrm{rev} := (p^\mathrm{rev}\vee p_C)\circ(p_1\vee \mathrm{id}_{S_2\vee R_2})\in\mathcal{P}_{S_1\vee R_1\vee S_2\vee R_2}$ is well-defined, where $\mathrm{id}_{S_2\vee R_2}$ is the identity process on the corresponding initial states of $p_C$ and $p^\mathrm{rev}$ for $R_2$ and $S_2$, respectively. See Figure~\ref{fig:revpiapp} for an illustration of the actions of the relevant processes.
Furthermore, this process is a reverse process for $p_2\vee \mathrm{id}_{S_1\vee R_1}$ on the composite system $S_1\vee R_1\vee S_2\vee R_2$ that is catalytic on $S_1\vee R_1$. 
Hence, according to the catalysis postulate there exists a process $\tilde p_2^\mathrm{rev}\in\mathcal{P}_{S_2\vee R_2}$ with the same thermodynamic properties, which then must be a reverse process for $p_2$. Thus $p_2$ (and if we do the same for $p_1$ also this one) must be reversible.\\
We are left with the task to show the above claims about $p_2^\mathrm{rev}$.
The fact that it is well-defined (i.e.\ the concatenation of the two processes exists) follows from checking that the input and output states of the two processes match.
Regarding the cyclicity on $S_1\vee R_1$ the consideration is equally easy.
Finally, we check the net work costs of $S_1$ and $R_1$
\begin{align}
W_{S_1}(p_2^\mathrm{rev}) &= W_{S_1}(p^\mathrm{rev}) +W_{S_1}(p_1) 
= W_{S_1}(p^\mathrm{rev}) +W_{S_1}(p) 
=0\,,
\label{eq:S1cat}\\
W_{R_1}(p_2^\mathrm{rev}) &= W_{R_1}(p_C) = 0\,.
\end{align}
In the last equality of Eq.~(\ref{eq:S1cat}) we used the fact that the work flows in $p^\mathrm{rev}$ are exactly opposite to the ones during $p$.
Hence, $p_2^\mathrm{rev}$ is indeed catalytic on $S_1\vee R_1$. 
\end{proof}

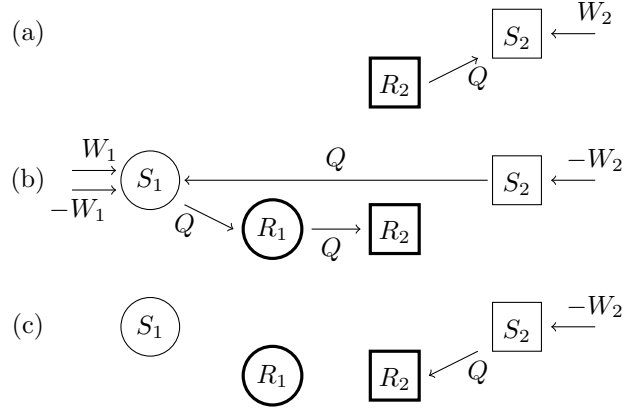
\begin{figure}
\begin{center} 
	\begin{tikzpicture}[scale=.65]

		\draw[](-2,0) node[] {(a)};
		
		\draw[->, xshift=5cm] (1.2,-1) -- (2.2,-.5) node[below] {$Q$};
		\draw[->, xshift=5cm] (4.6,0) node[above] {$W_2$} -- (3.7,0);

		\draw[xshift=5cm, yshift=-1cm, very thick] (.5,-.5) node[above] {$R_2$} -- (1,-.5) -- (1,.5) -- (0,.5) -- (0,-.5) -- (.5,-.5);
		\draw[xshift = 5cm] (3,-.5) node[above] {$S_2$} -- (3.5,-.5) -- (3.5,.5) -- (2.5,.5) -- (2.5,-.5) -- (3,-.5);
		
		\begin{scope}[yshift=-3cm]
		\draw[](-2,0) node[] {(b)};
		
		\draw[] (.5,0) node[] {$S_1$} circle (.6cm);
		\draw[very thick] (3,-1) node[] {$R_1$} circle (.6cm);
		\draw[<-] (1.2,0) -- (4.3,0) node[above] {$Q$} -- (7.4,0);
		\draw[->, yshift=.2cm, xshift=-.1cm] (-1,0) node[above right] {$W_1$} -- (-.1,0);
		\draw[->, yshift=-.2cm, xshift=-.1cm] (-1,0) -- (-.1,0) node[below left] {$-W_1$};
		\draw[->] (3.8,-1) node[below right] {$Q$} -- (4.8,-1);
		\draw[->] (1.2,-.5) node[below] {$Q$} -- (2.2,-1);
		\draw[->, xshift=5cm] (4.6,0) node[above] {$-W_2$} -- (3.7,0);

		\draw[xshift=5cm, yshift=-1cm, very thick] (.5,-.5) node[above] {$R_2$} -- (1,-.5) -- (1,.5) -- (0,.5) -- (0,-.5) -- (.5,-.5);
		\draw[xshift = 5cm] (3,-.5) node[above] {$S_2$} -- (3.5,-.5) -- (3.5,.5) -- (2.5,.5) -- (2.5,-.5) -- (3,-.5);
		\end{scope}

		\begin{scope}[yshift=-6cm]
		\draw[](-2,0) node[] {(c)};
		
		\draw[] (.5,0) node[] {$S_1$} circle (.6cm);
		\draw[very thick] (3,-1) node[] {$R_1$} circle (.6cm);
		\draw[<-, xshift=5cm] (1.2,-1) -- (2.2,-.5) node[below] {$Q$};
		\draw[->, xshift=5cm] (4.6,0) node[above] {$-W_2$} -- (3.7,0);

		\draw[xshift=5cm, yshift=-1cm, very thick] (.5,-.5) node[above] {$R_2$} -- (1,-.5) -- (1,.5) -- (0,.5) -- (0,-.5) -- (.5,-.5);
		\draw[xshift = 5cm] (3,-.5) node[above] {$S_2$} -- (3.5,-.5) -- (3.5,.5) -- (2.5,.5) -- (2.5,-.5) -- (3,-.5);
		\end{scope}

	\end{tikzpicture}
\end{center}
\caption{(a) Under $p_2\in\mathcal{P}_{S_2\vee R_2}$ the heat $Q$ flows from $R_2$ to $S_2$, according to Definition~\ref{def:heatattapp}. 
(b) In the construction of $p_2^\mathrm{rev}\in\mathcal{P}_{S_1\vee R_1\vee S_2\vee R_2}$ the heat $Q$ flows from $S_2$ via $S_1$ and $R_1$ to $R_2$ and the state change is exactly the opposite, as well as the work done on $S_2$. The systems $S_1$ and $R_1$ are catalytic.
(c) Essentially, under $p_2^\mathrm{rev}$ the heat $Q$ flows from $S_2$ to $R_2$ while nothing is done on the systems $S_1$ and $R_1$. Formally, Postulate~\ref{post:freedom} on the freedom of description then guarantees that the process can also be seen as a work process $\tilde p_2^\mathrm{rev}\in\mathcal{P}_{S_2\vee R_2}$ alone. Since the state changes are indeed opposite to the ones under $p_2$, we have found a reverse process for it.
}
\label{fig:revpiapp}
\end{figure}
%

Based on the previous lemma it is possible to show that non-zero reversible heat flows an assigned temperature is unique.

\begin{lemma}[Unique temperature for reversible heat flows]
\label{lemma:revheatuniquetapp}
Consider again the setting described in Definition~\ref{def:heatattapp}, in particular the heat flow is non-zero, $Q\neq0$.
If $p$ is reversible, then the heat flow in the reverse process has the same temperature.
Furthermore, this temperature is unique. 
\end{lemma}

\begin{proof}
We know from Lemma~\ref{lemma:revpiapp} that for a reversible process $p$ fulfilling Definition~\ref{def:heatattapp} the corresponding divided processes $p_i$ are also reversible. Hence, for the reverse process, $p^\mathrm{rev}$, the corresponding reverse processes of $p_i$ will fulfil the definition as well, now of course for the reverse heat flow $-Q$. 
Since the reservoirs in the reverse processes of the $p_i$ are the same, we have shown that for the particular reverse process $p^\mathrm{rev}$ the reverse heat flow also has temperature $T$.\\
However, this alone does not exclude that more than one temperature could be assigned to the reverse heat flow (as well as the forward heat flow). 
So let us assume that the heat flow $Q$ exchanged under the reversible process $p\in\mathcal{P}_{S_1\vee S_2}$ can be assigned the two temperatures $T$ and $T'$. According to Definition~\ref{def:heatattapp} this means that there exist processes $p_i\in\mathcal{P}_{S_i\vee R_i}$ and $p_i'\in\mathcal{P}_{S_i\vee R_i'}$ that divide the process $p$, where $R_1\sim R_2$ and $R_1'\sim R_2'$. Since $p$ is reversible, all four processes $p_1,p_2,p_1',p_2'$ are reversible, too. Call the reverse processes $p_i^\mathrm{rev}$ and $p_i^\mathrm{rev}{}'$, respectively. If we now concatenate the processes $p_1$ and $p_1^\mathrm{rev}{}'$ to a new process 
$(\mathrm{id}_{R_1}\vee p_1^\mathrm{rev}{}') \circ (p_1\vee \mathrm{id}_{R_1'})$, as depicted in Figure~\ref{fig:revheatunique}, we can construct a reversible engine $S_1$ operating between $R_1$ and $R_1'$. 
Reversibility follows from the fact that all processes used to construct the engine are reversible. 
According to Carnot's Theorem, which applies to this situation, we find that 
\begin{align}
\frac{T}{T'} = -\frac{-Q}{Q} = 1\,,
\end{align}
i.e.\ $T=T'$. This implies that $R_1\sim R_2\sim R_1'\sim R_2'$.
Remember that for this conclusion it is important that $Q\neq0$, which is a condition in Definition~\ref{def:heatattapp}.
Thus, for reversible processes, if a temperature can be assigned to a heat flow, this temperature is unique. 
\end{proof}

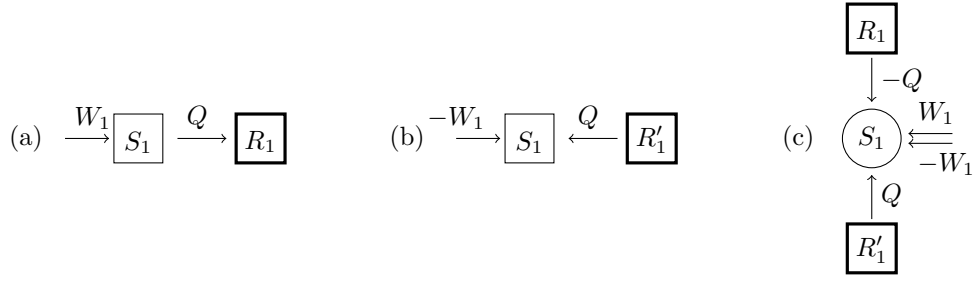
\begin{figure}
\begin{center} 
	\begin{tikzpicture}[scale=.65]
	
		\draw[](-1.8,0) node[] {(a)};
		
		\draw[] (.5,-.5) node[above] {$S_1$} -- (1,-.5) -- (1,.5) -- (0,.5) -- (0,-.5) -- (.5,-.5);
		\draw[very thick] (3,-.5) node[above] {$R_1$} -- (3.5,-.5) -- (3.5,.5) -- (2.5,.5) -- (2.5,-.5) -- (3,-.5);
		\draw[->] (1.3,0) node[above right] {$Q$} -- (2.3,0);
		\draw[->] (-1,0) node[above right] {$W_1$} -- (-.1,0);
		
		\begin{scope}[xshift=8cm]
		\draw[](-2,0) node[] {(b)};
		
		\draw[] (.5,-.5) node[above] {$S_1$} -- (1,-.5) -- (1,.5) -- (0,.5) -- (0,-.5) -- (.5,-.5);
		\draw[very thick] (3,-.5) node[above] {$R_1'$} -- (3.5,-.5) -- (3.5,.5) -- (2.5,.5) -- (2.5,-.5) -- (3,-.5);
		\draw[<-] (1.3,0) node[above right] {$Q$} -- (2.3,0);
		\draw[->] (-1,0) node[above] {$-W_1$} -- (-.1,0);
		\end{scope}		
		
		\begin{scope}[xshift=15cm, yshift=-2.75cm]
		\draw[](-1,2.75) node[] {(c)};

		\draw[very thick] (.5,4.5) node[above] {$R_1$}-- (1,4.5) -- (1,5.5) -- (0,5.5) -- (0,4.5) -- (.5,4.5); 
		\draw[very thick] (.5,0) node[above] {$R_1'$} -- (1,0) -- (1,1) -- (0,1) -- (0,0) -- (.5,0); 
		\draw[] (.5,2.75) node[] {$S_1$} circle (.6cm);

		\draw[->] (.5,4.4) -- (.5,3.5) node [above right] {$-Q$};
		\draw[->] (.5,1.1) -- (.5,2) node[below right] {$Q$};
		\draw[->,yshift=.1cm] (2.15,2.75) -- (1.25,2.75) node[above right] {$W_1$};
		\draw[->,yshift=-.1cm] (2.15,2.75) -- (1.25,2.75) node[below right] {$-W_1$};
		\end{scope}
	\end{tikzpicture}
\end{center}
\caption{(a) Under $p_1\in\mathcal{P}_{S_1\vee R_1}$ the state change in system $S_1$ is exactly the opposite compared to the state change under $p_1^\mathrm{rev}{}'\in\mathcal{P}_{S_1\vee R_1'}$ depicted in (b). If we concatenate the two reversible processes the total process will be cyclic on $S_1$ and reversible, since the construction with $p_1^\mathrm{rev}$ and $p_1'$ will yield a reverse process. Hence Carnot's Theorem can be directly applied to the situation in (c).
}
\label{fig:revheatunique}
\end{figure}

We conclude that a non-zero reversible heat flow can either be assigned a unique temperature or no temperature at all. \\

Referring to Example~\ref{ex:heatrevdifft} from the main text, one can investigate the relation between different reversible heat flows inducing the same state change on a system. 
%

%

\begin{lemma}[Different reversible heat flows inducing the same state change]
\label{lemma:heatratioQTapp}
Consider the setting as described in Definition~\ref{def:heatattapp} with a reversible process $p\in\mathcal{P}_{S_1\vee S_2}$. In addition, consider a different reversible process $p'\in\mathcal{P}_{S_1\vee S_2}$ with $\lfloor p'\rfloor_{S_1}=\lfloor p\rfloor_{S_1}$ 
and $\lceil p'\rceil_{S_1}=\lceil p\rceil_{S_1}$,
i.e.\ both $p$ and $p'$ reversibly induce the same state change on $S_1$.
If the heat $Q$ under $p$ is exchanged at temperature $T$ and $Q'$ under $p'$ at temperature $T'$, and both $Q\neq0$ and $Q'\neq0$,
then:
\begin{align}
\frac{Q}{T} = \frac{Q'}{T'}\,.
\end{align}
\end{lemma}

\begin{proof}
Let $p_1\in\mathcal{P}_{S_1\vee R_1}$ and $p_1'\in\mathcal{P}_{S_1\vee R_1'}$ be the divided processes for $p$ and $p'$, respectively, and let $p_1^\mathrm{rev}$ and 
$p_1^\mathrm{rev}{}'$ be reverse processes for them. These exist since $p_1$ and $p_1'$ are reversible according to Lemma~\ref{lemma:revpiapp}.

Just like in the proof of Lemma~\ref{lemma:revheatuniquetapp} we now construct a reversible process
$(p_1^\mathrm{rev}\vee\mathrm{id}_{R_1'}) \circ (p_1'\vee \mathrm{id}_{R_1})$, which is essentially a Carnot process on $S_1$ between the reservoirs $R_1$ and $R_1'$.
The concatenation is well-defined since the final states of $p_1$ and $p_1'$ match on $S_1$.
During this process, heat $Q$ flows from $R_1$ to $S_1$ and heat $Q'$ flows from $S_1$ to $R_1'$. 
$S_1$ itself will end up in its initial state since the initial states of $p_1$ and $p_1'$ also match on $S_1$. 
Hence, $S_1$ indeed acts as a cyclic machine between the two reservoirs. 
Finally, the constructed process is reversible since all its parts are.
According to Carnot's Theorem, we must have
\begin{align}
\frac{Q}{Q'} = \frac{T}{T'}\,,
\end{align}
which concludes the proof.
\end{proof}



Importantly, even though Definition~\ref{def:heatattapp} explicitly talks about both systems $S_1$ and $S_2$, the ratio $\tfrac{Q}{T}$ only depends on the state change on $S_1$. Thus, Lemma~\ref{lemma:heatratioQTapp} does not make any reference to the state changes on $S_2$ under $p$ and $p'$. Only the state change on $S_1$ needs to match.
This is an important observation which, after all, is key to be able to use the ratio $\tfrac{Q}{T}$ to define a state function (the entropy).

\section{Clausius' Theorem and thermodynamic entropy}
\label{app:clausius}

We here prove that entropy is additive under composition of disjoint systems. 

\begin{lemma}[Additivity of entropy]
\label{lemma:entropyaddapp}
Let $S_1,S_2\in\mathcal{S}$ be disjoint systems and write $S=S_1\vee S_2$. 
Furthermore, let $\sigma,\sigma'\in\Sigma_S$ be such that $\sigma= \sigma_1\vee \sigma_2$ and $\sigma'= \sigma_1'\vee \sigma_2'$ with $\sigma_i, \sigma_i' \in\Sigma_{S_i}$.
Write $\Delta S_S := S_S(\sigma') - S_S(\sigma)$ 
and $\Delta S_{S_i} := S_{S_i}(\sigma_i') - S_{S_i}(\sigma)$.
Then
\begin{align}
\label{eq:entropyaddapp}
\Delta S_S = \Delta S_{S_1} + \Delta S_{S_2}\,,
\end{align}
i.e.\ entropy differences are additive under composition of disjoint systems.
\end{lemma}

\begin{proof}
Let $\{p_i\}_{i=1}^N\subset\mathcal{P}$ be a sequence of processes that allows us to compute the entropy difference $\Delta S_{S_1}$, i.e. the processes $p_i\in\mathcal{P}_{S_1\vee R_i^{(1)}}$ in the sequence are concatenable and reversible, and under $p_i$ the heat $Q_{S_1}(p_i)$ is exchanged at temperature $T_i$. The total initial and final states of the sequence are $\sigma_1$ and $\sigma_1'$. 
We assume that such sequences exist between any two states of a system.
Let $\{q_j\}_{j=1}^M$ be an analogous sequence for system $S_2$ and the states $\sigma_2$ and $\sigma_2'$, where the temperature of the heat flows are denoted $\tilde T_j$. 
By the conditions on $\{p_i\}_i$ and $\{q_j\}_j\subset\mathcal{P}$ it holds
\begin{align}
\Delta S_{S_1} = \sum_i \frac{Q_{S_1}(p_i)}{T_i}
\quad \text{and} \quad
\Delta S_{S_2} = \sum_j \frac{Q_{S_2}(p_j)}{\tilde T_j}\,.
\end{align}
W.l.o.g.\ we can assume that the sequences of processes $\{p_i\}_i$ and $\{q_j\}_j$ act on disjoint systems $\tilde S_1$ and $\tilde S_2$ which include $S_1$ and $S_2$, respectively. This can be done due to Postulate~\ref{post:copies} on the existence of copies of systems. 
If a reservoirs was acted on by both sequences, replace it in one of them with a copy. 
We now construct $p\in\mathcal{P}_{\tilde S_1 \vee\tilde S_2}$ as the sequence of processes consisting of the joined sequences $\{p_i\}_i$ and $\{q_j\}_j$,
\begin{align}
p = ((p_N\vee\mathrm{id}_{S_2})\circ \cdots \circ (p_1\vee\mathrm{id}_{S_2})) \circ ((q_M\vee\mathrm{id}_{S_1})\circ \cdots \circ (q_1\vee\mathrm{id}_{S_1})) \,,
\end{align}
where $\mathrm{id}_{S_i}$ is a fitting identity on $S_i$.
This process is well-defined.
The total initial and final states on $S=S_1\vee S_2$ of $p$ are $\sigma = \sigma_1\vee\sigma_2$ and $\sigma'=\sigma_1'\vee\sigma_2'$.
The sequence defining $p$ fulfils all requirements necessary for it to be used to determine the entropy difference of its initial and final states on $S$ (reversibility, heat flows at well-defined temperature) and thus we find
\begin{align}
\label{eq:entropyaddproofapp}
\begin{split}
\Delta S_S 
&= S_{S_1\vee S_2}(\sigma') - S_{S_1\vee S_2}(\sigma) \\
&= \left( \sum_i \frac{Q_{S_1\vee S_2}(p_i\vee \mathrm{id}_{S_2})}{T_i} \right) 
+ \left( \sum_j \frac{Q_{S_1\vee S_2}(q_j\vee\mathrm{id}_{S_1})}{\tilde T_j} \right) \\
&= \left( \sum_i \frac{Q_{S_1}(p_i)}{T_i} \right) 
+ \left( \sum_j \frac{Q_{S_2}(q_j)}{\tilde T_j} \right)\\
&= \Delta S_{S_1} + \Delta S_{S_2}\,.
\end{split}
\end{align}
In the second to last equality we used that $Q_{S_1\vee S_2}(p_i\vee\mathrm{id}_{S_2}) = Q_{S_1}(p_i)$ since $\mathrm{id}_{S_2}$ is an identity on $S_{2}$, and likewise for exchanged roles of $S_1$ and $S_2$.
Due to Clausius' Theorem we know that $\Delta S_S=\Delta S_{S_1} + \Delta S_{S_2}$ also holds for any other valid way of computing $\Delta S_S$.
Therefore, entropy differences must be additive.
\end{proof}

By choosing the reference entropies of the systems accordingly, one can even achieve that not just entropy differences but entropy as a state variable is additive.\\

We conclude this section with a discussion of reversible processes with a net heat flow of zero.
Consider a reversible process $p\in\mathcal{P}_{S\vee S'}$ with $Q_S(p)=0$. 
If $S'\in\mathcal{R}$ is a reservoir and $W_{S'}(p)=0$, as is the case for the processes in Clausius' Theorem (Theorem~\ref{thm:clausius}), then $p$ is essentially a work process. 
In this case $S'$ is a catalytic system and according to Postulate~\ref{post:freedom} on the freedom of description there exists a work process on $S$ that does exactly the same, $\tilde p\in\mathcal{P}_S$, and acts on $S$ alone. 
For a work processes $\tilde p$ on $S$, we truly have $\tfrac{Q_S(\tilde p)}{T}=0$. 

If on the other hand $S'\in\mathcal{S}$ is an arbitrary thermodynamic system this is not true, as the following example shows. 

\begin{example}[Reversible net heat flow of zero]
\label{ex:zeroheat}
Consider an ideal gad $S$ undergoing a process $p=p_2\circ p'\circ p_1$, with two reservoirs $R_1$ and $R_2$, where the temperatures of the reservoirs are different, say $T_1 > T_2$.
Let $p_1\in\mathcal{P}_{S\vee R_1}$ be an isothermal compression of the gas, in which the heat $Q_S(p_1)>0$ is transferred from $S$ to $R_1$. This heat flows at temperature $T_1$. 
Next, the work process $p'\in\mathcal{P}_S$ is a reversible expansion of the gas, such that the final temperature of the gas is $T_2<T_1$. 
Now, apply $p_2\in\mathcal{P}_{S\vee R_2}$, which is an isothermal expansion at temperature $T_2$, which is such that the heat $Q_S(p_2) =-Q_S(p_1)<0$ flows from $S$ to $R_2$. 

In total, the heat $Q_S(p)=Q_S(p_1)+Q_S(p_2) = 0$ flows into $S$, i.e.\ we have a net heat flows of zero between $S$ and $R_1\vee R_2$. 
Nevertheless, according to Definition~\ref{def:entropy}, the entropy change under $p$ is 
\begin{align}
\tfrac{Q_S(p_1)}{T_1} + \tfrac{Q_S(p_2)}{T_2} < 0\,.
\end{align}
\end{example}

This example allows us to observe that it would not be a clever idea to define ``adiabatic processes'' on $S$ as those thermodynamic processes with $Q_S=0$.
Naively, having the traditional definition of an adiabatic process in mind, this sounds plausible.
However, as the example shows, during such processes the entropy of $S$ can decrease, i.e., such processes fail to fulfil Theorem~\ref{thm:entropythm} (Entropy Theorem), while adiabatic processes in the traditional sense are usually thought of those which fulfil the Entropy Theorem.\\

The above example also shows why the definition of entropy, Definition~\ref{def:entropy}, asks the sequence of processes $\{p_i\}_i$ connecting two states on $S$ to be such that $p_i\in\mathcal{P}_{S\vee R_i}$ and $W_{R_i}(p_i)=0$, where $R_i$ is a heat reservoir. 
Only in this setting a net heat flow of zero during $p_i$ implies that $\tfrac{Q}{T}=0$ is actually true. 
Beautifully enough, this definition still allows the simple extension to arbitrary heat flows between arbitrary systems to determine the entropy difference between two states, as long as these heat flows have a well-defined temperature. 
This follows from Definition~\ref{def:heatatt} and the considerations discussed in Section~\ref{sec:theatflow}.

\newpage
\addcontentsline{toc}{part}{References}
\bibliographystyle{unsrt}
\bibliography{references}

\begin{thebibliography}{10}

\bibitem{Clausius50}
R.~Clausius.
\newblock {\"U}ber die bewegende {K}raft der {W}\"arme und die {G}esetze,
  welche sich daraus fü\"ur die {W}\"armelehre selbst ableiten lassen.
\newblock {\em Annalen der Physik}, 79:368--397, 500--524, 1850.

\bibitem{Rankine50}
W.G. Rankine.
\newblock {\em {\"U}ber die mechanische {T}heorie der {W}\"arme}.
\newblock Wiley, 1850.

\bibitem{Kelvin51}
W.~Thomson.
\newblock On the {D}ynamical {T}heory of {H}eat, with numerical results deduced
  from {M}r. {J}oule's equivalent of a {T}hermal {U}nit, and {M}r. {R}egnault's
  {O}bservations on {S}team.
\newblock {\em Transactions of the Royal Society of Edinburgh}, XX:261--268,
  289--298, 1851.

\bibitem{Clausius54}
R.~Clausius.
\newblock {U}eber eine ver\"anderte {F}orm des zweiten {H}auptsatzes der
  mechanischen {W}\"armetheorie.
\newblock {\em Annalen der Physikun Chemie}, 12:31, 1854.

\bibitem{Maxwell71}
J.C. Maxwell.
\newblock {\em The Theory of Heat}.
\newblock Longmans, Green, and Co., 1871.

\bibitem{Bardeen73}
J.M. Bardeen, B.~Carter, and S.W. Hawking.
\newblock The {F}our {L}aws of {B}lack {H}ole {M}echanics.
\newblock {\em Commun. math. Phys.}, 31:161--170, 1973.

\bibitem{Carnot24}
S.~Carnot.
\newblock {\em R\'eflexions sur la puissance motrice du feu et sur les machines
  propres \`a d\'evelopper cette puissance}.
\newblock Bechelier, 1824.

\bibitem{Planck97}
M.~Planck.
\newblock {\em Vorlesungen \"uber Thermodynamik}.
\newblock Veit \& Comp., 1897.

\bibitem{Caratheodory09}
C.~Carath\'eodory.
\newblock {U}ntersuchung \"uber die {G}rundlagen der {T}hermodynamik.
\newblock {\em Math. Annalen}, 67:355--386, 1909.

\bibitem{Fowler39}
R.~Fowler and E.A. Guggenheim.
\newblock {\em Statistical Thermodynamics: A version of Statistical Mechanics
  for Students of Physics and Chemistry}.
\newblock Cambridge University Press, 1939.

\bibitem{Planck14}
M.~Planck.
\newblock {\em The Theory of Heat Radiation}.
\newblock P. Blakiston's Son \& Co., 1914.

\bibitem{Fermi56}
E.~Fermi.
\newblock {\em Thermodynamics}.
\newblock Dover, 1956.

\bibitem{Feynman63}
R.P. Feynman, B.~Leighton, and L.~Sands.
\newblock {\em The Feynman Lectures on Physics}.
\newblock Reading, Mass: Addison-Wesley Pub. Co, 1963.

\bibitem{Giles64}
R.~Giles.
\newblock {\em Mathematical Foundations of Thermodynamics}.
\newblock Pergamon, 1964.

\bibitem{Buchdahl66}
H.A. Buchdahl.
\newblock {\em The Concepts of Classical Thermodynamics}.
\newblock Cambridge University Press, 1966.

\bibitem{Zemansky68}
M.W. Zemansky.
\newblock {\em Heat and Thermodynamics}.
\newblock McGraw-Hill Book Company, 1968.

\bibitem{Pauli73}
W.~Pauli.
\newblock {\em Thermodynamics and the {K}inetic {T}heory of {G}ases}.
\newblock MIT Press, Cambridge, 1973.

\bibitem{Adkins83}
C.J. Adkins.
\newblock {\em Equilibrium {T}hermodynamics}.
\newblock Cambridge University Press, 1983.

\bibitem{Callen85}
H.~B. Callen.
\newblock {\em Thermodynamics and an {I}ntroduction to {T}hermostatistics}.
\newblock Wiley, 1985.

\bibitem{Huang87}
K.~Huang.
\newblock {\em Statistical {M}echanics}.
\newblock Wiley, 1987.

\bibitem{Neumaier07}
A.~Neumaier.
\newblock On the foundations of thermodynamics, 2007.
\newblock arxiv:0705.3790v1.

\bibitem{Thess11}
A.~Thess.
\newblock {\em The {E}ntropy {P}rinciple: {T}hermodynamics for the
  {U}nsatisfied}.
\newblock Springer-Verlag, Berlin, 2011.

\bibitem{Hulse18}
A.~Hulse, B.~Schumacher, and M.D. Westmoreland.
\newblock Axiomatic information thermodynamics, 2018.
\newblock arxiv:1801.05015.

\bibitem{Graf05}
G.M. Graf.
\newblock Theorie der {W}\"arme.
\newblock Lecture Notes, 2005.

\bibitem{Salem06}
W.K.~Abou Salem and J.~Fr\"ohlich.
\newblock Status of the {F}undamental {L}aws of {T}hermodynamics, 2006.
\newblock arxiv:math-ph/0604067v3.

\bibitem{LY99}
E.H. Lieb and J.~Yngvason.
\newblock The physics and mathematics of the second law of thermodynamics.
\newblock {\em Physics Reports}, pages 1--96, 1999.

\bibitem{LY02}
E.H. Lieb and J.~Yngvason.
\newblock The {M}athematical {S}tructure of the {S}econd {L}aw of
  {T}hermodynamics.
\newblock {\em Current Developments in Mathematics, 2001}, pages 89--130, 2002.

\bibitem{Talkner07}
P.~Talkner, E.~Lutz, and P.~H\"anggi.
\newblock Fluctuation theorems: {W}ork is not an observable.
\newblock {\em Phys. Rev. E}, 75:050102, 2007.

\bibitem{Aberg14}
J.~\AA{}berg.
\newblock Catalytic {C}oherence.
\newblock {\em Phys. Rev. Lett.}, 113:150402, 2014.

\bibitem{Perarnau17}
M.~Perarnau-Llobet, E.~B\"aumer, K.~Hovhannisyan, M.~Huber, and A.~Ac\'in.
\newblock No-{G}o {T}heorem for the {C}haracterization of {W}ork {F}luctuations
  in {C}oherent {Q}uantum {S}ystems.
\newblock {\em Phys. Rev. Lett.}, 118:070601, 2017.

\bibitem{Kammerlander18}
P.~Kammerlander and R.~Renner.
\newblock The zeroth law of thermodynamics is redundant, 2018.
\newblock arxiv:1804.09726.

\bibitem{Perarnau15}
M.~Perarnau-Llobet, K.~Hovhannisyan, M.~Huber, P.~Skrzypczyk, N.~Brunner, and
  A.~Ac\'in.
\newblock {E}xtractable {W}ork from {C}orrelations.
\newblock {\em Phys. Rev. X}, 5:041011, 2015.

\bibitem{Krivachy17}
T.~Kriv\'achy.
\newblock When are two system isomorphic in thermodynamics?
\newblock Master's thesis, ETH Zurich, March 2017.

\bibitem{Rudin76}
W.~Rudin.
\newblock {\em Principals of Mathematical Analysis}.
\newblock McGraw-Hill, Inc., 1976.

\bibitem{script19}
R.~Renner.
\newblock Theorie der {W}\"arme.
\newblock Lecture Notes, 2019.

\bibitem{Turner61}
L.A. Turner.
\newblock Zeroth {L}aw of {T}hermodynamics.
\newblock {\em American Journal of Physics}, 29:71, 1961.

\bibitem{Turner63}
L.A. Turner.
\newblock Temperature and {C}arath\'eodory's {T}reatment of {T}hermodynamics.
\newblock {\em The Journal of Chemical Physics}, 38:1163, 1963.

\bibitem{Ehrlich81}
P.~Ehrlich.
\newblock The concept of temperature and its dependence on the laws of
  thermodynamics.
\newblock {\em American Journal of Physics}, 49:622, 1981.

\bibitem{Buchdahl86}
H.A. Buchdahl.
\newblock On the redundancy of the zeroth law of thermodynamics.
\newblock {\em J. Phys. A: Math. Gen}, 19:L561--L564, 1986.

\bibitem{Miller52}
A.R. Miller.
\newblock The {C}oncept of {T}emperature.
\newblock {\em American Journal of Physics}, 20:488, 1952.

\bibitem{Walter89}
J.~Walter.
\newblock On {H} {B}uchdahl's project of a thermodynamics without empirical
  temperature as a primitive concept.
\newblock {\em J. Phys. A: Math. Gen.}, 22:341--342, 1989.

\bibitem{Buchdahl89}
H.A. Buchdahl.
\newblock Reply to comment by {J} {W}alter on `{O}n the redundancy of the
  zeroth law of thermodynamics'.
\newblock {\em J. Phys A: Math. Gen.}, 22:343, 1989.

\bibitem{Turner05}
L.A. Turner.
\newblock Further {R}emarks on the {Z}eroth {L}aw.
\newblock {\em American Journal of Physics}, 30:804, 2005.

\bibitem{Helsdon82}
R.M. Helsdon.
\newblock The zeroth law of thermodynamics.
\newblock {\em Phys. Educ.}, 17:114, 1982.

\bibitem{Clayton82}
D.G. Clayton and R.M. Helsdon.
\newblock Zeroth law of thermodynamics.
\newblock {\em Phys. Educ.}, 17:251, 1982.

\bibitem{Home77}
D.~Home.
\newblock Concept of temperature without the zeroth law.
\newblock {\em American Journal of Physics}, 45:1203, 1977.

\end{thebibliography}

%

\end{document}